\documentclass[11pt,reqno]{amsart}

\usepackage[margin=1.0in]{geometry}
\usepackage{amsmath}
\usepackage{amsthm, amsfonts,amssymb,euscript}
\usepackage{mathtools}
\usepackage{enumitem}
\usepackage{bm}
\usepackage{amssymb}
\usepackage{graphicx}
\usepackage{caption}
\usepackage{subcaption}
\usepackage{amsfonts}
\usepackage{float}
\usepackage{color}
\usepackage{slashed}
\usepackage[affil-it]{authblk}
\usepackage{mathrsfs}
\usepackage{upgreek}
\usepackage{pifont}
\usepackage{bbm}
\usepackage[colorlinks=true]{hyperref}

\usepackage[usenames,dvipsnames,svgnames,table]{xcolor}
\hypersetup{linkcolor=BrickRed, urlcolor=green, citecolor=blue, linktoc=page}

\numberwithin{equation}{section}

\newtheorem*{proposition*}{Proposition}
\newtheorem*{theorem*}{Theorem}
\newtheorem*{conjecture*}{Conjecture}
\newtheorem*{claim*}{Claim}
\newtheorem*{lemma*}{Lemma}
\newtheorem*{corollary*}{Corollary}
\newtheorem{theorem}{Theorem}[section]
\newtheorem{proposition}[theorem]{Proposition}

\newtheorem{lemma}[theorem]{Lemma}
\newtheorem{corollary}[theorem]{Corollary}

\newtheorem*{definition*}{Definition}
\newtheorem{definition}{Definition}[section]
\newtheorem*{assumption*}{Assumption}
\newtheorem{assumption}{Assumption}
\newtheorem*{remark*}{Remark}
\newtheorem{remark}{Remark}[section]

\newtheorem{thmx}{Theorem}

\newcommand{\R}{\mathbb{R}}
\newcommand{\s}{\mathbb{S}}
\newcommand{\C}{\mathbb{C}}

\newcommand{\Z}{\mathbb{Z}}
\newcommand{\N}{\mathbb{N}}

\newcommand{\snabla}{\slashed{\nabla}}

\newcommand{\sD}{\slashed{\Delta}}

\newcommand{\Divs}{\slashed{\textnormal{div}}_{\mathbb{S}^2} }

\newcommand{\Lbar}{\underline{L}}

\newcommand{\hpsi}{\widehat{\psi}}
\newcommand{\hphi}{\widehat{\phi}}
\newcommand{\tphi}{\widetilde{\varphi}}
\newcommand{\hfol}{\mathbbm{h}}
\newcommand{\tomega}{\widetilde{\omega}}
\newcommand{\tZ}{\widetilde{Z}}
\newcommand{\tX}{\widetilde{X}}
\newcommand{\etalchar}[1]{$^{#1}$}

\allowdisplaybreaks

\usepackage[foot]{amsaddr}

\begin{document}

\title{Azimuthal instabilities on extremal Kerr}
\author{Dejan Gajic}
\address{\small Institut f\"ur Theoretische Physik, Universit\"at Leipzig, Br\"uderstrasse 16, 04103 Leipzig, Deutschland}
\email{dejan.gajic@uni-leipzig.de}
\date{}
\maketitle
\begin{abstract}
We prove the existence of instabilities for the geometric linear wave equation on extremal Kerr spacetime backgrounds, which describe stationary black holes rotating at their maximally allowed angular velocity. These instabilities can be associated to non-axisymmetric azimuthal modes and are stronger than the axisymmetric instabilities discovered by Aretakis in \cite{aretakis4}. The existence of non-axisymmetric \emph{instabilities} follows from a derivation of very precise \emph{stability} properties of solutions: we determine therefore the precise, global, leading-order, late-time behaviour of solutions supported on a bounded set of azimuthal modes via energy estimates in both physical and frequency space. In particular, we obtain sharp, uniform decay-in-time estimates and we determine the coefficients and rates of inverse-polynomial late-time tails everywhere in the exterior of extremal Kerr black holes. We also demonstrate how non-axisymmetric instabilities leave an imprint in the radiation on future null infinity via the coefficients appearing in front of slowly decaying and oscillating late-time tails.
\end{abstract}

\tableofcontents

\section{Introduction}
Extremal black holes are critical members of families of stationary black holes. They sit exactly at a phase transition from black hole to ``naked singularity''.\footnote{In this context, naked singularities are spacetimes that, from an evolutionary perspective, have geodesically incomplete future and past null infinities and arise from \emph{singular} initial data. These are also called ``super-extremal'' spacetimes.} Criticality can be understood geometrically via the vanishing of the surface gravity associated to Killing vector fields that generate black hole event horizons. 

In the context of the vacuum Einstein equations, 
\begin{equation}
\label{eq:ee}
\textnormal{Ric}[g]=0,
\end{equation}
such black holes are called \emph{extremal Kerr black holes} \cite{kerr63} and criticality can be interpreted as the black holes rotating at their maximally allowed speed. 

In this paper, we perform a mathematically rigorous investigation of the dynamical properties of extremal black holes in the context of the geometric, linear wave equation:
\begin{equation}
\label{eq:waveeqintro}
\square_{g_{M,a}}\phi=0,
\end{equation}
on extremal Kerr spacetime backgrounds $(\mathcal{M},g_{M,a})$ with mass $M$ and specific angular momentum\\ $|a|=M$. A mathematical study of \eqref{eq:waveeqintro} constitutes a first step towards addressing the dynamics of black holes that arise in the evolution of small perturbations of extremal Kerr initial data for \eqref{eq:ee}.

Even the most elementary dynamical question, the question of \emph{stability}, still remains widely open for small perturbations of extremal Kerr initial data. This may be contrasted with the problem of stability of \emph{sub}-extremal Kerr, corresponding to $|a|<M$ and describing black holes that rotate at speeds strictly below the maximal speed limit, for which tremendous progress has recently been achieved; see the discussion and references in \S \ref{sec:prevwork}.

Already at the level of the linear problem \eqref{eq:waveeqintro}, the criticality of extremal Kerr leads to a significant departure from the sub-extremal case. In \cite{aretakis4}, Aretakis showed that extremal Kerr black holes feature an \emph{axisymmetric, asymptotic instability} along their future event horizon $\mathcal{H}^+$: higher-order transversal derivatives of axisymmetric $\phi$ generically blow up polynomially as time\\ $\tau\to \infty$. The same type of instability is also present for linear wave equations on the simpler, non-rotating but electromagnetically charged (``extremal Reissner--Nordstr\"om'') black holes \cite{aretakis1, aretakis2}.

An asymptotic instability at the \emph{linear} level may lead to a genuine instability at the \emph{nonlinear level}, i.e.\ the spacetimes arising from small perturbations of Kerr data could deviate significantly from extremal Kerr at late times. An asymptotic instability is, however, also compatible with \emph{stability} at the nonlinear level. That is to say, while derivatives of the metric might blow up as $\tau\to \infty$, the metric itself could stay close to the unperturbed metric. Indeed, in the context of extremal Reissner--Nordstr\"om black holes, an analysis of a nonlinear analogue of \eqref{eq:waveeqintro} with a similar nonlinear structure to \eqref{eq:ee}  has shown that the asymptotic instability remains asymptotic in the nonlinear setting and the solutions decay to zero \cite{aagnonlin}. A similar result is expected to hold in the setting of \emph{axisymmetric} solutions to analogous nonlinear wave equations on extremal Kerr.

In this paper, \textbf{we go \emph{beyond axisymmetry} and derive the existence of \emph{stronger} asymptotic instabilities for non-axisymmetric solutions to \eqref{eq:waveeqintro}}. These \emph{azimuthal instabilities} are derived via a sharp analysis of the \emph{stability} properties of $\phi$ in the form of the precise leading-order late-time behaviour (``late-time tails'') of azimuthal mode solutions $\phi_m$ to \eqref{eq:waveeqintro} along $\mathcal{H}^+$, with $m$ an azimuthal number.\footnote{These are solutions with an azimuthal angular dependence $e^{im\varphi}$, where $m\in \Z$.} These tails feature \emph{slower decay in time} when compared to the axisymmetric or sub-extremal settings.

Away from the event horizon and along future null infinity $\mathcal{I}^+$, we moreover derive the existence of slowly-decaying, oscillating late-time tails, which decay like $u^{-1}$ and oscillate to leading-order with a frequency $m \upomega_+$,
where $\upomega_+$ is the angular velocity of the black hole and $u$ a standard Bondi time coordinate. This late-time behaviour is markedly different from the faster-decaying, non-oscillating sub-extremal late-time tails (``Price's law''), as derived in \cite{hintzprice, aagkerr}.

While the nonlinear analysis of \cite{aagnonlin} suggests stability of extremal Kerr within the restricted class of axisymmetric perturbations\footnote{Since the initial data in this setting is necessarily geodesically incomplete and extremal Kerr data has no trapped surfaces, ``stability'' in this context would correspond to one of the following scenarios:  the convergence of the perturbed metric in a suitable choice of coordinates to the black hole exterior of 1) a nearby sub-extremal Kerr spacetime, 2) a nearby extremal Kerr spacetime, with an asymptotic instability at the level of metric derivatives  along the event horizon (``horizon hair''), or 3) the failure to collapse to a black hole: the perturbed spacetime is close to the causal past of a finite subset of $\mathcal{I}^+$ in extremal Kerr. See also \cite[\S IV.2]{dhrt21} for a related discussion on the formulations of nonlinear stability of extremal Reissner--Nordstr\"om. }, the slow decay and the stronger nature of the instabilities derived in the present paper goes beyond the methods of \cite{aagnonlin} and tantalizes the possibility of a \emph{departure from the Kerr family} in the end-point of the nonlinear evolution of small perturbations of extremal Kerr initial data.

The quantitative analysis of late-time tails along $\mathcal{I}^+$ in the present paper also has observational implications, as it also implies the existence of characteristic slowly decaying, oscillating late-time tails at \emph{late, but finite} time intervals in sub-extremal Kerr black holes that are sufficiently close to extremality (``near-extremal Kerr black holes'').  One may think of such \emph{intermediate late time tails} as carrying an imprint of the azimuthal instability that would be occurring at the event horizon of an extremal black hole.

\subsection{Main theorems}
We now give the first, rough formulations of the main two theorems that are proved in the present paper. More precise versions are stated in \S \ref{sec:precstatthm}, after an introduction of the relevant geometric notation in \S\S \ref{sec:prelim1}--\ref{sec:prelim2}.
\begin{theorem}[The existence of azimuthal instabilities on extremal Kerr; rough version]
\label{thm:azimuthal}
Non-axisymmetric solutions to \eqref{eq:waveeqintro}, arising from suitably regular and decaying initial data feature at least one of the following asymptotic instabilities:  generically,\footnote{The initial data for which \underline{none} of these instabilities occur can be thought of as having infinite codimension.}
\begin{enumerate}[label=\emph{\textbf{(\alph*)}}]
\item first-order transversal derivatives of $\phi$ grow at least like $\tau^{\frac{1}{2}}$ in time along a sequence of times $\tau_n\to \infty$ along the event horizon 
and the $H^1_{\rm space}$-norm (the non-degenerate energy) is non-decaying in time, or
\item $ \phi \notin L^2_{\rm time}L^2_{\rm space}$, or 
\item $\partial \phi \notin  L^2_{\rm time}L^2_{\rm space; loc}$ away from the event horizon.
\end{enumerate}
\end{theorem}

We refer to the instability \textbf{(a)} as an \emph{azimuthal instability} as it originates from azimuthal modes $\phi_m$ with $m\neq 0$, which are solutions with an azimuthal angular dependence $e^{im\varphi}$. In sub-extremal Kerr, or for axisymmetric solutions on extremal Kerr, all the derivatives $\partial \phi$ can be uniformly bounded in time and the spacetime integrals of $|\phi|^2$ and $|\partial \phi|^2$ away from the horizon can be bounded via integrated energy (Morawetz) estimates. Hence, none of the instabilities \textbf{(a)}, \textbf{(b)} or \textbf{(c)} occur in these settings.

The key ingredient for proving Theorem \ref{thm:azimuthal} is the following precise analysis of the late-time behaviour of azimuthal modes $\phi_m$:

\begin{theorem}[Late-time tails and stability for azimuthal modes on extremal Kerr; rough version]
\label{thm:tails}
Let $|m|\geq 2$. Under a \emph{qualitative} global integrability assumption for $\phi_m$, the $H^1_{\rm space}$-norm of $\phi_m$ is bounded and $r \phi_m$ decays \emph{quantitatively} in time $\tau$ with the sharp rate $\tau^{-\frac{1}{2}}$ along the future event horizon $\mathcal{H}^+$ and $\tau^{-1}$ away from the event horizon. 

More precisely, with respect to coordinates $(v,\theta,\tphi)$ along $\mathcal{H}^+$ that are adapted to the Killing vector field generating the null hypersurface $\mathcal{H}^+$:
\begin{equation*}
\phi_m|_{\mathcal{H}^+}(v,\theta,\tphi)\sim v^{-\frac{1}{2}}\sum_{\ell\::\: \alpha_{m\ell}<0} \mathfrak{h}_{m\ell} \left(a_{m\ell}\cdot e^{-i\left(\sqrt{-\alpha_{m\ell}}+m\right)\log v}+b_{m\ell}\cdot e^{i\left(\sqrt{-\alpha_{m\ell}}-m\right)\log v}\right)e^{im\tphi}S_{m\ell}(\theta).
\end{equation*}
With respect to Bondi coordinates $(u,\theta,\varphi)$ along future null infinity,
\begin{equation*}
r\phi_m|_{\mathcal{I}^+}(u,\theta,\varphi)\sim u^{-1}e^{-im\upomega_+ u}\sum_{\ell\::\: \alpha_{m\ell}<0} \mathfrak{h}_{m\ell} \left(c_{m\ell}\cdot e^{-i \sqrt{-\alpha_{m\ell}}\log u}+d_{m\ell}\cdot  e^{i \sqrt{-\alpha_{m\ell}}\log u}\right)e^{im\varphi}S_{m\ell}(\theta),
\end{equation*}
where 
\begin{itemize}
\item $\alpha_{m\ell}=\Lambda_{+,m\ell}-2m^2+\frac{1}{4}$,
\item $\Lambda_{+, m\ell}$ is the eigenvalue of a natural angular operator defined on the 2-sphere $\mathcal{H}^+\cap \{v=0\}$, with eigenfunctions $e^{im\tphi}S_{m\ell}(\theta)$,
\item $\upomega_+$ is the angular velocity of the black hole,
\item  $\mathfrak{h}_{m\ell}$ is an integral of the initial data, projected to $e^{im\tphi}S_{m\ell}(\theta)$, which is generically non-zero,\\
\item  $a_{m\ell},b_{m\ell},c_{m\ell},d_{m \ell}\in \C$ are constants that are independent of the initial data.
\end{itemize}
\end{theorem}

We now give several remarks regarding Theorems \ref{thm:azimuthal} and \ref{thm:tails}.

\begin{remark}
The instabilities \textbf{(b)} and \textbf{(c)} in Theorem \ref{thm:azimuthal} can be excluded with the assumption of appropriate, global, qualitative integrability assumptions for $\phi$, as in the statement of Theorem \ref{thm:tails}. These assumptions can be be verified if one extends the analysis in the present paper from exactly extremal Kerr backgrounds to near-extremal Kerr backgrounds; see also \S \ref{sec:remainingq}. This will be explored in future work.
\end{remark}

\begin{remark}
By combining the azimuthal instability \textbf{(a)} from Theorem \ref{thm:azimuthal} with the decay results in Theorem \ref{thm:tails}, we can infer that the energy of azimuthal modes $\phi_m$ with $|m|\geq 2$ (as measured by $H^1$-norms along a foliation of spacelike constant-$\tau$ hypersurfaces) \textbf{concentrates at the horizon} as $\tau\to\infty$, since it is non-decaying by Theorem \ref{thm:azimuthal}, but decays when restricted away from the horizon by Theorem \ref{thm:tails}.
\end{remark}

\begin{remark}
\label{rm:diffwithaxi}
The azimuthal instability \textbf{(a)} differs from the axisymmetric instability following from \cite{aretakis4,aretakis3,aretakis2012} (see Theorem \ref{thm:are} in \S \ref{sec:prevwork}) and the analogous extremal Reissner--Nordstr\"om instabilities from \cite{aretakis1,aretakis2} in the following ways: 
\begin{enumerate}[label=\arabic*)]
\item Axisymmetric solutions have bounded, but non-decaying first-order derivatives along $\mathcal{H}^+$, as shown in \cite{aretakis3} and, in view of the related analysis in extremal Reissner--Nordstr\"om in \cite{paper4}, non-degenerate energies, or $H^1_{\rm space}$ norms, are expected to decay in time.
\item The axisymmetric instability mechanism relies on the existence of axisymmetric conservation laws along $\mathcal{H}^+$. In that case, the coefficient $H_0$ appearing in front of the terms that blow up along $\mathcal{H}^+$ is the value of the corresponding conserved charge (the ``horizon hair'') and is therefore defined at the horizon. These conserved charges along $\mathcal{H}^+$ form a close analogue of the Newman--Penrose constants \cite{np2} along $\mathcal{I}^+$. In contrast, the azimuthal instability in Theorem \ref{thm:azimuthal} does not follow from conservation laws, but instead follows from the precise knowledge of the late-time tails of azimuthal modes along $\mathcal{H}^+$. 
\item In the non-axisymmetric setting, the role of $H_0$ is taken on by the constants $\mathfrak{h}_{m \ell}$ appearing in Theorem \ref{thm:tails}, which are not defined at the horizon but are integrals along the entire initial data hypersurface. As a result, the azimuthal instability does \underline{not} become weaker by simply considering initial data supported away from $\mathcal{H}^+$, as is the case in axisymmetry; see \cite{aretakis2012}. Instead, it becomes weaker for data which has vanishing constants $\mathfrak{h}_{m \ell}=0$.
\end{enumerate}
\end{remark}

\begin{remark}
The constants $\mathfrak{h}_{m \ell}$, which are responsible for the azimuthal instability (see Remark \ref{rm:diffwithaxi}), can be observed in the late-time behaviour of radiation fields $r\phi_m|_{\mathcal{I}^+}$, which also carries information about the rotation of the black hole via the oscillating factor $\frac{e^{-im\upomega_+ u}(1+\log u\,O(u^{-1}))}{u}$. Since this late-time behaviour differs significantly from the more rapid, non-oscillating inverse polynomial decay in the sub-extremal setting \cite{hintzprice, aagkerr}, it provides an \emph{observational signature} of extremality and the existence of azimuthal instabilities.

For sub-extremal black holes that are close to extremality, the azimuthal instability will instead become a ``transient instability'' by Cauchy stability, i.e.\ quantities that blow up in extremal Kerr will instead become very large, but they will eventually decay asymptotically in time. This transient instability will leave an \emph{imprint} along $\mathcal{I}^+$ in the form of a transient, oscillating late-time tail at late, but finite time intervals.

This observational signature may be compared with the setting of extremal Reissner--Nordstr\"om, where (near)-extremality \underline{cannot} immediately be inferred from the late-time decay rate along a single far-away causal curve, but requires additionally the consideration of integrals along the spheres foliating $\mathcal{I}^+$ \cite{extremal-prl}. See also \cite{burko19, burko21} for additional numerical analysis on the observational signature in this setting.
\end{remark}
\subsection{Previous work}
\label{sec:prevwork}
In this section, we give an overview of previous works relevant to Theorems \ref{thm:azimuthal} and \ref{thm:tails}.\\
\\
\paragraph{\textbf{Extremal black holes}}
A mathematical study of the dynamics of axisymmetric solutions to \eqref{eq:waveeqintro} was pioneered in the works of Aretakis  \cite{aretakis3, aretakis4, aretakis2012}, where the following theorem was proved:
\begin{thmx}[Axisymmetric instability on extremal Kerr; \cite{aretakis3, aretakis4, aretakis2012}]
\label{thm:are}
Consider suitably regular and decaying initial data for \eqref{eq:waveeqintro}. Then
\begin{enumerate}[label=\emph{(\roman*)}]
\item Generically, second-order derivatives of $\phi$ grow at least with the rate $\tau$ as $\tau\to \infty$ along the future event horizon $\mathcal{H}^+$.
\item If the initial data are supported away from the future event horizon $\mathcal{H}^+$, then the growth in \emph{(i)} occurs instead at the level of third-order derivatives.
\item Axisymmetric solutions $\phi_0$ decay at least as fast as $\tau^{-\frac{1}{2}}$ along $\mathcal{H}^+$ and future null infinity $\mathcal{I}^+$, and at least as fast as $\tau^{-1}$ in bounded regions of space away from $\mathcal{H}^+$, and their first-order derivatives are bounded in time.
\end{enumerate}
\end{thmx}
Theorem \ref{thm:are} demonstrates the existence of \emph{axisymmetric instabilities} for higher-order transversal derivatives along $\mathcal{H}^+$. The underlying mechanism for instabilities is based on the existence of conservation laws along $\mathcal{H}^+$, together with the assumption of a \emph{qualitative} decay statement for $\phi$ and its tangential derivatives along $\mathcal{H}^+$. Under the global assumption of qualitative decay for $\phi$, the axisymmetric instabilities have been shown to hold in more general settings \cite{hm2012,murata2012}. See Remark \ref{rm:diffwithaxi} for a comparison of Theorem \ref{thm:are} with Theorem \ref{thm:azimuthal}.

In view of the existence of a bijection between the behaviour of $\phi_{m=0}$ along $\mathcal{H}^+$ and $r\phi_{m=0}$ in $\mathcal{I}^+$ (valid only in axisymmetry) \cite{couch, bizon2012}, the conservation law of Aretakis along $\mathcal{H}^+$ may be interpreted as a manifestation of the Newman--Penrose conservation law along $\mathcal{I}^+$ that is present in general asymptotically flat settings \cite{np2}. The relation between conservation laws and characteristic initial data gluing problems for wave equations was investigated further in \cite{aretakisglue} and has recently been analyzed in the context of the Einstein equations in \cite{acr21a,acr21b, acr21c, czro22}.

The first analysis of non-axisymmetric solutions to \eqref{eq:waveeqintro} was performed by Teixeira da Costa in \cite{costa20}, where it was shown that there are no exponentially growing mode solutions on extremal Kerr, i.e.\ there is mode stability on extremal Kerr.

A related problem is the study of the geometric wave equation on extremal Reissner--Nordstr\"om black holes, which are charged but non-rotating analogues of extremal Kerr black holes. Due to the lack of rotation, the analysis on these spacetimes is greatly simplified and it mirrors the analysis of axisymmetric solutions to \eqref{eq:waveeqintro}. In particular, these spacetimes also feature an instability mechanism based on conservation laws, and it was shown in \cite{aretakis1, aretakis2} that analogues of (i) and (ii) from Theorem \ref{thm:are} also hold in this setting. Via a gluing analysis along the event horizon, spherically symmetric black hole spacetimes have also been constructed that are isometric to extremal Reissner--Nordstr\"om in the black hole exterior for sufficiently late times, but arise from regular, geodesically complete initial data \cite{keru22}.

In the setting of extremal Reissner--Nordstr\"om, precise late-time decay properties and the existence of late-time tails were obtained in \cite{paper4}:
\begin{thmx}[Late-time tails on extremal Reissner--Nordstr\"om; \cite{paper4}]
\label{thm:ERNtails}
Consider suitably regular and decaying initial data for the geometric wave equation on extremal Reissner--Nordstr\"om spacetimes. Then the corresponding solutions $\phi$ display the following late-time behaviour:
\begin{align*}
\phi|_{\mathcal{H}^+}(v,\theta,\varphi)\sim&\: H_0v^{-1},\\
r\phi|_{\mathcal{I}^+}(u,\theta,\varphi)\sim&\:  (H_0+Q)u^{-2},
\end{align*}
with $H_0$ the value of the conserved quantity at $\mathcal{H}^+$ and $Q$ proportional to the time-inverted Newman--Penrose constant.
\end{thmx}

Theorem \ref{thm:ERNtails} demonstrates that despite the presence of asymptotic instabilities along $\mathcal{H}^+$, precise late-time decay results can be obtained. In particular, the late-time tails along $\mathcal{I}^+$ can be split into two parts: a part related to the conservation law and instability along $\mathcal{H}^+$, which comes with a constant $H_0$, and a part which is also present in the sub-extremal setting and comes with a constant $Q$ that may be interpreted as the Newman--Penrose conserved quantity at $\mathcal{I}^+$ associated to the \emph{integral in time} of the spherically symmetric part of $\phi$. Unlike $H_0$, the constant $Q$ can be expressed as an integral along the entire initial data hypersurface.

Due to the strong similarity of the analysis of axisymmetric solutions to \eqref{eq:waveeqintro} with general solutions to the geometric wave equation on extremal Reissner--Nordstr\"om, it is expected that the results of Theorem \ref{thm:ERNtails} can be extended to axisymmetric solutions to \eqref{eq:waveeqintro}. The rates appearing in Theorem \ref{thm:ERNtails} may be contrasted with the more intricate and slower decay in Theorem \ref{thm:tails} and the absence of the time-inverted Newman--Penrose constant $Q$ in the leading-order asymptotics.

Recently, the analogue of the instability in (i) of Theorem \ref{thm:are} was also shown to hold in the context of the linearized Einstein--Maxwell equations on extremal Reissner--N\"ordstrom \cite{ap22}. In the nonlinear setting, the techniques developed in the proof of Theorem \ref{thm:ERNtails} have been applied to show that the instability remains asymptotic in the context of a nonlinear problem that models the nonlinear structure present in the Einstein equations \cite{aagnonlin}.\footnote{The asymptotic instability can also lead to a finite-in-time instability for nonlinear equations that have a different nonlinear structure; see \cite{are13}.}

The implications of the behaviour in Theorem \ref{thm:ERNtails} to the singularity properties in the black hole interior of extremal Reissner--Nordstr\"om and axisymmetric solutions to \eqref{eq:waveeqintro} was studied in \cite{gajic, gajic2, dejanjon1}. In particular, it was shown that the Cauchy horizon in this setting is \emph{more stable} than in the sub-extremal setting and the regularity of solutions at the Cauchy horizon may depend rather delicately on the leading and next-to-leading order late-time tails along $\mathcal{H}^+$.\\
\\
\paragraph{\textbf{Sub-extremal black holes}}
The mathematical analysis of the geometric wave equation on sub-extremal black holes has a rich history. Uniform boundedness and decay estimates for solutions to the wave equation on the full sub-extremal range of Kerr parameters were obtained by Dafermos--Rodnianski--Shlapentokh-Rothman in \cite{part3}; see also the references therein for a comprehensive list of intermediate and related results. 

In \cite{part3}, the following key mechanisms were exploited to establish weak, integrated energy estimates or \emph{Morawetz estimates} in sub-extremal Kerr: 1) the red-shift effect (first discussed as a decay mechanism in \cite{redshift}), 2) the decoupling of \emph{null geodesic trapping} and \emph{superradiance} in Fourier space, and 3) quantitative mode stability, first established qualitatively in \cite{whiting} and then quantitatively in \cite{sr15}. 

Passing from integrated energy energy estimates to inverse-polynomial energy decay in time involves an application of the $r$-weighted estimates of Dafermos--Rodnianski \cite{newmethod}, which are not restricted to Kerr spacetimes, but hold on general asymptotically flat spacetimes, as shown in \cite{moschidis1}. The analysis in \cite{part3} has recently also been extended to the setting of the Teukolsky equations of linearized gravity \cite{shlcosta20,shlcosta23}.

The step from uniform upper bound decay estimates to precise leading-order decay estimates and late-time tails (analogous to Theorems \ref{thm:tails} and \ref{thm:ERNtails}) was first made in \cite{paper1,paper2} in the context of spherically symmetric, sub-extremal black hole spacetimes and then extended to sub-extremal Kerr in \cite{aagkerr}, where it was moreover explored how the late-time tails corresponding to different angular modes interact; see also \cite{hintzprice} for a different approach that involves precise resolvent estimates near the zero time-frequency.

A useful subset of the sub-extremal Kerr family to consider first are the slowly rotating Kerr spacetimes, which have a rotation parameter $a$ that satisfies $|a|\ll M$. These include the Schwarzschild spacetimes for which $a=0$. Due to the smallness of the parameter $a$, the issue of superradiance and mode stability, which are absent when $a=0$, are easier to handle. Significant advances have been made in this setting regarding decay and stability at the level of the nonlinear Einstein equations. The nonlinear stability of Schwarzschild was proved for polarized axisymmetric initial data in \cite{klainerman17}. Data without symmetry were considered in \cite{dhrt21} and the full, codimension 3 subset of initial data perturbations were identified for which stability holds. More recently, a proof of nonlinear stability of slowly rotating Kerr without symmetry restrictions was completed in a series of works \cite{ks21, ks22a, ksg22,ks22b, sh22}.\\
\\
\paragraph{\textbf{Inverse-square potentials}}
A setting that roughly models the situation ``between'' sub-extremal Kerr and extremal Kerr are wave equations with negative inverse-square potentials on Schwarzschild spacetimes. Here, future null infinity $\mathcal{I}^+$ plays the role of the event horizon $\mathcal{H}^+$ in extremal Kerr. The absence of conservation laws for non-axisymmetric solutions along $\mathcal{H}^+$ in extremal Kerr is reflected in the breakdown of the Newman--Penrose conservation law at $\mathcal{I}^+$. The proofs for deriving the precise late-time behaviour of solutions to the wave equation without potential (see \cite{aagkerr} and the earlier works \cite{paper1, paper2}) rely crucially on the existence of this conservation law, and they therefore break down. In \cite{gaj22a}, a new, more robust, method was developed for dealing with the absence of conservation laws and deriving late-time tails which depend on the coefficient in the leading-order term of the inverse square potential. See also upcoming work on establishing these tails in a different way for wave equations with inverse-square potentials and electromagnetically charged wave equations on Minkowski \cite{gvdm23} and see \cite{brm22} for recent work on late-time tails for wave equations with inverse-square potentials on Minkowski.\\
\\
\paragraph{\textbf{Numerics and heuristics}}
The mathematical analysis described in the above paragraphs was preceded by a large number of works with numerical and heuristic analysis in the physics literature. In particular, the pointwise blow-up of transversal derivations in the azimuthal instability of Theorem \ref{thm:azimuthal} was predicted from a heuristic analysis of the Teukolsky equation in frequency space near the frequency $\omega=m \upomega_+$ by Casals--Gralla--Zimmerman in \cite{zimmerman1}, who moreover predicted the decay rates of late-time tails along $\mathcal{H}^+$ consistent with Theorem \ref{thm:tails}. Growth along $\mathcal{H}^+$ with consistent rates was also observed numerically in \cite{hbb13}. See also the related heuristics in \cite{harveyeffective} on instabilities in the near-horizon spacetimes corresponding to extremal Kerr.

The late-time behaviour in Theorem \ref{thm:tails} was first predicted (in bounded regions in space, away from $\mathcal{H}^+$) by Andersson--Glampedakis \cite{mhighinsta} via a heuristic analysis of slowly decaying \emph{quasinormal modes} (``zero-damped modes''; see for example \cite{hod08} for a discussion) on near-extremal Kerr black holes with complex time-frequencies that are expected to accumulate to the real frequency $\omega=m\upomega_+$. While each quasinormal mode decays exponentially in time, their decay rate goes to zero in the extremal limit. The analysis in \cite{mhighinsta} suggested that the \emph{sum} over these modes should give rise to an inverse-polynomial tail with decay rate $t^{-1}$ and an overall  oscillating factor $e^{-i m \upomega_+ t}$. 

The analysis of \cite{mhighinsta} was later revisited, heuristically and numerically, in \cite{zeni13}. The numerics in \cite{zeni13} predicted moreover that the slowly-decaying late-time tail could constitute a dominant part of $\phi$ at relatively \emph{early} times, for initial data that are relevant when modelling black hole mergers. This stands in contrast with the expectation that the late-time tails in sub-extremal Kerr become dominant only at relatively \emph{late} times. This strengthens further the case for slowly decaying late-time tails in extremal Kerr serving as a promising observational signature in the gravitational radiation emitted by astrophysical systems settling down to near-extremal Kerr black holes.

Note that \emph{nonlinear} aspects of the dynamics of extremal Reissner--Nordstr\"om black holes were probed numerically in \cite{harvey2013}.

\subsection{The main difficulties compared to the axisymmetric and sub-extremal settings}
\label{sec:maindiff}
In this section, we give an overview of the key difficulties enountered when proving Theorems \ref{thm:azimuthal} and \ref{thm:tails}. We will revisit these difficulties and outline how they are resolved in the present paper in \S \ref{sec:sketchproof}.\\
\\
\paragraph{\textbf{Integrated energy estimates: absence of (degenerate) red-shift}}
A key ingredient towards establishing energy decay estimate on sub-extremal black holes is the red-shift estimate, which, schematically, gives the following integrated energy estimate for solutions $\phi$ to the wave equation:
\begin{multline*}
\kappa \int_{0}^{T}\left[\int_{\Sigma_{\tau}\cap \{r_+\leq r\leq r_++\epsilon\}}|\phi|^2+|\partial \phi|^2 d\mu_{\Sigma_{\tau}}\right]\,d\tau\lesssim \int_{\Sigma_{0}\cap \{r_+\leq r\leq r_++2\epsilon\}}|\phi|^2+|\partial \phi|^2 d\mu_{\Sigma_{\tau}}\\
+\int_{0}^{T}\left[\int_{\Sigma_{\tau}\cap \{r_++\epsilon\leq r\leq r_++2\epsilon\}}|\phi|^2+|\partial \phi|^2 d\mu_{\Sigma_{\tau}}\right]\,d\tau.
\end{multline*}
Here $r_+$ is the radius of the event horizon $\mathcal{H}^+$, $\epsilon\ll r_+$ and $d\mu_{\Sigma_{\tau}}$ is the induced volume form on the spacelike hypersurface or time slice $\Sigma_{\tau}$. The constant $\kappa>0$ denotes the surface gravity associated to the future event horizon $\mathcal{H}^+$, which appears as a proportionality constant in the following equation:
\begin{equation*}
\nabla_KK=\kappa K,
\end{equation*}
where $K$ is a Killing vector field which is a null generator of the event horizon. The above integrated estimate allows one to extend integrated energy estimates valid in $\{r>r_++\epsilon\}$ up to $\{r\geq r_+\}$.

In extremal Kerr, $\kappa=0$, so the red-shift estimate breaks down. Nevertheless, degenerate versions of the red-shift estimate remain valid for axisymmetric solutions to \eqref{eq:waveeqintro}. For example,
\begin{multline*}
\kappa \int_{0}^{T}\left[\int_{\Sigma_{\tau}\cap \{r_+\leq r\leq r_++\epsilon\}}|\phi|^2+(r-r_+)|\partial \phi|^2 d\mu_{\Sigma_{\tau}}\right]\,d\tau\lesssim \int_{\Sigma_{0}\cap \{r_+\leq r\leq r_++2\epsilon\}}|\phi|^2+|\partial \phi|^2 d\mu_{\Sigma_{\tau}}\\
+\int_{0}^{T}\left[\int_{\Sigma_{\tau}\cap \{r_++\epsilon\leq r\leq r_++2\epsilon\}}|\phi|^2+|\partial \phi|^2 d\mu_{\Sigma_{\tau}}\right]\,d\tau.
\end{multline*}
For non-axisymmetric solutions, \underline{even the degenerate versions of the red-shift estimate break down}. Furthermore, while in the case of axisymmetric solutions $\phi_0$, the energy current corresponding to $K$ is non-negative near $\mathcal{H}^+$, i.e.\
\begin{equation*}
\mathbb{T}[\phi_0](K,\mathbf{n}_{\Sigma_{\tau}})\geq 0,
\end{equation*}
with $\mathbb{T}$ the energy momentum tensor and $\mathbf{n}_{\Sigma_{\tau}}$ the normal to the spacelike hypersurfaces $\Sigma_{\tau}$, the non-negativity property breaks down for non-axisymmetric solutions, since $K$ is not causal in any open neighbourhood of $\mathcal{H}^+$.\\
\\
\paragraph{\textbf{Integrated energy estimates: superradiance}}
Let $T$ denote the Killing vector field that is timelike and normalized at $r=\infty$. In all rotating ($a\neq 0$) spacetimes, $T$ fails to be causal in the \emph{ergoregion}
\begin{equation*}
r<M\left(1+\sqrt{1+a^2\cos \theta}\right),
\end{equation*}
with $r$ and $\theta$ Boyer--Lindquist coordinates. The existence of an ergoregion leads to \emph{superradiance}, which is an energy amplification effect that may be interpreted as ``energy extracted from the black hole''. At the level of the Fourier transform of azimuthal modes $\phi_m$ along integral curves of $T$, superradiance affects the following frequency range:
\begin{equation*}
0< m\omega< m^2\upomega_+.
\end{equation*}
Note in particular that for axisymmetric $\phi$, $m=0$, so the above frequency range is empty and superradiance is not present. For bounded $m$, we can divide solutions into those with \emph{high} angular frequencies $\Lambda$ and those with bounded \emph{angular} frequencies $\Lambda\lesssim m^2$ (see \S \ref{sec:freqspacean} for a precise definition of the angular frequency $\Lambda$). In both cases, the analysis from the sub-extremal setting \emph{breaks down} and new techniques are required to obtain the desired estimates in Fourier space.
\\
\\
\paragraph{\textbf{Instabilities: absence of conservation laws}}
The instability mechanism underlying the proof of (i) of Theorem \ref{thm:are} relies on the existence of conservation laws along $\mathcal{H}^+$ for axisymmetric $\phi$ in extremal Kerr. For this reason, asymptotic instabilities for transversal derivatives $\partial^2 \phi_0$ along $\mathcal{H}^+$ follow immediately from the conservation of an appropriate linear combination of $\phi$ and $\partial \phi$, together with a \emph{qualitative} decay statement for $\phi$ and its tangential derivatives along $\mathcal{H}^+$.

This stands in stark contrast with non-axisymmetric solutions, for which there is no conservation law available along $\mathcal{H}^+$. Deriving an instability in this setting requires instead sharp, \emph{quantitative} decay statements for $\phi$ and its tangential derivatives along $\mathcal{H}^+$. That is to say, one needs to derive the precise late-time tails for $\phi$ before deriving the existence of instabilities.\\
\\
\paragraph{\textbf{Late-time tails and time integrals}}
Once weak, integrated estimates have been established for $\phi$, one needs a decay mechanism to derive sharp decay estimates and to derive the precise leading-order behaviour (late-time tails) of $\phi$. In the sub-extremal setting, such a mechanism was developed in \cite{paper1, paper2, aagkerr} and it revolved around \emph{time integrals of $\phi$}, which are integrals of $\phi$ along the integral curves of the vector field $T$, together with appropriate hierarchies of $r$-weighted energy estimates far away from the black hole, first derived in \cite{newmethod}, and an exploitation of conservation laws along future null infinity $\mathcal{I}^+$.

In the case of axisymmetric $\phi$ and, analogously, solutions to the wave equation on extremal Reissner--Nordstr\"om, a similar mechanism can be exploited, which depends also on the existence of hierarchies of degenerate energy estimates near $\mathcal{H}^+$ and conservation laws along $\mathcal{H}^+$.

In view of the aforementioned breakdown of the degenerate red-shift estimates and the absence of conservation laws for non-axisymmetric $\phi$, \underline{the sub-extremal decay mechanism breaks down}. Furthermore, whereas there is an unambiguous notion of ``time integral'' for axisymmetric $\phi$ in extremal Kerr, associated to the vector field $T$, which agrees with $K$ when acting on axisymmetric $\phi$, there are multiple possible notions of time integrals in the non-axisymmetric, extremal setting.

\subsection{Remaining questions}
\label{sec:remainingq}
We list here some remaining questions and problems, which are not addressed in the present paper, but will be explored in future investigations.

\begin{itemize}
\item \emph{Uniform frequency-space analysis on near-extremal Kerr.} Extending the frequency-space analysis in the present paper to $|a|\leq M$ would allow one to prove the qualitative integrability assumptions of Theorem \ref{thm:tails} and to rule out instabilities \textbf{(b)} and \textbf{(c)} in Theorem \ref{thm:azimuthal}. While most of the frequency-space analysis in the present paper includes near-extremal spacetimes, it remains to revisit the estimates in the bounded, frequency regime, with $|\omega-m\upomega_+|\ll 1$ and extend them to the $|a|\leq M$ setting.
\item \emph{Uniform estimates in $m$.} The analysis in the present paper is concerned with fixed azimuthal modes, or bounded $m$. This means that we do not keep careful track of the $m$-dependence in all the estimates. Extending the results of the present paper uniformly in $m$ would require revisiting the analysis in the frequency regime $|\omega-m\upomega_+|\ll |m|$, which is then no longer a bounded frequency regime. In view of the existence of trapped null geodesics arbitrarily close to $\mathcal{H}^+$, one would expect any extension of the integrated energy estimates to unbounded $m$ to come with an additional loss of azimuthal $\partial_{\varphi}$ derivatives. As this frequency regime is both superradiant and is affected by trapping, it could also lead to more severe deviations from the bounded-$m$ setting and perhaps even to stronger instabilities that are not visible for bounded $m$.
\item \emph{Nonlinear analysis.} In order to address the effects on nonlinearities in the Einstein equations to the presence of azimuthal instabilities in the linearized setting, one could consider model nonlinear wave equations of the form:
\begin{equation*}
\square_g\phi= A(\phi)\cdot g^{-1}(d\phi,d\phi),
\end{equation*}
which satisfy the \emph{null condition} and are therefore consistent with the nonlinear structure present in the Einstein equations. Even when considering initial data supported on a fixed azimuthal mode $\phi_m$, all the other modes are excited due to the presence of the nonlinearity and one needs to first understand uniformly what happens to the estimates in the present paper when $m\to \infty$.

The additional difficulties faced when considering $m\to \infty$ notwithstanding, we can make the optimistic supposition for a moment that the growth and decay estimates in Theorems \ref{thm:azimuthal} and \ref{thm:tails} remain valid uniformly in $m$.

Since the nonlinearity is of the form $g^{-1}(d\phi,d\phi)$, growing transversal derivatives along the event horizon along $\mathcal{H}^+$ are not immediately fatal as will be multiplied by decaying tangential derivatives, so it is \emph{in principle} possible to obtain global existence and uniqueness of solutions and show that the behaviour of the linear problem persists. Indeed, the above null structure of the nonlinearity was used in \cite{aagnonlin} to show that the instability in extremal Reissner--Nordstr\"om remains asymptotic, just like in the linear problem.

Given the slower decay and stronger azimuthal instabilities in the present paper, the methods of \cite{aagnonlin} break down and new methods need to be developed to perform an analysis of nonlinear wave equations satisfying the null condition. The expected loss of derivatives in integrated energy estimates when considering the limit $m\to \infty$ and the resulting loss of derivatives in (degenerate) energy boundedness estimates would form additional difficulties for the nonlinear problem.
\end{itemize}
\subsection{Outline of the remainder of the paper}
We provide here an outline of the structure of the remainder of the paper.

\begin{itemize}
\item In \S\ref{sec:prelim1} and \S \ref{sec:prelim2} we introduce the necessary geometric preliminaries and basic facts regarding linear wave equations on Kerr spacetimes, as well as the main notation in the rest of the paper. We moreover derive some properties of the coefficients $\alpha_{\ell m}$ from Theorem \ref{thm:tails} and discuss stationary, or $K$-invariant solutions to the wave equation on extremal Kerr.
\item Then, we state precise versions of Theorems \ref{thm:azimuthal} and \ref{thm:tails} in \S \ref{sec:precstatthm} and sketch how they are proved in \S \ref{sec:sketchproof}.
\item The frequency-space analysis of the paper is fully contained in \S \ref{sec:freqspacean}, resulting in integrated energy estimates in physical space, which are derived in \S \ref{sec:intenestphys}.
\item The key hierarchies of $r$-weighted (in far-away regions) and $(r-r_+)$-weighted (in near-horizon regions) energy estimates are derived in \S \ref{sec:rweight}. Together with the integrated energy estimates in  \S \ref{sec:intenestphys}, these form the key ingredient for deriving energy decay estimates for $K$-derivatives in \S \ref{sec:edecay}.
\item In \S \ref{sec:ellipticKinv}, we construct time integrals of $\phi_m$ along integral curves of the vector field $K$ and we derive the necessary elliptic-type estimates. Then we combine these with the energy decay estimates from \S \ref{sec:edecay} in \S \ref{sec:decay} to obtain decay (energy and pointwise) estimates for $\phi_m$.
\item Finally, we use the precise energy and pointwise decay estimates from \S \ref{sec:decay} to derive azimuthal instabilities in \S \ref{sec:instab}.
\end{itemize}

\subsection{Acknowledgements}
The author would like to thank Rita Teixeira da Costa and Yakov Shlapentokh-Rothman for useful discussions.

\section{Preliminaries: geometry and foliations}
\label{sec:prelim1}
In this section, we introduce the geometric properties that are required for the analysis in the paper.
\subsection{The spacetime metric}
In this section, we introduce the 2-parameter family of \emph{Kerr spacetimes} $(\mathcal{M}_{M,a},g_{M,a})$ and the corresponding geometric notation. 

Let $M>0$ and $a\in \R$ with $|a|\leq M$. Then $\mathcal{M}_{M,a}$ is the following manifold-with-boundary:
\begin{align*}
\mathcal{M}_{M,a}=&\:\R\times [r_+,\infty)\times \s^2,\\
r_+=&\:M+\sqrt{M^2-a^2}.
\end{align*}
The constant $M$ is the \emph{mass} of the Kerr black hole and $a\cdot M$ can be considered its \emph{specific angular momentum}, or its angular momentum/mass.

We can cover $\mathcal{M}_{M,a}$ with coordinates $(v,r,\theta,\varphi_*)$, where $v\in \R$, $r\in [r_+,\infty)$ and $(\theta,\varphi_*)$ denote standard spherical coordinates on the unit round 2-sphere $\s^2$ (which degenerate at a half circle connecting the north and south pole). In these coordinates, the Lorentzian metric $g_{M,a}$ takes the following form:
\begin{multline}
\label{eq:metric}
g_{M,a}=-\rho^{-2}(\Delta -a^2 \sin^2\theta)dv^2+2dvdr-4Ma r\rho^{-2}\sin^2\theta dv d\varphi_*-2a\sin^2\theta drd\varphi_*\\
+\rho^2d\theta^2+\rho^{-2}((r^2+a^2)^2-a^2\Delta\sin^2\theta)\sin^2\theta d\varphi_*^2,
\end{multline}
with
\begin{align*}
\rho^2=&\:r^2+a^2\cos^2\theta,\\
\Delta=&\: r^2-2Mr+a^2=(r-r_+)(r-r_-),\\
r_-=&\:M-\sqrt{M^2-a^2}.
\end{align*}
The above local expression for $g_{M,a}$ is called the \emph{ingoing Eddington--Finkelstein} form of the Kerr metric.

The 1-parameter subfamily $g_{|a|,a}$ of Kerr spacetimes for which $M=|a|$ are known as the \emph{extremal Kerr spacetimes}.

The inverse metric $g^{-1}_{M,a}$ corresponding to $g_{M,a}$ then takes the following form:
\begin{multline}
\label{eq:invmetric}
g^{-1}_{M,a}=a^2\rho^{-2}\sin \theta^2\theta \partial_v\otimes \partial_v+\rho^{-2}(r^2+a^2)[\partial_v\otimes\partial_r+\partial_r\otimes \partial_v]+\rho^{-2}\Delta\partial_r\otimes \partial_r\\
+a\rho^{-2}[(\partial_v+\partial_r)\otimes \partial_{\varphi_*}+\partial_{\varphi_*}\otimes (\partial_v+\partial_r)]+\rho^{-2}[\partial_{\theta}\otimes\partial_{\theta}+\sin^{-2}\theta\partial_{\varphi_*}\otimes \partial_{\varphi_*} ].
\end{multline}

We introduce the following notation for submanifolds of $\mathcal{M}_{M,a}$:
\begin{align*}
\mathcal{H}^+:=&\:\{r=r_+\}\quad \textnormal{is called the (future) event horizon},\\
\mathring{\mathcal{M}_{M,a}}:=&\:\mathcal{M}_{M,a}\setminus \mathcal{H}^+,\\
S^2_{v',r'}:=&\:\{r=r'\}\cap\{v=v'\}\:\:\textnormal{with $v'\in \R$ and $r'\in [r_+,\infty)$ are diffeomorphic to $\s^2$}.
\end{align*}
We will moreover consider the following functions:
\begin{align*}
r_*=&\:\int_{3M}^r \frac{r'^2+a^2}{\Delta}\,dr' \quad \textnormal{on $\mathring{\mathcal{M}}_{M,a}$},\\
t=&\:v-r_*\quad \textnormal{on $\mathring{\mathcal{M}}_{M,a}$},\\
u=&\:v-2r_*=t-r_*\quad \textnormal{on $\mathring{\mathcal{M}}_{M,a}$},\\
\varphi=&\: \varphi_*+\int_r^{\infty}\frac{a}{\Delta}\,dr'\:\:\textnormal{mod}\,2\pi \quad \textnormal{on $\mathring{\mathcal{M}}_{M,a}$},\\
\widetilde{\varphi}=&\:\varphi_*-\upomega_+ v\:\:\textnormal{mod}\,2\pi \quad \textnormal{on $\mathcal{M}_{M,a}$},
\end{align*}
where $\upomega_+=\frac{a}{r_+^2+a^2}$ can be interpreted as the ``angular velocity of the Kerr black hole'', assigned by stationary observers at infinity. The coordinates $(v,r,\theta,\tphi)$ may be interpreted as ``rotating with the Kerr black hole'', because the integral curves of the null generator of $\mathcal{H}^+$ correspond to the level sets $\{(v,r,\theta,\tphi)=(v,r_+,\theta_0,\tphi_0)\}$, with $(\theta_0,\tphi_0)$ constant.

The functions $(t,r,\theta,\varphi)$ define coordinates on the manifold $\mathring{\mathcal{M}}_{M,a}$ and are called \emph{Boyer--Lindquist coordinates}.

To simplify the notation, we will frequently omit one parameter in the 2-parameter family of spacetimes by fixing $r_+=1$, which implies that:
\begin{align*}
r_-=&\:a^2,\\
M=&\:\frac{1}{2}(1+a^2).
\end{align*}
Then extremal Kerr spacetimes correspond simply to the choices $a=\pm 1$.

\subsection{Foliations and conformal coordinates}
\label{sec:foliations}
In this section, we introduce a foliation of $\mathcal{M}_{M,a}$ by appropriate isometric spacelike hypersurfaces, i.e.\ a ``time-slicing''.

Let $\mathbbm{h}: [r_+,\infty)\to\R_+$ be a smooth function, such that for all $r\in [r_+,\infty)$:
\begin{align*}
\frac{2(r^2+a^2)}{\Delta}-\mathbbm{h}(r)=&\:O_{\infty}(r^{-2}),\\
\mathbbm{h}(r)\Delta \left(\frac{2(r^2+a^2)}{\Delta}-\mathbbm{h}(r)\right)>&\:a^2\sin^2\theta,\\
\end{align*}
Then we define the \emph{time function} $\tau$ on $\mathcal{M}_{M,a}$ as follows:
\begin{equation*}
\tau=v-\int_{r_+}^r \mathbbm{h}(r')\,dr'-v_0,
\end{equation*}
with $v_0\in \R_+$. Define
\begin{equation*}
\Sigma_{\tau'}=\{\tau=\tau'\}.
\end{equation*}
\textbf{The analysis in the present paper, will be restricted to the spacetime subset:}
\begin{equation*}
\mathcal{R}:=\{\tau\geq 0\}.
\end{equation*}
We also denote
\begin{equation*}
\Sigma:=\Sigma_{0}.
\end{equation*}
See Figure \ref{fig:penrose} below for a picture.
\begin{figure}[H]
	\begin{center}
\includegraphics[scale=0.75]{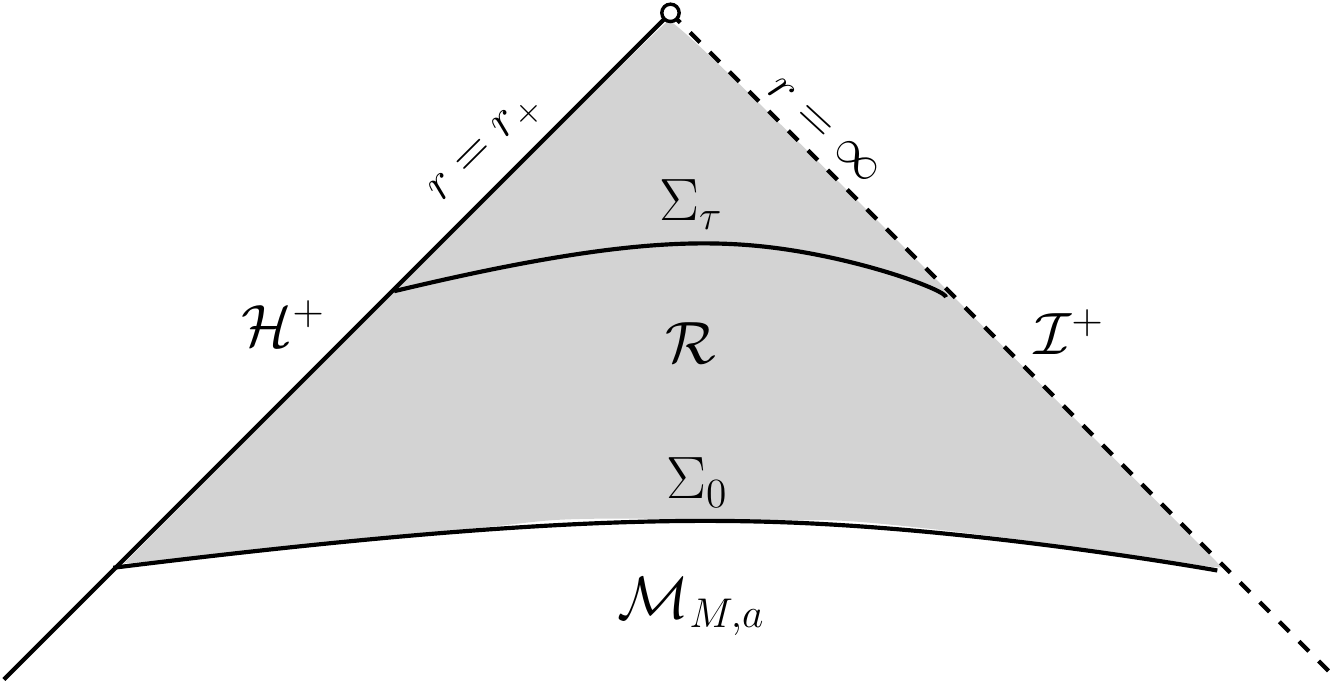}
\end{center}
\caption{A 2-dimensional representation of the region under consideration in the spacetime $\mathcal{M}_{M,a}$, with the hypersurfaces $\Sigma_{\tau}$ and the shaded region depicting $\mathcal{R}$. Each point in the picture represents a 2-sphere, composed of all points with $(v,r)$ constant.}
	\label{fig:penrose}
\end{figure}

\subsection{Key vector fields}
\label{sec:vf}
Kerr spacetimes are stationary and axisymmetric, which is reflected in the $v$- and $\varphi_*$-independence of the metric components in \eqref{eq:metric}. The following vector fields are therefore Killing vector fields:
\begin{align*}
T:=&\:\partial_v,\\
\Phi:=&\:\partial_{\varphi_*},
\end{align*}
with respect to the coordinates $(v,r,\theta,\varphi_*)$.

The linear combination:
\begin{equation*}
K:=T+\upomega_+\Phi
\end{equation*}
plays a particularly important role, as it is a null generator of $\mathcal{H}^+$. Note that $K=\partial_v$ in $(v,r,\theta,\tphi)$ coordinates.

We moreover introduce:
\begin{align*}
Y=&\:\partial_r,\\
\underline{L}=&-\frac{\Delta}{2(r^2+a^2)}Y,\\
L=&\: T+\frac{a}{r^2+a^2}\Phi-\underline{L},\\
\end{align*}
Note that $L(v)=-\Lbar(u)=1$. It is straightforward to verify the following causality properties for $r_+=1$:
\begin{align*}
g(T,T)<&\:0\quad \textnormal{in}\quad \{\Delta-a^2\sin^2\theta>0\} \subset \{r<2M\},\\
g(L,L)=&\:g(\Lbar,\Lbar)=0,\\
g(K,K)=&\:-\frac{(2+a^2+a^2\cos(2\theta))^2}{4(1+a^2)^2}(1-a^2)(r-1)+\left(-1+\frac{a^2\sin^2\theta(6+2a^2-a^2\sin^2\theta)}{1+a^2}\right)(r-1)^2\\
&+O((r-1)^3)\quad \textnormal{when fixing $r_+=1$}.
\end{align*}
Note therefore that $K$ is null along $\mathcal{H}^+$ ($r=1$) and, in sub-extremal Kerr spacetimes, it extends as a timelike vector field in a neighbourhood of the form $\{r-r_+<\delta_a\}$, with $\delta_a>0$ suitably small (depending on $a$). In extremal Kerr spacetimes, on the other hand, $K$ is not uniformly timelike or spacelike in a region of the form $\{r-r_+\leq \delta\}$, no matter how small we take $\delta>0$.

We denote further:
\begin{align*}
X:=&\:Y+\mathbbm{h}T,\\
\widetilde{X};=&\:Y+\mathbbm{h}K=X+\upomega_+\mathbbm{h}\Phi.
\end{align*}
The vector fields $X$ and $\tX$ are tangential to $\Sigma_{\tau}$. Note that $X=\partial_r$ in $(\tau,r,\theta,\varphi_*)$ coordinates and $\widetilde{X}=\partial_r$ in $(\tau,r,\theta,\widetilde{\varphi})$ coordinates.

Finally, we denote
\begin{equation*}
{\mathfrak{X}}_*=\frac{\Delta}{r^2+a^2}Y+T+\frac{a}{r^2+a^2}\Phi.
\end{equation*}
We can express ${\mathfrak{X}}_*=\partial_{r_*}$ in $(t,r_*,\theta,\varphi)$ coordinates.

We moreover introduce the Carter operator
\begin{equation*}
Q:=\slashed{\Delta}_{\s^2}-\Phi^2+\sin^2\theta T^2.
\end{equation*}
From standard elliptic estimates on $\s^2$ we have the following inequality for $f\in H^2(\s^2)$:
\begin{equation*}
\int_{\s^2}|f|^2+|\snabla_{\s^2}f|^2+|\snabla_{\s^2}^2f|^2\,d\sigma\leq C\int_{\s^2}|Qf|^2+|\Phi^2f|^2+|K\Phi f|^2+|K^2 f|^2\,d\sigma.
\end{equation*}
\subsection{Differential operators on $\s^2$}
\label{sec:sphere}
We introduce the following second-order differential operator on the unit round 1-sphere $\s^2$:
\begin{equation}
\label{eq:anghorop}
\slashed{\mathcal{D}}(\cdot):=\slashed{\Delta}_{\s^2}(\cdot)+a^2 \upomega_+^2\sin^2\theta \Phi^2(\cdot),
\end{equation}
where $\slashed{\Delta}_{\s^2}$ is the Laplacian on the unit 1-sphere. 

The operator $\slashed{\mathcal{D}}$ appears in the purely angular part of the wave operator $\square_{g_{M,a}}$. This follows by writing down $\square_g(\cdot)$ in terms of the vector fields $(K,Y,\Phi)$; see (iv) of Lemma \ref{lm:exprboxg}.

We will now relate $\slashed{\mathcal{D}}$ to the differential operator corresponding to \emph{oblate spheroidal harmonics}. 

Let $\nu\in \R$ and consider the following operator on $L^2(\s^2)$:
\begin{equation*}
P(\nu)=-\frac{1}{\sin \theta} \partial_{\theta}(\sin \theta \partial_{\theta}(\cdot))-\frac{1}{\sin^2\theta}\partial_{\varphi}^2(\cdot)-\nu^2\cos^2\theta(\cdot).
\end{equation*}
Note that $P(\nu)$ is self-adjoint, so there exists an orthonormal basis of eigenfunctions. Since $[P(\nu),\partial_{\varphi}]=0$, we can decompose these eigenfunctions by employing eigenfunctions $\{e^{im\varphi}\}_{m\in \Z}$ of $\partial_{\varphi}$:
\begin{equation*}
S_{m \ell}(\theta;\nu)e^{i m\varphi},
\end{equation*}
with $\ell\in \N_0$ and $m\in \Z$ labels, such that $-\ell\leq m\leq \ell$. We denote the corresponding eigenvalues of $P(\nu)$ by $\lambda_{m \ell}^{(\nu)}$, with $\lambda_{m\ell}^{(\nu)}<\lambda_{m(\ell+1)}^{(\nu)}$, and note that they satisfy the inequalities:
\begin{align}
\label{eq:angev1}
\lambda_{m\ell}^{(\nu)}+\nu^2\geq&\: |m|(|m|+1),\\
\label{eq:angev2}
\lambda_{m\ell}^{(\nu)}+\nu^2\geq &\:2|m\nu|.
\end{align}

We will moreover use the following shorthand notation: let $|a|\leq M$, then
\begin{align*}
P:=&\:P\left(a\upomega_+\right)=-\frac{1}{\sin \theta} \partial_{\theta}(\sin \theta \partial_{\theta}(\cdot))-\frac{1}{\sin^2\theta}\partial_{\varphi}^2(\cdot)-a^2\upomega_+^2m^2\cos^2\theta(\cdot),\\
S_{m \ell}(\theta):=&\:S_{m \ell}\left(\theta; a m\upomega_+\right),\\
\lambda_{m\ell}:=&\:\lambda_{m\ell}^{\left(a\upomega_+\right)}.
\end{align*}
	
We can relate $P$ as follows to $\slashed{\mathcal{D}}$
\begin{equation*}
-\slashed{\mathcal{D}}(\cdot)=P(\cdot)-a^2\upomega_+^2\Phi^2,
\end{equation*}
so by \eqref{eq:angev1}, the following inequality holds:
\begin{equation*}
-\int_{\s^2}\Re(\slashed{\mathcal{D}}(S_{m \ell}e^{im\tphi}) \overline{S}_{m\ell}e^{-im\tphi})\,d\sigma\geq  |m|(|m|+1) \int_{\s^2}|S_{m\ell}|^2\,d\sigma.
\end{equation*}
We moreover have that for all $\theta\in (0,\pi)$
\begin{equation*}
\frac{1}{\sin^2\theta}+a^2\upomega_+^2\sin^2\theta\geq 1+a^2\upomega_+^2,
\end{equation*}
so we also have that
\begin{equation*}
\begin{split}
-\int_{\s^2}\Re(\slashed{\mathcal{D}}S_{m \ell}e^{im\tphi} \overline{S}_{m\ell}e^{-im\tphi})\,d\sigma=&\:\int_{\s^2}|\partial_{\theta}S_{m\ell}|^2+\left(\frac{1}{\sin^2\theta}+a^2\upomega_+^2\sin^2\theta\right)m^2|S_{m \ell}|^2\,d\sigma\\
\geq&\:  \left( 1+a^2\upomega_+^2\right)m^2\int_{\s^2}|S_{m \ell}|^2\,d\sigma.
\end{split}
\end{equation*}

Denote $\Lambda_{+,m \ell}=\lambda_{m\ell}+a^2\upomega_+^2m^2$. Then 
\begin{equation*}
\slashed{\mathcal{D}}(e^{im\varphi}S_{m\ell}(\theta))=-\Lambda_{+, m \ell}e^{im\tphi}S_{m \ell}(\theta).
\end{equation*}
and, by the above,
\begin{equation}
\label{eq:mainineqLambda}
\Lambda_{+, m\ell}\geq \max\{ (1+a^2\upomega_+^2)|m|,|m|+1\}|m|.
\end{equation}
In particular, for $r_+=1$ and $|a|=1$,
\begin{equation*}
\Lambda_{+, m\ell}\geq \max\left\{ \frac{5}{4}|m|, |m|+1\right\}|m|.
\end{equation*}

Define
\begin{equation*}
|\widetilde{\snabla}_{\s^2}(\cdot)|^2:=|{\snabla}_{\s^2}(\cdot)|^2+a^2 \upomega_+^2\sin^2\theta|\Phi(\cdot)|^2,
\end{equation*}
with $\snabla_{\s^2}$ denoting the covariant derivative with respect to the metric on the unit 1-sphere.

Let $r_+=1$ and $|a|=1$. Then the following quantity will play a key role in the paper:
\begin{equation*}
\alpha_{m \ell}:=\Lambda_{+, m\ell}-2m^2+\frac{1}{4}.
\end{equation*}
The late-time asymptotics and sharp decay rates that we will derive follow from the fact that $\alpha_{m |m|}>0$ for $|m|\geq 2$. This is proved in the lemma below.
\begin{lemma}
\label{lm:negalpha}
For $|m|=1$, $\alpha_{|m|m}\geq \frac{1}{4}$ and for $|m|\geq 2$, $\alpha_{m |m|}<0$.
\end{lemma}
\begin{proof}
The inequality for $|m|=1$ follows immediately from \eqref{eq:mainineqLambda}. Now let $|m|\geq 2$. We appeal to the Rayleigh--Ritz quotient: let $f: [0,\pi]\to \R$, then
\begin{equation*}
\Lambda_{m |m|}=\inf_{e^{im\tphi}f\in L^2(\s^2)\setminus \{0\}}\frac{\int_{\s^2}|\partial_{\theta}f|^2\sin\theta+(\frac{1}{4}\sin^2\theta+\sin^{-2}\theta) \sin\theta |f|^2\,d\theta d\tphi}{\int_{\s^2} |f|^2\,\sin \theta\,d\theta d\tphi}
\end{equation*}
Let $f(\theta)=\sin\theta$, then we obtain via the above formula,
\begin{equation*}
m^{-2}\Lambda_{m |m|}\leq \frac{17}{10}+\frac{1}{2m^2}.
\end{equation*}
Hence,
\begin{equation*}
m^{-2}\alpha_{m|m|}\leq -\frac{3}{10}+\frac{3}{4m^2}.
\end{equation*}
For $|m|\geq 2$, the above expression therefore implies that
\begin{equation*}
m^{-2}\alpha_{m|m|}\leq -\frac{9}{80}.
\end{equation*}
\end{proof}

We will denote $d\sigma=\sin\theta d\theta d\varphi_*$ and $d\sigma=\sin\theta d\theta d\tphi$ when integrating functions along constant $(\tau,r)$ spheres with respect to $(\tau,r,\theta,\varphi_*)$ and $(\tau,r,\theta,\tphi)$  coordinates, respectively.

Let $f\in L^2(\s^2)$. Then we can write $f=\sum_{m\in \Z}f_m$, with
\begin{equation*}
 f_m(\theta,\tphi)=\left[\int_{\s^2}f(\theta,\tphi')e^{-im \tphi'}\,\sin \theta d\theta d\tphi'\right]e^{im\tphi}.
\end{equation*}
We refer to $f_m$ as the $m$-th azimuthal mode of $f$. We can straightforwardly extend the definition of $f_m$ to functions on $\mathcal{R}$.

Note that the eigenfunctions $S_{m\ell}(\theta)$ and the eigenvalues $\Lambda_{+,m\ell}$ may be interpreted as the oblate spheroidal harmonics and angular frequencies evaluated at the frequency $\omega=m\upomega_+$ corresponding to a Carter separation of the Fourier transform of $\phi$ in time $t$; see \S \ref{sec.carter}.

\subsection{Notation}
We will make use of the notation $O(|f|)$ to indicate functions on $\mathcal{M}$ that can be bounded uniformly by $|f|$ with a constant that depends only on $M$ and $\mathbbm{h}$. We denote with $O_{n}(|f|)$ functions $g$ on $\mathcal{M}$ such that $\partial_r^k g$ is bounded uniformly by $|\partial_r^kf|$ for all $k\leq n\in \N_0$.

We use the notation $O^m_k(|f|)$ for functions with azimuthal angular dependence $e^{im\tphi}$, when the uniform constant is allowed to depend additionally on $m$.

We will use the notation $c$ and $C$ to indicate constants that depend only on $M,a$ and $\mathbbm{h}$. If the constants depends on an additional set of parameters $\beta_1,\ldots,\beta_s\in \C$, we will instead use the notation  $c_{\beta_1,\ldots,\beta_s}$ and $C_{\beta_1,\ldots,\beta_s}$. To simplify notation, we make use the following ``algebra of constants'':
\begin{equation*}
C\cdot C=C+C=C,\quad c\cdot c=c+c=c.
\end{equation*}
\section{Preliminaries: wave equation}
\label{sec:prelim2}
\subsection{The wave operator}
In this section, we express the wave operator $\square_g$ using different sets of vector fields and we state the relevant local-in-time weighted energy estimates.
\begin{lemma}
\label{lm:exprboxg}
Let $\phi$ be a solution to the inhomogeneous wave equation:
\begin{equation}
\label{eq:inhomwaveeq}
\square_{g_{M,a}}\phi=G.
\end{equation}
Then, we can express \eqref{eq:inhomwaveeq} in the following different ways (interpreted in a distributional sense):
\begin{enumerate}[label=\emph{(\roman*)}]
\item In terms of the vector fields $(T,Y)$, we can express:
\begin{equation}
\label{eq:waveeqY}
\rho^2G=Y(\Delta Y\phi)+2aY\Phi\phi+\slashed{\Delta}_{\s^2}\phi+2(r^2+a^2)TY\phi+a^2\sin^2\theta T^2\phi+2aT\Phi\phi+2rT\phi.
\end{equation}
\item In terms of the vector fields $(T,X)$, we can express:
\begin{equation}
\label{eq:waveeqX}
\begin{split}
\rho^2G=&\:X(\Delta X\phi)+2aX\Phi\phi+\slashed{\Delta}_{\s^2}\phi+2(r^2+a^2-\hfol\Delta)TX\phi+[a^2\sin^2\theta+\hfol (\Delta \hfol- 2(r^2+a^2))]T^2\phi\\
&\:+2a(1-\hfol)T\Phi\phi+(2r-\frac{d}{dr}(\Delta \hfol))T\phi\\
=&\:X(\Delta X\phi)+2aX\Phi\phi+\slashed{\Delta}_{\s^2}\phi-2(r^2+a^2)^{\frac{1}{2}}(1+O(r^{-1}))TX\psi\\
&\:+[a^2\sin^2\theta+\hfol (\Delta \hfol- 2(r^2+a^2))]T^2\phi+2a(1-\hfol)T\Phi\phi+O(1)T\phi.
\end{split}
\end{equation}
\item In terms of the vector fields $(K,Y)$, we can express:
\begin{multline}
\label{eq:waveeqK}
\rho^2G=Y(\Delta Y\phi)-2\upomega_+(r^2-r_+^2)Y\Phi\phi+(\slashed{\mathcal{D}}-2a\upomega_+ \Phi^2)\phi-2\upomega_+r \Phi \phi+2(r^2+a^2)KY\phi\\
+a^2\sin^2\theta (K^2\phi-2\upomega_+ K\Phi\phi)+2a K\Phi\phi+2rK\phi.
\end{multline}
\item In terms of the vector fields $(K,\widetilde{X})$ and the rescaled quantity $\psi:=(r^2+a^2)^{\frac{1}{2}}\phi$, we can express:
\begin{multline}
\label{eq:waveeqtildeX}
\rho^2G=\widetilde{X}(\Delta \widetilde{X}\phi)-2\upomega_+(r^2-r_+^2)\widetilde{X}\Phi\phi+(\slashed{\mathcal{D}}-2a\upomega_+ \Phi^2)\phi-2\upomega_+r \Phi \phi+2(r^2+a^2-\hfol\Delta)K\widetilde{X}\phi\\
+[a^2\sin^2\theta+\hfol (\Delta \hfol- 2(r^2+a^2))]K^2\phi+2(a+\hfol\upomega_+(r^2-r_+^2)-a^2\upomega_+\sin^2\theta)K\Phi\phi+\left(2r-\frac{d}{dr}(\Delta \hfol)\right)K\phi\\
=\widetilde{X}(\Delta \widetilde{X}\phi)-2\upomega_+(r^2-r_+^2)\widetilde{X}\Phi\phi+(\slashed{\mathcal{D}}-2a\upomega_+ \Phi^2)\phi-2\upomega_+r \Phi \phi-2(r+O(r^{-1}))KX\psi\\
+[a^2\sin^2\theta+O(1)]K^2\phi+2(O(1)-a^2\upomega_+\sin^2\theta)K\Phi\phi+O(1)K\phi.
\end{multline}
\item In terms of the principal null fields $(L,\Lbar)$, we can express:
\begin{multline}
\label{eq:waveeqnull}
(r^2+a^2)^{\frac{1}{2}}\rho^2G=-4\frac{(r^2+a^2)^2}{\Delta}L\Lbar \psi+a^2\sin^2 \theta  K^2\psi+2a(1-a\upomega_+\sin^2\theta) K\Phi \psi+(\slashed{\mathcal{D}}-2a\upomega_+ \Phi^2)\psi\\
-2\frac{ar}{r^2+a^2}\Phi\psi-\frac{(r^2+a^2)^2}{2 r\Delta} \frac{d}{dr}\left(r^2(r^2+a^2)^{-3}\Delta^2\right)\psi.
\end{multline}
\end{enumerate}
\end{lemma}
\begin{proof}
We write:
\begin{equation*}
\square_g\phi=\frac{1}{\sqrt{-\det g_{M,a}}}\sum_{0\leq \alpha,\beta\leq 3}\partial_{\alpha}(\sqrt{-\det g_{M,a}} (g^{-1})^{\alpha \beta} \partial_{\beta}\phi)
\end{equation*}
and evaluate this expression in $(v,r,\theta,\varphi_*)$, using \eqref{eq:invmetric}, in order to obtain (i). The remaining equations follow in a straightforward manner by using the relations between $T,K,Y,Y,\widetilde{X}$ and $\Phi$ from \S \ref{sec:vf}.
\end{proof}
By standard local-in-time energy estimates for linear wave equations, we obtain global existence and uniqueness for solutions to the initial value problem corresponding to \eqref{eq:inhomwaveeq}; see also \cite[Proposition 3.4]{aagkerr}.
\begin{theorem}
Let $k\in \N_0$ and $G\in L^1_{\tau}H^{k}(\Sigma_{\tau};\C)$ and $(\phi_{\rm data},\phi_{\rm data}')\in H^{k+1}_{\rm loc}(\Sigma_0;\C)\times H^{k}_{\rm loc}(\Sigma_0;\C)$. Then there exists a solution $\phi$ to \eqref{eq:inhomwaveeq} (in the distributional sense), such that for all $\tau \geq 0$:
\begin{align*}
\phi\in &\:C^0([0,\tau ], H^{k+1}_{\rm loc}(\Sigma_0;\C) \cap C^1([0,\tau ], H^{k}_{\rm loc}(\Sigma_0;\C),\\
\phi|_{\Sigma_0}=&\: \phi_{\rm data},\\
K\phi|_{\Sigma_0}=&\:\phi_{\rm data}'.
\end{align*}
Furthermore, if $(\phi_{\rm data},\phi_{\rm data}')\in C^{\infty}_c(\Sigma_0;\C)\times C^{\infty}_c(\Sigma_0;\C)$, then $\phi\in C^{\infty}(\mathcal{R};\C)$ and for all $j,k,l\in \N_0$
\begin{equation*}
\lim_{R\to\infty} \int_{\Sigma_{\tau}\cap \{r=R\}} (r^2X)^j K^k \Phi^l\psi d\sigma<\infty,
\end{equation*}
with $\psi$ defined in \emph{(v)} of Lemma \ref{lm:exprboxg}.
\end{theorem}

In the $a=|M|$ case, we define the following \emph{energy densities}, with weights near $r=r_+=M$:
\begin{equation}
\label{eq:edens}
\mathcal{E}_p[\phi]:=(1-Mr^{-1})^{2-p}|X\phi|^2+|T\phi|^2+r^{-2}|\snabla_{\s^2}\phi|^2\quad \textnormal{with $0\leq p\leq 2$}.
\end{equation}

\begin{proposition}
\label{prop:locenest}
Assume $a=|M|$. Let $T_0\geq 0$ and let $\phi$ be a solution to \eqref{eq:inhomwaveeq} with $G\in H^{k}(\Sigma_{\tau};\C)$ for all $\tau\in [0,T_0]$, such that $\int_{\Sigma_0}\mathcal{E}_p[\phi]\,r^2d\sigma dr<\infty$ for some $0\leq p\leq 2$. Then there exists a constant $C_{T_0}>0$ such that for all $\tau\in [0,T_0]$
\begin{multline*}
\int_{\Sigma_{\tau}}\mathcal{E}_p[\phi]\,r^2d\sigma dr+\int_{\mathcal{I}^+\cap\{0\leq \tau'\leq\tau\}} |T\psi_m|^2\,d\sigma d\tau+\int_{\mathcal{H}^+\cap\{0\leq \tau'\leq\tau \}} |K\psi_m|^2\,d\sigma d\tau\\
\leq C_{T_0}\left[\int_{\Sigma_{0}}\mathcal{E}_p[\phi]\,r^2d\sigma dr+\int_{0}^{\tau}\int_{\Sigma_{\tau} }|G|^2 r^2\, d\sigma dr d\tau\right].
\end{multline*}
\end{proposition}
\begin{proof}
Let $V_p$ be a smooth vector field such that $V=L+(r-M)^{-p}\underline{L}$ in $r\leq 3M$ and $N=T$ in $r\geq 4M$. Then
\begin{equation*}
\begin{split}
g_{M,\pm |M|}(V_p,V_p)=&\:2(r-M)^{-p}g(L,\underline{L})=(r-M)^{-p}g(L+\underline{L},L+\underline{L})\\
=&\:(r-M)^{-p}g\left(T+\frac{M}{r^2+M^2}\Phi,T+\frac{a}{r^2+a^2}\Phi\right)\\
=&\:\rho^{-2}(r-M)^{-p}(r^2+M^2)^{-1}\big[M^2(r^2+M^2)\sin^2\theta-\Delta(r^2+M^2)-4M^3 r\sin^2\theta\\
&\:+M^2(r^2+M^2)\sin^2\theta-M^4(r^2+M^2)^{-1} \Delta\sin^4\theta\big]\\
=&\:-\rho^{-2}(r-M)^{2-p} (r^2+M^2\cos(2\theta)+M^4(r^2+M^2)^{-1}\sin^4\theta).
\end{split}
\end{equation*}
Hence, $V_p$ is timelike in $\{r>M\}$ for all $0\leq p\leq 2$, timelike along $\mathcal{H}^+$ if $p=2$ and null along $\mathcal{H}^+$ if $p<2$.

It then follows that:
\begin{equation*}
\mathbb{T}[\phi](V_p,\mathbf{n}_{\tau})\sim \mathcal{E}_p[\phi],
\end{equation*}
with $\mathbb{T}[\phi]:=\Re(d\phi\otimes d\overline{\phi})-\frac{1}{2} |\nabla \phi|^2 g$ the energy-momentum tensor corresponding to $\square_g\phi$ and $\mathbf{n}_{\tau}$ the unit normal to $\Sigma_{\tau}$.
\end{proof}

The main theorems concern solutions to:
\begin{equation}
\label{eq:waveeq}
\square_{g_{M,a}}\phi=0
\end{equation}
with $|a|=M$.

\subsection{Stationary solutions in extremal Kerr}
\label{sec:statsol}
A key role in the analysis will be played by \emph{stationary}, $K$-invariant solutions to to \eqref{eq:inhomwaveeq} with $G=0$ and $|a|=M$. If we express these with respect to the coordinates $(t,r,\theta,\varphi)$, they take the form
\begin{equation*}
\mathfrak{w}_{m \ell}(r)e^{-i\upomega_+ t+im \varphi }S_{m\ell}(\theta).
\end{equation*}
Assume for the sake of convenience that $r_+=1$, then 
\begin{equation*}
K(\mathfrak{w}_{m \ell}e^{-i\upomega_+ t+im \varphi }S_{m\ell})=0\quad \textnormal{and}\quad  \square_{g_{M,a}} (\mathfrak{w}_{m \ell}e^{-i\upomega_+ t+im \varphi }S_{m\ell})=0.
\end{equation*}

is equivalent to the following ODE for $\mathfrak{w}_{m \ell}(r(r_*))$:
\begin{align}
\label{eq:staticfreqweight}
\mathfrak{w}''-\widetilde{V}_{\rm stat}(r(r_*))\mathfrak{w}=&\:0,\\ \nonumber
\widetilde{V}_{\rm stat}(r):=&\:-\frac{1}{4}m^2+(r^2+1)^{-2}\left[m^2a^2(2r-1)+\Lambda_{+,m\ell} \Delta\right]+(r^2+1)^{-3}\Delta(3r^2-4r+1)\\ \nonumber
&-3(r^2+1)^{-4}\Delta^2r^2.
\end{align}

Denote $\alpha_{m\ell}=\Lambda_{+,m\ell}-2m^2+\frac{1}{4}$.

\begin{lemma}
\label{lm:propfracw}
Let $r_+=1$ and $a=1$. Assume that there exist $b,B>0$ such that $\Lambda_{+,m\ell}\leq Bm^2$ and $\alpha_{m\ell}\neq 0$ for all $m\in \Z\setminus\{0\}$. Then for all $m\in \Z\setminus\{0\}$, there exists two unique solutions $\mathfrak{w}_{\pm,m\ell}$ to \eqref{eq:staticfreqweight} that satisfy:
\begin{align}
\label{eq:wfreqexp}
\mathfrak{w}_{\pm,m\ell}(r)=&\:(r-1)^{-\frac{1}{2}\pm i\sqrt{-\alpha_{m\ell}}}\left[1+O_{\infty}((r-1))\right].
\end{align}
Furthermore,
\begin{align}
\label{eq:estderw}
\mathfrak{w}_{\pm,m\ell}^{-1}\frac{d\mathfrak{w}_{\pm,m\ell}}{dr}=&\:-\frac{1}{2}(r-1)^{-1}\pm i\sqrt{-\alpha_{m\ell}}(r-1)^{-1}+O(1).
\end{align}
When $\alpha_{m\ell}=0$, there exist again two unique solutions $\mathfrak{w}_{\pm}$, satisfying:
\begin{align}
\label{eq:wfreqexpexcmin}
\mathfrak{w}_{-,m\ell}(r)=&\:(r-1)^{-\frac{1}{2}}\left[1+O_{\infty}^m((r-1))\right],\\
\label{eq:wfreqexpexcplus}
\mathfrak{w}_{+,m\ell}(r)=&\:(r-1)^{-\frac{1}{2}}\log\left(\frac{1}{r-1}\right)[1+O_{\infty}^m((r-1)]\\
\label{eq:estderwexcmin}
\mathfrak{w}_{-,m\ell}^{-1}\frac{d\mathfrak{w}_{-,m\ell}}{dr}=&\:-\frac{1}{2}(r-1)^{-1}(r-1)^{-1}+O^m(1),\\
\label{eq:estderwexcplus}
\mathfrak{w}_{+,m\ell}^{-1}\frac{d\mathfrak{w}_{+,m\ell}}{dr}=&\:-\left(\frac{1}{2}+\frac{1}{\log(r-1)^{-1}}\right)(r-1)^{-1}+\log(r-1)O^m(1).
\end{align}
\end{lemma}
\begin{proof}
We will suppress the subscript $\pm, m\ell$ in $\mathfrak{w}_{\pm,m\ell}$ below. Note first that we can express \eqref{eq:wfreqexp} as follows:
\begin{align*}
\frac{d^2\mathfrak{w}}{dr^2}+2(r-1)^{-1}[1+O_{\infty}((r-1))]\frac{d\mathfrak{w}}{dr}-(r-1)^{-4}(r^2+a^2)^2\widetilde{V}_{\rm stat}\mathfrak{w}=0.
\end{align*}
By applying \eqref{eq:staticfreqweight} and expanding $\widetilde{V}_{\rm stat}$ in $r-1$ around $r=1$ (see \eqref{eq:Vstatest}), it follows that:
\begin{equation*}
\frac{d^2\mathfrak{w}}{dr^2}+2(r-1)^{-1}[1+O_{\infty}((r-1))]\frac{d\mathfrak{w}}{dr}+(r-1)^2\left[2m^2-\Lambda_{+,m\ell}+O_{\infty}(r-1)\right]\mathfrak{w}=0.
\end{equation*}
so $r=1$ is a regular singular point for \eqref{eq:staticfreqweight} and we can apply the Liouville--Green or WKB approximation in \cite[\S 6.2, Theorem 2.2]{olv74} to obtain \eqref{eq:wfreqexp} and \eqref{eq:estderw}. Here we make use of the condition $\alpha_{m\ell}\neq 0$ to ensure that the Liouville--Green approximation does not break down.

In the case $\alpha_{m\ell}\neq 0$, we instead apply standard ODE theory for expansions of solutions around regular singular points \cite[\S 5.4, Theorem 4.1]{olv74} which introduces $m$ dependence in the expansions \eqref{eq:wfreqexpexcmin}--\eqref{eq:estderwexcplus}.
\end{proof}

We now introduce the functions $w_{\pm, m \ell}$, which are related to $\mathfrak{w}_{\pm,m \ell}$, defined in Lemma \ref{lm:propfracw}, as follows:
\begin{equation*}
w_{\pm, m \ell}(r)e^{-i\upomega_+ \tau+im\tilde{\varphi}}=\frac{\sqrt{2}}{(r^2+1)^{\frac{1}{2}}}\mathfrak{w}_{\pm,m \ell}e^{-i\upomega_+ t+im \varphi }.
\end{equation*}

By Lemma \ref{lm:propfracw}, we can expand uniformly in $m$ for $\alpha_{m\ell}\neq 0$ :
\begin{align}
\label{eq:wfreqexpelliptic}
w_{\pm}(r)=&\:(r-1)^{-\frac{1}{2}-im\pm i\sqrt{2m^2-\Lambda_{+,m \ell}-\frac{1}{4}}}\left[1+O_{\infty}((r-1))\right],\\
\label{eq:estderwelliptic}
w_{\pm,m \ell}^{-1}\frac{dw_{\pm,m \ell}}{dr}=&\:-\frac{1}{2}(r-1)^{-1}+i\left(-m \pm \sqrt{2m^2-\Lambda_{+,m \ell}-\frac{1}{4}}\right)(r-1)^{-1}+O_{\infty}(1).
\end{align}
Furthermore, there exist constants $A_{\pm, m \ell}, \widetilde{A}_{\pm, m \ell}\in \C$ such that the following expansions hold uniformly in $m$ in $r\geq 2$:
\begin{equation}
\label{eq:largerasympw}
w_{\pm, m \ell}=\frac{\widetilde{A}_{\pm, m \ell}}{r}+\frac{{A}_{\pm, m \ell}e^{im r_*}}{r}+O_{\infty}(r^{-2}).
\end{equation}
When $2m^2-\Lambda_{+,m \ell} -\frac{1}{4}=0$, there exist again two unique solutions $w_{\pm,m \ell}$, satisfying \eqref{eq:largerasympw} and:
\begin{align}
\label{eq:wfreqexpexcminalt}
w_{-,m \ell}(r)=&\:(r-1)^{-\frac{1}{2}-im}\left[1+O_{\infty}^m((r-1))\right],\\
\label{eq:wfreqexpexcplusalt}
w_{+,m \ell}(r)=&\:(r-1)^{-\frac{1}{2}-im}\log\left(\frac{1}{r-1}\right)[1+O_{\infty}^m((r-1)],\\
\label{eq:estderwexcminalt}
w_{-,m \ell}^{-1}\frac{dw_{-,m \ell}}{dr}=&\:-\frac{1}{2}(1+2im)(r-1)^{-1}+O_{\infty}^m(1),\\
\label{eq:estderwexcplusalt}
w_{+,m \ell}^{-1}\frac{dw_{+,m \ell}}{dr}=&\:-\left(\frac{1}{2}+im+(\log(r-1)^{-1})^{-1})\right)(r-1)^{-1}+\log(r-1)O_{\infty}^m(1).
\end{align}
Finally, we denote with $w_{\infty}(r)$ and $\widetilde{w}_{\infty}(r)$ the linear combinations of $w_{\pm, m \ell}$ such that $w_{\infty,m \ell}(r)=r^{-1}e^{im r_*}+O(r^{-2})$ and $\widetilde{w}_{\infty,m \ell}(r)=r^{-1}+O(r^{-2})$. We can express:
\begin{align*}
w_{\infty,m \ell}(r)=&\:B_+ w_{+,m\ell}(r)+B_-w_{-,m\ell}(r),\\
\widetilde{w}_{\infty,m \ell}(r)=&\:\widetilde{B}_+ w_{+,m\ell}(r)+\widetilde{B}_-w_{-,m\ell}(r),\\
w_{\pm, m \ell}(r)=&\:\widetilde{A}_{\pm, m \ell}\widetilde{w}_{\infty,m \ell}(r)+{A}_{\pm, m \ell}w_{\infty,m \ell}(r),
\end{align*}
with $\widetilde{B}_{\pm }=\widetilde{A}_{\pm}(1+A_{\pm} B_{\pm})$ or $\widetilde{B}_{\mp }=\widetilde{A}_{\pm}^{-1}A_{\pm} B_{\mp}$.

\section{Precise statements of the main theorems}
\label{sec:precstatthm}
In this section, we will state precisely the main theorems in the present paper. To be able to state a precise version of Theorem \ref{thm:tails}, we will introduce the function $\Pi$ that will give the leading-order late-time asymptotics of $\phi$ globally. We will assume without loss of generality that $a>0$. The $a<0$ can be included in the analysis by a redefinition $\varphi_*\to-\varphi_*$. Then we set (without loss of generality) $r_+=1$ and $a=1$.

Recall the definition of $\alpha_{m\ell}$: $\alpha_{m\ell}=\Lambda_{+,m\ell}-2m^2+\frac{1}{4}$. We will consider in this section $m\in \Z\setminus \{-1,0,1\}$ and $\ell\in \N_0$, $l\geq |m|$, such that
\begin{equation}
\label{eq:asmsphev}
\alpha_{m \ell}\neq 0.
\end{equation}

Then we define the following function, which will describe the global, late-time asymptotic behaviour of $\phi$:
\begin{align*}
\Pi(\tau,r,\theta,\varphi)=&\:\sum_{\substack{|m|\leq \ell\leq l_m\\ \alpha_{m \ell}\neq 0}} \Pi_{m \ell}(\tau,r)e^{im \tphi}S_{m \ell}(\theta),
\end{align*}
where $l_m$ is the largest value of $\ell$ such that $\Lambda_{+,m \ell}\leq 2m^2$ (or equivalently, $\alpha_{m \ell}\leq \frac{1}{4}$) and $\Pi_{m \ell}$ is defined as follows: let ${\mathfrak{h}}_{m \ell}\in \C$ and let $A_{\pm},\widetilde{A}_{\pm},B_{\pm},\widetilde{B}_{\pm}$ be the constants appearing in the expansions of the stationary solutions defined in \S \ref{sec:statsol}, then
\begin{enumerate}
\item If $\alpha_{m \ell}<0$, $\widetilde{B}_-\neq 0$ and $\widetilde{B}_+\neq 0$, we define:
\begin{equation}
\label{eq:exprPi1}
 \Pi_{m \ell}(\tau,r):={\mathfrak{h}}_{m \ell}w_{\infty, m\ell}(r)\left[f_+^{m \ell}(\tau,r)+f_-^{m \ell}(\tau,r)\right]+{\mathfrak{h}}_{m \ell}\widetilde{w}_{\infty, m\ell}(r)\left[E_+f_{+,{\rm exp}}^{m \ell}(\tau,r)+E_-f_{-,{\rm exp}}^{m \ell}(\tau,r)\right],
\end{equation}
with $E_+=-\widetilde{B}_-^{-1}B_-$ and $E_-=-\widetilde{B}_+^{-1}B_+$, ${\mathfrak{h}}_{m \ell}\in \C$, and
\begin{align*}
f_{\pm}^{m \ell}(\tau,r):=\left(\tau+\frac{4}{r-1}+1\right)^{-\frac{1}{2}\mp \sqrt{\alpha_{m \ell}}+im}(\tau+1)^{-\frac{1}{2}\mp \sqrt{\alpha_{m \ell}}-im },\\
f_{\pm,{\rm exp}}^{m \ell}(\tau,r):=\left(\tau+\frac{4}{r-1}+e^{r-1}\right)^{-\frac{1}{2}\mp \sqrt{\alpha_{m \ell}}+im}(\tau+1)^{-\frac{1}{2}\mp \sqrt{\alpha_{m \ell}}-im }.
\end{align*}
\item If $\alpha_{m \ell}<0$, $\widetilde{B}_+=0$, then $\widetilde{A}_+=0$, and we define:
\begin{equation}
\label{eq:exprPi2}
 \Pi_{m \ell}(\tau,r):={\mathfrak{h}}_{m \ell}w_{\infty, m\ell}(r)f_+^{m \ell}(\tau,r)={\mathfrak{h}}_{m \ell}A_{+}^{-1}w_{+, m\ell}(r)f_+^{m \ell}(\tau,r).
\end{equation}
\item If $\alpha_{m \ell}<0$, $\widetilde{B}_-=0$, then $\widetilde{A}_-=0$ and we define:
\begin{equation}
\label{eq:exprPi3}
 \Pi_{m \ell}(\tau,r):={\mathfrak{h}}_{m \ell}w_{\infty, m\ell}(r)f_-^{m \ell}(\tau,r)={\mathfrak{h}}_{m \ell}A_{-}^{-1}w_{-, m\ell}(r)f_-^{m \ell}(\tau,r).
\end{equation}
\item If $\alpha_{m \ell}>0$, then we define:
\begin{equation}
\label{eq:exprPi4}
 \Pi_{m \ell}(\tau,r):={\mathfrak{h}}_{m \ell}w_{\infty, m\ell}(r)f_+^{m \ell}(\tau,r).
\end{equation}
\end{enumerate}
Denote
\begin{multline*}
F_{\Pi}:= 2(r^2+1-\hfol \Delta)\tX \Pi|_{\Sigma_{0}}+2\left(1+\frac{1}{2}\hfol (r^2-1)-\frac{1}{2}\sin^2\theta\right)\Phi \Pi|_{\Sigma_{0}}\\
+\left(\hfol(\Delta \hfol-2(r^2+1))-\sin^2\theta\right)K\Pi|_{\Sigma_{0}}-\left(2r-\frac{d}{dr}(\Delta \hfol)\right) \Pi|_{\Sigma_{0}}-\int_{0}^{\infty}\square_g\Pi (\tau',r,\theta,\tphi)\,d\tau'.
\end{multline*}

We consider the following initial data quantity
\begin{multline*}
F_{\rm hom}[\phi]=-2(r^2+1-\hfol \Delta)\tX\phi|_{\Sigma_{0}}-2\left(1+\frac{1}{2}\hfol (r^2-1)-\frac{1}{2}\sin^2\theta\right)\Phi \phi|_{\Sigma_{0}}\\
-\left(\hfol(\Delta \hfol-2(r^2+1))-\sin^2\theta\right)K\phi|_{\Sigma_{0}}+\left(2r-\frac{d}{dr}(\Delta \hfol)\right) \phi |_{\Sigma_{0}}
\end{multline*}
and the following initial data energy norms along $\Sigma_0$.
\begin{multline*}
 D^{(1)}[\phi]=\sum_{|m|\leq \ell \leq l_m}|{\mathfrak{h}}_{m \ell}|^2+\sum_{k+l\leq 4,l\leq 2}\int_{\Sigma_{0}} \mathcal{E}_{1+\epsilon}[K^{k+l}\phi]\, r^2d\sigma dr\\
 +\int_{\Sigma_{0} \cap\{r\geq 2\}}  \sum_{k\leq 2} r^{2-\delta} | LK^{k}\psi|^2+\sum_{k\leq 1}r^{2-\delta} | L((r^2L)K^k\psi)|^2\,d\sigma dr\\ 
 +\sum_{\substack{k\leq 2\\j\leq 1}}\sup_{r'\leq 2}\int_{\Sigma_0\cap\{r=r'\}}|\tX^j((r-1)\tX)^kF_{\rm hom}[\phi])|^2\,d\sigma+\sup_{r'\geq 2}\int_{\Sigma_0\cap\{r=r'\}}|r  F_{\rm hom}[\phi]|^2\,d\sigma\\
 +\int_{\Sigma_0\cap\{r\geq 2\}} r^2 |\tX F_{\rm hom}[\phi]|^2\,d\sigma dr
\end{multline*}
and
\begin{equation*}
 D_2^{(1)}[\phi]=\int_{\Sigma_0 \cap\{r\geq r_I\}}r  |L(r^2L)^2 K \hpsi_{|m|\geq 2}|^2\,d\sigma dr+D^{(1)}[\phi]+D^{(1)}[K\phi].
\end{equation*}

Let $r_0>1$ and let $\mathfrak{h}_{m\ell}$ be chosen such that the following cancellation occurs:
\begin{equation*}
\int_{1}^{\infty}w_{\infty,m \ell}(r)e^{-im\int_{r_0}^r \frac{x+1}{x-1}\,dx}((F_{\Pi})_{m \ell}+(F_{\rm hom})_{m \ell}[\phi])\,dr\neq 0.
\end{equation*}
To guarantee this can be done for general initial data, we need that (for $\mathfrak{h}_{m\ell}\neq 0$):
\begin{equation*}
\int_{1}^{\infty}w_{\infty,m \ell}(r)e^{-im\int_{r_0}^r \frac{x+1}{x-1}\,dx}(F_{\Pi})_{m \ell}\,dr\neq 0.
\end{equation*}
We will assume, without loss of generality, that the above condition holds for all $\ell\leq l_m$. If it did not, we could modify the functions $f_{\pm }^{m\ell}$ and $f_{\pm,{\rm exp}}^{m \ell}$ by replacing $\left(\tau+\frac{4}{r-1}+1\right)$ with $\left(\tau+\frac{4}{r-1}+\beta\right)$ and $\left(\tau+\frac{4}{r-1}+e^{r-1}\right)$ with $\left(\tau+\frac{4}{r-1}+\beta e^{r-1}\right)$, respectively, and choose $\beta\in \R$ appropriately, as this modification does not affect the leading-order asymptotic behaviour of $f_{\pm }^{m\ell}$ and $f_{\pm,{\rm exp}}^{m \ell}$.

Observe that $\mathfrak{h}_{m\ell}$ is therefore a constant multiple of $\int_{1}^{\infty}w_{\infty,m \ell}(r)e^{-im\int_{r_0}^r \frac{x+1}{x-1}\,dx}(F_{\rm hom})_{m \ell}[\phi]\,dr$.

Define
\begin{equation*}
	\mathfrak{h}_m:=\sqrt{\sum_{\substack{|m|\leq \ell\leq l_m\\ \alpha_{m\ell}<0}}|{\mathfrak{h}}_{m \ell}|^2}.
		\end{equation*}
		
		By Lemma \ref{lm:negalpha}, the set of $\ell$ for which $\alpha_{m\ell}<0$ is non-empty for all $m$, so $\mathfrak{h}_m\neq 0$ for generic data (i.e.\ when restricted to the azimuthal number $m$, this will be the complement of a codimension $l_m+1-|m|$ set of data).

We first state the following precise version of Theorem \ref{thm:tails}:
\begin{theorem}[Late-time tails and stability for azimuthal modes on extremal Kerr; precise version]
			\label{thm:latetimeasym}
			Let $|m|\geq 2$ and consider solutions $\phi_m$ to \eqref{eq:waveeq} arising from initial data satisfying\\ $\sum_{2k+l\leq 2}D_2^{(1)}[Q^lK^k\phi_m]<\infty$. Assume that $\alpha_{m\ell}\neq 0$ for all $\ell$ and define
			\begin{equation*}
			\nu_m:=\min \left( \{\alpha_{m \ell}\,|\, \alpha_{m \ell}> 0\}\cup\{1\}\right).
			\end{equation*}
			Assume that for all $1<r_0<r_1$
			\begin{align}
			\label{asmphi:thm1}
\int_{0}^{\infty}\int_{\Sigma_{\tau}\cap\{r\leq r_1\}} |\phi_m|^2\,d\sigma dr d\tau<&\:\infty,\\
	\label{asmphi:thm2}
 \int_{0}^{\infty}\int_{\Sigma_{\tau}\cap\{r_0\leq r\leq r_1\}} |X\phi_m|^2+|T \phi_m|^2+ |\snabla_{\s^2}\phi_m|^2\,d\sigma dr d\tau<&\:\infty.
\end{align}
			Let $0<\eta<\frac{\nu_m}{4}$ be arbitrarily small. Then there exists a constant $C_m>0$, such that for all $\tau\geq 0$:
			\begin{align*}
			\left|\phi_m-\Pi \right|(\tau,r,\theta,\tphi)&\\
			\leq &\:C_mr^{-1}((\tau+1)(1-r^{-1})+1)^{-\frac{1}{2}}(\tau+1)^{-\frac{1}{2}-\frac{\nu_m}{4}+\eta}\sqrt{\sum_{2k+l\leq 2}D_2^{(1)}[Q^lK^k\phi]},\\
		\int_{\Sigma_{\tau}} \mathcal{E}_{2}[\phi-\Pi]\, r^2d\sigma dr\leq&\:  C_m (1+\tau)^{-\nu_m+\eta} D^{(1)}[\phi].
			\end{align*}
			\end{theorem}
			\begin{proof}
			The theorem is proved by combining Proposition \ref{prop:nondegedecay} and Corollary \ref{cor:latetimeasym}, which are proved in \S \ref{sec:decay}.
			\end{proof}

We will now arrive at a precise version of Theorem \ref{thm:azimuthal}. First, Theorem \ref{thm:preciseinstab} below shows that if the global assumptions \eqref{asmphi:thm1} and \eqref{asmphi:thm2} hold, the azimuthal instability \textbf{(a)} must hold.
\begin{theorem}[The existence of azimuthal instabilities on extremal Kerr; precise version]
	\label{thm:preciseinstab}
	Let $|m|\geq 2$ and consider solutions $\phi_m$ satisfying the assumptions from Theorem \ref{thm:latetimeasym}. If there exists an $|m|<\ell_*\leq l_m$ such that $\alpha_{m\ell_*}=0$, restrict to initial data for which correspondingly $\mathfrak{h}_{m\ell_*}=0$.
	\begin{enumerate}[label=\emph{(\roman*)}]
	\item  Let $\eta>0$ be suitably small. Then there exist constants $c_m>0$ and $C_m>0$, such that
	\begin{equation*}
	\int_{\Sigma_{\tau}}\mathcal{E}_2[\phi_m]\,r^2d\sigma dr\geq c_m\mathfrak{h}_m^2-C_m(1+\tau)^{-\nu_m+\eta}D^{(1)}[\phi_m].
	\end{equation*}
	\item Let $k\in \N$. There exists a constant $c_m=c_m(m,k)>0$ and a monotonically increasing sequence $\{\tau_n\}$ in $ \R_+$ with $\tau_n\to\infty$ along which
	\begin{equation*}
		||Y^k \phi_m||^2_{L^{\infty}(\Sigma_{\tau_n}\cap \mathcal{H}^+)}\geq\frac{1}{4\pi}\int_{\Sigma_{\tau_n}\cap \mathcal{H}^+}|Y^k \phi_m|^2\,d\sigma\geq  c_{m} \mathfrak{h}_m^2(1+\tau_n)^{1+2k}.
	\end{equation*}
		\end{enumerate}
	\end{theorem}
	\begin{proof}
	The statements (i) and (ii) follow from a combination of Propositions \ref{eq:energyinstab}, \ref{prop:pointwinstab} and \ref{prop:hopointwinstab}.
	\end{proof}
	
	From Theorem \ref{thm:preciseinstab}, it follows easily deduce a general instability result, which includes the instability scenarios (b) and (c) of Theorem \ref{thm:latetimeasym}.
	\begin{corollary}
		There exist smooth, compactly supported initial data $(\phi_*|_{\Sigma},K\phi_*|_{\Sigma})$, such that the corresponding solution $\phi_*$ satisfies at least one of the following instabilities:
		\begin{enumerate}[label=\emph{(\alph*)}]
	\item Let $k\in \N$. There exists a constant $b>0$, depending on initial data, and two sequences of times $\{\tau_n\}$ and $\{\tau_n'\}$ with $\tau_n,\tau_n'\to \infty$, such that
\begin{align}
\label{eq:corinstab1}
	\int_{\Sigma_{\tau_n'}}\mathcal{E}_2[\phi_*]\,r^2d\sigma dr&\: \geq b^2,\\
	\label{eq:corinstab2}
||Y^k \phi_*||_{L^{\infty}(\Sigma_{\tau_n}\cap \mathcal{H}^+)}& \:\geq b \tau_n^{\frac{1}{2}+k}.
\end{align}
	\item There exists an $|m|\geq 2$ for which
	\begin{equation}
	\label{eq:corinstab3}
		(\phi_*)_m\notin L^2_{\tau}L^2(\Sigma_{\tau}).
	\end{equation}
	\item There exists an $|m|\geq 2$ and radii $1<r_0\leq r_1<\infty$ for which
\begin{equation}
\label{eq:corinstab4}
		|(Y\phi_*)_m|+|(K\phi_*)_m|+|(\snabla_{\s^2}\phi_*)_m|\notin L^2_{\tau}L^2(\Sigma_{\tau}\cap\{r_0\leq r\leq r_1\}).
	\end{equation}
		\end{enumerate}
Furthermore, the subset of smooth compactly supported data which does \underline{not} satisfy any of the instabilities (a), (b) and (c) has infinite codimension.
	\end{corollary}
	\begin{proof}
	Suppose that there do not exist smooth, compactly supported data for which the instabilities (b) and (c) hold. Then	 for all $|m|\geq 2$, smooth compactly supported data supported on the $m$th azimuthal mode must satisfy \eqref{asmphi:thm1} and \eqref{asmphi:thm2}. We then take $(\phi_*|_{\Sigma},K\phi_*|_{\Sigma})$ to be any smooth and compactly supported data with $\mathfrak{h}_{m|m|}\neq 0$. If $\alpha_{m\ell_*}=0$ for some $\ell_*>|m|$, we invoke the additional restriction $\mathfrak{h}_{m\ell_*}=0$. Then the instability (a) follows from Theorem \ref{thm:preciseinstab}, with $b$ proportional to $\mathfrak{h}_m$.
	
	To conclude infinite codimensionality, we first observe that $|m|\geq 2$ above was arbitrary, so we can in fact find an infinite linearly independent set of data, supported on different $m$, with each element resulting in instability (a).
	
	Consider now a solution $\phi$ arising from smooth, compactly supported data $(\phi|_{\Sigma},K\phi|_{\Sigma})$, for which none of the instabilities (a)--(c) hold. By linearity of the wave equation, the data $(\phi|_{\Sigma}+\lambda \phi_*|_{\Sigma},K\phi|_{\Sigma}+\lambda K\phi_*|_{\Sigma})$, with $\lambda\in \R$, correspond to the solution $\phi+\lambda \phi_*$. If $\phi_*$ satisfies (b) or (c), then so will $\phi+\lambda \phi_*$ for any $\lambda\neq 0$. 
	
	Suppose that $\phi_*$ satisfies (a) with the sequences $\{\tau_n\}$ and $\{\tau_n'\}$. If there exists a $\lambda_*\in \R$ such that $\phi+\lambda_* \phi_*$ does not satisfy \eqref{eq:corinstab1} (along any subsequence of $\{\tau_n'\}$), then there exists a subsequence $\{\tau_{n_l}'\}$ of $\{\tau_n'\}$ along which $\int_{\Sigma_{\tau_{n_l}'}} \mathcal{E}_2[\phi+\lambda_*\phi_*]\,r^2d\sigma dr\to 0$ as $l\to \infty$. By taking a further subsequence, it follows that $\int_{\Sigma_{\tau}} \mathcal{E}_2[\phi+\lambda\phi_*]\,r^2d\sigma dr\geq \frac{b^2(\lambda-\lambda_*)^2}{2}$ along a subsequence of $\{\tau_n'\}$. 
	
	Similarly, we can conclude that if for some $k\in \N$ and $\lambda_*\in \R$, $\phi+\lambda_* \phi_*$ does not satisfy \eqref{eq:corinstab2} (along any subsequence), then there exist a sequence of times along which \eqref{eq:corinstab2} does hold for $\phi+\lambda \phi_*$ with $b$ replaced by $\frac{b}{2}|\lambda-\lambda_*|$. Hence, the subset of initial data for which (a), (b) and (c) all do not hold has codimension at least 1. Since we can repeat this argument with $\phi_*$ supported on different azimuthal modes, this subset must have infinite codimension.
	\end{proof}

	\section{A sketch of the proofs and an overview of the main new ideas}
	\label{sec:sketchproof}
	In this section, we give an overview of the main steps in the proofs of Theorems \ref{thm:latetimeasym} and \ref{thm:preciseinstab}. We will highlight the new techniques that are introduced in the present paper and sketch how the novel difficulties listed in \S \ref{sec:maindiff} are overcome. We will again assume for the sake of convenience that $r_+=1$ and $a=1$.
	\subsection{Part I: Integrated energy estimates in frequency space}
	The first step towards proving energy decay is to control spacetime integrals of the form
	\begin{equation*}
	\int_{0}^{\infty} \int_{\Sigma_{\tau}} r^{-p}(1-r^{-1})^q [|\phi_m|^2+|\partial \phi_m|^2]\,r^2d\sigma dr,
	\end{equation*}
	with $p,q>0$ appropriate exponents, in terms of a weighted energy flux along $\Sigma_0$. In view of the presence of superradiance and trapped null geodesics that cannot be easily localized in physical space, the strategy for controlling the above integrals involves a Fourier transform in time $t$, which introduces a time frequency $\omega$. Here, we make use of the qualitative assumptions \eqref{asmphi:thm1} and \eqref{asmphi:thm2}, which guarantee sufficient integrability of $\phi$ and its derivatives to make sense of the Fourier transform. The main estimates in Part I therefore take place in \emph{Fourier space}.\\
	\\
	\paragraph{\textbf{Carter separation and the frequency triple $(\omega, m, \Lambda)$}}
	To analyze the Fourier transform, we make use of the \emph{Carter separation property} for \eqref{eq:waveeq}, which implies that if we decompose the Fourier transform of $\phi$ in $t$ into appropriate angular functions $e^{im\varphi}S_{m\ell}(\theta;a\omega)$, we can separate variables and the wave equation becomes equivalent to an ODE for the coefficients $u_{m \ell}(r_*;\omega)$ in the expansion of the Fourier transform, which takes the form of a time-independent Schr\"odinger equation:
	\begin{equation}
	\label{eq:introode}
	u''+(\omega^2-V)u=H,
	\end{equation} 
	with $V=V_{\ell m}(r_*;\omega)$ a potential function and the inhomogeneity $H=H_{m\ell}(r;\omega)$ determined by the choice of initial data for $\phi$ and $(\cdot)'=\frac{d}{dr*}(\cdot)$. The angular functions $e^{im\varphi}S_{m\ell}$ are eigenfunctions of oblate spheroidal operators with eigenvalues $\lambda_{m\ell}$, with $\lambda_{m\ell}=\ell(\ell+1)$ in the special case $a=0$; see \S \ref{sec:sphere}. The translated angular frequency $\Lambda_{m\ell}=\lambda_{m\ell}+a^2\omega^2$ forms part of the frequency triple $(\omega,m,\Lambda)$ that can be used to characterize the relevant phenomena like superradiance and trapping, which are reflected in the different shapse of the potential $V$, depending on $m,\omega,\Lambda$. 
	
	Integrated energy estimates are obtained at the level of $u$ by multiplying the equation \eqref{eq:introode} with $f(r;\omega)u+g(r;\omega)u'$, for appropriate (frequency dependent) choices of $f,g$ and then integrating by parts in $r$. Different frequency ranges are dealt with using different multipliers that are adapted to the monotonicity properties of $V$ in the frequency regimes under consideration.
	\\
	\paragraph{\textbf{Key frequency regimes}}
	The main differences with the frequency-space analysis in the sub-extremal setting from \cite{part3} appear for \emph{near-superradiant frequencies}:
	\begin{equation*}
	|\omega-m\upomega_+|\leq \epsilon \sqrt{\Lambda }
	\end{equation*}
	with $\epsilon\ll 1$. These frequencies can be separated into two sub-regimes: \textbf{1)} $\Lambda\gg m^2$ and \textbf{2)} $\Lambda \lesssim m^2$. An important fact about regime \textbf{1)} is that it does not feature the difficulties associated to the trapping of null geodesics. The main difference with the $|a|<M$ case is the need for \emph{additional degeneracies} in the estimates at the event horizon, related to the higher degeneracy of $V-\omega^2$ at $\omega=m\upomega_+$ in the $|a|=M$ case. The relevant estimates are proved in \S \ref{eq:Lambdasr}.
	
	In general, an important new difficulty in the $|a|=M$ case is that the phenomenon of trapping and superradiance are coupled in regime \textbf{2)}. However, since we are assuming $m$ to be bounded or fixed, regime  \textbf{2)} is a \underline{bounded frequency regime}. Trapping of null geodesics affects unbounded $\Lambda$ and it is therefore not an issue in regime \textbf{2)} in the present paper. See the comments on the unbounded $m$ case in \S \ref{sec:remainingq}.
	\\
	\paragraph{\textbf{Bounded, near-superradiant frequencies}}
	The main new techniques introduced in the frequency-space analysis in the present paper are concerned with the regime \textbf{2)} for bounded $m$. The strategy here is to express:
	\begin{equation}
	\label{eq:wronskidintro}
u(r)=W^{-1}\left[u_{\rm inf}(r)\int_{1}^{r} u_{\rm hor}(r')H(r')\,\frac{r'^2+a^2}{\Delta}dr'+u_{\rm hor}(r)\int_{r}^{\infty} u_{\rm inf}(r')H(r')\,\frac{r'^2+a^2}{\Delta}dr'\right],
	\end{equation}
	with $u_{\rm hor}$ and $u_{\rm inf}$ solutions to \eqref{eq:introode} with $H=0$, with appropriate boundary conditions at $r=r_+$ and $r=\infty$, respectively. Estimates for $u$ follow from estimates for $u_{\rm hor}$ and $u_{\rm inf}$, combined with an estimate for $W^{-1}$. The main difficulties, when compared with the sub-extremal setting, are the following:
	\begin{enumerate}
	\item As part of the proof of mode stability, \cite{costa20} derived uniform bounds for $W^{-1}$. In contrast with the analogous estimate for sub-extremal Kerr in \cite{sr15}, these bounds feature constants that blow up as $\omega\to m\upomega_+$. Indeed, in \cite[Proposition 6.3]{costa20} it is shown that
	\begin{equation*}
	|W^{-1}|\lesssim \frac{1}{\sqrt{|\omega-m\upomega_+|}}.
	\end{equation*}
	\item In the limit $\omega\to m\upomega_+$, $(V-\omega^2)(r)$ has a quadratic zero at $r=r_+$, whereas it is linear in the sub-extremal case. For this reason, $r=r_+$ changes character from an \emph{irregular} singular point to a \emph{regular} singular point of \eqref{eq:introode}.
	\end{enumerate}
Since the singular point $r=r_+$ or $r_*=-\infty$ of the ODE \eqref{eq:introode} changes character at $m\upomega_+$, one can not appeal directly to standard ODE estimates to derive pointwise bounds for $u_{\rm inf}$ and $u_{\rm hor}$ for $r\in (r_+,\infty)$. By considering instead the resulting ODE with respect to the rescaled variable $x:=\frac{|\omega-m\upomega_+| r_*}{|m|}$, it is possible to obtain uniform pointwise estimates in the near-horizon region $|r_*|\gtrsim \frac{|m|}{|\omega-m\upomega_+| }$ that do not breakdown in the limit $\omega\to m\upomega_+$ (however, the corresponding region in space shrinks in the limit).

We are left with obtaining estimates for the homogeneous solutions $u_{\rm inf}$ in the region $|r_*|\lesssim \frac{|m|}{|\omega-m\upomega_+| }$. Here we face the difficulty that $V-\omega^2$ \underline{changes sign arbitrarily close to $r=r_+$} for $|\omega-m\upomega_+|$ sufficiently small. For this reason, it is significantly more difficult to derive integrated estimates in the bounded frequency regime with $|\omega-m\upomega_+|\ll 1$, compared to the bounded frequency regime with $|\omega|\ll 1$. Indeed at $\omega=0$, the singular point at $r=\infty$ similarly changes character from an irregular to a regular singular point. However, in that case, $V(r)\sim \Lambda r^{-2}$ as $r\to \infty$, so one can exploit that $V-\omega^2$ is positive for large but finite $r$-intervals and $\omega^2$ suitably small.

The key idea for resolving the difficulty of the change in sign of $V-\omega^2$ is to rescale $u$ by considering
\begin{equation*}
\check{u}_{\pm,\rm inf}=\frac{u_{\rm inf}}{\mathfrak{w}_{\pm }},
\end{equation*}
	where $\mathfrak{w}_{\pm }$ give the radial part of the $K$-independent solutions to \eqref{eq:waveeq} that were introduced in \S \ref{sec:statsol}. In frequency space, $\mathfrak{w}_{\pm }$ can be considered solutions to \eqref{eq:introode} with $\omega=m\upomega_+$, with $H\equiv 0$ and appropriate boundary conditions at $r=r_+$. The effect of this rescaling is that the ODE satisfied by $\check{u}_{\pm,\rm inf}$ features a potential which factors of $\omega-m\upomega_+$ that has a good sign in the region $|r_*|\lesssim \frac{|m|}{|\omega-m\upomega_+| }$. A physical space analogue of this rescaling trick was used in the integrated energy estimates of \cite{gaj22a}. The analysis of $\check{u}_{\pm,\rm hor}$ and $\check{u}_{\pm,\rm inf}$ is carried out in \S \ref{sec:homsoln}.
	
	Employing the identity \eqref{eq:wronskidintro}, together with estimates for $u_{\rm inf}$ in the region $|r_*|\gtrsim \frac{|m|}{|\omega-m\upomega_+| }$ and global estimates for $u_{\rm inf}$, obtained via the strategy above, then results in estimates for $u(r)$ that are also valid in the region of the form $|r_*|\gtrsim \frac{|m|}{|\omega-m\upomega_+| }$. Extending these to  $|r_*|\lesssim \frac{|m|}{|\omega-m\upomega_+| }$ involves considering again the rescaled variables $\check{u}_{\pm}=\frac{u}{\mathfrak{w}_{\pm }}$, as discussed in the estimates for $\check{u}_{\pm,\rm inf}$.
	
	The main conclusion of the above estimates in frequency space, when combined with Plancherel's theorem, is the following integrated energy estimate and energy boundedness estimate: for all $\epsilon>0$, there exists a constant $C_m>0$, such that
	\begin{multline}
	\label{eq:intromorawetz1}
	\int_{\Sigma_{\tau_2}}\mathcal{E}_{1-\epsilon}[\phi]\,r^2d\sigma dr +\int_{\tau_1}^{\tau_2} \int_{\Sigma_{\tau}\cap\{r\leq 2M\}}(1-Mr^{-1})^{\epsilon}\mathcal{E}_0[\phi]+(1-Mr^{-1})^{-1+\epsilon}|K\phi|^2\,d\sigma dr d\tau\\
				\leq C_m \int_{\Sigma_{\tau_1}}\mathcal{E}_{1+\epsilon}[\phi]\,r^2d\sigma dr,
					\end{multline}
					which is derived in \S \ref{sec:intenestphys}.
					
	Note in particular that the energy controlled in the spacetime integral on the LHS degenerates at $r=M$ and the energy flux term on the LHS features an energy with an additional factor $(1-Mr^{-1})^{2\epsilon}$, when compared to the energy flux term on the RHS. There is no loss of derivatives on the RHS due to the presence of trapped null geodesics, as for bounded $m$, there is no trapping in the region $r\leq 2M$.  We complement the above estimate with an integrated energy estimate that does feature the region containing trapped null geodesics and hence loses a derivative on the RHS: let $r_I>0$ be arbitrarily large, then
	\begin{equation}
	\label{eq:intromorawetz2}
	\int_{\tau_1}^{\tau_2} \int_{\Sigma_{\tau}\cap\{2M\leq r\leq r_I\}}\mathcal{E}_2[\phi]\,d\sigma dr d\tau\leq C\sum_{k\leq 1} \int_{\Sigma_{\tau_1}}\mathcal{E}_{1+\epsilon}[K^k\phi]\,r^2d\sigma dr.
					\end{equation}

	\subsection{Part II: Energy decay, $K$-derivatives and $K$-integrals}
	The remaining part of the proof concerns energy estimates in physical space, starting from the integrated and energy boundedness estimates \eqref{eq:intromorawetz1} and \eqref{eq:intromorawetz2}.\\
	\\
	\paragraph{\textbf{More $K$-derivatives lead to more decay }}
	The first new difficulty one faces when trying to convert the integrated energy estimates \eqref{eq:intromorawetz1} and \eqref{eq:intromorawetz2} into inverse-polynomial decay estimates for (weighted) energies along $\Sigma_{\tau}$ is the fact that there is no hierarchy of estimates with weights $(r-M)^{-p}$ near $r=M$, as is the case in the region $r\gg M$ for energies with $r^p$-weights, or in the extremal Reissner--Nordstr\"om setting near $r=M$. 
	
	Indeed, by applying the mean-value theorem to \eqref{eq:intromorawetz1} and \eqref{eq:intromorawetz2}, combined with an $r^p$-estimate in $r\gg M$, the best that one can get is $\tau^{-1}$ decay for the strongly degenerate energy with energy density $(1-Mr^{-1})^{\epsilon}\mathcal{E}_0[\phi]$ along some sequence $\{\tau_n\}$ of times. To get the desired pointwise estimates, one would need decay for all times $\tau\geq 0$, also for less degenerate energies, with better decay rates.
	
	The key idea here is that there is an \underline{alternative way} to obtain a hierarchy of additional, less degenerate estimates near $r=M$, namely,  by replacing $\phi$ with the time derivative $K\phi$. Indeed, from \eqref{eq:waveeq} if follows that, schematically,
	\begin{equation*}
	\mathcal{E}_p[K\phi]\lesssim \mathcal{E}_{p-2}[(r-M)Y\phi]+\mathcal{E}_{p-2}[\snabla_{\s^2}\phi]+\mathcal{E}_{p-2}[K\phi].
	\end{equation*}
	Hence, $K$ derivatives can be traded for additional degenerate factors of $r-M$, at the expense of requiring estimates additionally estimates for the derivatives $(r-M)Y\phi$ and $\snabla_{\s^2}\phi$. Fortunately, the estimates for $\phi$ extend naturally to $(r-M)Y\phi$ and $\snabla_{\s^2}\phi$. In \cite{gaj22a}, energy decay with the desired rates is similarly obtained by replacing $\phi$ with time derivatives $T^k\phi$. The additional difficulty in the present setting is that, in contrast with \cite{gaj22a}, there is no (degenerate) energy boundedness estimate available without additional degeneracies on the LHS.
	
	Note that the decay rate obtained by repeatedly applying $K$ is limited by the length of the hierarchy of $r^p$-estimates near infinity, so it is not possible to obtain arbitrarily fast decay for a suitable number of $K$-derivatives. The energy decay estimates for $K$-derivatives of $\phi$ are derived in \S \ref{sec:edecay}. \\
	\\
		\paragraph{\textbf{The difference function $\phi-\Pi$}}
		It remains to apply the energy decay estimates for $K\phi$ to obtain energy decay estimates for $\phi$. In view of the analogous steps in \cite{gaj22a}, it would be natural to think of $\phi$ as a $K$-derivative: $\phi=K(K^{-1}\phi)$, where, formally,
		\begin{equation*}
		K^{-1}\phi(\tau,r,\theta,\tphi)=-\int_{\tau}^{\infty}\phi(\tau',r,\theta,\tphi)\,d\tau'
		\end{equation*}
		and use the energy decay estimates for $K$-derivatives from the previous step.
		
		This strategy does not quite work. Indeed, in view of the decay rates in Theorem \ref{thm:tails}, we will prove $\tau^{-1}$-decay for $\phi$ away from $r=M$, but this is not integrable, so the time integral does not make sense.
		
		Instead, we consider
		\begin{equation*}
		\hphi=\phi-\Pi
		\end{equation*}
		with $\Pi$ the function defined in \S \ref{sec:precstatthm}, which will encode the precise, global leading-order late-time behaviour of $\phi$. As can be seen in Theorem \ref{thm:latetimeasym}, $\hphi$ \emph{does} decay with an integrable rate in time and we are able to make sense of $K^{-1}\hphi$. Since $\hphi$ satisfies an inhomogeneous wave equation \eqref{eq:inhomwaveeq} with $G=-\square_g\Pi$, the inhomogeneity has to decay suitably fast in $\tau$ and $r$ in order to be able to repeat the integrated energy estimates and energy decay estimates for $\hphi$ instead of $\phi$.
		
		The consideration of $\hphi$ moreover allows us to go beyond (sharp) upper bound decay estimates for $\phi$ and to derive the \emph{precise} leading-order late-time behaviour of $\phi$. Indeed, if we can show that $\hphi$ decays \emph{faster} than $\Pi$, it means that the leading-order late-time asymptotics or $\phi$ coincide with the leading-order late-time asymptotics of $\Pi$.
		\\
		\\
			\paragraph{\textbf{$K$-inversion and elliptic estimates}}
	Since we do not know \emph{a priori} that $\hphi$ decays suitably fast to be able to make sense of $K^{-1}\phi$, defined as an integral in time, we instead construct the initial data for $K^{-1}\phi$ by solving an elliptic type equation:
	\begin{equation*}
	\mathcal{L}K^{-1}\hphi|_{\Sigma_0}=F[\phi|_{\Sigma_0}],
	\end{equation*}
	where $F[\phi|_{\Sigma_0}]$ can be split as follows:
	\begin{equation*}
	F[\phi|_{\Sigma_0}]=F_{\rm hom}[\phi]+F_{\Pi},
	\end{equation*}
	with $F_{\rm hom}[\phi|_{\Sigma_0}]$ depending on the initial data for $\phi$ along $\Sigma_0$ and $F_{\Pi}$ depending only on $\Pi$.
	
	More concretely, we construct the projections $(K^{-1}\hphi)_{\ell m}$ to the angular functions $S_{m\ell}(\theta)$, defined in \S \ref{sec:sphere}, and derive uniform elliptic estimates for the operator $\mathcal{L}$ for large values of $\ell$ to be able to sum $(K^{-1}\hphi)_{\ell m}$ over $\ell$. These elliptic estimates are adapted to the vector fields $(K,\tX)$ and require more decay of the initial data for $\phi$ as $r\to \infty$ than elliptic estimates adapted to the $(T,X)$ vector fields.
	
	The construction of $(K^{-1}\hphi)_{\ell m}$ amounts to solving an ODE. This ODE can easily be solved by considering instead the rescaled variables
	\begin{equation*}
	\frac{(K^{-1}\hphi)_{\ell m}}{w_{\infty, m \ell}},
	\end{equation*}
	with $w_{\infty, m \ell}(r)$ the radial parts of the $K$-invariant stationary solutions introduced in to \eqref{eq:waveeq} that were introduced in \S \ref{sec:statsol}, which are closely related to the functions $\mathfrak{w}_{\pm}(r)$ described in Part I above; see \S \ref{sec:statsol}. The ODE for these rescaled variables can be put in the folllowing form for $a=M=1$:
	\begin{equation*}
\tX\left( \Delta  w_{\infty,m \ell}^{2}e^{-im\int_{r_0}^r \frac{x+1}{x-1}\,dx}\tX \left(\frac{(K^{-1}\hphi)_{\ell m}}{w_{\infty, m \ell}}\right)\right)=e^{-im\int_{r_0}^r \frac{x+1}{x-1}\,dx}w_{\infty,m \ell}F_{m \ell},
\end{equation*}
	and can therefore be solved by simply integrating twice in the $\tX$-direction. The requirement that the integral of the RHS over $r\in (M,\infty)$ vanishes for $\alpha_{m\ell}\leq \frac{1}{4}$, fixes the constants $\mathfrak{h}_{m\ell}$ appearing in the definition of $\Pi$ and therefore also the initial-data-dependent factors in the coefficients in front of the late-time tails of $\phi$. The decay rates in the late-time tails for $\phi$ are then encoded in the behaviour of $(K^{-1}\hphi)$ as $r\downarrow M$. The data corresponding to the $K$-integral of $\phi_m$ and the corresponding elliptic estimates are derived in \S \ref{sec:ellipticKinv}.
	
	\subsection{Part III: Instability follows from stability}
Part I and Part II above lead to a proof of Theorem \ref{thm:latetimeasym}, which is the main ingredient in the proof of Theorem \ref{thm:preciseinstab}.

First, non-decay of the non-degenerate energy $\int_{\Sigma_{\tau}}\mathcal{E}_2[\phi]\,r^2d\sigma dr$ follows from decay of the energy of the difference function $\int_{\Sigma_{\tau}}\mathcal{E}_2[\hphi]\,r^2d\sigma dr$, combined with
non-decay of
\begin{equation*}
\int_{\Sigma_{\tau}}\mathcal{E}_2[\Pi]\,r^2d\sigma dr,
\end{equation*}
which in turn follows by plugging in the explicit expressions for $\Pi$.

To obtain pointwise blow-up of $Y\phi_m$ we use that $Y\phi_m$ satisfies the following equation along $\mathcal{H}^+$, obtained by simply evaluating \eqref{eq:waveeq} along $\mathcal{H}^+$:
\begin{equation*}
K(-4Y \phi_{m \ell})=(\Lambda_{m\ell} +m^2-im)\phi_{m \ell}+(\sin^2\theta (K^2\phi- im K\phi))_{m \ell}+2(1+ im) K\phi_{m \ell}.
\end{equation*}
Since Theorem \ref{thm:latetimeasym} gives the precise late-time tails of $\phi_{m \ell}$ and implies that $K^2\phi$ and $K\phi$ decay faster, we integrate the above equation to obtain the precise late-time behaviour of $Y\phi_m$. Due to the possible oscillation present in the late-time behaviour of $|Y\phi_m|$, it will not satisfy a uniform lower bound, but we can identify a sequence of times $\{\tau_n\}$ along which blow-up occurs. Stronger blow-up of higher-order $Y$-derivatives can similarly be obtained (along a sequence of times) by plugging in the late-time asymptotics of lower-order $Y$-derivatives into the equation along $\mathcal{H}^+$ that are obtained by commuting \eqref{eq:waveeq} with $Y^k$ and evaluating the resulting equation along $\mathcal{H}^+$.

The instability results in the paper are all derived in \S \ref{sec:instab}.
\section{Frequency-space analysis}
\label{sec:freqspacean}
In this section, we derive the estimates in frequency space that are required to obtain integrated energy estimates. In order to be able to take a Fourier transform, we will introduce a \emph{global assumption} on solutions $\phi$ arising from smooth and compactly supported initial data on $\Sigma_0$. Without loss of generality, we will assume in this section that $r_+=1$ and $a\geq 0$. Indeed, the $a<0$ can be included in the analysis by transforming $\varphi_*\to-\varphi_*$. 

First we make the following integrability assumption on the inhomogeneity $G$:
\begin{assumption}
\label{asm:G}
Let $m\in \Z$ and let $G_m$ denote the projection of $G$ to the $m$-th azimuthal mode. Then for all $n_1,n_2,n_3\in \N_0$:
\begin{align}
\label{eq:assmG}
\int_{0}^{\infty}\int_{\Sigma_{\tau}} |\rho^2G_m|^2\,d\sigma dr d\tau<&\:\infty,\\
\label{eq:assmG2}
\sup_{r\in [r_0,r_1]}\int_{0}^{\infty} \int_{\s^2} |\snabla_{\s^2}^{n_1}X^{n_2}T^{n_3}G_m|^2(\tau,r,\theta,\varphi)\,d\sigma d\tau<&\:\infty
\end{align}
\end{assumption}

We make the following global assumption on solutions $\phi$ to \eqref{eq:inhomwaveeq}:
\begin{assumption}
\label{asm:phi}
Let $\phi$ be a solution to \eqref{eq:inhomwaveeq} with $G=0$, arising from initial data $(\phi_0,\phi_0')\in C_c^{\infty}(\Sigma_0)\times C_c^{\infty}(\Sigma_0)$. Then for all $m\in \Z$ and $r_1>r_0>r_+$:
\begin{align}
\label{eq:assmphi1}
\int_{0}^{\infty}\int_{\Sigma_{\tau}\cap\{r\leq r_1\}} |\phi_m|^2\,d\sigma dr d\tau<&\:\infty,\\
\label{eq:assmphi2}
 \int_{0}^{\infty}\int_{\Sigma_{\tau}\cap\{r_0\leq r\leq r_1\}} |X\phi_m|^2+|T \phi_m|^2+ |\snabla_{\s^2}\phi_m|^2\,d\sigma dr d\tau<&\:\infty.
\end{align}
\end{assumption}

Note that for $m=0$ and $|a|=M$, the Assumptions \ref{eq:assmphi1} and \ref{eq:assmphi2} hold by the decay results in \cite{aretakis3}. For $|a|<M$, the Assumptions \ref{eq:assmphi1} and \ref{eq:assmphi2} hold for all $m\in \Z$ by the decay results in \cite{part3}.

\subsection{Carter separation}
\label{sec.carter}
We first introduce the class of \emph{sufficiently integrable functions}:
\begin{definition}
\label{def:suffint}
A function $f: \mathring{\mathcal{M}}_{M,a}\to \C$ is \emph{sufficiently integrable} if: for all $n_1,n_2,n_3\in \N_0$ and $r_1>r_0>r_+$:
\begin{align}
\label{eq:suffint1}
\sup_{r\in [r_0,r_1]}\int_{\R} \int_{\s^2}\left[ |\snabla_{\s^2}^{n_1}{\mathfrak{X}}_*^{n_2}T^{n_3}f|^2+ |\snabla_{\s^2}^{n_1}{\mathfrak{X}}_*^{n_2}T^{n_3}\square_{g_{M,a}}f|^2 \right](t,r_*,\theta,\varphi)\,d\sigma dt<&\:\infty,\\
\label{eq:suffint2}
\int_{r_+}^{r_0}\int_{\R} \int_{\s^2} |({\mathfrak{X}}_*-K)f|^2\,\frac{r^2+a^2}{\Delta}d\sigma dt dr+\int_{r_1}^{\infty}\int_{\R} \int_{\s^2} |({\mathfrak{X}}_*+T)f|^2\,d\sigma dt dr<&\:\infty.
\end{align}
\end{definition}
Now, let $\phi: \mathcal{M}_{M,a}\to \C$ be a solution to \eqref{eq:inhomwaveeq} with $G=0$ arising from smooth and compactly supported initial data, and let $\Pi: \mathcal{M}_{M,a}\to \C$ be a function that will be fixed later in \S \ref{sec:constKinvdata}. We denote $\widehat{\phi}=\phi-\Pi$ and $G=-\square_g\Pi$. Note that $\widehat{\phi}$ is a solution to  \eqref{eq:inhomwaveeq} with this choice of $G$.

Let $\xi: \mathcal{M}_{M,a}\to \R$ be a smooth cut-off function, such that $\xi\equiv 0$ in $\{\tau\leq 0\}$ and $\xi\equiv 1$ in $\{\tau\geq 1\}$. Define:
\begin{equation*}
\upphi=\xi\cdot \hphi.
\end{equation*}
Then
\begin{equation}
\label{eq:waveeqxi}
\square_{g_{M,a}}\upphi=F_{\xi}+\xi G=:F,
\end{equation}
with
\begin{multline}
\label{eq:Fxi}
\rho^2F_{\xi}=2(r^2+a^2-\hfol\Delta)T(\xi)X\phi+[a^2\sin^2\theta+\hfol (\Delta \hfol- 2(r^2+a^2))](2T(\xi)T\phi+T^2(\xi)\phi)+2a(1-\hfol)T(\xi)\Phi\phi\\
+(2r-(\Delta \hfol)')T(\xi)\phi\\
=-2 T(\xi) (r^2+a^2)^{\frac{1}{2}}(1+O(r^{-1}))X\psi+O(1)T(\xi)\phi+[a^2\sin^2\theta+\hfol (\Delta \hfol- 2(r^2+a^2))](2T(\xi)T\phi+T^2(\xi)\phi)\\
+2a(1-\hfol)T(\xi)\Phi\phi\ ,
\end{multline}
which is supported in $\{0\leq \tau\leq 1\}$.

\begin{proposition}
Assuming \eqref{eq:assmphi1} and \eqref{eq:assmphi2}, the function $\upphi_m$ satisfying \eqref{eq:waveeqxi} is sufficiently integrable if $\xi \cdot \Pi$ is sufficiently integrable and $G=-\square_g\Pi$ satisfies \eqref{eq:assmG}.
\end{proposition}
\begin{proof}
By standard elliptic estimates in the bounded region $1<r_0<r_1\leq \infty$ (away from the event horizon and for bounded $m$) and the fundamental theorem of calculus, we apply the assumption on $\Pi$ and \eqref{eq:assmG} to conclude that
\begin{equation*}
\sum_{n_1+n_2+n_3\leq N}\sup_{r\in [r_0,r_1]}\int_{\R} \int_{\s^2} |\snabla_{\s^2}^{n_1}\partial_{r_*}^{n_2}\partial_t^{n_3}\upphi_m|^2(t,r_*,\theta,\varphi)\,dt<\infty
\end{equation*}
holds if \eqref{eq:assmphi2} holds with $\phi_m$ replaced by $K^n\phi_m$, where $n\leq N$. Since $K^n\phi_m$ is also a solution to \eqref{eq:inhomwaveeq} with $G=0$ arising from smooth, compactly supported data, \eqref{eq:assmphi2} applies also to $K^n\phi_m$ and, applying \eqref{eq:assmG} once more to estimate $\square_g\upphi_m$, we conclude that \eqref{eq:suffint1} holds for $f=\upphi_m$.

We have additionally that
\begin{equation*}
\int_{r_1}^{\infty}\int_{\R} \int_{\s^2} |({\mathfrak{X}}_*+T)\upphi |^2\,d\sigma dt dr<\infty
\end{equation*}
by the sufficient integrability assumption on $\Pi$ together with a standard $r^p$-estimate with $p=1$ in the region $r\geq r_I$, for $r_I$ sufficiently large.

We similarly obtain 
\begin{equation*}
\int_{r_1}^{\infty}\int_{\R} \int_{\s^2} |({\mathfrak{X}}_*-K)\upphi |^2\,\frac{r^2+a^2}{\Delta}d\sigma dt dr<\infty
\end{equation*}
by applying Corollary \ref{cor:intassmpart} together with Assumption \ref{eq:assmphi1}.

\end{proof}
\subsection{The frequency-localized ODE}
For sufficiently integrable $\upphi$, the Fourier transform in $t$ along constant $(r,\theta,\varphi)$ of $\upphi$ is well-defined, and we will denote it by $\widehat{\upphi}$. Then we can express:
\begin{equation*}
\upphi(t,r,\theta,\varphi)=\frac{1}{\sqrt{2\pi}}\int_{\R} e^{i\omega t}\widehat{\upphi}(\omega,r,\theta,\varphi)\,d\omega,
\end{equation*}
where the above identity is to be interpreted in the function space $L^2_tL^2_{\s^2}$.

We will moreover define:
\begin{equation*}
\widetilde{\omega}:=\omega-m\upomega_+.
\end{equation*}
The frequency $\omega$ is the time-frequency with respect to the flow along $T$, whereas the shifted frequency $\widetilde{\omega}$ may be interpreted as the time-frequency with respect to the flow along $K$.

Subsequently, we can decompose $\widehat{\upphi}$ as follows in terms of the orthonormal basis of oblate spheroidal harmonics in $L^2(\s^2)$ (see \S \ref{sec:sphere}):
\begin{equation*}
\widehat{\upphi}(\omega,r,\theta,\varphi)=\sum_{\substack{\ell\in \N_0\\ |m|\leq \ell}} \widehat{\upphi}^{(a\omega)}_{m \ell}(r) S_{m\ell}(\theta;a\omega)e^{im\varphi},
\end{equation*}
which is to be interpreted as an identity in the function space $L^2_{\omega}L^2_{\s^2}$. To simplify the notation, we introduce the shorthand notation:
\begin{equation*}
\sum_{\substack{\ell\in \N_0\\ |m|\leq \ell}} \to \sum_{m\ell}.
\end{equation*}

Hence, we obtain the following identity in the space $L^2_tL^2_{\s^2}$:
\begin{equation*}
\upphi(t,r,\theta,\varphi)=\frac{1}{\sqrt{2\pi}}\int_{\R} \sum_{m\ell} \widehat{\upphi}^{(a\omega)}_{m \ell}(r) S_{m\ell}(\theta;a\omega)e^{i(m\varphi+\omega t)}\,d\omega.
\end{equation*}

A remarkable property of the Kerr spacetime is that the equation \eqref{eq:waveeqxi} reduces to an ODE for the variable $\widehat{\upphi}^{(a\omega)}_{m \ell}(r)$. This is known as \emph{Carter's separation}.
\begin{proposition}[Carter's separation]
\label{prop:carter}
Let $\upphi$ be a sufficiently integrable solution to \eqref{eq:waveeqxi}. Then $\widehat{\upphi}^{(a\omega)}_{m \ell}$ satisfies the following equality in $L^2_{\omega}l^2_{ml}$:
\begin{equation*}
\Delta (\rho^2 F)^{(a\omega)}_{m\ell}=\Delta\frac{d}{dr}\left(\Delta \frac{d\widehat{\upphi}^{(a\omega)}_{m \ell}}{dr}\right)+[a^2m^2+(r^2+a^2)^2\omega^2-4Mr a\omega m-\Delta(\lambda^{(a\omega)}_{m\ell}+a^2\omega^2)]\widehat{\upphi}^{(a\omega)}.
\end{equation*}
Furthermore, if we rescale
\begin{align*}
u^{a\omega}_{m\ell}:=&\:\sqrt{r^2+a^2}\widehat{\upphi}^{a\omega}_{m\ell},\\
H^{(a\omega)}_{m\ell}:=&\:\Delta (r^2+a^2)^{-\frac{3}{2}}(\rho^2 F)^{(a\omega)}_{m\ell},
\end{align*}
then $u^{a\omega}_{m\ell}$ satisfies the following in $L^2_{\omega}l^2_{ml}$:
\begin{equation}
\label{eq:odeU}
\frac{d^2 u^{(a\omega)}_{m\ell}}{dr_*^2}+(\omega^2-V_{ml}^{(a\omega)})u^{a\omega}_{m\ell}=H^{(a\omega)}_{m\ell},
\end{equation}
with
\begin{multline*}
V_{m\ell}^{(a\omega)}=(r^2+a^2)^{-2}\left[4Ma\omega r-a^2m^2+\Delta (\lambda^{(a\omega)}_{m\ell} +a^2\omega^2)\right]+(r^2+a^2)^{-3}\Delta \left[3r^2-4Mr+a^2\right]\\
-3r^2\Delta^2(r^2+a^2)^{-4}. 
\end{multline*}
\end{proposition}

We denote:
\begin{equation*}
\Lambda_{m\ell}^{(a\omega)}=\lambda^{(a\omega)}_{m\ell}+a^2\omega^2.
\end{equation*}
From \S \ref{sec:sphere}, it follows that for $\omega=m\upomega_+$:
\begin{equation*}
\Lambda_{m\ell}^{(am\upomega_+)}=\Lambda_{+,m\ell}.
\end{equation*}
We moreover have that
\begin{align}
\label{eq:Lambdalimit1}
\lim_{m\tomega\to 0}\frac{1}{a^2m \tomega}(\Lambda_{m\ell}^{(a \omega)}-\Lambda_{m\ell}^{(am\upomega_+)})<&\:\infty,\\
\label{eq:Lambdalimit2}
\lim_{a\to M}\frac{1}{(a^2-M^2)m^2 \omega^2}\left(\Lambda_{m\ell}^{(a \omega)}-\Lambda_{m\ell}^{(\omega)}\right)<&\:\infty,
\end{align}
by the smooth dependence of $\Lambda_{m\ell}^{(a\omega)}$ on $a \omega$; see for example \cite[\S 3.22, Proposition 1]{ms54}.

\subsection{Boundary conditions}
In this section, we will make use of the regularity assumptions at $\mathcal{H}^+$ and $\mathcal{I}^+$ encoded in the definition of sufficient integrability in Definition \ref{def:suffint}, in order to derive boundary conditions for the ODE \eqref{eq:odeU}.
\begin{lemma}
\label{lm:bc}
Let $\upphi$ be a sufficiently integrable solution to \eqref{eq:waveeqxi}. Then $u^{(a\omega)}_{m\ell}$ satisfies the following boundary conditions: there exists sequences $\{r_i\}$ and $\{\tilde{r}_i\}$, with $r_i\to \infty$ and $\tilde{r_i}\downarrow 1$ as $i\to \infty$, such that:
\begin{align}
\label{eq:bcinfpre}
\lim_{i\to \infty}\left|\frac{d u^{(a\omega)}_{m\ell}}{dr_*}-i\omega u^{(a\omega)}_{m\ell}\right|(r_i)=&\:0,\\
\label{eq:bchorpre}
\lim_{i\to \infty}\left|\frac{d u^{(a\omega)}_{m\ell}}{dr_*}+i\widetilde{\omega} u^{(a\omega)}_{m\ell}\right|(\tilde{r}_i)=&\:0,
\end{align}
for almost every $\omega\in \R$.
\end{lemma}
\begin{proof}
By \eqref{eq:suffint2} combined with Plancherel, it follows that: for $r_0>r_+$ and $r_1<\infty$
\begin{equation*}
\int_{r_+}^{r_0}\int_{\R}\left|\frac{d u^{(a\omega)}_{m\ell}}{dr_*}+i\tomega u^{(a\omega)}_{m\ell}\right|^2\,d\omega dr+\int_{r_1}^{\infty}\int_{\R}\left|\frac{d u^{(a\omega)}_{m\ell}}{dr_*}-i\omega u^{(a\omega)}_{m\ell}\right|^2\,d\omega dr<\infty.
\end{equation*}
Let $\rho_{n}$, $i\geq 1$ be a strict monotonically increasing sequence with $\rho_1=r_1$ and $\rho_i\to \infty$ as $i\to \infty$. By applying the mean-value theorem with $r\in [\rho_i,\rho_{i+1}]$, we can find a sequence $\rho'_i$ such that $\rho_i'\to \infty$ and
\begin{equation*}
\lim_{i\to \infty}\int_{\R}\left|\frac{d u^{(a\omega)}_{m\ell}}{dr_*}-i\omega u^{(a\omega)}_{m\ell}\right|^2(\rho_i')\,d\omega= 0.
\end{equation*}
Convergence of a sequence with respect to the $L^2_{\omega}(\R)$-norm implies that there exists a subsequence along which the sequence convergence pointwise for almost every $\omega$. This concludes \eqref{eq:bcinfpre}.

The limit \eqref{eq:bchorpre} follows similarly, by taking instead $\rho_{n}$, $i\geq 1$ to be a strict monotonically decreasing sequence with $\rho_1=r_0$ and $\rho_i\to 1$ as $i\to \infty$.
\end{proof}

As the goal of Carter's separation is to arrive at integrated energy estimates, it is sufficient to restrict the consideration to a subset of frequencies $\omega$ with full Lebesgue measure. In the proposition below, we show that the limits in Lemma \ref{lm:bc} can be interpreted as boundary conditions for smooth solutions to the \eqref{eq:odeU}, for \underline{almost every} $\omega\in \R$.

\begin{proposition}
\label{prop:smoothnessu}
Let $\upphi$ be a sufficiently integrable solution to \eqref{eq:waveeqxi}. Then, for all $m,\ell$, with $\ell\in \N_0$ and $m\in \Z_{\geq -\ell}\cap\Z_{\leq \ell}$, and for almost every $\omega\in \R$:
\begin{enumerate}[label=\emph{(\roman*)}]
\item $H^{(a\omega)}_{m\ell}$ is smooth on $(\infty,\infty)_{r_*}$ and the corresponding solution $u^{(a\omega)}_{m\ell}$ to \eqref{eq:odeU} is smooth and satisfies the following boundary conditions:
\begin{align}
\label{eq:bcinf2}
\lim_{r_*\to \infty}\left|\frac{d u^{(a\omega)}_{m\ell}}{dr_*}-i\omega u^{(a\omega)}_{m\ell}\right|(r_*)=&\:0,\\
\label{eq:bchor2}
\lim_{r_*\to -\infty}\left|\frac{d u^{(a\omega)}_{m\ell}}{dr_*}+i\widetilde{\omega} u^{(a\omega)}_{m\ell}\right|(r_*)=&\:0.
\end{align}
\item $H^{(a\omega)}_{m\ell}$ satisfies:
\begin{equation}
\label{eq:Hintbound}
\int_{r_+}^{\infty} \Delta^{-2}r^4|H^{(a\omega)}_{m\ell}|^2\,dr<\infty.
\end{equation}
\end{enumerate}
\end{proposition}
\begin{proof}
We refer to \cite[Lemma 5.4.1]{part3} for a proof of (i), with the minor modification that we can restrict to the region $r\in [r_0,\infty)$, since we are proving smoothness in $(-\infty,\infty)_{r_*}$.

Note that by Plancherel
\begin{equation*}
\lim_{n\to \infty}\sum_{m \ell}\int_{\R}\int_{r_++\frac{1}{n}}^{n}  \Delta^{-2}r^4|H^{(a\omega)}_{m\ell}|^2\,dr d\omega=\int_{0}^{\infty}\int_{\Sigma_{\tau}}r^{-2} |\rho^2 F|^2\,d\sigma dr d\tau.
\end{equation*}

By the local-in-time weighted energy estimates from Proposition \ref{prop:locenest} and the assumption \eqref{eq:assmG}, we can further estimate
\begin{equation*}
\begin{split}
\int_{0}^{\infty}\int_{\Sigma_{\tau}}r^{-2} |\rho^2 F|^2\,d\sigma dr d\tau\leq &\: \int_{0}^{1}\int_{\Sigma_{\tau}}r^{-2} |\rho^2 F_{\xi}|^2\,d\sigma dr d\tau+C\int_{0}^{\infty}\int_{\Sigma_{\tau}} |rG|^2\,d\sigma dr d\tau\\
\leq &\: C \int_{\Sigma_0} \mathcal{E}_2[\hphi]\, r^2d\sigma dr +C\int_{0}^{\infty}\int_{\Sigma_{\tau}} |rG|^2\,d\sigma dr d\tau<\infty.
\end{split}
\end{equation*}
Hence, $\sum_{m \ell}\int_{r_++\frac{1}{n}}^{n}  \Delta^{-2}r^4|H^{(a\omega)}_{m\ell}|^2\,dr$ is Cauchy in $L^1_{\omega}(\R)$ and therefore has a subsequence that converges pointwise for almost every $\omega$. As the original sequence is monotonically increasing in $n$, it must also converge pointwise for almost every $\omega$, which concludes \eqref{eq:Hintbound}.
\end{proof}

\subsection{The potential $V_{ml}^{(a\omega)}$}
In this section, we derive the main properties and decompositions of the potential $V_{ml}^{(a\omega)}$ that play a role in the rest of the analysis in this section.
\subsubsection{The main potentials}
\label{sec:mainpot}
With the choice $r_+=1$, we can express:
\begin{align*}
M=&\:\frac{1}{2}(1+a),\\
\upomega_+=&\:\frac{a}{1+a^2},\\
\Delta=&\: (r-1)(r-a^2).
\end{align*}
and $0\leq a\leq 1$ remains our only free parameter. In the sections below we will suppress the super- and subscripts in the notation $V_{m\ell}^{(a\omega)}$ and $\lambda_{m\ell}^{(a\omega)}$; we will simply write $V$ and $\Lambda$, respectively.\\

\paragraph{\underline{The potentials $V_0$ and $V_1$}}\mbox\\

The potential $V(r)$ can be split as follows:
\begin{align}
\label{eq:defpot}
V(r)=&\:V_0(r)+V_1(r),\quad \textnormal{with}\\
\label{eq:defpot0}
V_0(r)=&\: (r^2+a^2)^{-2}\left[2ma(1+a^2)\omega r-a^2m^2+\Lambda \Delta\right],\\
\label{eq:defpot1}
V_1(r)=&\:(r^2+a^2)^{-3}\Delta(3r^2-2(1+a^2)r+a^2)-3(r^2+a^2)^{-4}\Delta^2r^2,\\ \nonumber
\Lambda:=&\:\lambda+a^2\omega^2.
\end{align}
Note that $V_1\geq 0$.\\ 

\paragraph{\underline{The potential $\widetilde{V}$}}\mbox\\
\\

We define the shifted potential:
\begin{equation*}
\widetilde{V}-\widetilde{\omega}^2:=V-\omega^2.
\end{equation*}
We can express:
\begin{align*}
\widetilde{V}(r)=&\:\widetilde{V}_0(r)+\widetilde{V}_1(r),\quad \textnormal{with}\\
\widetilde{V}_0(r)=&\: -2(1+a^2)^{-1}ma\tomega-a^2m^2(1+a^2)^{-2}+(r^2+a^2)^{-2}\left[2ma(1+a^2)\tomega r+m^2a^2(2r-1)+\Lambda \Delta\right],\\
\widetilde{V}_1(r)=&\:V_1(r).
\end{align*}

\paragraph{\underline{The potential $V_{q,\varpi}$}}\mbox\\
\\

Let $q,\varpi\in \R$. Consider the following modified potential that arises when considering the ODE for the rescaled variable $v^{(a\omega)}_{m\ell}=e^{i(\varpi+\tomega)r_*}\Delta^{iq}u^{(a\omega)}_{m\ell}$ (see \S \ref{eq:freqloccurr} for the motivation for considering such a rescaling):
\begin{equation*}
V_{q,\varpi}(r):=\widetilde{V}-\tomega^2+\left[(\tomega+\varpi)^2+2q (\tomega+\varpi) \frac{d\Delta}{dr}(r^2+a^2)^{-1}+q^2\left(\frac{d\Delta}{dr}\right)^2(r^2+a^2)^{-2}\right]
\end{equation*}
In order for the leading-order terms in $r-1$ of $V_{q,\varpi}(r)$ to cancel, we will need to make the following choices for the constants $q$ and $\varpi$:
\begin{align}
\label{eq:choiceq}
q=&\:\frac{(3-a^2)m}{4a},\\
\label{eq:choicevarpi}
\varpi=&\:-\frac{1-a^2}{1+a^2}q=-\frac{(1-a^2)(3-a^2)m}{4a(1+a^2)^2}.
\end{align}

\paragraph{\underline{The potentials $\widetilde{V}_{\rm stat}$ and $\check{V}$}}\mbox\\
\\

We introduce the shorthand notation:
\begin{equation*}
\Lambda_+=\Lambda_{+, m\ell}
\end{equation*}
and split $\widetilde{V}$ as follows:
\begin{equation*}
	\widetilde{V}=\widetilde{V}_{\rm stat}+\check{V},
\end{equation*}
with
\begin{multline}
\label{eq:Vstat}
(1+a^2)^{2}(r^2+a^2)^{2}\widetilde{V}_{\rm stat}(r):=\sum_{k=2}^4\frac{1}{k!}\frac{d^k((r^2+a^2)^{2}(1+a^2)^{2}\widetilde{V_0})}{dr^k}\Big|_{r=1, \tomega=0, \Lambda=\Lambda_+}\cdot (r-1)^k+V_1(r)|_{a=1}
	\end{multline}
When $a=1$, the potential $
V_{\rm stat}$ may be interpreted as the potential corresponding to stationary solutions with respect to the vector field $K$ ($\tomega=0$); see \S \ref{sec:statsol}.
\subsubsection{Properties of the potentials}
\label{sec:proppots}
In this section, we will analyze the main quantitative and qualitative properties of the potentials introduced in \S \ref{sec:mainpot}. 

The lemma below concerns precise quantitiative properties of the potentials in the form of their asymptotic behaviour in the limits $r\downarrow 1$ and $r\to \infty$.
\begin{lemma}
\label{prop:potentials}
The potentials defined in \S \ref{sec:mainpot} have the following asymptotic behaviour:
\begin{align}
\label{eq:V0nearhor}
\widetilde{V}_0(r)=&\:(1+a^2)^{-3}\left[-8 a m \tomega+2(\Lambda-m^2)(1-a^2)+(2m^2-\Lambda+2a m 
\tomega)(1-a^2)^2\right](r-1)+\\ \nonumber
&+\: (1+a^2)^{-4}\Big[4 (\Lambda-2m^2) +6(3m^2-2 \Lambda+ 4 a m \tomega) (1-a^2) \\ \nonumber
&+  (5\Lambda-10m^2-12 am \tomega) (1-a^2)^2\Big](r-1)^2\\ \nonumber
&+\Lambda O((r-1)^3),\\
\label{eq:V0nearhor2}
\widetilde{V}_0(r)=&\:(1+a^2)^{-3}\left[-8 a m \tomega-2m^2(1-a^2)+2(m^2+ a m \tomega)(1-a^2)^2\right](r-1)\\ \nonumber
&+\: (1+a^2)^{-4}\left[-8m^2 +2(9m^2+12m a \tomega) (1-a^2) -2(5m^2+6am\tomega)(1-a^2)^2\right](r-1)^2\\ \nonumber
&+\: \Lambda (r^2+a^2)^{-2}\Delta+(m^2+|m\tomega|)O((r-1)^3),\\
\label{eq:V1nearhor}
\widetilde{V}_1(r)=&\:(1+a^2)^{-3}(1-a^2)^2(r-1)+2(1+a^2)^{-4}\left[3(1-a^2)-4(1-a^2)^2+(1-a^2)^3\right](r-1)^2\\ \nonumber
&+O((r-1)^3),\\
\label{eq:Vtildenearhor}
V_{q,\varpi}(r)=&\:\frac{\left(\Lambda-\frac{2a^2m^2}{(1+a^2)}\right) \Delta}{(r^2+a^2)^2}-m\tomega(r-1)^2+(r-1) \Delta O(m^2+|m\tomega|),\\
\label{eq:Vchecknearhor}
\check{V}=&\: \left[- m \tomega+\frac{1}{4}(1-a^2) (\Lambda- m^2+O(|m\tomega|))+(1-a^2)^2O(|\Lambda|)\right](r-1)\\ \nonumber
&+O(|m\tomega|+\Lambda(1-a^2)+|\Lambda-\Lambda_+|)(r-1)^2,\\
\label{eq:Vstatest}
\widetilde{V}_{\rm stat}(r)=&\:\left[\frac{1}{4}(\Lambda_+-2m^2)+\Lambda_+  O((1-a^2))\right](r-1)^2+\Lambda_+O((r-1)^3),\\
\label{eq:Vnearinf}
V(r)=&\:\Lambda (r^{-2}+O_1(r^{-3})).
\end{align}
\end{lemma}
\begin{proof}

To obtain \eqref{eq:V0nearhor}--\eqref{eq:Vchecknearhor}, we perform a Taylor expansion in $r-1$.

Similarly, \eqref{eq:Vnearinf} follows by Taylor expanding in $\frac{1}{r}$ at $r=\infty$. The term $O(r^{-3})$ does not depend on $m,\omega,\Lambda$ because we are using that $\frac{m^2+|m\omega|+1}{\Lambda}$ is uniformly bounded, since $\Lambda\geq \max\{|m|(|m|+1),2|m\omega|\}$ by \eqref{eq:angev1} and \eqref{eq:angev2}.
\end{proof}

The following lemma concerns the qualitative behaviour of the potential $V_0$, via an investigation of the possible shapes of its graph. This lemma forms the analogue of \cite[Lemma 6.3.1]{part3} and \cite[Lemma 6.3.2]{part3} and includes a more refined analysis of the $a=1$ case.

\begin{lemma}
\label{prop:basicpropV}
The following properties hold:
\begin{enumerate}[label=\emph{(\roman*)}]
\item The potential $V_0$ is \underline{either}:
\begin{enumerate}[label=\emph{(\alph*)}]
\item strictly decreasing;
\item has a unique critical point $r^0_{\rm max}\geq 1$ corresponding to a global maximum, with equality iff $a=1$ and $m\omega=m\upomega_+=\frac{m^2}{2}$ and $\Lambda\leq 2m^2$;
\item has exactly two critical points $1\leq r^0_{\rm min}<r^0_{\rm max}$, corresponding to a local minimum and local maximum, respectively, with equality iff $a=1$ and $m\omega=m\upomega_+=\frac{m^2}{2}$ and $\Lambda>2m^2$.
\end{enumerate}
\item There exists a numerical constant $B>0$ (independent of any parameters), such that
\begin{equation*}
1\leq r^0_{\rm max}<B.
\end{equation*}
\item If we restrict to $m\omega< m\upomega_+$, then
\begin{equation*}
\frac{dV_0}{dr}(1)=\frac{dV_0}{dr}(1)\geq 0,
\end{equation*}
so $V_0(r)$ has exactly one critical point $r_{\max}^0$.
\end{enumerate}
\end{lemma}
\begin{figure}[h]
        \centering
        \begin{subfigure}[b]{0.2\textwidth}
            \centering
            \includegraphics[width=\textwidth]{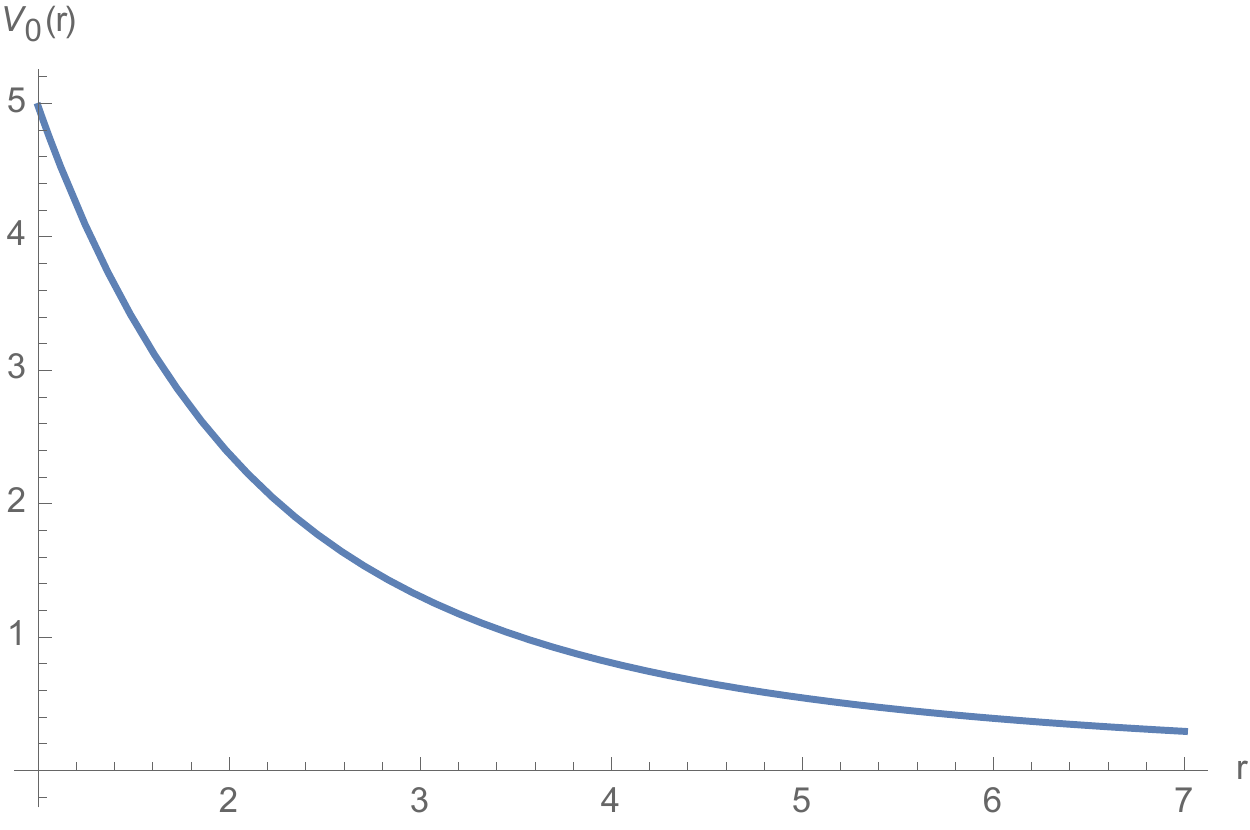}
            \caption[]%
            {{\small $V_0$ strictly decreasing (Case (i)(a)).}}    
            \label{fig:V0plot1}
        \end{subfigure}
        \hskip\baselineskip
        \begin{subfigure}[b]{0.25\textwidth}  
            \centering 
            \includegraphics[width=\textwidth]{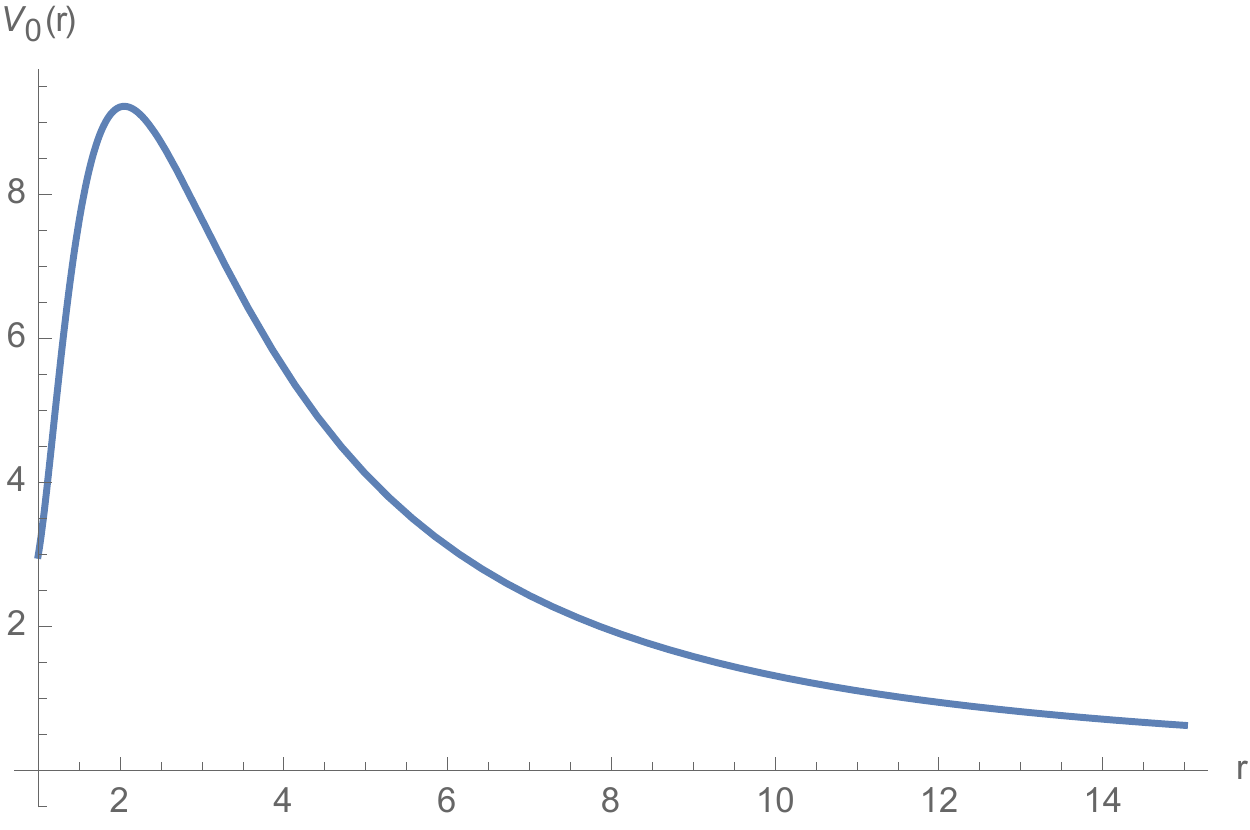}
            \caption[]%
            {{\small $V_0$ with one maximum and no minima (Case (i)(b)).}}    
            \label{fig:V0plot2}
        \end{subfigure}
        \hskip\baselineskip
            \begin{subfigure}[b]{0.25\textwidth}  
            \centering 
            \includegraphics[width=\textwidth]{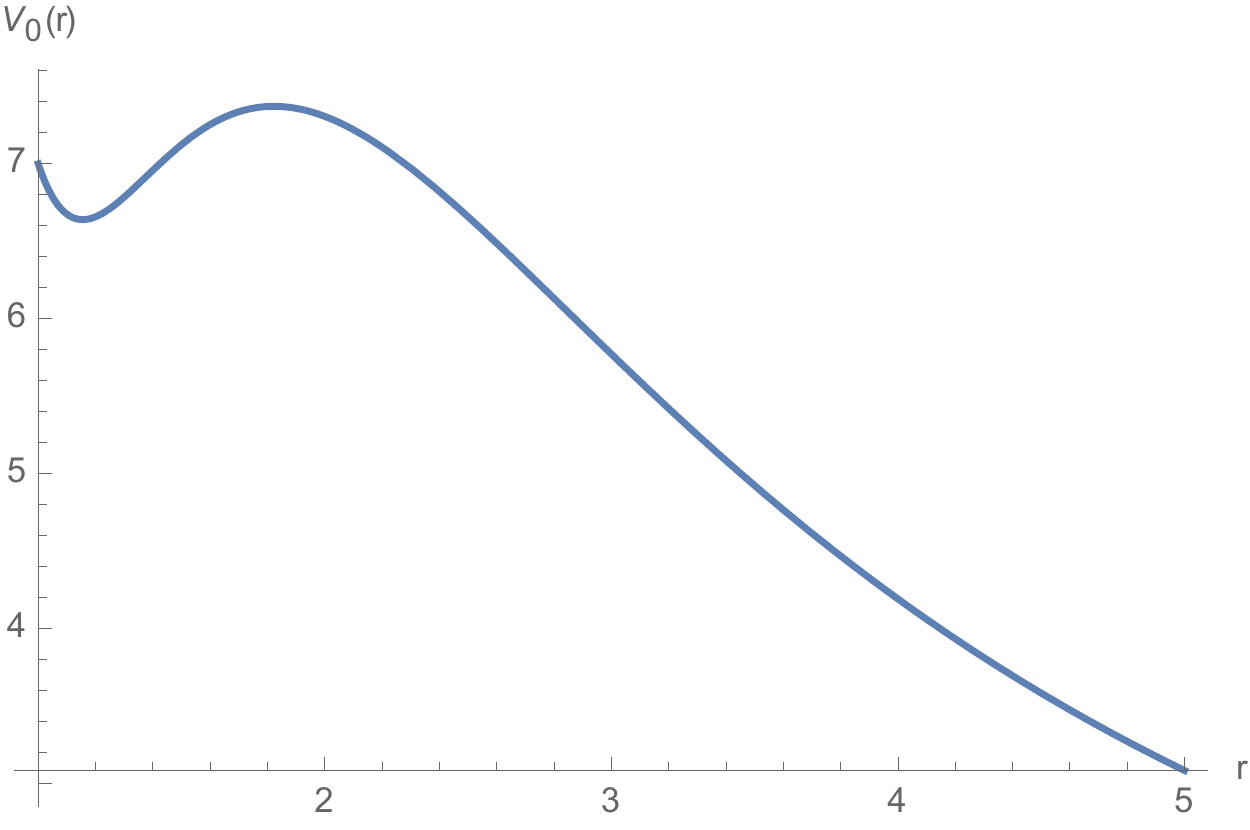}
            \caption[]%
            {{\small $V_0$ with one minimum away from $r=1$ and one maximum (Case (i)(c)). }}    
            \label{fig:V0plot3}
        \end{subfigure}
        \vskip\baselineskip
        \begin{subfigure}[b]{0.25\textwidth}   
            \centering 
            \includegraphics[width=\textwidth]{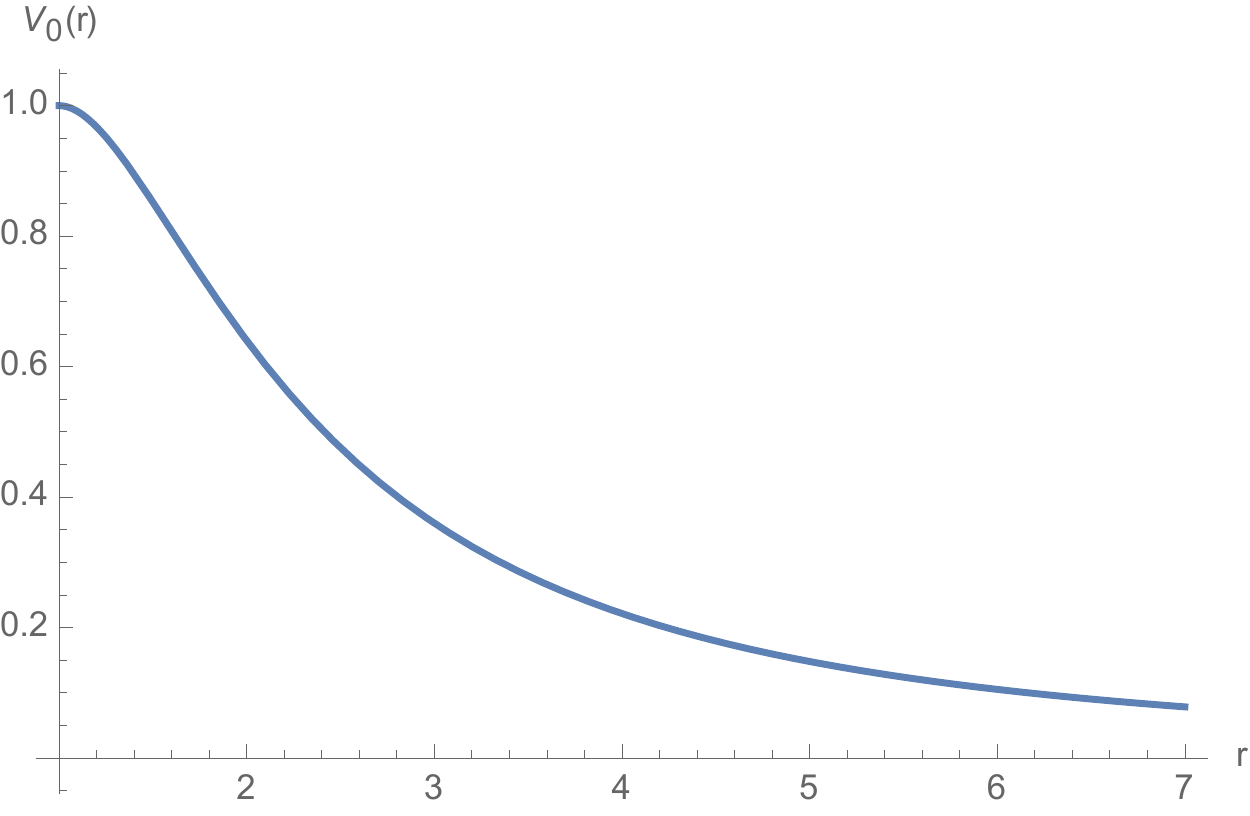}
            \caption[]%
            {{\small $V_0$ with one maximum at \\$r=1$ and no minima (Case (i)(b)).}}    
            \label{fig:V0plot4}
        \end{subfigure}
         \hskip\baselineskip
        \begin{subfigure}[b]{0.25\textwidth}   
            \centering 
            \includegraphics[width=\textwidth]{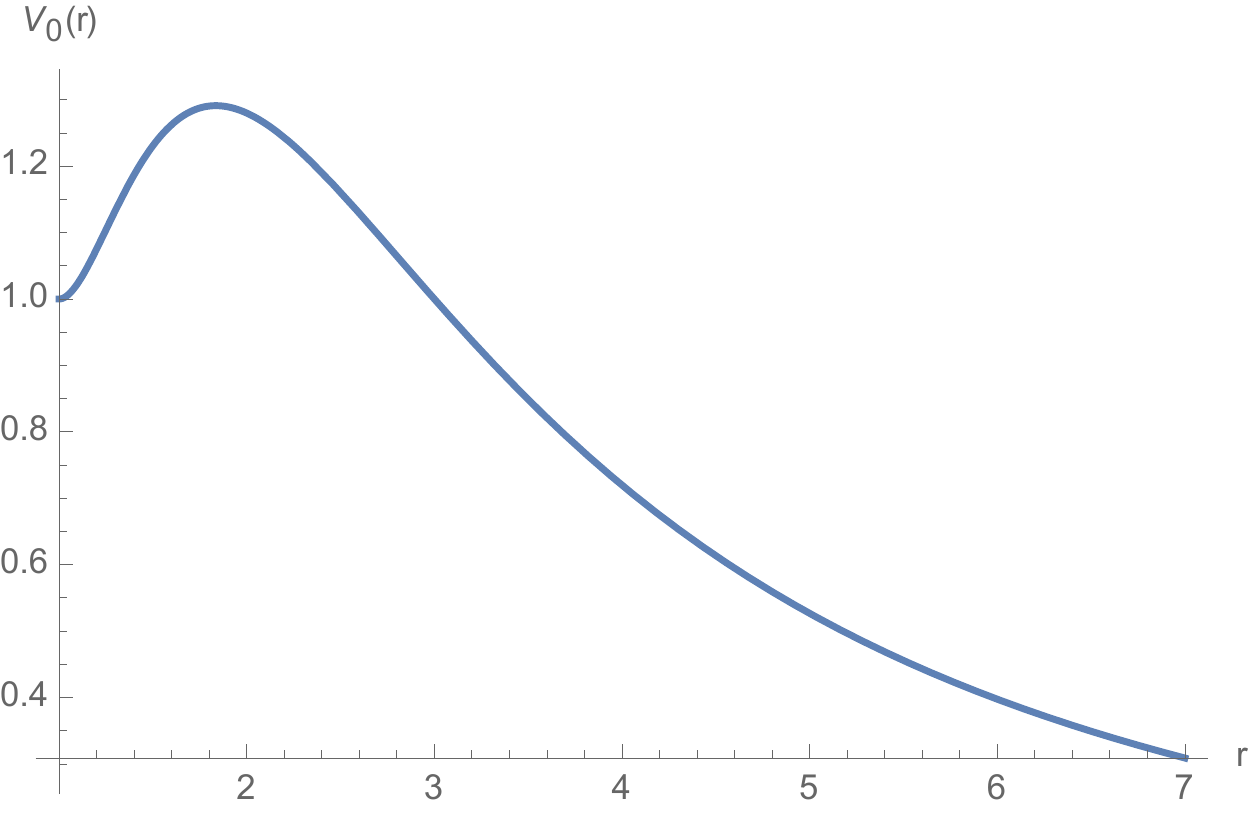}
            \caption[]%
            {{\small $V_0$ with one minimum at $r=1$ and one maximum (Case (i)(c)).}}    
            \label{fig:V0plot5}
            
        \end{subfigure}
        \caption[ The average and standard deviation of critical parameters ]
        {\small The different shapes of the graph of $V_0(r)$ corresponding to the cases outlined in Lemma \ref{prop:basicpropV}.} 
        \label{fig:V0plot}
    \end{figure}
\begin{proof}
We can write
\begin{align*}
(r^2+a^2)^{3}\frac{dV_0}{dr}=&\:2(1+a^2)am\omega (a^2-3r^2)+4m^2 a^2r+\Lambda[(3-2r)r^2-a^4+a^2(r-1)(3r+1)],\\
\frac{d}{dr}\left((r^2+a^2)^{3}\frac{dV_0}{dr}\right)=&\: -12(1+a^2) a m \omega r-6\Lambda r(r-1)+a^2 \Lambda (6r-2)+4a^2m^2.
\end{align*}
Let $r_{1}<r_2$ be the solutions of $\frac{d}{dr}\left((r^2+a^2)^{3}\frac{dV_0}{dr}\right))(r_i)=0$. Then
\begin{equation*}
r_{1,2}=\frac{1}{2}(1+a^2)(1-2m a\omega \Lambda^{-1})\mp \sqrt{\frac{1}{4}(1+a^2)^2(2ma\omega \Lambda^{-1}-1)^2-\frac{a^2}{3}(1-2m^2\Lambda^{-1})}.
\end{equation*}
Suppose $m\omega>0$. Then clearly $r_1<1$ if it is real. Suppose $m\omega\leq 0$. Then we can express:
\begin{equation*}
r_1=\frac{1}{2}(1+a^2)(1+2 a |m\omega| \Lambda^{-1})- \sqrt{\frac{1}{4}(1+a^2)^2(2a|m\omega| \Lambda^{-1}+1)^2-\frac{a^2}{3}(1-2m^2\Lambda^{-1})}
\end{equation*}
so, using that
\begin{equation*}
\begin{split}
\frac{1}{4}(1+a^2)^2(2a|m\omega| \Lambda^{-1}+1)^2-\frac{a^2}{3}(1-2m^2\Lambda^{-1})\geq &\:\frac{1}{4}(1+a^2)^2(2a|m\omega| \Lambda^{-1})^2+\frac{1}{4}(1+a^2)^2-\frac{1}{3}a^2\\
\geq &\:\frac{1}{4}(1+a^2)^2(2a|m\omega| \Lambda^{-1})^2
\end{split}
\end{equation*}
 we once again have that $r_1<\frac{1}{2}(1+a^2)<1$.

In both cases, $(r^2+a^2)^{3}\frac{dV_0}{dr}$ can therefore have at most one critical point in $[1,\infty)$. We moreover have that
\begin{equation*}
\lim_{r\to \infty}(r^2+a^2)^{3}\frac{dV_0}{dr}=-\infty.
\end{equation*}
Hence, $\frac{dV_0}{dr}$ either: A) strictly negative, B) vanishes at a unique point $r^0_{\rm max}$, C) vanishes at two points $r^0_{\rm min}<r^0_{\rm max}$. By the above observations together with the fact that $\lim_{r\to \infty}r^2V_0(r)=\Lambda$, we can moreover infer that in the case: $V_0(r)$ is 1) strictly decreasing, 2) $V_0(r)$ has a unique critical point $r^0_{\rm max}$ which is a global maximum, or 3) has two critical points, a local minimum with agrees with $r^0_{\rm min}$ and a local maximum which agrees with $r^0_{\rm max}$.

From \eqref{eq:V0nearhor} it follows that in the special case $a=1$ and $\omega=\upomega_+$: $r=1$ is either: A) a local minimum if $\Lambda>2m^2$, B) a global maximum if $\Lambda<2m^2$, C) a global maximum if $\Lambda=2m^2$.

Finally, we observe that there exists a numerical constant $A>0$, such that
\begin{equation*}
|(r^2+a^2)^{3}\frac{dV_0}{dr}+2\Lambda r^{3}-3(1+a^2)(\Lambda-2am\omega)r^2|\leq A (\Lambda+|m\omega|+m^2r^{-1})r
\end{equation*}
Hence,
\begin{equation*}
|\Lambda r^0_{\rm max}-(1+a^2)(\Lambda-2am\omega)|\leq A (\Lambda+|m\omega|+m^2).
\end{equation*}
We divide the above inequality with $\Lambda$ and use that $\Lambda\geq |m|(|m|+1)$ and $\Lambda\geq 2|m\omega|$ to conclude that there exists a numerical constant $B>0$ such that $r^0_{\rm max}<B$.

It follows from \eqref{eq:V0nearhor} and \eqref{eq:V1nearhor} that when $m\tomega<0$:
\begin{equation*}
\frac{dV_0}{dr}(1)=\frac{d\widetilde{V}_0}{dr}(1)\geq 0,
\end{equation*}
so it immediately follows that $V_0$ cannot have a minimum.
\end{proof}

The lemma below concerns the quantitative behaviour of $V$ around its extrema.

\begin{lemma}
\label{prop:highLambdapotential}
The potential $V$ has at most two critical points $r_{\min}<r_{\max}$, corresponding to a local minimum and a global maximum. Furthermore:
\begin{enumerate}[label=\emph{(\roman*)}]
\item There exists a numerical constant $C>0$, such that
\begin{align*}
|r_{\max}-(1+\sqrt{2})|\leq &\: C (|m\omega|+m^2)\Lambda^{-1}+C(1-a^2),\\
0\leq r_{\min}-1\leq &\:C (|m\omega|+m^2)\Lambda^{-1}+C(1-a^2) \qquad \textnormal{if $m\omega\geq m\upomega_+$},\\
\textnormal{$V(r)$ has no minimum in $[1,\infty)$ if  $m\omega< m\upomega_+$}.
\end{align*}
\item For sufficiently small $\frac{|m\omega|+m^2+1}{\Lambda}$ and $\frac{\omega^2}{\Lambda}$, there exist constants $b>0$ and $\delta>0$, such that
\begin{equation*}
V(r)-\omega^2\geq b \Lambda\qquad \textnormal{for all $r\in (r_{\max}-\delta,r_{\max}+\delta)$}.
\end{equation*}
\item For sufficiently small $\frac{|m\omega|+m^2+1}{\Lambda}$, there exist constants $b,B>0$, such that for all $r\geq 1$:
\begin{align*}
(r_{\max}-r)\frac{dV}{dr}(r)\geq &\: b \Lambda r^{-5}(r-1)(r-r_{\max})^2\qquad \textnormal{if $m\omega< m\upomega_+$},\\
(r_{\max}-r)\frac{dV}{dr}(r)\geq &\: b \Lambda r^{-5}(r-1)(r-r_{\max})^2-B\left(m\omega-m\upomega_+\right)r^{-4}(r-r_{\max})^2\qquad \textnormal{if $m \omega \geq m\upomega_+$}.
\end{align*}
\item Let $0<\eta<\frac{1}{4}$ and $m\neq 0$. Consider $m\omega \leq m\upomega_+-\eta m^2$. Then there exists constants $b_{\eta},\delta_{\eta}>0$ such that:
\begin{align*}
(r_{\max}-r)\frac{dV}{dr}(r)\geq &\:  b_{\eta} r^{-4}(r-r_{\max})^2 m^2  \quad \textnormal{for $r\in [1,\infty)$},\\
V(r)-\omega^2\geq &\: b_{\eta} \Lambda  \quad \textnormal{for $r\in (r_{\max}-\delta_{\eta},r_{\max}+\delta_{\eta})$},
\end{align*}
with $r_{\rm max}-\delta_{\eta}>1$.
\end{enumerate}
\end{lemma}
\begin{proof}
First of all, since
\begin{equation*}
\begin{split}
(r^2+a^2)^{3}\frac{dV_0}{dr}=&\:2(1+a^2)am\omega (a^2-3r^2)+4m^2 a^2r+\Lambda[(3-2r)r^2-a^4+a^2(r-1)(3r+1)]
\end{split}
\end{equation*}
and there exists a uniform constant $A>0$ such that
\begin{equation*}
(r^2+a^2)^{3}\frac{dV_1}{dr}\leq A r^2,
\end{equation*}
it follows that there exists a numerical constant $C>0$, such that
\begin{align*}
|r^0_{\min}-r_{\min}|\leq C\frac{|m\omega|+m^2+1}{\Lambda},\\
|r^0_{\max}-r_{\max}|\leq C\frac{|m\omega|+m^2+1}{\Lambda}.
\end{align*}

We can express $(r^2+a^2)^{3}\frac{dV_0}{dr}$ as follows by expanding around $a^2=1$:
\begin{equation}
\label{eq:V0maxdeterm}
\begin{split}
(r^2+a^2)^{3}\frac{dV_0}{dr}=&\: 2\Lambda (r - 1) (1 - r^2 + 2 r)-8 am \widetilde{\omega}-2 m (r-1) (m + 3 m r + 6 a \widetilde{\omega} (1 + r))\\
&+\Lambda(1-a^2)[3+r(3r-2)-(1-a^2)] \\
&+\left[-4 m^2 r+ 6 a m \widetilde{\omega} (r^2-1) +2 m^2 \left(r^2-\frac{2}{3}\right)\right] (1-a^2) \\
&+ 2 m (a \widetilde{\omega}  + m) (1-a^2)^2.
\end{split}
\end{equation}

We use that by Lemma \ref{prop:basicpropV}, $r_{\max}^0\leq B$, so we can estimate
\begin{equation*}
(r_{\square}^0-1)((r_{\square}^0)^2-2r_{\square}^0-1)\leq C (|m\omega|+m^2)\Lambda^{-1}+C(1-a^2),
\end{equation*}
with $C>0$ and $\square\in\{\min,\max\}$.

We conclude that 
\begin{align*}
|r_{\max}^0-(1+\sqrt{2})|\leq &\:C (|m\omega|+m^2)\Lambda^{-1}+C(1-a^2),\\
|r_{\min}^0-1|\leq &\:C (|m\omega|+m^2)\Lambda^{-1}+C(1-a^2).
\end{align*}
We conclude that (i) holds by applying in addition (iii) of Lemma \ref{prop:basicpropV}. 

Then, for $1-a^2$ suitably small, (ii) follows immediately from the above estimate for $r_{\max}$ and continuity of $V(r)$, together with the smallness assumption on $\frac{|m\omega|+m^2+\omega^2+1}{\Lambda}$. For the remaining $a$, we refer to \cite[Lemma 8.3.1]{part3}.\footnote{In the present paper, only the case of small $1-a^2$ is relevant, but we include the remaining $a$ for the sake of completeness.}

We will now establish (iii). We will consider the case where $1-a^2$ is suitably small. From \eqref{eq:V0maxdeterm}, we obtain
\begin{equation*}
\begin{split}
(r^2+a^2)^3\frac{dV_0}{dr}\geq  &\:2\Lambda (r-1)(r_{\max}-r)(r+\sqrt{2}-1)-2\Lambda (r-1)(r_{\max}-(1+\sqrt{2}))(r+\sqrt{2}-1)\\
&-8 am \widetilde{\omega}-2 (r-1) (m^2 + 3 m^2 r + 6 a m\widetilde{\omega} (1 + r))+(1-a^2)\left(4\Lambda-3\frac{2}{3}m^2\right)
\end{split}
\end{equation*}
and hence, noting that $\Lambda\geq m^2$ and applying (i),
\begin{equation*}
\begin{split}
(r^2+a^2)^{3}r^{-1}(r_{\max}-r)\frac{dV_0}{dr}\geq &\:2\Lambda (r-1)(r_{\max}-r)^2-B(1+m^2+|m\omega|)|r_{\max}-r|(r-1)\\
&-8a m \widetilde{\omega}r^{-1}(r_{\max}-r)
\end{split}
\end{equation*}
For $r\in [1,1+\frac{3}{4}\sqrt{2}]$, we use that for $1-a^2$ and $\frac{|m\omega|+m^2+1}{\Lambda}$ suitably small, $r_{\max}>1+\frac{3}{4}\sqrt{2}$ and there exists a constant $b>0$ such that:
\begin{equation*}
\begin{split}
(r^2+a^2)^{3}r^{-1}(r_{\max}-r)\frac{dV}{dr}\geq &\:b\Lambda (r-1)-8a m \widetilde{\omega}r^{-1}(r_{\max}-r)
\end{split}
\end{equation*}
Note that when $r\in [1+\frac{3}{4}\sqrt{2},\infty)$, we can estimate:
\begin{equation*}
\frac{d}{dr}\left((r^2+a^2)^3\frac{dV_0}{dr}\right)\leq -c\Lambda r^2,
\end{equation*}
for some constant $c>0$, by using that the critical point $r_2$ defined in the proof of Lemma \ref{prop:basicpropV} satisfies:
\begin{equation*}
\left|r_2-\frac{1}{2}\left(1+a^2+\sqrt{1 - \frac{a^2}{3}}\right)\right|\leq C\frac{|m\omega|+m^2+1}{\Lambda},
\end{equation*}
so $r_2<1+\frac{3}{4}\sqrt{2}$ for $1-a$ and $\frac{|m\omega|+m^2+1}{\Lambda}$ sufficiently small. Since moreover
\begin{equation*}
\frac{d}{dr}\left((r^2+1)^3\frac{dV_1}{dr}\right)\leq C r,
\end{equation*}
we can integrate $\frac{d}{dr}\left((r^2+1)^3\frac{dV}{dr}\right)$ from $r=r_{\rm max}$ to conclude that for any $r\in [1+\frac{3}{4}\sqrt{2},\infty)$:
\begin{equation*}
\begin{split}
(r^2+a^2)^{3}r^{-1}(r_{\max}-r)\frac{dV}{dr}\geq &\:b\Lambda(r_{\max}-r)^2.
\end{split}
\end{equation*}
Combining the above lower bound estimates for $r^5(r_{\max}-r)\frac{dV}{dr}$ then gives (iii) for $1-a$ sufficiently small. The statement with remaining $a$ follows in particular from \cite[Lemma 8.3.1]{part3}.

We will now prove (iv). Consider first the case $m\omega>0$ and $a=1$. Since $0<m\omega<m\upomega_+-\eta m^2=\frac{1-2\eta}{2}m^2$, we have that
\begin{equation*}
r_0=\frac{m}{2\omega}=\frac{m^2}{2m\omega}>1.
\end{equation*}
Now consider the subcase $m\omega>\epsilon m^2$. Then we additionally have that $r_0<R$ for some $R$ depending on $\epsilon$.

We can write for $a=1$:
\begin{equation*}
\begin{split}
(r_0^2+1)^2(V_0(r_0)-\omega^2)=&\:4m\omega r_0-m^2+\Lambda (r_0-1)^2-\omega^2(r_0^2+1)^2\\
=&\:m^2+\Lambda \omega^{-2} \left(\frac{m}{2}-\omega\right)^2-\omega^{-2}\left(\frac{m^2}{4}+\omega^2\right)^2\\
=&\:\omega^{-2}\left(m^2\omega^2+\Lambda \tomega^2-(\tomega^2+\omega m)^2 \right)\\
=&\: \tomega^2\omega^{-2}\left(\Lambda-\tomega^2-2\omega m \right)\\
=&\: \tomega^2\omega^{-2}\left(\Lambda-\omega^2-\frac{m^2}{4}-\omega m \right)\\
\geq &\: \tomega^2\omega^{-2}\left(\Lambda-(1-\eta^2)m^2 \right).
\end{split}
\end{equation*}
By \eqref{eq:angev1}, we have that $\Lambda\geq |m|(|m|+1)$. Hence,
\begin{equation*}
(V_0(r_0)-\omega^2)\geq c_{\epsilon,\eta}\Lambda.
\end{equation*}
Since $V_1\geq 0$, the above inequality holds also with $V_0$ replaced by $V$.

Now consider the case $m\omega\leq \epsilon m^2$. Then we apply \eqref{eq:Vnearinf} to conclude that for $R_0>1$ suitably large, there exists a $\epsilon>0$ suitably small, such that
\begin{equation*}
V(R_0)-\omega^2\geq \frac{\Lambda}{2 R_0^2}.
\end{equation*}

By combining the above and applying continuity of $V$ in $a$, we can therefore conclude that there exists a $b_{\eta}>0$, such that for all $1-a\ll 1$
\begin{equation*}
V(r_{\rm max})-\omega^2\geq b_{\eta} \Lambda.
\end{equation*}

We now consider $\frac{dV}{dr}$. For $1-a\ll1$, it follows from \eqref{eq:V0nearhor} that there exists a constant $c_{\eta}>0$ such that $\frac{dV}{dr}(1)\geq c_{\eta}m^2>0$. Hence, $V$ has exactly one critical point, which is the maximum $r_{\rm max}$. Furthermore, $r_{\rm max}-1\geq r_0-1=-(m\omega)^{-1}m\tomega\geq \frac{2\eta}{1-2\eta}+O(1-a^2)>\frac{\eta}{2}$ for $1-a$ sufficiently small compared to $\eta$.

By taking $\delta_{\eta}>0$ sufficiently small, we also have that $r_{\max}-\delta_{\eta}>1$, so after applying Taylor's theorem around $r_{\rm max}$, there exists a constant $b_{\eta}>0$ such that:
\begin{equation*}
(r_{\max}-r)\frac{dV}{dr}(r)\geq  b_{\eta} r^{-4}(r-r_{\max})^2  m^2,
\end{equation*}
for $r\in (r_{\max}-\delta_{\eta}, r_{\rm max}+\delta_{\eta})$. 

We now consider the case $m\omega<0$ and $1-a\ll 1$. Suppose $\Lambda<Bm^2$ for some arbitrarily large constant $B>0$. Then by \eqref{eq:V0nearhor}
\begin{equation*}
\frac{dV}{dr}(1)=-8a m \tomega +O(1-a^2)(\Lambda+m^2+|m\tomega|)\geq \left(4+O(1-a^2)\right)\Lambda.
\end{equation*}
We infer that $V$ has a unique critical point $r_{\rm max}$ which is a maximum. Furthermore, since there exists a constant $C>0$ such that for $1\leq r\leq 2$
\begin{equation*}
\frac{d^2V}{dr^2}(r)\geq  -C \Lambda,
\end{equation*}
we have that there exists a $\delta>0$ such that $r_{\rm max}>1+\delta$. We moreover have that
\begin{equation*}
(r_{\max}-r)\frac{dV}{dr}(r)\geq  b r^{-4}(r-r_{\max})^2  m^2,
\end{equation*}
for all $r\geq 1$. If $\Lambda \geq Bm^2$ for $B>0$ arbitrarily large, we conclude that (iv) holds.

If instead $\Lambda <Bm^2$, we can simply apply (iii) to conclude (iv).

In the case $a<a_0<1$, we refer to \cite[Lemma 8.3.1]{part3} and  \cite[Proposition 8.5.1]{part3}  for a proof of (iv).
\end{proof}

\subsection{Frequency-localized estimates}
\label{sec:freqlocest}
In this section, we consider smooth solutions $u:\R_{r_*}\to \R$ to the boundary value problem:
\begin{align}
\label{eq:ode}
u''+(\omega^2-V)u=&\:H,\\
\label{eq:bchor}
[u'+i\left(\omega-m \upomega_+\right)u](-\infty)=&\:0,\\
\label{eq:bcinf}
[u'-i\omega u](+\infty)=&\:0,
\end{align}
where we used the notation $(\cdot)':=\frac{d}{dr_*}=\frac{\Delta}{r^2+a^2}\frac{d}{dr}$ and $V(r)=V(r;a,m,\omega,\Lambda)$ is the potential function satisfying \eqref{eq:defpot}--\eqref{eq:defpot1}, where we take the parameters $\omega,m,\Lambda\in \R$ to denote \emph{admissible frequencies}, according to the definition below.
\begin{definition}
We say that a frequency triple $(\omega,m,\Lambda)$, with $m\in \Z$, $|m|\leq m_0$, with $m_0\in \N$, $\Lambda,\omega\in \R$ is \underline{admissible} iff:
\begin{equation*}
\Lambda\geq \max\{|m|(|m|+1),2|m\omega|\}
\end{equation*}
and $\Lambda(a,\omega,m,\ell)\to \Lambda_+|_{a=1}$ as $|a|\uparrow 1$ and $\omega \to \upomega_+$.
\end{definition}

In the remainder of \S \ref{sec:freqlocest}, we will restrict the consideration to  the ``near-extremal'' setting, i.e.\ we will assume that $a_0\leq a\leq 1$, with $a_0$ suitably close to 1.

\subsubsection{Frequency-localized currents}
\label{eq:freqloccurr}
We wil derive integrated estimates by systematically integrating by parts via the use of \emph{currents}. We define the following currents:
\begin{align*}
j_1^h[u]:=&\:h \Re(u' \overline{u})-\frac{1}{2}h'|u|^2,\\
j_2^y[u]:=&\:y(|u'|^2+(\omega^2-V)|u|^2),\\
j_3^f[u]:=&\:j_1^{f'}[u]+j_2^f[u]=f(|u'|^2+(\omega^2-V)|u|^2)+f' \Re(u' \overline{u})-\frac{1}{2}f''|u|^2,\\
j^{g_{\infty}}_{\infty}[u]:=&\:\frac{1}{2}g_{\infty}|u'-i\omega u|^2-\frac{1}{2}g_{\infty}V|u|^2,\\
j^{g_+}_{+}[u]:=&\:\frac{1}{2}g_+|u'+i\tomega u|^2-\frac{1}{2}g_+\widetilde{V}|u|^2,\\
j^T[u]:=&\:-\omega \Re(i u'\overline{u}),\\
j^K[u]:=&\:-\tomega \Re(i u'\overline{u}).
\end{align*}
By applying \eqref{eq:ode}, we obtain the following expressions for the derivatives of the above currents:
\begin{align*}
\left(j_1^h[u]\right)'=&\:h(|u'|^2+(V-\omega^2)|u|^2)-\frac{1}{2}h''|u|^2+h \Re(u\overline{H}),\\
\left(j_2^y[u]\right)'=&\:y'(|u'|^2+(\omega^2-V)|u|^2)-yV'|u|^2+2y\Re(u' \overline{H}),\\
\left(j_3^f[u]\right)'=&\:2f'|u'|^2-fV'|u|^2-\frac{1}{2}f'''|u|^2+\Re(f'u \overline{H}+2fu' \overline{H}),\\
\left(j^{g_{\infty}}_{\infty}[u]\right)'=&\: \frac{1}{2}g_{\infty}'|u'-i\omega u|^2-\frac{1}{2}(g_{\infty}V)'|u|^2+\Re(g_{\infty} (u'-i\omega u)\overline{H}),\\
\left(j^{g_+}_{+}[u]\right)'=&\:\frac{1}{2}g_{+}'|u'+i\tomega u|^2-\frac{1}{2}(g_{+}\widetilde{V})'|u|^2+\Re(g_{+} (u'+i\tomega u)\overline{H}) ,\\
\left(j^T[u]\right)'=&\:-\omega \Re(i H \overline{u}),\\
\left(j^K[u]\right)'=&\:-\tomega \Re(i H \overline{u}).
\end{align*}

It will be convenient to define an alternative to $j_{+}^{g_+}[u]$ near $r=1$. For this purpose, we introduce the following rescaled variable
\begin{equation*}
v:= e^{i (\widetilde{\omega}+\varpi) r_*}\Delta^{iq} u
\end{equation*}
Note that
\begin{align*}
u'=(e^{-i (\widetilde{\omega}+\varpi) r_*}\Delta^{-i q} v)'=&-i\left[ \widetilde{\omega}+\varpi+q \frac{d \Delta}{dr}(r^2+a^2)^{-1}\right]e^{-i (\widetilde{\omega}+\varpi) r_*}\Delta^{-i q}v\\
&+e^{-i  (\widetilde{\omega}+\varpi) r_*}\Delta^{-i q} v',\\
e^{i (\widetilde{\omega}+\varpi) r_*}\Delta^{iq} u''=&\:v''-2i\left[  (\widetilde{\omega}+\varpi)+q\frac{d \Delta}{dr}(r^2+a^2)^{-1}\right]v'\\
&-\left[  (\widetilde{\omega}+\varpi)^2 +2q( \widetilde{\omega}+\varpi) \frac{d \Delta}{dr}(r^2+a^2)^{-1}+q^2 \left(\frac{d \Delta}{dr}\right)^2(r^2+a^2)^{-2}\right]v\\
&-iq\frac{\Delta}{(r^2+a^2)} \left[ \frac{d^2 \Delta}{dr^2}(r^2+a^2)^{-1}-2r \frac{d \Delta}{dr}(r^2+a^2)^{-2}\right]v.
\end{align*}
Recall from \S \ref{sec:mainpot} that
\begin{equation*}
\begin{split}
V_{q,\varpi}(r)=&\:\widetilde{V}(r)-\widetilde{\omega}^2+\left[  (\widetilde{\omega}+\varpi)^2 +2q (\widetilde{\omega}+\varpi) \frac{d \Delta}{dr}(r^2+a^2)^{-1}+q^2 \left(\frac{d \Delta}{dr}\right)^2(r^2+a^2)^{-2}\right].
\end{split}
\end{equation*}

Then \eqref{eq:ode} is equivalent to the following ODE:
\begin{equation}
\label{eq:ODEnearhor}
\begin{split}
v''&-2i\left[\widetilde{\omega}+\varpi+q\Delta \frac{d \Delta}{dr}(r^2+a^2)^{-1}\right]v'-V_{q,\varpi}v-iq\frac{\Delta}{(r^2+a^2)} \left[ \frac{d^2 \Delta}{dr^2}(r^2+a^2)^{-1}-2r \frac{d \Delta}{dr}(r^2+a^2)^{-2}\right]v\\
=&\: e^{i (\widetilde{\omega}+\varpi) r_*}\Delta^{iq} H.
\end{split}
\end{equation}
We will make the following choices of $q,\varpi$:
\begin{align*}
q=&\:\frac{(3-a^2)m}{4a},\\
\varpi=&\:-\frac{1-a^2}{1+a^2}q=-\frac{(1-a^2)(3-a^2)m}{4a(1+a^2)^2}.
\end{align*}

We now define the following current:
\begin{align*}
j_{++}^{g_{++}}[v]:=&\:\frac{1}{2}g_{++}|v'|^2-\frac{1}{2}g_{++} V_{q,\varpi}|v|^2,\qquad \textnormal{with}\\
\left(j_{++}^{g_{++}}[v]\right)':=&\:\frac{1}{2}g_{++}'|v'|^2-\frac{1}{2}\left(g_{++} V_{q,\varpi}\right)'|v|^2-\frac{\Delta g_{++}}{r^2+a^2}\left[ \frac{d^2 \Delta}{dr^2}(r^2+a^2)^{-1}-2r \frac{d \Delta}{dr}(r^2+a^2)^{-2}\right]\Re(iq v\frac{d \overline{v}}{dr})\\
&+\Re(g_{++} v'e^{-i (\widetilde{\omega}+\varpi) r_*}\Delta^{-iq}\overline{H}).
\end{align*}

We will define the functions $h,y,f,g_{\infty},g_+,g_{++}: \R_{r_*}\to \R$ in \S\S \ref{sec:rweightfreq}--\ref{sec:boundfreqaway}. We will occasionally abuse notation and interpret these functions as functions on $(1,\infty)_r$, implicitly interpreting $r_*$ as a function of $r$.

\subsubsection{$r$-weighted and $(r-1)$-weighted estimates}
\label{sec:rweightfreq}
In this section, we derive estimates with $r$-weights in the region $r\gg 1$ and $(r-1)$-weights in the region $r-1\ll 1$.
\begin{proposition}
\label{prop:rweightedest}
	Let $0<p< 2$. 
	\begin{enumerate}[label=\emph{(\roman*)}]
		\item  There exists a sufficiently large $R_{\infty}>1$ and a constant $C>0$ such that:
	\begin{multline}
	\label{eq:rpest}
		\int_{R_{\infty}}^{\infty} \frac{1}{2}p|u'-i\omega u|^2+\frac{1}{4}(2-p)\Lambda r^{p-3}|u|^2\,dr\leq C\int_{R_{\infty}-1}^{R_{\infty}}\Lambda r^{p-2}|u|^2\,dr-\int_{\R} r^p\chi_{\infty} \Re((u'-i\omega u)\overline{H})\,dr_*.
		\end{multline}
		\item There exists a $R_+>1$ with $R_+-1$ sufficiently small and a constant $b>0$ such that:
		\begin{multline}
	\label{eq:rmin1pestfreq}
		b\int_{1}^{R_+} p\Delta^{-\frac{p}{2}-1}(r-1)|u'+i\tomega u|^2+(2-p)\Lambda \Delta^{-\frac{p}{2}}(2(r-1)+1-a^2)|u|^2\,dr\\
		\leq \int_{1}^{R_{+}}m^2 \Delta^{-\frac{p}{2}}(2(r-1)+1-a^2)|u|^2\,dr+\int_{R_{+}}^{R_{+}+1}\Lambda r^{p-2}|u|^2\,dr\\
		-\int_{1}^{R_++1} \Delta^{-\frac{p}{2}}\chi_{+} \Re\left(\frac{dv}{dr}e^{-i (\widetilde{\omega}+\varpi) r_*}\Delta^{-iq}\overline{H}\right)\,dr.
		\end{multline}
		\item Let $\frac{m^2+|m\tomega|}{\Lambda}\ll 2-p$, then there exists a $R_+>1$ with $R_+-1$ sufficiently small and a constant $b>0$ such that:
		\begin{multline}
	\label{eq:rmin2pest}
		b\int_{1}^{R_+} p\Delta^{-\frac{p}{2}-1}(r-1)|u'+i\tomega u|^2+(2-p)\Lambda \Delta^{-\frac{p}{2}}(2(r-1)+1-a^2)|u|^2\,dr\\
		\leq \int_{R_{+}}^{R_{+}+1}\Lambda r^{p-2}|u|^2\,dr-\int_{1}^{R_++1} \Delta^{-\frac{p}{2}}\chi_+ \Re\left(\frac{dv}{dr}e^{-i (\widetilde{\omega}+\varpi) r_*}\Delta^{-iq}\overline{H}\right)\,dr.
		\end{multline}
	\end{enumerate}
\end{proposition}
\begin{proof}
Consider first the current $j^{g_{\infty}}_{\infty}$ with $g_{\infty}(r)=r^p\chi_{\infty}(r)$, where $\chi$ denotes a smooth cut-off function that satisfies: $\chi_{\infty}(r)=1$ for $r\geq R_{\infty}$ and $\chi_{\infty}(r)=0$ for $r\leq R_{\infty}-1$, where we will take $R_{\infty}>2$ suitably large. Integrating $(j^{g_{\infty}}_{\infty})'$ then gives the identity:
	
	\begin{multline}
	\label{eq:rpestaux}
		\int_{1}^{\infty} \frac{1}{2}\frac{d}{dr}(r^p\chi_{\infty})|u'-i\omega u|^2-\frac{1}{2}\frac{d}{dr}(r^pV\chi_{\infty})|u|^2\,dr\\
		=\frac{1}{2}\lim_{r_*\to \infty}\left[r^p|u'-i\omega u|^2(r_*)-\frac{1}{2}r^pV|u|^2\right]-\int_{\R} r^p\chi_{\infty} \Re((u'-i\omega u)\overline{H})\,dr_*.
		\end{multline}
Since $(u'-i\omega u)(\infty)=0$, we can appeal to the asymptotic expansion of $u'-i\omega u$ at $r=\infty$ to conclude that
\begin{equation*}
	\lim_{r_*\to \infty}r^q|u'-i\omega u|^2(r_*)=0\quad \textnormal{for $q<4$}.
\end{equation*}
Since $\lim_{r\to \infty}r^2V=\Lambda$, the other limit appearing on the right-hand side of \eqref{eq:rpestaux} has a good sign.

We furthermore apply \eqref{eq:Vnearinf} to conclude that:
	\begin{equation*}
	\frac{1}{2}\frac{d}{dr}(r^pV)\geq \frac{1}{4}(2-p)\Lambda r^{p-3}.
		\end{equation*}
		for $R_{\infty}$ sufficiently large. Then \eqref{eq:rpest} follows.
	
	Consider now $j^{g_{++}}_{++}$ with $g_+=-\Delta^{-\frac{p}{2}}\chi_+$, where $\chi_+$ denotes a smooth cut-off function satisfying $\chi_+(r)=1$ for $r\leq R_+$ and $\chi_+(r)=0$ for $r\geq R_{+}+1$, where we will take $R_+-1$ suitably small.
	
	Integrating $(j^{g_{++}}_{++})'$ results in the identity:
	\begin{multline}
	\label{eq:rmin1pestaux}
		\int_{1}^{\infty} \frac{1}{2}\frac{d}{dr}(-\Delta^{-\frac{p}{2}}\chi_+)|v'|^2+\frac{1}{2}\frac{d}{dr}(\Delta^{-\frac{p}{2}}V_{q,\varpi}\chi_+)|v|^2\\
		+\Delta^{-\frac{p}{2}}\chi_+ \left[ \frac{d^2 \Delta}{dr^2}(r^2+a^2)^{-1}-2r \frac{d \Delta}{dr}(r^2+a^2)^{-2}\right]\Re(iq v\frac{d \overline{v}}{dr})\,dr\\
		=\frac{1}{2}\lim_{r_*\to -\infty}\left[\Delta^{-\frac{p}{2}}|v'|^2(r_*)-\frac{1}{2}\Delta^{-\frac{p}{2}}V_{q,\varpi}|v|^2(r_*)\right]+\int_{\R} \Delta^{-\frac{p}{2}}\chi_+\Re(v'e^{-i (\widetilde{\omega}+\varpi) r_*}\Delta^{-iq}\overline{H})\,dr_*,
		\end{multline}
		with
		\begin{align*}
q=&\:\frac{(3-a^2)m}{4a},\\
\varpi=&-\frac{(1-a^2)(3-a^2)m}{4a(1+a^2)^2}.
\end{align*}
By the boundary condition on $u$ at $r=-\infty$ and \eqref{eq:Vtildenearhor}, the limit on the right-hand side of \eqref{eq:rmin1pestaux} is equal to zero for $p<2$. Furthermore, there exists a constant $b>0$ such that
\begin{equation*}
	\frac{1}{2}\frac{d}{dr}(\Delta^{-\frac{p}{2}}V_{q,\varpi})\geq \frac{1}{8}(2-p)\Lambda\left(1-\frac{m^2}{\Lambda}-4\frac{m\tomega}{\Lambda}\right )\frac{d\Delta}{dr} \Delta^{-\frac{p}{2}}.
\end{equation*}
If we assume that $\Lambda>(1+\epsilon)m^2$ for some $\epsilon>0$ and that $R_+-1$ is suitably small, and, in the case $m\tomega >0$, we assume additionally that $\frac{m\tomega}{\Lambda}\ll 1$, we in fact obtain:
\begin{equation}
\label{eq:pothighfreqLambda}
	\frac{1}{2}\frac{d}{dr}(\Delta^{-\frac{p}{2}}V_{q,\varpi})\geq b(2-p)\Lambda\frac{d\Delta}{dr}\Delta^{-\frac{p}{2}}.
\end{equation}

We apply Young's inequality to estimate
\begin{multline*}
	\left|\Delta^{-\frac{p}{2}}\left[ \frac{d^2 \Delta}{dr^2}(r^2+a^2)^{-1}-2r \frac{d \Delta}{dr}(r^2+a^2)^{-2}\right]\Re(iq v\frac{d\overline{v}}{dr})\right|\\
	\leq 2r\Delta^{-\frac{p}{2}}(1-2(r-1)-(1-a^2))\frac{|m|(1-2(1-a)+O(1-a^2))}{2}|v|\left|\frac{dv}{dr}\right|\\
	\leq r\Delta^{-\frac{p}{2}}|m||v|\left|\frac{dv}{dr}\right|\\
	\leq \frac{p}{8}\Delta^{-\frac{p}{2}-1} \frac{d\Delta}{dr}|v'|^2+ \frac{8 m^2r^2\Delta}{p(\frac{d\Delta}{dr})^2}\Delta^{-\frac{p}{2}}\frac{d\Delta}{dr}|v|^2
\end{multline*}
and note that we can bound uniformly in $a\leq 1$:  $\frac{\Delta}{(\frac{d\Delta}{dr})^2}\leq \frac{1}{4}$. We then obtain:
	\begin{multline*}
		b\int_{1}^{R_+} p\Delta^{-\frac{p}{2}-1}(2(r-1)+1-a^2)|v'|^2+(2-p)\Lambda \Delta^{-\frac{p}{2}}(2(r-1)+1-a^2)|v|^2\,dr\\
		\leq \int_{1}^{R_{+}}m^2 \Delta^{-\frac{p}{2}}(2(r-1)+1-a^2)|v|^2\,dr+\int_{R_{+}}^{R_{+}+1}\Lambda r^{p-2}|v|^2\,dr-\int_{1}^{R_++1} \Delta^{-\frac{p}{2}}\chi \Re\left(\frac{dv}{dr}e^{-i (\widetilde{\omega}+\varpi) r_*}\Delta^{-iq}\overline{H}\right)\,dr.
		\end{multline*}

We conclude \eqref{eq:rmin1pestfreq} by noting that $|v|=|u|$ and
\begin{multline*}
	\Delta^{-\frac{p}{2}-1}(r-1)|v'|^2\geq \frac{1}{2}\Delta^{-\frac{p}{2}-1}(r-1)|u'+i\tomega u|^2-C(1-a)^2m^2\Delta^{-\frac{p}{2}-1}(r-1)|u|^2-Cm^2\Delta^{-\frac{p}{2}-1}(r-1)^3|u|^2\\
\geq 	\frac{1}{2}\Delta^{-\frac{p}{2}-1}(r-1)|u'+i\tomega u|^2-Cm^2\Delta^{-\frac{p}{2}}(2(r-1)+(1-a^2))|u|^2
\end{multline*}

Finally, \eqref{eq:rmin2pest} follows from \eqref{eq:rmin1pestfreq} by absorbing the integral over $[1,R_+]$ on the right-hand side, using that $m^2\Lambda^{-1}\ll 2-p$ and applying \eqref{eq:pothighfreqLambda}.

\end{proof}

\subsubsection{Partitioning the frequency ranges}
We will split up the ranges of the frequencies $(\omega, \Lambda, m)$, in order to separate different difficulties that show up when trying to establish integrated energy estimates.

 We consider $a\geq a_0$, with $a_0$ suitably large.  Let $\alpha,\epsilon_{\rm width},\omega_{\rm high},\eta_{\flat}>0$ be constants. We will choose $\eta_{\flat}$ suitably small, $\epsilon_{\rm width}<\alpha^2\ll\eta_{\flat}^2$ and $\omega_{\rm high}^2\gg \epsilon_{\rm width}^{-1}$. Then we define:
\begin{align*}
\mathcal{G}^{\sharp}_{\rm sr}:=&\:\bigg\{(\omega,m,\Lambda)\:\textnormal{admissible}\,\big|\,  \Lambda\geq \left(\frac{a}{1+a}+\alpha\right)^{-2} \omega^2_{\rm high},\: m\omega\in \left(0,m \upomega_++\alpha \sqrt{\Lambda m^2}\right),\:\textnormal{and:} \\ 
 & \Lambda\geq \left(\frac{a}{1+a}+\alpha\right)^{-2} \omega^2_{\rm high} m^2\,\textnormal{or}\, \left|m\widetilde{\omega}\right|\geq \eta_{\flat} m^2 \bigg\},\\
\mathcal{G}^{\sharp, \omega}_{\rm nsr}:=&\:\left\{(\omega,m,\Lambda)\:\textnormal{admissible}\,\big|\,  |\omega|\geq  \omega_{\rm high},\: \Lambda<\epsilon_{\rm width} \omega^2,\: m\omega\notin \left(0,m \upomega_++\alpha  \sqrt{\Lambda m^2}\right)\right\},\\
\mathcal{G}^{\sharp, \Lambda}_{\rm nsr}:=&\:\left\{(\omega,m,\Lambda)\:\textnormal{admissible}\,\big|\, \Lambda\geq  \epsilon^{-1}_{\rm width}\omega_{\rm high}^2,\: \Lambda>\epsilon_{\rm width}^{-1}\omega^2,\: m\omega\notin \left(0,m \upomega_++\alpha  \sqrt{\Lambda m^2}\right)\right\},\\
\mathcal{G}^{\sharp, \sim}_{\rm nsr}:=&\:\bigg\{(\omega,m,\Lambda)\:\textnormal{admissible}\,\big|\, |\omega|\geq \omega_{\rm high},\:\epsilon_{\rm width} \Lambda \leq \omega^2 \leq \epsilon^{-1}_{\rm width}\Lambda, \: m\omega\notin \left(0,m \upomega_++\alpha  \sqrt{\Lambda m^2}\right)\bigg\},\\
\mathcal{G}^{\flat}:=&\:\left\{(\omega,m,\Lambda)\:\textnormal{admissible}\,\big|\,  |\omega|< \omega_{\rm high}|m|,\: \Lambda\leq \epsilon_{\rm width}^{-1}\omega^2_{\rm high}m^2\right\}\setminus (\mathcal{G}^{\sharp}_{\rm sr}\cup \mathcal{G}^{\sharp, \sim}_{\rm nsr}).
\end{align*}
It is straightforward to verify that, given the assumptions on $\alpha,\epsilon_{\rm width},\omega_{\rm high}$, the frequency ranges $\mathcal{G}^{\sharp}_{\rm sr},\mathcal{G}^{\sharp, \Lambda}_{\rm nsr},\mathcal{G}^{\sharp, \sim}_{\rm nsr}$ and $\mathcal{G}^{\flat}$ partition the full set of admissible frequencies; see \cite[Lemma 8.1.1]{part3} proving this in the case of a similar partition.

We will split $\mathcal{G}^{\flat}$ further into two parts:
\begin{align*}
	\mathcal{G}^{\flat}_{<}:=&\:\left\{(\omega,m,\Lambda)\in \mathcal{G}^{\flat}\,\big|\,  |\omega|\leq \omega_{\rm high} |m|,\: \Lambda\leq \epsilon_{\rm width}^{-1}\omega^2_{\rm high}m^2,\: \left|m\widetilde{\omega}\right|<\eta_{\flat}m^2\right\},\\
\mathcal{G}^{\flat}_{>}:=&\:\left\{(\omega,m,\Lambda)\in \mathcal{G}^{\flat}\,\big|\,  |\omega|\leq \omega_{\rm high},\: \Lambda\leq \epsilon_{\rm width}^{-1}\omega^2_{\rm high},\: \left|m\widetilde{\omega}\right|\geq \eta_{\flat} m^2\right\}.
\end{align*}
Observe that $\mathcal{G}^{\flat}_{>}$ is a bounded frequency range and $\mathcal{G}^{\flat}_{<}$ is only a bounded frequency range if $|m|$ is bounded.

We will exclude the case $m=0$, as that has already been treated in \cite{aretakis3}. Note the following points:
\begin{itemize}
\item In \S \ref{eq:Lambdasr}--\S \ref{eq:nsrtrapped}, \textbf{we will not make any further (boundedness) assumptions on $m$}. 
\item The estimates in \S \S \ref{sec:homsoln}--\ref{sec:freqnearsr} will only be useful with the additional assumption of boundedness of $m$. That is to say, \textbf{we will not keep careful track of the dependence $m$-dependence of constants in the estimates derived in those sections}.
\end{itemize}

We will employ the following notation to indicate integration and summation over the bounded frequency ranges at fixed $m$:
\begin{align*}
\int_{\mathcal{B}_<}\sum_{\ell} (\cdot)\,d\tomega=:&\:\int_{\{|\omega| \leq  \omega_{\rm high} |m|\}\cap\{\left|m\widetilde{\omega}\right|<\eta_{\flat}m^2\}}\sum_{\{\ell\in \N_{|m|}\,|\,\Lambda\leq \epsilon_{\rm width}^{-1}\omega^2_{\rm high}m^2 \}} (\cdot)\,d\tomega,\\
\int_{\mathcal{B}_>}\sum_{\ell}(\cdot)\,d\tomega=:&\:\int_{\{|\omega| \leq  \omega_{\rm high} |m|\}\cap\{\left|m\widetilde{\omega}\right|\geq \eta_{\flat}m^2\}}\sum_{\{\ell\in \N_{|m|}\,|\,\Lambda\leq \epsilon_{\rm width}^{-1}\omega^2_{\rm high}m^2\} } (\cdot)\,d\tomega.
\end{align*}

\subsubsection{$\mathcal{G}^{\sharp}_{\rm sr}$: superradiant, $\Lambda$-dominated frequencies}
\label{eq:Lambdasr}
In this section, we deal with the case of superradiant frequencies for which either $\Lambda\gg m^2$ or $m\omega$ is bounded away from the superradiant threshold frequency $m\upomega_+$.
\begin{proposition}
	\label{prop:Lambdadomsr}
	Let $(\omega,m,\Lambda)\in \mathcal{G}^{\sharp}_{\rm sr}$. For $a\leq 1$ with $1-a$ suitably small, $\epsilon>0$ and $E,\omega_{\rm high}$ suitably large, there exists a constant $c>0$ such that:
	\begin{enumerate}[label=\emph{(\roman*)}]
	\item For $\Lambda>\left(\frac{a}{1+a}+\alpha\right)\omega_{\rm high}^2m^2$:
	\begin{multline}
\label{eq:Ldomestwithbt}
	c\int_{R_+}^{R_{\infty}}|u'|^2+(\Lambda+\tomega^2)|u|^2\,dr\\
	\leq \int_{\R}\Re((f'+h)u \overline{H}+2fu' \overline{H})+(r^2-a^2)^{-1}\Delta^{\frac{\epsilon}{2}}\chi_+ \Re\left(-\frac{dv}{dr}e^{-i (\widetilde{\omega}+\varpi) r_*}\Delta^{-iq}\overline{H}\right) \,dr_*\\
	+E\int_{\R}(\chi_K\tomega+\chi_T\omega)\Re(i u \overline{H})\,dr_*,
\end{multline}
where $f,h,\chi_K,\chi_T$ are defined in the proof below and $\chi_+$ is defined in the proof of Proposition \ref{prop:rweightedest}.

	\item For $\Lambda\leq \left(\frac{a}{1+a}+\alpha\right)\omega_{\rm high}^2m^2$ and $|m\tomega|\geq \eta_{\flat}m^2$:
	\begin{equation}
\label{eq:Ldomestwithbt2}
	c\int_{1}^{R_{\infty}}|u'|^2+(\Lambda+\tomega^2) |u|^2\,dr\\
	\leq \int_{\R}\Re((f'+h)u \overline{H}+2fu' \overline{H})+E\int_{\R}(\chi_K\tomega+\chi_T\omega)\Re(i u \overline{H})\,dr_*,
\end{equation}
with the functions $f,h,\chi_K,\chi_T$ defined in the proof below.
\end{enumerate}
\end{proposition}
\begin{proof}
	The proof consists of \textbf{three steps} which each feature a different type of current.\\
	
	\paragraph{\textbf{Step 1}: a $j^f_3$-current}
	For $\omega_{\rm high}$ suitably large, properties (i) and (iii) of Lemma \ref{prop:highLambdapotential} are satisfied for frequencies $(\omega,m,\Lambda)\in \mathcal{G}^{\sharp}_{\rm sr}$ with $\Lambda>\left(\frac{a}{1+a}+\alpha\right)\omega_{\rm high}^2m^2$, which implies: for $\omega_{\rm high}$ suitably large and $\alpha$ suitably small, there exists constants $b,B>0$ such that
	\begin{align*}
	(r_{\rm max}-r)\frac{dV}{dr}(r)\geq &\: b\Lambda r^{-5}(r-1)(r-r_{\rm max})^2\quad  \textnormal{if $m\tomega\leq 0$},\\
	(r_{\rm max}-r)\frac{dV}{dr}(r)\geq &\: b\Lambda r^{-5}(r-1)(r-r_{\rm max})^2-Bm\tomega r^{-4}(r-r_{\rm max})^2 \quad \textnormal{if $m\tomega>0$}
	\end{align*}
	with $r_{\rm max}$ arbitrarily close to $1+\sqrt{2}$, given $\omega_{\rm high}$ suitably large and $\alpha$ suitably small. We will refer to the condition $\Lambda>\left(\frac{a}{1+a}+\alpha\right)\omega_{\rm high}^2m^2$ as \textbf{Condition A}.
	
	We refer to Figure 1(B) for the shape of the potential when $m\tomega<0$ (as $m\tomega \uparrow 0$, the shape approaches Figure 1(E)) and Figure 1(C) when $m\tomega>0$ (approaching the Figure 1(E) as $m\tomega \downarrow 0$).
	
	If $(\omega,m,\Lambda)\in \mathcal{G}^{\sharp}_{\rm sr}$ with $|m\tomega|\geq \eta_{\flat}m^2$ and $\Lambda \leq \left(\frac{a}{1+a}+\alpha\right)\omega_{\rm high}^2m^2$, we use that $\alpha\ll\eta_{\flat}$ to infer that $m\tomega<0$, so property (iv) of Lemma \ref{prop:highLambdapotential} is satisfied. Hence,
	\begin{equation*}
	(r_{\max}-r)\frac{dV}{dr}(r)\geq  b_{\eta_{\flat}} r^{-4}(r-r_{\max})^2  \Lambda,
	\end{equation*}
	with $1<r_{\max}=r_{\max}(\eta_{\flat})<B$ for some suitably large constant $B>0$. We will refer to the condition  $|m\tomega|\geq \eta_{\flat}m^2$ and $\Lambda \leq \left(\frac{a}{1+a}+\alpha\right)\omega_{\rm high}^2m^2$ as \textbf{Condition B}.
	
	We will construct a function $f$ that is adapted to $r_{\rm max}$, but does not take into account the possible presence of $r_{\rm min}$ in the $m\tomega>0$ case given Condition A.

	Consider Condition A with $m\tomega<0$ or Condition B. We can then construct an appropriate $f:(1,\infty)\to \R$ as follows: let $R_{\infty}>3$ be arbitrarily large and define $c_0=r_0^{-1}\frac{\arctan(r_*(r_1)-r_*(r_{\rm max}))}{\arctan(r_*(r_{\rm max})-r_*(r_0))}$, then $c_0>0$ and we require
	\begin{equation*}
	f(r):=\begin{cases}
	-r^{-1} & 1\leq r\leq 1+\frac{1}{4}(r_{\rm max}-1)=:r_0\\
	\frac{1}{2}r_0^{-1}\frac{\arctan(r_*(r)-r_*(r_{\rm max}))}{\arctan(r_*(r_{\rm max})-r_*(r_0))}& 1+\frac{1}{2}(r_{\rm max}-1)\leq r\leq 1+2(r_{\rm max}-1)=:r_1 \\
	\frac{3}{4}c_0+\frac{c_0}{8}\frac{r-r_2}{R_{\infty}-r_2}&   r_2:=1+3(r_{\rm max}-1) <r\leq R_{\infty} \\
	c_0 & r> R_{\infty}+1
	\end{cases}
	\end{equation*}
	It is possible to construct a smooth $f$ satisfying the above properties with moreover $\frac{df}{dr}>0$ for $r\leq R_{\infty}$ via a standard mollification argument. 
	
	In the case of Condition A with $m\tomega>0$, we require instead
	\begin{equation*}
	f(r):=\begin{cases}
	(r-1) & 1\leq r\leq 1+\frac{1}{2}\delta,\\
	-r^{-1} & 1+\delta\leq r\leq 1+\frac{1}{4}(r_{\rm max}-1)=:r_0\\
	\frac{1}{2}r_0^{-1}\frac{\arctan(r_*(r)-r_*(r_{\rm max}))}{\arctan(r_*(r_{\rm max})-r_*(r_0))}& 1+\frac{1}{2}(r_{\rm max}-1)\leq r\leq 1+2(r_{\rm max}-1)=:r_1 \\
	\frac{3}{4}c_0+\frac{1}{8}c_0\frac{r-r_2}{R_{\infty}-r_2}&   r_2:=1+3(r_{\rm max}-1) <r\leq R_{\infty} \\
	c_0 & r> R_{\infty}+1.
	\end{cases}
	\end{equation*}
	Note that we do not have that $f'>0$ in $1\leq r\leq 1+\delta$. We do have that $|\frac{df}{dr}|\leq \frac{1}{2}B\delta^{-1}$ for some constant $B>0$.
	
	We now integrate $(j^f_3[u])'$ to obtain:
	\begin{multline}
	\label{eq:fcurrentaux}
		\int_{\R} 2f'|u'|^2-fV'|u|^2-\frac{1}{2}f'''|u|^2\,dr_*=c_0(|u'|^2+\omega^2|u|^2)(\infty)+(|u'|^2+\tomega^2|u|^2)(-\infty)\\
		+\int_{\R}\Re(f'u \overline{H}+2fu' \overline{H})\,dr_*.
	\end{multline}
	 Furthermore, there exist constants $c,C>0$ such that under Condition A
		\begin{align*}
	-f\frac{dV}{dr}(r)\geq &\: c\Lambda r^{-5}(r-1)(r-r_{\rm max})^2\quad  \textnormal{if $m\tomega\leq 0$},\\
-f\frac{dV}{dr}(r)\geq &\: c\Lambda r^{-5}(r-1)(r-r_{\rm max})^2-C\frac{m\tomega}{\Lambda} \Lambda r^{-4}(r-r_{\rm max})^2 \quad \textnormal{if $m\tomega>0$}.
	\end{align*}
	We moreover have that $f'''<0$ in $(r_{\rm max}-\delta,r_{\rm max}+\delta)$, for $\delta>0$ suitably small and $f'''$ is dominated by $-fV'$ when $r\notin (r_{\rm max}-\delta,r_{\rm max}+\delta)$, since $\Lambda\gg 1$. Note finally, that in the case $m\tomega>0$, $\frac{m\tomega}{\Lambda}\leq \frac{m\tomega}{\sqrt{m^2\Lambda}}<\alpha$ can be made arbitrarily small by taking $\alpha$ suitably small. Hence, given $R_+$ as in Proposition \ref{prop:rweightedest}, there exists $\alpha$ suitably small such that:
	\begin{align*}
	-f\frac{dV}{dr}(r)\geq &\: c\Lambda r^{-5}(r-1)(r-r_{\rm max})^2 & \textnormal{if $m\tomega>0$ and $r\geq R_+$}\\
	\left|f\frac{dV}{dr}(r)\right|\leq &\:C\alpha \Lambda & \textnormal{if $m\tomega>0$ and $1\leq r\leq R_+$}.
	\end{align*}
	
	Under Condition B, we obtain
	\begin{equation*}
	-f\frac{dV}{dr}(r)\geq  c  \Lambda r^{-4}(r-r_{\rm max})^2.
	\end{equation*}
	As in the case of Condition A,  $f'''<0$ in $(r_{\rm max}-\delta,r_{\rm max}+\delta)$, for $\delta>0$ suitably small and $f'''$ is dominated by $-fV'$ when $r\notin (r_{\rm max}-\delta,r_{\rm max}+\delta)$.\\

\paragraph{\textbf{Step 2}: a $j^h_1$-current}
By (ii) and (iv) of Lemma \ref{prop:highLambdapotential}, for either Condition A or B, there exists a $\delta>0$ and $b>0$ suitably small, such that
\begin{equation*}
	V(r)-\omega^2\geq b\Lambda \quad \textnormal{for all $r\in (r_{\rm max}-\delta,r_{\rm max}+\delta)$.}
\end{equation*}
We now let $h:(1,\infty)\to [0,\infty)$ be a smooth function, such that:
\begin{equation*}
	h(r)=\begin{cases}
		A & (r_{\rm max}-\frac{\delta}{2},r_{\rm max}+\frac{\delta}{2})\\
		0 & r\notin (r_{\rm max}-\delta,r_{\rm max}+\delta),
	\end{cases}
\end{equation*}
with $A>0$ a constant that will be chosen suitably below.
By adding $(j_1^h[u])'$ to the integral of $(j_3^f[u])'$, we obtain:
\begin{multline*}
		\int_{\R} (2f'+h)|u'|^2-fV'|u|^2+h(V-\omega^2)|u|^2-\frac{1}{2}f'''|u|^2-\frac{1}{2}h''|u|^2\,dr_*=c_0(|u'|^2+\omega^2|u|^2)(\infty)\\
		+(|u'|^2+\tomega^2|u|^2)(-\infty)+\int_{\R}\Re((f'+h)u \overline{H}+2fu' \overline{H})\,dr_*.
	\end{multline*}
Let $R_+>1$. By choosing $\omega_{\rm high}$ suitably large compared to $A$ and $\alpha$ suitably small compared to $R_+-1$, the left-hand side above is non-negative negative for $r\in [R_+,\infty)$ and we obtain particular under Condition A:
\begin{multline}
\label{eq:LdomestwithbtA0}
	c\int_{R_+}^{R_{\infty}}|u'|^2+\Lambda|u|^2\,dr+A\int_{r_{\rm max}-\frac{\delta}{2}}^{r_{\rm max}+\frac{\delta}{2}}\frac{r^2+a^2}{\Delta}(|u'|^2+b \Lambda|u|^2)\,dr\\
	\leq (|u'|^2+\omega^2|u|^2)(\infty)+(|u'|^2+\tomega^2|u|^2)(-\infty)\\
		-C(\alpha+\delta) \int_{1}^{1+\delta}\Lambda (r-1) |u|^2+ B\delta^{-1}\int_{1+\frac{\delta}{2}}^{1+\delta} |u'|^2\,dr+\int_{\R}\Re((f'+h)u \overline{H}+2fu' \overline{H})\,dr_*.
\end{multline}
Note moreover that we can further estimate
\begin{equation}
\label{eq:auxnearhorerror}
\begin{split}
B\delta^{-1}\int_{1+\frac{\delta}{2}}^{1+\delta} |u'|^2\,dr\leq &\:  2B\int_{1+\frac{\delta}{2}}^{1+\delta} \delta^{-1}|u'+i\tomega u|^2+\delta^{-1}\tomega^2 |u|^2\,dr\\
\leq &\: 2B\int_{1+\frac{\delta}{2}}^{1+\delta} \delta^{1-\epsilon}(r-1)^{-2+\epsilon}|u'+i\tomega u|^2+\alpha^2\left(\frac{\delta}{2}\right)^{-1-\epsilon}(r-1)^{\epsilon}\Lambda |u|^2 \,dr.
\end{split}
\end{equation}

Under Condition B, we obtain instead
\begin{multline}
\label{eq:LdomestwithbtB}
	c\int_{1}^{R_{\infty}}|u'|^2+\Lambda|u|^2\,dr\leq (|u'|^2+\omega^2|u|^2)(\infty)+(|u'|^2+\tomega^2|u|^2)(-\infty)+\int_{\R}\Re((f'+h)u \overline{H}+2fu' \overline{H})\,dr_*.
\end{multline}

We further estimate the right-hand side of \eqref{eq:LdomestwithbtA0}, using \eqref{eq:auxnearhorerror}, by combining it with \eqref{eq:rmin2pest} with $p=1-\epsilon$ and taking $\alpha\ll \delta$ suitably small to conclude:
\begin{multline}
\label{eq:LdomestwithbtA}
	c\int_{R_+}^{R_{\infty}}|u'|^2+\Lambda|u|^2\,dr+A\int_{r_{\rm max}-\frac{\delta}{2}}^{r_{\rm max}+\frac{\delta}{2}}\frac{r^2+a^2}{\Delta}(|u'|^2+b \Lambda|u|^2)\,dr\\
	\leq (|u'|^2+\omega^2|u|^2)(\infty)+(|u'|^2+\tomega^2|u|^2)(-\infty)\\
		+\int_{\R}\Re((f'+h)u \overline{H}+2fu' \overline{H})-(r^2-a^2)^{-1}\Delta^{\frac{\epsilon}{2}}\chi_+ \Re\left(\frac{dv}{dr}\overline{H}\right) \,dr_*.
\end{multline}
Note finally that we can add $\tomega^2|u|^2$ inside the integral on the left-hand side, since $\tomega^2\leq \alpha^2 \Lambda$.\\

\paragraph{\textbf{Step 3}: $j^K$- and $j^T$-currents}
In the final step of the proof, we estimate the boundary terms at $r=-\infty$ and $r=\infty$ appearing on the right-hand side of \eqref{eq:LdomestwithbtB} and \eqref{eq:LdomestwithbtA}. For $m\tomega>0$ this immediately follows by integrating $(j^K[u])'=0$ and $(j^T[u])'=0$. For $m\tomega<0$ we need to modify the currents.

First, we consider $(E\chi_Kj^K[u])'$, where $E\geq 2c_0$ is a constant and $\chi_K$ denotes a smooth cut-off function such that
\begin{equation*}
	\chi_K(r)=\begin{cases}
		1 & r\in [1,r_{\rm max}-\frac{\delta}{2}]\\
		0 & r\in [r_{\rm max}+\frac{\delta}{2},\infty),
	\end{cases}
	\end{equation*}
	with $\left|\frac{d\chi}{dr}\right|\leq B\delta^{-1}$ for some $B>0$ independent of $\delta$.
	Then we can express:
	\begin{multline*}
		(|u'|^2+\tomega^2|u|^2)(-\infty)=E\int_{\R}(\chi_K j^K[u])'\,dr_*=E\int_{r_{\rm max}-\frac{\delta}{2}}^{r_{\rm max}+\frac{\delta}{2}}\frac{d\chi}{dr} \tomega \Re(iu' \overline{u})\,dr+E\int_{\R}\chi_K\tomega\Re(i u \overline{H})\,dr_*\\
		\leq EB\delta^{-1}\int_{r_{\rm max}-\frac{\delta}{2}}^{r_{\rm max}+\frac{\delta}{2}} |u'|^2+\tomega^2|u|^2\,dr+E\int_{\R}\chi_K\tomega\Re(i u \overline{H})\,dr_*\\
		\leq \frac{A}{4}\int_{r_{\rm max}-\frac{\delta}{2}}^{r_{\rm max}+\frac{\delta}{2}}\frac{r^2+a^2}{\Delta}(|u'|^2+b \Lambda|u|^2)\,dr+2\int_{\R}\chi_K\tomega\Re(i u \overline{H})\,dr_*,
		\end{multline*}
		where the last inequality follows by taking $A$ suitably large compared to $EB\delta^{-1}$, which in turn requires $\omega_{\rm high}$ to be sufficiently large.
		
		Analogously, we consider we consider $(E\chi_Tj^T[u])'$, where $\chi_T$ denotes a smooth cut-off function such that
\begin{equation*}
	\chi_T(r)=\begin{cases}
		0 & r\in [1,r_{\rm max}-\frac{\delta}{2})\\
		1 & r\in [r_{\rm max}+\frac{\delta}{2},\infty),
	\end{cases}
	\end{equation*}
	with $\left|\frac{d\chi}{dr}\right|\leq B\delta^{-1}$ and we repeat the argument above to control $(|u'|^2+\omega^2|u|^2)(\infty)$.
		\end{proof}
	\subsubsection{$\mathcal{G}^{\sharp,\omega}_{\rm nsr}$: non-superradiant, $\omega$-dominated frequencies}
	In this section, we consider non-superradiant frequencies with $\omega^2\gg \Lambda$ and we make use of the favourable sign of $\omega^2-V$ in this frequency regime.
	\begin{proposition}
	\label{prop:omegadom}
	Let $(\omega,m,\Lambda)\in \mathcal{G}^{\sharp,\omega}_{\rm nsr}$. For $a\leq 1$, $E\geq 2$, $\delta>0$, and for $\epsilon_{\rm high}$ suitably small, there exists a constant $c>0$ such that
	\begin{equation}
\label{eq:Ldomestwithbtother}
	c\int_{1}^{R_{\infty}}(1-r^{-1})^{-1+\delta}r^{-1-\delta}|u'|^2+(\omega^2+\Lambda)(1-r^{-1})^{-1+\delta}r^{-3}|u|^2\,dr\\
	\leq -\int_{\R} 2y\Re(u'\overline{H})\,dr_*	+E\int_{\R}\omega \Re(i u \overline{H})\,dr_*,
\end{equation}
with the function $y$ is defined in the proof below.
\end{proposition}
	\begin{proof}
		The proof will proceed in \textbf{two steps}.\\
		
		\paragraph{\textbf{Step 1}: a $j^y_2[u]$-current}
		By \eqref{eq:Vnearinf}, there exists a $R_{\infty}>2$, arbitrarily large and uniform in $a,\Lambda, m,\omega$, such that:
		\begin{equation*}
			\frac{dV}{dr}(r)\leq -\Lambda r^{-3}
		\end{equation*}
		for all $r\geq R_{\infty}-1$. We now define $y:(1,\infty)\to \R_{+}$ as follows:
		\begin{equation*}
			y(r)=\begin{cases}
			1+\frac{1}{2}\frac{(r-1)^{\delta}}{(R_{\infty}-1)^{\delta}} & r\leq R_{\infty}\\
			2-r^{-\delta}& r\geq R_{\infty}+1,
			\end{cases}
		\end{equation*}
		with $\delta>0$. Via a standard mollification argument, we can ensure $y$ is smooth and $\frac{dy}{dr}>0$ for $r\leq R_{\infty}$. Integrating $(j^y_2[u])'$ gives
		\begin{multline}
		\label{eq:omegadomaux}
		\int_1^{\infty} \frac{dy}{dr}(|u'|^2+(\omega^2-V)|u|^2)-y\frac{dV}{dr}|u|^2\,dr+|u'|^2(-\infty)+\tomega^2|u|^2(-\infty)=2|u'|^2(\infty)+2\omega^2|u|^2(\infty)\\
		-\int_{\R} 2y\Re(u'\overline{H})\,dr_*	
		\end{multline}
		For $\omega_{\rm high}$ sufficiently large, depending on $\delta$ and $R_{\infty}$:
		\begin{equation*}
		\frac{dy}{dr}(\omega^2-V)|u|^2\geq  c (r-1)^{-1+\delta}(\omega^2+\Lambda) \quad \textnormal{for $r\in [1,R_{\infty})$}.
		\end{equation*}
		\paragraph{\textbf{Step 2}: the $j^T[u]$-current}
		It remains to estimate the boundary terms on the right-hand side of \eqref{eq:omegadomaux}. Since the frequency range is non-superradiant, we can simply consider the current $EJ^T[u]$, with $E\geq 2$ a constant. Indeed, this gives:
		\begin{equation*}
			(|u'|^2+\omega^2|u|^2)(\infty)+E\omega \tomega |u|^2(-\infty)\leq E\int_{\R}\omega \Re(iu \overline{H})\,dr_*.
		\end{equation*}
		Note that $\omega \tomega=\omega(\omega-m\upomega_+ )>0$, since $\omega^2\gg\Lambda>m^2$, so the left-hand side above is non-negative definite.
	\end{proof}
	
	\subsubsection{$\mathcal{G}^{\sharp,\Lambda}_{\rm nsr}$: non-superradiant, $\Lambda$-dominated frequencies}
	In this section, we consider non-superradiant frequencies with $\Lambda\gg \omega^2$. The strategy for obtaining integrated estimates may be thought of as a simplification of the strategy employed in the proof of Proposition \ref{prop:Lambdadomsr}.
	\begin{proposition}
	\label{prop:Lambdadomnsr}
	Let $(\omega,m,\Lambda)\in \mathcal{G}^{\sharp, \Lambda}_{\rm nsr}$. For $a\leq 1$ with $1-a$ suitably small, $E\geq 2$ and $\omega_{\rm high}$ suitably large, there exists a constant $c>0$ such that
	\begin{equation}
\label{eq:Ldomestwithbtnsr}
	c\int_{R_+}^{R_{\infty}}|u'|^2+(\Lambda+\omega^2)|u|^2\,dr\\
	\leq \int_{\R}\Re((f'+h)u \overline{H}+2fu' \overline{H}) \,dr_*,
\end{equation}
with the functions $f$ and $h$ defined in the proof of Proposition \ref{prop:Lambdadomsr}.
\end{proposition}
\begin{proof}
Suppose $m\omega\geq 0$. Then $m\omega\geq m\upomega_++\alpha \sqrt{\Lambda m^2}$, which implies that $\omega^2>\alpha^2\Lambda$. But we also have that $\omega^2 < \epsilon_{\rm width}\Lambda$ for $(\omega,m,\Lambda)\in \mathcal{G}^{\sharp, \Lambda}_{\rm nsr}$.  For $0<\epsilon_{\rm width}<\alpha^2$ suitably small, this is a contradiction. It must therefore hold that $m\omega<0$.

Hence, as in the proof of Proposition \ref{prop:Lambdadomsr} (Condition B), we can apply (iv) of Lemma \ref{prop:highLambdapotential} to infer that there exist $\delta>0$ and $b>0$ suitably small, such that
\begin{align*}
(r_{\max}-r)\frac{dV}{dr}(r)\geq &\:  b r^{-4}(r-r_{\max})^2 m^2\quad \textnormal{for $r\in [1,\infty)$},\\
V(r)-\omega^2\geq &\: b \Lambda\quad \textnormal{for $r\in (r_{\rm max}-\delta,r_{\rm max}+\delta)$}.
\end{align*}
We can therefore repeat the steps in the proof of Proposition \ref{prop:Lambdadomsr} (Condition B), with the additional simplification that the boundary terms at $r_*=\pm \infty$ can immediately be estimated, without the need for cut-off functions, since $m\omega<0$ guarantees that we are outside of the superradiant frequency regime.
\end{proof}
	
	\subsubsection{$\mathcal{G}^{\sharp, \sim}_{\rm nsr}$: non-superradiant, trapped frequencies}
\label{eq:nsrtrapped}
By construction, the frequencies in $\mathcal{G}^{\sharp, \sim}_{\rm nsr}$ are non-superradiant. This frequency range is affected by the existence of trapped null geodesics (away from the event horizon). The corresponding integrated estimates from \cite[Proposition 8.6.1]{part3} in the case $a<1$ hold without modification when $a\leq 1$.
\begin{proposition}
	\label{prop:trappedfreq}
	Let $(\omega,m,\Lambda)\in \mathcal{G}^{\sharp, \sim}$. For $a\leq 1$, $E\geq 2$ and $\omega_{\rm high}$ suitably large, there exist constants $b,c>0$ and $r_1\in (1+b,\infty]$ such that:
		\begin{enumerate}[label=\emph{(\roman*)}]
		\item
 For $r\in [1,r_1)$
	\begin{equation*}
	\omega^2-V(r)\geq b \Lambda.
	\end{equation*}
	Furthermore, if $r_1<\infty$, $V$ has a unique nondegenerate maximum $r_{\max}\in [r_1,\infty)$. Denote with $R_{\rm dec}>1$ a constant such that $\frac{dV}{dr}<0$ for $r\geq R_{\rm dec}$.
	\item
	\begin{multline}
\label{eq:Ldomestwithbtmore}
	c\int_{1}^{R_{\infty}}|u'|^2+(\Lambda+\omega^2)(r-1)(r-r_{\rm trap})^2|u|^2+|u|^2\,dr\\
	\leq \int_{\R}-\Re(f_{\sim}'u \overline{H}-2f_{\sim}u' \overline{H})+ 2\hat{y}\Re(u'\overline{H})\,dr_*\\
	+E\int_{\R}\omega\Re(i u \overline{H})\,dr_*,
\end{multline}
where $f_{\sim},\hat{y}: (1,\infty)\to \R$ are smooth functions that satisfy: for $r_1<R_{\rm dec}<\infty$, $r_{\rm trap}=r_{\rm max}$ and
\begin{align*}
f_{\sim}(1)=&\:0,\\
\frac{df_{\sim}}{dr}(r)\geq &\: 0\quad \textnormal{for $r\in [1,\infty)$},\\
-f_{\sim}V'-\frac{1}{2}f_{\sim}'''> &\: b \Lambda \Delta r^{-7}(r-r_{\rm max})^2\quad \textnormal{for $r\in [r_1,\infty)$},\\
f_{\sim}(r)=&\:1\quad \textnormal{for $r\in [R_{\rm dec},\infty)$},\\
|f_{\sim}|+\left|\frac{df_{\sim}}{dr}\right|\leq&\: B,\\
\hat{y}(r)=&\:1-e^{C(r_1-r)} \quad \textnormal{for $r\in [1,r_1)$},\\
\hat{y}(r)=&\:0\quad \textnormal{for $r\in [r_1,\infty)$}
\end{align*}
for appropriate positive constants $\delta,B,C>0$ with $E\gg C$ and $\delta>0$, $R_+-1$ suitably small.

For $R_{\rm dec}\leq r_1\leq \infty$, $r_{\rm trap }=0$ and 
\begin{align*}
f_{\sim}(r)=&\:0\quad \textnormal{for $r\in [1,R_{\rm dec}]$},\\
\frac{df_{\sim}}{dr}(r)\geq &\: 0\quad \textnormal{for $r\in [r_1,R_{\infty}]$},\\
f_{\sim}(r)=&\:1\quad \textnormal{for $r\in [R_{\rm dec}+1,\infty)$},\\
|f_{\sim}|+\left|f'_{\sim}\right|+\left|f'''_{\sim}\right|\leq&\: B,\\
\hat{y}(r)=&\:1-e^{\tilde{C}(R_{\rm dec}+2-r)}\quad \textnormal{for $r\in [1,R_{\rm dec}+2)$},\\
\hat{y}(r)=&\:0\quad \textnormal{for $r\in [R_{\rm dec}+2,\infty)$},
\end{align*}
for appropriate positive constants $B,\tilde{C}>0$ with $E\gg \tilde{C}$.
\item If $\Lambda > \omega_{\rm high} m^2$, then
\begin{equation*}
|r_{\rm max}-(1+\sqrt{2})|\leq C(\omega_{\rm high}^{-1}+(1-a^2)).
\end{equation*}
\end{enumerate}
\end{proposition}

	\begin{proof}
	Part (i) and (iii) follows from \cite[Lemma 8.6.1]{part3} and \cite[Proposition 8.6.1]{part3}. Part (iii) follows from (i) of Lemma \ref{prop:highLambdapotential}.
	\end{proof}

	\subsubsection{Homogeneous solutions}
	\label{sec:homsoln}
	In this section, we consider solutions to the homogeneous ODE:
	\begin{equation}
		\label{eq:homode}
		u''+(\tomega^2-\widetilde{V})u=0,
	\end{equation}
	with $\omega\in \R \setminus \{0,m\upomega_+\}$. We will derive weighted $L^{\infty}$ and $L^2$ estimates for solutions to \eqref{eq:homode}. 
	
	These constitute a key part towards deriving integrated estimates in the bounded frequency range $\mathcal{G}^{\flat}$ and they form the greatest departure from the frequency-space analysis in the $|a|<M$ setting in \cite{part3}.
	
	 \textbf{In this section, we will not keep track of the $m$-dependence in the uniform constants appearing in our integrated and pointwise estimates.} Because of this, our final integrated estimates in frequency space will lead to integrated energy estimates in physical space for azimuthal modes with bounded $m$, i.e.\ the estimates we derive here are not sufficient to be able to consider the sum over all $m$.
	
	We will denote with $u_{\rm hor}$ and $u_{\rm inf}$ solutions to \eqref{eq:homode} that satisfy the following boundary conditions:
	\begin{align}
		\label{eq:hombc1}
		\lim_{r_*\to -\infty} (u_{\rm hor}'+i\tomega u_{\rm hor})(r_*)=&\:0,\\
		\label{eq:hombc2}
		\lim_{r_*\to +\infty} (u_{\rm inf}'-i\omega u_{\rm inf})(r_*)=&\:0,\\
		\label{eq:hombc3}
		|u_{\rm hor}|(-\infty)=|u_{\rm inf}|(+\infty)=1.
	\end{align}
	Note in particular that the norms of $u_{\rm hor}$ and $u_{\rm inf}$ are fixed at $r_*=-\infty$ and $r_*=+\infty$, respectively.
	
		The main aim of Lemma \ref{lm:prelimhomest} below is to derive $L^{\infty}$-estimates for $u_{\rm hor}$ and $u_{\rm inf}$ sufficiently close to the event horizon at $r=1$. Since the irregular singular point $r=1$ for $\tomega\neq 0$ becomes a regular singular point in the limit $\tomega\to 0$, these $L^{\infty}$-estimates do not follow immediately from the expansions obtained from standard ODE theory. 
		
		The main strategy for resolving this difficulty is to rescale $r-1$ in such a way that the singular point stays irregular in the limit $\tomega\to 0$ in terms of the rescaled variable. As a result, the expansions obtained from standard ODE theory remain valid in the limit. This comes at the expense of the region of validity of the expansions \emph{shrinking to zero} as $\tomega\to 0$. More precisely, the $L^{\infty}$ estimates near $r=1$ obtained in this way will only be valid in the interval $(1,1+z_0|m|^{-1}|\tomega|)$, where $z_0$ can be taken arbitrarily large.
		
		Although we only require the estimates of Lemma \ref{lm:prelimhomest} for $a=1$, we will actually derive them uniformly for $a\leq 1$ with $1-a\geq 0$ sufficiently small. 
	
	\begin{lemma}
		\label{lm:prelimhomest}
Assume $|\omega|\geq \omega_0>0$ for some $\omega_0>0$. Let $z_0>0$ be arbitrarily large and assume that there exists a $B>0$ such that $\Lambda+|\omega|^2\leq Bm^2$. Then, there exists a numerical constant $B_{\rm hom}=B_{\rm hom}(B,z_0,m,\omega_0)>0$, such that for $1-a$ suitably small, solutions $u$ to \eqref{eq:homode} satisfy the bounds:
\begin{align}
\label{eq:hombound1}
||u_{\rm hor}||_{L^{\infty}_r(1,1+z_0|m|^{-1}|\tomega|)}+|\tomega|^{-1}||u'_{\rm hor}||_{L^{\infty}_r(1,1+z_0|m|^{-1}|\tomega|)}\leq &\: B_{\rm hom},\\
\label{eq:hombound2}
||u_{\rm inf}||_{L^{\infty}_r(1,1+z_0|m|^{-1}|\tomega|)}+|\tomega|^{-1}||u'_{\rm inf}||_{L^{\infty}_r(1,1+z_0|m|^{-1}|\tomega|)}\leq &\: \frac{B_{\rm hom}(|W|+\sqrt{|\omega|}\sqrt{|\tomega|})}{|\tomega|},\\
\label{eq:hombound3}
||u_{\rm inf}||_{L^{\infty}_{r}(1+z_0^{-1},\infty)}+||u_{\rm inf}'||_{L^{\infty}_{r}(1+z_0^{-1},\infty)}\leq &\: B_{\rm hom},\\
\label{eq:hombound4}
||u_{\rm hor}||_{L^{\infty}_r(1+z_0^{-1},\infty)}+|\omega|^{-1}||u'_{\rm hor}||_{L^{\infty}_r(1+z_0^{-1},\infty)}\leq &\: \frac{B_{\rm hom}(|W|+\sqrt{|\tomega|}\sqrt{|\omega|})}{|\omega|}.
\end{align}
	\end{lemma}
	\begin{proof}
		We can express
\begin{equation*}
r_*=(1+a^2)\left[-\frac{a^2}{1-a^2}\log\left(1+\frac{1-a^2}{r-1}\right)+\log(r-1)\right]+r-1,
\end{equation*}
so that
\begin{equation*}
	e^{\frac{-(1-a^2)}{a^2(1+a^2)}r_*}=\left(1+\frac{1-a^2}{r-1}\right)e^{-a^{-2}(1-a^2)(\log(r-1)+(r-1))}.
	\end{equation*}
By expanding the exponentials, we obtain
\begin{multline}
\label{eq:rmin1rstar}
	(r-1)^{-1}\left[1+\frac{r-1}{1-a^2}(e^{-a^{-2}(1-a^2)(\log(r-1)+(r-1))}-1)\right]\\
	=\left[-a^{-2}(1+a^{2})^{-1}r_*\right]\left(1+\sum_{k=1}^{\infty}\frac{(-a^{-2}(1+a^{2})^{-1}(1-a^2)r_*)^k}{(k+1)!}\right).
\end{multline}
By considering separately the cases $(r-1)\leq \epsilon (1-a^2)$ and $(r-1)> \epsilon(1-a^2)$ with $\epsilon>0$ arbitrarily small, and taking $\delta$ suitably small compared to $\epsilon$, we can estimate for $r-1<\delta$:
\begin{equation*}
\frac{r-1}{1-a^2}(1-e^{-a^{-2}(1-a^2)(\log(r-1)+(r-1))})\leq \epsilon.
\end{equation*}
Therefore, for $r_*<-1$, we can express:
\begin{equation}
\label{eq:rmin1rstar2}
	(r-1)\leq \frac{1+\epsilon}{a^{-2}(1+a^2)^{-2}|r_*|\left(1+\sum_{k=1}^{\infty}\frac{(a^{-2}(1+a^{2})^{-1}(1-a^2)|r_*|)^k}{(k+1)!}\right)}\leq \frac{(1+\epsilon)a^2(1+a^2)}{|r_*|},
\end{equation}
from which it also follows that for any $k\in \N_0$:
\begin{equation*}
(1-a^2)^k(r-1)=O(|r_*|^{-k-1}).
\end{equation*}
By taking derivatives in $r_*$ on both sides of \eqref{eq:rmin1rstar}, we moreover infer that:
\begin{equation}
\label{eq:rmin1rstar2better}
(1-a^2)^k(r-1)=O_{\infty}(|r_*|^{-k-1}).
\end{equation}
Consider the rescaled variable $x=|m|^{-1}|\tomega| r_*$. Then \eqref{eq:homode} is equivalent to:
		\begin{equation}
			\label{eq:homodex}
			\frac{d^2u}{dx^2}+m^2(1-\tomega^{-2}\widetilde{V})u=0. 
		\end{equation}
By applying \eqref{eq:rmin1rstar2better}, we obtain with respect to the coordinate $x$:
\begin{equation}
\label{eq:relrx}
(1-a^2)^k(r-1)=O_{\infty}(|m|^{-1-k} |\tomega|^{1+k}|x|^{-1-k}).	
\end{equation}
So, by \eqref{eq:V0nearhor2} and \eqref{eq:V1nearhor}, it follows that
\begin{align*}
\tomega^{-2}\widetilde{V}(r)=&\:O_{\infty}(|x|^{-1})
\end{align*}

	The expressions for the rescaled potentials above imply therefore that $x=-\infty$ is a irregular singular point for \eqref{eq:homodex} for all $a\leq 1$, so by standard ODE theory for asymptotic expansions around irregular singular points, see for example \cite[\S 7.2, Theorem 2.1]{olv74}, it admits two solutions $u_{\pm}$ that can be written as follows:
	\begin{align}
	\label{eq:odeerrorest1}
	u_{\pm}(x)=&\:e^{\pm i|m| x}(1\pm \varepsilon_{\pm}(x)),\\
	\label{eq:odeerrorest2}
	|\varepsilon_{\pm}(x)|	\leq &\: C_m|x|^{-1},\\
	\label{eq:odeerrorest3}
	\left|\frac{d\varepsilon_{\pm}}{dx}(x)\right| \leq &\: C_m|x|^{-2}.
	\end{align}
	By the boundary conditions for $u_{\rm hor}$, it therefore follows immediately that:
	\begin{equation*}
	u_{\rm hor}(r_*)=e^{-i\tomega r_*}(1+O_m(|x|^{-1})).
\end{equation*}
so we can conclude that
	\begin{equation*}
	||u_{\rm hor}||_{L^{\infty}_{r_*}(-\infty,-|m||\tomega|^{-1}x_0)}+|\tomega|^{-1}||u'_{\rm hor}||_{L^{\infty}_{r_*}(-\infty,-|m||\tomega|^{-1}x_0)}\leq B_{\rm hom}.
\end{equation*}
We expand:
\begin{equation*}
	u_{\rm inf}(r_*)=A_+u_{\rm hor}(r_*)+A_-\overline{u_{\rm hor}(r_*)}=A_+e^{-i\tomega r_*}(1+O_m(|x|^{-1}))+A_-e^{+i\tomega r_*}(1+O_m(|x|^{-1})),
\end{equation*}
with $A_{\pm}(\tomega) \in \C$.

By integrating $(j^T[u_{\rm inf}]))'$ and applying the boundary conditions for $u_{\rm inf}$ we can express 
\begin{equation*}
\begin{split}
\omega^2=\omega^2|u_{\rm inf}|^2(\infty)=&-\omega \Re(i u_{\rm inf}' \overline{u}_{\rm inf})(-\infty)=-\tomega(|A_+|^2-|A_-|^2),
\end{split}
\end{equation*}
So we have that:
\begin{equation}
\label{eq:wronksrel1}
|A_+|^2-|A_-|^2=-\frac{\omega}{\tomega}.
\end{equation}

We can also express the Wronskian $W$ as follows:
\begin{equation}
\label{eq:wronksrel2}
	W=[u_{\rm inf}' u_{\rm hor}-u_{\rm hor}' u_{\rm inf}](-\infty)=2i\tomega A_-.
\end{equation}

Hence, there exists a constant $B_{\rm hom}>0$ depending only on $x_0$, such that
	\begin{equation*}
	||u_{\rm inf}||_{L^{\infty}_{r_*}(-\infty,-|m||\tomega|^{-1}x_0)}+|\tomega|^{-1}||u'_{\rm inf}||_{L^{\infty}_{r_*}(-\infty,-|m||\tomega|^{-1}x_0)}\leq \frac{B_{\rm hom}(|W|+\sqrt{|\omega|}\sqrt{|\tomega|})}{|\tomega|},
\end{equation*}
with $B_{\rm hom}>0$ a constant depending on $m$, $x_0$ and $B$.

Note further that by \eqref{eq:rmin1rstar2}, we can find a constant $C>0$ such that
\begin{equation*}
	(r-1)|r_*|\leq C,
\end{equation*}
so we can restrict the intervals in the $L^{\infty}$-estimates from $(-\infty,-|m||\tomega|^{-1}x_0)_{r_*}$ to $(1,1+z_0|m|^{-1}|\tomega|)_r$ for an appropriately large $z_0$. This concludes the proof of \eqref{eq:hombound1} and \eqref{eq:hombound2}. 

The estimate \eqref{eq:hombound3} follows easily from a similar asymptotic analysis of \eqref{eq:homode} at $r_*=\infty$, where $u_{\rm inf}$ now takes the role of $u_{\rm hor}$ above and there is no need to rescale the variable $r_*$ as we are bounded away from $\omega=0$ by the assumption $|\omega|>\omega_0$. We have that
\begin{align*}
	u_{\inf}(x)=&\:e^{ i \omega r_*}(1+\varepsilon_{+}(r_*)),\\
	|\varepsilon_{+}(r_*)|	\leq &\: C_m|r |^{-1},\\
	\left|\frac{d\varepsilon_{+}}{dr_*}(r_*)\right| \leq &\: C_m|r|^{-1}.
	\end{align*}

We can then express in the region $r\geq 1+z_0^{-1}$:      
\begin{equation*}
	u_{\rm hor}(r_*)=D_+u_{\rm inf}(r_*)+D_-\overline{u_{\rm inf}(r_*)}=D_+e^{i\omega r_*}(1+|m|^{-1}O(r^{-1}))+D_-e^{-i\tomega r_*}(1+|m|^{-1}O(r^{-1})),
\end{equation*}
with $D_{\pm}(\tomega) \in \C$.
We integrate $(j^T[u_{\rm inf}]))'$ to express:
\begin{equation*}
\begin{split}
\tomega \omega =\tomega \omega |u_{\rm hor}|^2(-\infty)=\omega \Re(i u_{\rm hor}' \overline{u}_{\rm hor})(\infty)=-\omega^2 (|D_+|^2-|D_-|^2),
\end{split}
\end{equation*}
We can also express the Wronskian $W$ as follows:
\begin{equation*}
	W=[u_{\rm inf}' u_{\rm hor}-u_{\rm hor}' u_{\rm inf}](\infty)=2i\omega D_-.
\end{equation*}
From the above, it follows that \eqref{eq:hombound4} holds with $B_{\rm hom}$ depending only on $z_0$.
		\end{proof}

We will need to establish additionally appropriate control over $u_{\rm inf}$ in the region $r\geq 1+z_0|m|^{-1}|\tomega|$. For this reason, we need to consider the potential $\check{V}$, which is defined in \S \ref{sec:mainpot}. See \eqref{eq:Vchecknearhor} in Lemma \ref{prop:potentials} for the behaviour of $\check{V}$ near $r=1$. 

\textbf{In Lemma \ref{lm:propfracw} and Proposition \ref{prop:addestuout} below, we will restrict to $a=1$.} 

We now consider the weight functions $\mathfrak{w}$ satisfying \eqref{eq:staticfreqweight}, i.e.\ solutions to
\begin{align*}
\label{eq:staticfreqweight}
\mathfrak{w}''-\widetilde{V}_{\rm stat}\mathfrak{w}=0.
\end{align*}

Define the renormalized homogeneous solutions $\check{u}_{\rm inf}$ as follows:
\begin{equation*}
\check{u}_{\rm inf,\pm}=\mathfrak{w}_{\pm}^{-1}u_{\rm inf}.
\end{equation*}
Then $\check{u}_{\rm inf,\pm}$ satisfies the following equation for $a=1$:
\begin{equation}
\label{eq:renormode}
\check{u}_{\rm inf,\pm}''+2\mathfrak{w}_{\pm}^{-1}\mathfrak{w}'_{\pm}\check{u}_{\rm inf,\pm}'+\left[ \widetilde{\omega}^2-\check{V}\right]\check{u}_{\rm inf,\pm}=0.
\end{equation}
We introduce the corresponding current $\check{j}^y[\check{u}_{\rm inf,\pm}]$ as follows:
\begin{align*}
\check{j}^y[\check{u}_{\rm inf,\pm}]=&\:y\left(|\check{u}_{\rm inf,\pm}'|^2+\left[ \widetilde{\omega}^2-\check{V}\right]|\check{u}_{\rm inf,\pm}|^2\right)\\
\frac{d}{dr}\left(\check{j}^y[\check{u}_{\rm inf,\pm}]\right)=&\: \frac{dy}{dr}\left(|\check{u}_{\rm inf,\pm}'|^2+\left[ \widetilde{\omega}^2-\check{V}\right]|\check{u}_{\rm inf,\pm}|^2\right)-y\frac{d\check{V}}{dr}|\check{u}_{\rm inf,\pm}|^2-4y\Re\left(\mathfrak{w}_{\pm}^{-1}\frac{d\mathfrak{w}_{\pm}}{dr}\right)|\check{u}_{\rm inf,\pm}'|^2.
\end{align*}

In the proposition below, we will derive integrated estimates for $\check{u}_{\rm inf,+}$ or $\check{u}_{\rm inf,-}$  in the region $1+z_0|m|^{-1}|\tomega|\leq r<\infty$, using Lemma \ref{lm:prelimhomest} to estimate the boundary term at $r=1+z_0|m|^{-1}|\tomega|$.
		
		\begin{proposition}
\label{prop:addestuout}
Let $(\omega,m,\Lambda)\in \mathcal{G}^{\flat}_<$. Let $\epsilon>0$ be arbitrarily small and $z_0$ be arbitrarily large. Then, for $\eta_{\flat}$ suitably small, there exists a constant $C=C(\eta_{\flat},\epsilon,z_0)>0$ such that
\begin{align}
\label{eq:keynearsuperrest2}
\int_{1+z_0|m|^{-1}|\tomega|}^{\infty}& r^{-1-\epsilon}(1-r^{-1})^{\epsilon}((1-r^{-1})^{-2}|u_{\rm inf}'|^2+\tomega^2 m^2|u_{\rm inf}|^2)+m^2|u_{\rm inf}|^2]\,dr\\ \nonumber
\leq &\:C B^2_{\rm hom}m^4(1+|W|^2|\omega|^{-1}|\tomega|^{-1}),\\
\label{eq:keynearsuperrest1}
\int_{1+z_0|m|^{-1}|\tomega|}^{1+\delta}& (r-1)^{-3+\epsilon-2\Im\sqrt{2m^2-\Lambda_{+}-\frac{1}{4}}}(|\check{u}_{\rm inf,\pm}'|^2+\tomega^2|\check{u}_{\rm inf,\pm}|^2)+(r-1)^{-1+\epsilon}|\check{u}_{\rm inf,\pm}|^2]\,dr\\ \nonumber
\leq &\:C B^2_{\rm hom}m^2(1+|W|^2|\omega|^{-1}|\tomega|^{-1}),\\
\label{eq:keynearsuperrest2b}
\int_{1}^{1+z_0|m|^{-1}|\tomega|} &r^{-1-\epsilon}(1-r^{-1})^{-1+\epsilon}(|m||\tomega|^{-1}|u_{\rm inf}'|^2+|m\tomega||u_{\rm inf}|^2)\,dr\\ \nonumber
\leq &\:C B^2_{\rm hom}m^2(1+|W|^2|\omega|^{-1}|\tomega|^{-1}),
\end{align}
where \eqref{eq:keynearsuperrest1} is valid with $\check{u}_{\rm inf,+}$ if $m\tomega<0$ and with $\check{u}_{\rm inf,-}$ if $m\tomega>0$.
\end{proposition}
\begin{proof}
Let $0\leq \epsilon<1$ and choose $y_{\epsilon}(r)=(r-1)^{-2+\epsilon-\beta}\chi(r)$, with $\beta=2\Im\sqrt{2m^2-\Lambda_{+}-\frac{1}{4}}$ and $\chi$ a smooth cut-off function such that $\chi(r)=1$ for $r\leq 1+\delta$ and $\chi(r)=0$ for $r\geq 1+2\delta$, with $\delta>0$ suitably small.

We apply \eqref{eq:Vchecknearhor} to express:
\begin{equation*}
\begin{split}
\tomega^2-\check{V}(r)=&\:\tomega^2+m \tomega(r-1)+[|m\tomega|+|\Lambda-\Lambda_+|]O((r-1)^2).
\end{split}
\end{equation*}	
Hence,
\begin{multline}
\label{eq:keyexprVcheckest}
\frac{d}{dr}\left(y_{\epsilon}(\tomega^2-\check{V}(r))\right) =-(2-\epsilon-\beta ) \chi \tomega^2(r-1)^{-3+\epsilon-\beta }+(1-\epsilon+\beta)\chi(-m \tomega)(r-1)^{-2+\epsilon-\beta}\\
+[|m\tomega|+|\Lambda-\Lambda_+|]O((r-1)^{-1+\epsilon-\beta})+ \frac{d\chi}{dr}(r-1)^{-2+\epsilon-\beta} g(r)(\tomega^2-\check{V}(r)).
\end{multline}
Suppose first that $m\tomega<0$. Using that $(r-1)\geq z_0|\tomega||m|^{-1}$ and taking $\delta$ and $\eta_{\flat}$ suitably small and $z_0$ suitably large, together with \eqref{eq:Lambdalimit1}, we can estimate
\begin{equation*}
\frac{d}{dr}\left(y_{\epsilon}(\tomega^2-\check{V}(r))\right)\geq\frac{1}{2}|m\tomega|(r-1)^{-2+\epsilon-\beta}.
\end{equation*}

Furthermore, by \eqref{eq:estderw}, we have that when $2m^2-\Lambda_+-\frac{1}{4}\neq 0$,
\begin{equation}
\label{eq:dercheckucancel}
\frac{dy_{\epsilon}}{dr}-4y_{\epsilon}\Re\left(\mathfrak{w}_+^{-1}\frac{d\mathfrak{w}_+}{dr}\right)=(\epsilon+\beta)(r-1)^{-3+\epsilon}+O((r-1)^{-2+\epsilon}).
\end{equation}
When $2m^2-\Lambda_+-\frac{1}{4}=0$, we obtain instead
\begin{equation}
\label{eq:dercheckucancelexc}
\frac{dy_{\epsilon}}{dr}-4y_{\epsilon}\Re\left(\mathfrak{w}_+^{-1}\frac{d\mathfrak{w}_+}{dr}\right)=(\epsilon+\log((r-1)^{-1})^{-1})(r-1)^{-3+\epsilon}+\log(r-1)O((r-1)^{-2+\epsilon}).
\end{equation}

We can moreover estimate:
\begin{equation*}
	\frac{d\chi}{dr}(\tomega^2-\check{V})\geq -C(|m\tomega|+|\Lambda-\Lambda_+|) |\check{u}_{\rm inf}|^2
\end{equation*}
when $r\in [1+\delta,2+2\delta]$.

Combining the above observations, we obtain the following estimate for $\epsilon>0$, $\delta>0$ suitably small and $m\tomega<0$: let $r_0=1+z_0|m|^{-1}|\tomega|$, then
\begin{multline}
\label{eq:auxhomestcheck}
\int_{r_0}^{1+2\delta}(r-1)^{-3+\epsilon-\beta}\epsilon\chi |\check{u}'_{\rm inf,+}|^2 +(r-1)^{-2+\epsilon-\beta}\chi (\tomega^2(r-1)^{-1}+|m\widetilde{\omega}| )|\check{u}_{\rm inf,+}|^2\,dr\\
\leq  C\int_{1+\delta}^{1+2\delta}(|m\tomega|+|\Lambda-\Lambda_+|) |\check{u}_{\rm inf,+}|^2\,dr+C\left|m\widetilde{\omega}\right|(r-1)^{-1+\epsilon-\beta}|\check{u}_{\rm inf,+}|^2\Big|_{r=r_0}.
\end{multline}

By a Hardy inequality in $r$, we can also estimate:
\begin{equation}
	\label{eq:hardyVcheck}
	\begin{split}
\int_{r_0}^{1+2\delta}(r-1)^{-1+\epsilon-\beta} \chi|\check{u}_{\rm inf,+}|^2\,dr\leq &\:C\int_{r_0}^{1+2\delta}(\epsilon-\beta)^{-2}(r-1)^{-3+\epsilon-\beta}\chi |\check{u}'_{\rm inf,+}|^2\,dr\\
&+C\int_{1+\delta}^{1+2\delta}|\check{u}_{\rm inf,+}|^2\,dr.
\end{split}
\end{equation}
We can combine \eqref{eq:auxhomestcheck}, \eqref{eq:auxhomestcheck2}  and \eqref{eq:hardyVcheck} and estimate the resulting boundary terms at $r=r_0$ via \eqref{eq:hombound2} and Lemma \ref{lm:propfracw}, together with \eqref{eq:hombound2} and \eqref{eq:hombound3} to estimate the boundary terms and the integrals in $[1+\delta,1+2\delta]$, in order to obtain for $z_0$ suitably large and $\eta_{\flat}$ suitably small, depending on $\epsilon$:
\begin{multline}
\label{eq:keyestucheckhom}
\int_{1+z_0|m|^{-1}|\widetilde{\omega}|}^{1+2\delta}(r-1)^{-3+\epsilon-\beta}\chi |\check{u}'_{\rm inf,+}|^2+(r-1)^{-1+\epsilon-\beta} \chi|\check{u}_{\rm inf,+}|^2+(r-1)^{-2+\epsilon-\beta}\chi (|m\widetilde{\omega}|+(r-1)^{-1}\tomega^2 |\check{u}_{\rm inf,+}|^2\,dr\\
\leq CB_{\rm \hom}^2(1+|W|^2|\omega|^{-1}|\tomega|^{-1})m^2.
\end{multline}
We conclude that \eqref{eq:keynearsuperrest1} holds. In fact, we can now repeat the argument with $\epsilon=0$ by considering $y_0$.

By using the properties of $\mathfrak{w}_+$ from Lemma \ref{lm:propfracw}, we obtain
\begin{equation*}
\begin{split}
	|{u}_{\rm inf}'|^2=&\:|\mathfrak{w}_+|^2|\check{u}_{\rm inf,+}'+\mathfrak{w}_+^{-1}\mathfrak{w}_+'\check{u}_{\rm inf,+}|^2\\
	\leq &\: C(r-1)^{-1-\beta}[|\log(r-1)| |\check{u}_{\rm inf}'|^2+Cm^2(r-1)^2|\check{u}_{\rm inf,+}|^2],
	\end{split}
\end{equation*}
where the $|\log(r-1)|$ factor is only required when $2m^2-\Lambda_+-\frac{1}{4}=0$.

Hence,
\begin{multline}
\label{eq:fromucheckinftouinf}
\int_{1+z_0|m|^{-1}|\widetilde{\omega}|}^{1+2\delta}(r-1)^{-2+\epsilon}\chi  (|{u}'_{\rm inf}|^2+\tomega^2|u_{\rm inf}|^2)+m^2(r-1)^{\epsilon}\chi|{u}_{\rm inf}|^2\,dr\\
\leq CB_{\rm \hom}^2(1+|W|^2|\omega|^{-1}|\tomega|^{-1})m^4.
\end{multline}

To conclude \eqref{eq:keynearsuperrest2} for $m\tomega<0$, we use that we can add an appropriate integral in $[1+\delta,\infty)$ to \eqref{eq:fromucheckinftouinf} after applying \eqref{eq:hombound3}.

Now suppose that  $m\tomega>0$. Then we consider instead $\check{u}_{\rm inf,-}$ together with the current $\check{j}^{-y_0}[\check{u}_{\rm inf,-}]$, with $y_0=(r-1)^{-2+\beta}$. Using that for $\delta>0$ suitably small we now have:
\begin{equation*}
\frac{d}{dr}\left(-y_0(\tomega^2-\check{V}(r))\right)\geq \frac{1}{2}|m\tomega|(r-1)^{-2+\beta}.
\end{equation*}
and by the leading-order term on the right-hand side of \eqref{eq:dercheckucancel} has a good sign when $\epsilon=0$ and $\mathfrak{w}_+$ is replaced with $\mathfrak{w}_-$, we can conclude that for $m\tomega>0$:
\begin{multline}
\label{eq:auxhomestcheck2}
\int_{r_0}^{1+2\delta}|m\widetilde{\omega}| (r-1)^{-2+\beta}|\check{u}_{\rm inf,-}|^2\,dr\\
\leq  C\int_{1+\delta}^{1+2\delta}(|m\tomega|+|\Lambda-\Lambda_+|) |\check{u}_{\rm inf,-}|^2\,dr+C|m\widetilde{\omega}|(r-1)^{-1+\beta}|\check{u}_{\rm inf,-}|^2\Big|_{r=r_0}.
\end{multline}
From this it follows that
\begin{equation*}
\int_{1+z_0|m|^{-1}|\widetilde{\omega}|}^{1+2\delta} (r-1)^{-2-\beta}\chi |m\widetilde{\omega}| |\check{u}_{\rm inf,-}|^2\,dr\leq CB_{\rm \hom}^2(1+|W|^2|\omega|^{-1}|\tomega|^{-1})m^2.
\end{equation*}
Equipped with the above estimate, we can then continue the argument as in the case $m\tomega<0$ to conclude that \eqref{eq:keynearsuperrest1} and  \eqref{eq:keynearsuperrest2} hold also when $m\tomega>0$.

Finally, \eqref{eq:keynearsuperrest2b}  follows by simply integrating the estimates in \eqref{eq:hombound2}.
\end{proof}

We will additionally need an integrated estimate for $u_{\rm hor}$ in the region $1+z_0|m|^{-1}|\tomega|\leq r\leq 1+\delta$, with $\delta>0$ arbitrarily small.

\begin{proposition}
\label{prop:addestuhor}
Let $(\omega,m,\Lambda)\in \mathcal{G}^{\flat}_<$. Let $\epsilon>0$ be arbitrarily small and $z_0$ be arbitrarily large. Then, for $\eta_{\flat}$ suitably small, there exists a constant $C=C(\eta_{\flat},\epsilon,z_0)>0$ such that
\begin{equation}
\label{eq:keynearsuperresthor1}
\begin{split}
\int_{1+z_0|m|^{-1}|\tomega|}^{\infty}& r^{-1-\epsilon}(1-r^{-1})^{\epsilon}((1-r^{-1})^{-2}|u_{\rm hor}'|^2+\tomega^2 m^2|u_{\rm hor}|^2)+m^2|u_{\rm hor}|^2]\,dr\\
\leq &\:C B^2_{\rm hom} m^2(|m\tomega|+|W|^2).
\end{split}
\end{equation}
\end{proposition}
\begin{proof}
We proceed exactly as in the proof of Proposition \ref{prop:addestuout}, replacing $u_{\rm inf}$ with $u_{\rm hor}$ and applying the $L^{\infty}$ estimates for $u_{\rm hor}$ from Lemma \ref{lm:prelimhomest} instead of those for $u_{\rm inf}$.
\end{proof}
		
\subsubsection{$\mathcal{G}^{\flat}_{<}$: bounded frequencies near superradiant threshold}
\label{sec:freqnearsr}
	The solution $u$ to the inhomogeneous ODE \eqref{eq:odeU} can be expressed via integrals involving the Wronskian and solutions to the homogeneous ODE:
\begin{lemma}
\label{lm:explititinhom}
Let $a\leq 1$. Then the following identities hold for all $r\geq 1$ and for almost every $\omega$:
\begin{align}
\label{eq:explinhom1}
u(r)=&\:W^{-1}\left[u_{\rm inf}(r)\int_{1}^{r} u_{\rm hor}(r')H(r')\,\frac{r'^2+a^2}{\Delta}dr'+u_{\rm hor}(r)\int_{r}^{\infty} u_{\rm inf}(r')H(r')\,\frac{r'^2+a^2}{\Delta}dr'\right],\\
\label{eq:explinhom2}
u'(r)=&\:W^{-1}\left[u_{\rm inf}'(r)\int_{1}^{r} u_{\rm hor}(r')H(r')\,\frac{r'^2+a^2}{\Delta}dr'+u_{\rm hor}'(r)\int_{r}^{\infty} u_{\rm inf}(r')H(r')\,\frac{r'^2+a^2}{\Delta}dr'\right].
\end{align}
\end{lemma}
\begin{proof}
The right-hand sides of \eqref{eq:explinhom1} and \eqref{eq:explinhom2} are well defined since the integrals are bounded for almost every $\omega$ by (ii) of Proposition \ref{prop:smoothnessu} and the (non-quantitative) boundedness of $u_{\rm hor}$, $u_{\rm inf}$, and furthermore, $W^{-1}$ is bounded for $\omega\notin \{0,m\upomega_+\}$ by \cite{costa20}[Theorem 1.3]; see also Theorem \ref{thm:wronskianbound} below. 

By the smoothness of $H$ in $r_*$ established in (i) of Proposition \ref{prop:smoothnessu}, it therefore follows that the right-hand side of \eqref{eq:explinhom1} constitutes a smooth solution to \eqref{eq:odeU} for almost every $\omega$ that moreover satisfies the boundary conditions \eqref{eq:bchor} and \eqref{eq:bcinf}. By uniqueness, it must therefore be equal to $u$ for almost every $\omega$.
\end{proof}

To obtain estimates for \eqref{eq:odeU}, we will need the  following lower bound assumption on the Wronskian:
	\begin{equation*}
				|\omega \tomega||W^{-1}|^{2}\leq  W_0.
			\end{equation*}
			for $a=1$ and for all admissible $(\omega,m,\Lambda)$ satisfying $\Lambda\leq B_0m^2$, for all $m\in \Z\setminus\{0\}$ with $B_0>0$ a uniform constant.
			
			As an immediate corollary of \eqref{eq:wronksrel1} and \eqref{eq:wronksrel2}, we obtain the validity of the above Wronskian bound in the non-superradiant frequency regime. 
		\begin{corollary}
		Let $m\omega \notin (0,m\upomega_+)$. Then
		\begin{equation*}
		|\omega \tomega||W^{-1}|^{2}\leq \frac{1}{4}.
		\end{equation*}
		\end{corollary}
		\begin{proof}
		We combine \eqref{eq:wronksrel1} and \eqref{eq:wronksrel2} to obtain:
		\begin{equation*}
		|W|^{2}=4\tomega^2|A_+|^2+4\omega \tomega\geq 4\omega \tomega.
		\end{equation*}
		We conclude the bound on $|W|^{-1}$ by using that the assumption on $\omega$ implies that $\omega \tomega>0$.
		\end{proof}

In \cite{costa20}, a bound has been obtained for $W^{-1}$ that applies also to superradiant frequencies in the case $a=1$ and with a constant that may depend on $m$:
\begin{theorem}[Proposition 6.3 in \cite{costa20}]
\label{thm:wronskianbound}
Let $a=1$ and $(\omega,m,\Lambda)\in \mathcal{G}^{\flat}$ with $|\omega|\geq \omega_0$. Then there exists a constant $W_0(\omega_0,\omega_{\rm high}, \epsilon_{\rm width}, m)>0$ such that
\begin{equation}
\label{assm:wronk}
|\omega \tomega||W^{-1}|^{2}\leq W_0.
\end{equation}
\end{theorem}
		
		The aim of Proposition \ref{prop:horest} below is to combine the homogeneous estimates from Lemma \ref{lm:prelimhomest} and Proposition \ref{prop:addestuout} with the Wronskian bound \eqref{assm:wronk} and obtain estimates for solutions to the inhomogeneous ODE that are valid in the near-horizon region $r\leq 1+|m|^{-1}|\tomega| z_0$. 
		
		So far, we have made no assumptions on $H$ in the integrated, fixed-frequency estimates. Now, will assume that $H$ can be related to $F=F_{\xi}+G$ appearing on the right-hand side of \eqref{eq:waveeqxi}, using that $u$ can be related to $\upphi$ via a Fourier transform.
		
		\begin{proposition}
			\label{prop:horest}
			Let $a=1$ and $(\omega, m,\Lambda)\in \mathcal{G}^{\flat}$. Assume that $u$ is a solution to \eqref{eq:odeU}.
				
			Then for $\epsilon>0$ arbitrarily small and $z_0$ arbitrarily large, there exists a constant $C_m>0$ such that
			\begin{equation}
			\label{eq:nearhorest}
			\begin{split}
				\int_{{\mathcal{B}_<}}&\sum_{\ell}\mathbf{1}_{|\tomega|\geq z_0^{-1}|m|(r-1)}(r-1)(|u|^2+|\tomega|^{-2}|u'|^2)(r)\,d\tomega\\
				\leq &\: C_mB_{\rm hom}^4 W_0 \left[\int_{0}^1 \int_{\Sigma_{\tau}}\Delta^{\frac{1-\epsilon}{2}}r^{-1+\epsilon}r^2|F_{
		\xi}|^2\,d\sigma dr d\tau+\int_0^{\infty}\int_{\Sigma_{\tau}}(1+\tau)^{1+\delta}\Delta^{\frac{1-\epsilon}{2}}r^{-1+\epsilon}r^2|G_{
		\xi}|^2\,d\sigma drd\tau\right],
							\end{split}
			\end{equation}
			where $\mathbf{1}_{|\tomega|\geq z_0^{-1}|m|(r-1)}$ denotes an indicator function.
		\end{proposition}
		\begin{proof}
			Let $m$ be fixed. We will first estimate:
			\begin{equation*}
				\int_{{\mathcal{B}_<}}|m\tomega|\sum_{\ell}\left|W^{-1}u_{\rm inf}(r_0)\int_{1}^{r_0} u_{\rm hor}\frac{H}{\Delta}(r^2+a^2)\,dr\right|^2\,d\tomega.
			\end{equation*}
			for $r_0-1\leq z_0|m|^{-1}|\tomega|$.
			
			We can alternatively write this as
			\begin{equation}
			\label{eq:inhomesthor}
				\int_{{\mathcal{B}_<}}\mathbf{1}_{|\tomega|\geq z_0^{-1}|m|(r_0-1)}\sum_{\ell}|m\tomega|\left| W^{-1}u_{\rm inf}(r_0)\int_{1}^{r_0} u_{\rm hor}\frac{H}{\Delta}(r^2+a^2)\,dr\right|^2\,d\tomega.
			\end{equation}
			
			We first apply \eqref{eq:hombound2} and assumption \eqref{assm:wronk} to estimate:
			\begin{equation*}
				|m\tomega|| W^{-1}u_{\rm inf}|^2(r_0)\leq B_{\rm hom}^2W_0|m||\tomega|^{-1}.
			\end{equation*}
			
			By splitting $u_{\rm hor}=e^{-i\tomega r_*}(1+\varepsilon_+)$ according to \eqref{eq:odeerrorest1}, we will first exploit the oscillation of $e^{-i\tomega r_*}$ to estimate \eqref{eq:inhomesthor}.
			
			We apply an inverse Fourier transform and a change of variables, to obtain:
			\begin{equation*}
			\begin{split}
				&\left|\int_{1}^{r_0} e^{-i\tomega r_*}\frac{H}{\Delta}(r^2+a^2)\,dr\right|^2=\frac{1}{2\pi}\left|\int_{\s^2}\int_{1}^{r_0} \int_{\R_t}e^{-i\tomega r_*+i\omega t}F\, dt dr\, e^{-im\varphi}S_{m\ell}(\theta;a\omega)d\sigma\right|^2\\
				=&\:\frac{1}{2\pi}\left|\int_{\s^2}\int_{\R_u} \int_{\R_{\tau}}e^{i\tomega u+im\upomega_+ (u+v)}\mathbf{1}_{r\leq r_0}F\Delta(r^2+a^2)^{-1}\,du d\tau\, e^{-im\varphi}S_{m\ell}(\theta;a\omega)d\sigma\right|^2.
				\end{split}
			\end{equation*}
			Hence, by applying Plancherel on $\s^2$, we obtain:
			\begin{equation*}
			\begin{split}
				\int_{{\mathcal{B}_<}}&\mathbf{1}_{|\tomega|\geq z_0^{-1}|m|(r_0-1)}|m||\tomega|^{-1}\sum_{\ell}\left|\int_{1}^{r_0} e^{-i\tomega r_*}\frac{H}{\Delta}(r^2+a^2)\,dr\right|^2\,d\tomega\\
				=&\:\frac{1}{2\pi}\int_{\s^2}\int_{{\mathcal{B}_<}}\mathbf{1}_{|\tomega|\leq z_0^{-1}|m|(r_0-1)}|m||\tomega|^{-1}\left|\int_{\R_u} \int_{\R_{\tau}}e^{i\tomega u+im\omega_+ (u+v)}\mathbf{1}_{r\leq r_0}F\Delta(r^2+a^2)^{-1}\, du d\tau\right|^2\,d\tomega d\sigma.
				\end{split}
			\end{equation*}
			We now use that
			\begin{equation*}
				|m||\tomega|^{-1}\mathbf{1}_{|\tomega|\geq z_0^{-1}|m|(r_0-1)}\leq z_0^{-1}(r_0-1)^{-1}
			\end{equation*}
			and we split $F=F_{\xi}+\xi G$ and apply Cauchy--Schwarz in $\tau$, using that $\xi$ is supported in $\tau\geq 0$ and $\dot{\xi}$ is supported in $\{0\leq \tau\leq 1\}$ to further estimate the right-hand side above and obtain
			\begin{equation*}
			\begin{split}
				\int_{{\mathcal{B}_<}}&\sum_{\ell}\mathbf{1}_{|\tomega|\geq z_0^{-1}|m|(r_0-1)}|m||\tomega|^{-1}\left|\int_{1}^{r_0} e^{-i\tomega r_*}\frac{H}{\Delta}(r^2+a^2)\,dr\right|^2\,d\tomega\\
				\leq &\: C(r_0-1)^{-1}\int_{\R_{\tau}}\int_{\s^2}\int_{{\mathcal{B}_<}}\left|\int_{\R_u}e^{i(\tomega+m\omega_+)u}\mathbf{1}_{r\leq r_0}F_{\xi}\,\Delta(r^2+a^2)^{-1} du \right|^2\, d\tomega d\sigma d\tau\\
				+&\:C(r_0-1)^{-1}\int_{\R_{\tau}}(1+\tau^2)^{\frac{1+\delta}{2}}\int_{\s^2}\int_{{\mathcal{B}_<}} \left|\int_{\R_u}e^{i(\tomega+m\omega_+) u}\mathbf{1}_{r\leq r_0}\xi G\Delta(r^2+a^2)^{-1}\, du \right|^2\, d\tomega d\sigma d\tau.
				\end{split}
			\end{equation*}
			We apply Plancherel in $u$ to estimate the right-hand side further and conclude:
			\begin{equation}
			\label{eq:keyestHint}
			\begin{split}
				\int_{{\mathcal{B}_<}}&\mathbf{1}_{|\tomega|\geq z_0^{-1}|m|(r_0-1)}|m||\tomega|^{-1}\sum_{\ell}\left|\int_{1}^{r_0} e^{-i\tomega r_*}\frac{H}{\Delta}(r^2+a^2)\,dr\right|^2\,d\tomega\\
				\leq &\: C(r_0-1)^{-1}\int_{\R_{\tau}}\int_{\R_u}\int_{\s^2}	\mathbf{1}_{r\leq r_0}(|F_{\xi}|^2+(1+\tau^2)^{\frac{1+\delta}{2}}\xi^2|G|^2)\Delta^2(r^2+a^2)^{-4}\,d\sigma du d\tau\\
				\leq&\: C(r_0-1)^{-1}\int_{0}^1 \int_{\Sigma_{\tau}\cap\{r\leq r_0\}}\Delta(r^2+a^2)^{-1}|F_{
		\xi}|^2\,d\sigma dr d\tau\\
		&+ C(r_0-1)^{-1}\int_0^{\infty}\int_{\Sigma_{\tau}\cap\{r\leq r_0\}}\Delta(r^2+a^2)^{-1}(1+\tau)^{1+\delta}|G|^2\,d\sigma drd\tau\\
		\leq&\: C(r_0-1)^{q}\int_{0}^1 \int_{\Sigma_{\tau}\cap\{r\leq r_0\}}\Delta^{\frac{1-q}{2}}(r^2+a^2)^{-1}|F_{
		\xi}|^2\,d\sigma dr d\tau\\
		&+ C(r_0-1)^{q}\int_0^{\infty}\int_{\Sigma_{\tau}\cap\{r\leq r_0\}}\Delta^{\frac{1-q}{2}}(r^2+a^2)^{-1}(1+\tau)^{1+\delta}|G|^2\,d\sigma drd\tau.
							\end{split}
			\end{equation}

			We now estimate the contribution of $\epsilon_+$. We first integrate by parts to write:
			\begin{equation}
			\label{eq:inhomodeesterror}
			\begin{split}
				\int_{1}^{r_0} e^{-i\tomega r_*}\varepsilon_+(r)\frac{H}{\Delta}(r^2+a^2)\,dr=&\:\int_{1}^{r_0} \varepsilon_+(r)\frac{d}{dr}\int_1^r\frac{He^{-i\tomega r_*}}{\Delta}(r^2+a^2)\,dr\\
				=&\:\varepsilon_+(r_0)\int_1^{r_0}\frac{He^{-i\tomega r_*} }{\Delta}(r^2+a^2)\,dr\\
				&-\int_{1}^{r_0} \frac{d\varepsilon_+}{dr}(r)\int_1^r\frac{He^{-i\tomega r_*}}{\Delta}(r^2+a^2)\,dr.
				\end{split}
			\end{equation}
			
			By \eqref{eq:odeerrorest2} and \eqref{eq:odeerrorest3}, it follows that for $r_0-1\leq z_0\frac{|\tomega|}{|m|}$
			\begin{align*}
			|\varepsilon_+(r_0)|\leq &\:C_m ,\\
			\left|\frac{d\varepsilon_+}{dr_*}(r_*)\right|\leq &\:C_m|r_*|^{-1}.
			\end{align*}
			
			Hence, the first term on the right-hand side of \eqref{eq:inhomodeesterror} can be bounded via \eqref{eq:keyestHint} after integrating over $\mathcal{B}_<$. We estimate the second term on the right-hand side of \eqref{eq:inhomodeesterror} as follows:
			\begin{equation*}
			\begin{split}
			\Bigg|\int_{1}^{r_0}& \frac{d\varepsilon_+}{dr}(r)\int_1^r\frac{He^{-i\tomega r_*}}{\Delta}(r^2+a^2)\,dr'\,dr\Bigg|\\
			\leq &\: \int_{1}^{r_0} (r-1)^{\frac{1+q}{2}}\left|\frac{d\varepsilon_+}{dr}(r)\right|\,dr\cdot \sup_{1\leq r\leq r_0} (r-1)^{-\frac{1+q}{2}}\left|\int_1^r\frac{He^{-i\tomega r_*}}{\Delta}(r'^2+a^2)\,dr'\right|\\
			\leq &\: \int_{-\infty}^{r_*(r_0)} (r-1)^{\frac{q}{2}}\left|\frac{d\varepsilon_+}{dr_*}(r_*)\right|\,dr_*\cdot \sup_{1\leq r\leq r_0} (r-1)^{-\frac{q}{2}}\left|\int_1^r\frac{He^{-i\tomega r_*}}{\Delta}(r'^2+a^2)\,dr'\right|
			\end{split}
			\end{equation*}
which can again be bounded via \eqref{eq:keyestHint} if we take $q>0$ and integrate over $\mathcal{B}_<$.

			We will next estimate:
			\begin{equation*}
				\int_{{\mathcal{B}_<}}(r_0-1)|\sum_{\ell}\left|W^{-1}u_{\rm hor}(r_0)\int_{r_0}^{\infty} u_{\rm inf}\frac{H}{\Delta}(r^2+a^2)\,dr\right|^2\,d\tomega.
			\end{equation*}

			We apply \eqref{eq:hombound1} :
			\begin{equation*}
				| u_{\rm hor}|^2(r_0)\leq B_{\rm hom}^2.
			\end{equation*}
			
			We therefore need to estimate
	\begin{equation}
			\label{eq:inhomestinf}
				\int_{{\mathcal{B}_<}}(r_0-1)|W|^{-2}\sum_{\ell}\left|\int_{r_0}^{\infty} u_{\rm inf}\frac{H}{\Delta}(r^2+a^2)\,dr\right|^2\,d\tomega.
			\end{equation}

			We will use the shorthand notation $\mathfrak{w}$ for $\mathfrak{w}_{\pm}$ and $\check{u}_{\inf}$ for $\check{u}_{\inf,\pm}$. We first integrate by parts and denote $\beta=2\Im\sqrt{2m^2-\Lambda_{+}-\frac{1}{4}}$ to obtain:
			\begin{equation*}
			\begin{split}
				\int_{r_0}^{R_0} u_{\rm inf}\frac{H}{\Delta}(r^2+a^2)\,dr=&\:\int_{r_0}^{R_0} e^{i\tomega r_*}{u}_{\rm inf}\frac{d}{dr}\int_1^r\frac{He^{-i\tomega r_*}}{\Delta}(r^2+a^2)\,dr\\
				=&\:(e^{i\tomega r_*}{u}_{\rm inf})(R_0)\int_1^{R_0}\frac{He^{-i\tomega r_*}}{\Delta}(r^2+a^2)\,dr\\
				&-(e^{i\tomega r_*}{u}_{\rm inf} )(r_0)\int_1^{r_0}\frac{He^{-i\tomega r_*}}{\Delta}(r^2+a^2)\,dr\\
				&-\int_{r_0}^{R_0} \frac{d}{dr}(e^{i\tomega r_*}{u}_{\rm inf})\int_1^r\frac{He^{-i\tomega r_*}}{\Delta}(r^2+a^2)\,dr.
				\end{split}
			\end{equation*}
							Observe that the contribution of the first two terms on the very right-hand side to the integral over $\int_{{\mathcal{B}_<}}\sum_{\ell}$ can be estimated as above, exploiting the presence of the oscillating factor $e^{-i\tomega r_*}$ multiplying $H$.
				
				We estimate the third term on the right-hand side of the equation above as follows:
				\begin{multline*}
					(r_0-1)|W|^{-2}\Bigg|\int_{r_0}^{R_0} \frac{d}{dr}(e^{i\tomega r_*}{u}_{\rm inf})\int_1^r\frac{He^{-i\tomega r_*}}{\Delta}(r^2+a^2)\,dr\Bigg|^2\\
					\leq (r_0-1)|W|^{-2}\left(\int_{r_0}^{R_0}(r-1)^{\frac{q+1}{2}} \left|\frac{d}{dr}(e^{i\tomega r_*}{u}_{\rm inf} )\right|\,dr\right)^2\\
					 \times \sup_{r_0\leq r\leq R_0}(r-1)^{-1-q}\left|\int_1^r\frac{He^{-i\tomega r_*}}{\Delta}(r'^2+a^2)\,dr'\right|^2.
				\end{multline*}
				By the estimates above, we obtain:
		\begin{multline*}
				\int_{{\mathcal{B}_<}}(r-1)^{-1-q}\sup_{r_0\leq r\leq R_0}\sum_{\ell}\left|\int_{1}^{r} e^{-i\tomega r_*}\frac{H}{\Delta}(r^2+a^2)\,dr'\right|^2\,d\tomega\\
				\leq C\int_{0}^{\infty}\int_{1}^{r_0}\int_{\s^2}(|F_{\xi}|^2+(1+\tau)^{1+\delta}\xi^2|G|^2)\Delta^{\frac{1-q}{2}}(r^2+a^2)^{-2}\,d\sigma dr d\tau. 
		\end{multline*}
		
				It remains to estimate the integral involving ${u}_{\rm inf}$. We apply Lemma \ref{lm:prelimhomest} and the estimates \eqref{eq:keynearsuperrest2} and \eqref{eq:keynearsuperrest2b} to obtain:
		\begin{equation*}
		\begin{split}
			(r_0-1)|W|^{-2}&\Bigg(\int_{r_0}^{R_0} (r-1)^{\frac{q+1}{2}}\left|\frac{d}{dr}(e^{i\tomega r_*}{u}_{\rm inf} )\right|\,dr\Bigg)^2\\
			\leq&\:2(r_0-1)|W|^{-2}\Bigg(\int_{r_0}^{1+z_0|m|^{-1}|\tomega|} (r-1)^{\frac{q+1}{2}}\left|\frac{d}{dr}(e^{i\tomega r_*}{u}_{\rm inf})\right|\,dr\Bigg)^2\\
			&\:+2(r_0-1)|W|^{-2}\Bigg(\int_{1+z_0|m|^{-1}|\tomega|}^{R_0} (r-1)^{\frac{q+1}{2}}\left|\frac{d}{dr}(e^{i\tomega r_*}{u}_{\rm inf})\right|\,dr\Bigg)^2\\
			\leq &\: C_m\tomega (r_0-1) |W|^2\int_{r_0}^{1+z_0|m|^{-1}|\tomega|}(r-1)^{-2+q-\epsilon'}\,dr\\
			&\times \int_{1}^{1+z_0|m|^{-1}|\tomega|}(r-1)^{-1+\epsilon'}\frac{|m|}{|\tomega|}(\tomega^2|{u}_{\rm inf}|^2+|{u}_{\rm inf}'|^2)\,dr\\
			&\:+C_m(r_0-1)|W|^{-2}\int_{1+z_0|m|^{-1}|\tomega|}^{R_0}(r-1)^{-1+q-\epsilon'}\,dr\int_{1+z_0|m|^{-1}|\tomega|}^{R_0}(r-1)^{-2+\epsilon'}(\tomega^2|{u}_{\rm inf}|^2+|{u}_{\rm inf}'|^2)\\
			\leq &\: C_m (r_0-1)^{q-\epsilon'}W_0^2B_{\rm hom}^2 |m\tomega|+C_m(R_0-1)^{q-\epsilon'}B_{\rm hom}^2W_0^2\\
			\leq &\: C_mW_0^2B_{\rm hom}^2
			\end{split}
		\end{equation*}
		if $q>\epsilon'>0$.
		
		We can finally estimate
		\begin{multline*}
				(r_0-1)\int_{{\mathcal{B}_<}} |W|^{-2}\sum_{\ell}\left|\int_{R_0}^{\infty} u_{\rm inf}\frac{H}{\Delta}(r^2+a^2)\,dr\right|^2\,d\tomega\leq C_m W_0^2\int_{0}^1 \int_{\Sigma_{\tau}\cap\{r\geq R_0\}}r^2|F_{
		\xi}|^2\,d\sigma dr d\tau\\
		+ C_m W_0^2(r_0-1)^{-1}\int_0^{\infty}\int_{\Sigma_{\tau}\cap\{r\geq R_0\}}(1+\tau)^{1+\delta}r^2|G|^2\,d\sigma drd\tau,
			\end{multline*}
			with $R_0$ arbitrarily large, by writing $u_{\rm inf}=e^{i\omega r_*}+O(r^{-1})$ and exploiting oscillations as above.

		We combine the above estimates and to conclude the estimate involving $u$ in \eqref{eq:nearhorest}. The estimate for $u'$ can be obtained by using \eqref{eq:explinhom2} and repeating the steps above.
		\end{proof}
		
		In the proposition below, we derive additional estimates away from the horizon, in the region $r>1+z_0^{-1}$ with $z_0>0$ arbitrarily small.
		
		\begin{proposition}
			\label{prop:infest}
			Let $a=1$ and $(\omega, m,\Lambda)\in \mathcal{G}^{\flat}_<$ and assume that the Wronskian assumption \eqref{assm:wronk} holds.
				
			Then for $\epsilon>0$ arbitrarily small and $z_0$ arbitrarily large, there exists a constant $C_m>0$ such that for $r\geq 1+z_0^{-1}$:
			\begin{equation}
			\label{eq:nearinfest}
			\begin{split}
				\int_{{\mathcal{B}_<}}&\sum_{\ell} (|u|^2+|\omega|^{-2}|u'|^2)(r)\,d\tomega\\
				\leq &\: C_mB_{\rm hom}^4 W_0 \left[\int_{0}^1 \int_{\Sigma_{\tau}}\Delta^{\frac{1-\epsilon}{2}}r^{-1+\epsilon}r^2|F_{
		\xi}|^2\,d\sigma dr d\tau+\int_0^{\infty}\int_{\Sigma_{\tau}}(1+\tau)^{1+\delta}\Delta^{\frac{1-\epsilon}{2}}r^{-1+\epsilon}r^2|G|^2\,d\sigma drd\tau\right].
							\end{split}
			\end{equation}
		\end{proposition}
\begin{proof}
We will make use of the identity \eqref{eq:explinhom1} with $r>1+z_0^{-1}$ to derive \eqref{eq:nearinfest}. We estimate
\begin{equation}
\label{eq:keyestuhorinft}
\begin{split}
\int_{\mathcal{B}_{<}}\sum_{\ell}(|u|^2+\omega^{-2}|u'|^2)(r_0)\,d\tomega\leq &\:\int_{\mathcal{B}_{<}}\sum_{\ell}\frac{|u_{\rm inf}|^2(r_0)+\omega^{-2}|u_{\rm inf}'|^2(r_0)}{|W|^2}\left(\int_1^{r_0} u_{\rm hor}(r) \frac{H}{\Delta}(r^2+a^2)(r)\,dr\right)^2\\
+&\:\frac{|u_{\rm hor}|^2+\omega^{-2}|u_{\rm hor}'|^2(r_0)}{|W|^2}\left(\int_{r_0}^{\infty} u_{\rm inf}(r) \frac{H}{\Delta}(r^2+a^2)(r)\,dr\right)^2\,d\tomega.
\end{split}
\end{equation}
with $r_0>1+z_0^{-1}$. The derivation of \eqref{eq:nearinfest} now proceeds much like the derivation of \eqref{eq:nearhorest}, with the key difference being the application of Proposition \ref{prop:addestuhor} instead of Proposition \ref{prop:addestuout}. 

In particular, we integrate by parts to obtain:
			\begin{equation*}
			\begin{split}
				W^{-1}&\int_{1}^{r_0} u_{\rm hor}\frac{H}{\Delta}(r^2+a^2)\,dr=W^{-1}\int_{1}^{r_0} e^{i\tomega r_*}{u}_{\rm hor}\frac{d}{dr}\int_1^{r}\frac{He^{-i\tomega r_*}}{\Delta}(r'^2+a^2)\,dr'dr\\
				=&\:W^{-1}(e^{i\tomega r_*}{u}_{\rm hor})(r_0)\int_1^{r_0}\frac{He^{-i\tomega r_*}}{\Delta}(r'^2+a^2)\,dr'-W^{-1}\int_{1}^{r_0} \frac{d}{dr}(e^{i\tomega r_*}{u}_{\rm hor})\int_1^r\frac{He^{-i\tomega r_*}}{\Delta}(r'^2+a^2)\,dr'dr
				\end{split}
			\end{equation*}
			
			It follows from \eqref{eq:hombound4} and \eqref{assm:wronk} that:
			\begin{equation*}
			|W|^{-2}\left|e^{i\tomega r_*}{u}_{\rm hor}\right|^2(r_0)\leq C W_0 B_{\rm hom}^2 m^{-2}.
\end{equation*}
From the proof of Proposition \ref{prop:horest}, it follows moreover that for any $1\leq r_1\leq r_0$ and $q>0$:
\begin{equation}
			\label{eq:keyestHintv2}
			\begin{split}
				\int_{{\mathcal{B}_<}}&\sum_{\ell}(r_1-1)^{-1-q}\left|\int_{1}^{r_1} e^{-i\tomega r_*}\frac{H}{\Delta}(r^2+a^2)\,dr\right|^2\,d\tomega\\
				\leq &\: C(r_1-1)^{-1-q}\int_{\R_{\tau}}\int_{\R_u}\int_{\s^2}	\mathbf{1}_{r\leq r_1}(|F_{\xi}|^2+(1+\tau^2)^{\frac{1+\delta}{2}}\xi^2|G|^2)\Delta^2(r^2+a^2)^{-4}\,d\sigma du d\tau\\
				\leq&\: C(r_1-1)^{-1-q}\int_{0}^1 \int_{\Sigma_{\tau}\cap\{r\leq r_1\}}\Delta(r^2+a^2)^{-1}|F_{
		\xi}|^2\,d\sigma dr d\tau\\
		&+ C(r_1-1)^{-1-q}\int_0^{\infty}\int_{\Sigma_{\tau}\cap\{r\leq r_0\}}\Delta(r^2+a^2)^{-1}(1+\tau)^{1+\delta}|G|^2\,d\sigma drd\tau\\
		\leq&\: C\int_{0}^1 \int_{\Sigma_{\tau}\cap\{r\leq r_0\}}\Delta^{\frac{1-q}{2}}(r^2+a^2)^{-1}|F_{
		\xi}|^2\,d\sigma dr d\tau\\
		&+ C\int_0^{\infty}\int_{\Sigma_{\tau}\cap\{r\leq r_0\}}\Delta^{\frac{1-q}{2}}(r^2+a^2)^{-1}(1+\tau)^{1+\delta}|G|^2\,d\sigma drd\tau.
							\end{split}
			\end{equation}
Hence we can control the first term on the very right-hand side of \eqref{eq:keyestuhorinft}. We control the second term on the very right-hand side by observing that:
\begin{multline*}
					\Bigg|W^{-1}\int_{1}^{r_0} \frac{d}{dr}(e^{i\tomega r_*}{u}_{\rm hor})\int_1^r\frac{He^{-i\tomega r_*} }{\Delta}(r^2+a^2)\,dr'dr\Bigg|^2\\
					\leq |W|^{-2}\left(\int_{1}^{r_0}(r-1)^{\frac{q}{2}} \left|\frac{d}{dr}(e^{i\tomega r_*}{u}_{\rm hor} )\right|\,dr\right)^2\\
					\times  \sup_{1\leq r\leq r_0}(r-1)^{-1-q}\left|\int_1^r\frac{He^{-i\tomega r_*}}{\Delta}(r'^2+a^2)\,dr'\right|^2.
				\end{multline*}
The second factor on the right-hand side above can be estimated via \eqref{eq:keyestHintv2}, whereas the first factor can be estimated as follows:
	\begin{equation*}
		\begin{split}
			|W|^{-2}\Bigg(\int_{1}^{r_0}& (r-1)^{\frac{q}{2}}\left|\frac{d}{dr}(e^{i\tomega r_*}{u}_{\rm hor} )\right|\,dr\Bigg)^2\\
			=&\:|W|^{-2}\Bigg(\int_{1}^{1+z_0|m|^{-1}|\tomega|} (r-1)^{\frac{q}{2}}\left|\frac{d}{dr}(e^{i\tomega r_*}{u}_{\rm hor})\right|\,dr\Bigg)^2\\
			&+|W|^{-2}\Bigg(\int_{1+z_0|m|^{-1}|\tomega|}^{r_0} (r-1)^{\frac{q}{2}}\left|\frac{d}{dr}(e^{i\tomega r_*} u_{\rm hor} )\right|\,dr\Bigg)^2.
						\end{split}
		\end{equation*}
		
		The first term on the right-hand side has already been estimated in the proof of Proposition \ref{prop:horest}. The second term, we estimate further by applying Cauchy--Schwarz together with \eqref{eq:keynearsuperresthor1}:
		\begin{equation*}
		\begin{split}
		|W|^{-2}\Bigg(&\int_{1+z_0|m|^{-1}|\tomega|}^{r_0} (r-1)^{\frac{q}{2}}\left|\frac{d}{dr}(e^{i\tomega r_*}{u}_{\rm hor})\right|\,dr\Bigg)^2\\
			\leq &\: C|W|^{-2}\int_{1+z_0|m|^{-1}|\tomega|}^{r_0}(r-1)^{-1+q-\epsilon'}\,dr\\
			&\:\times \int_{1+z_0|m|^{-1}|\tomega|}^{r_0}(r-1)^{-4+\epsilon'}(\tomega^2+(r-1)^2)|{u}_{\rm hor}|^2+(r-1)^{-4+\epsilon'}|{u}_{\rm hor}'|^2\,dr\\
		\leq &\: C_m B^2_{\rm hom}W_0^2,
			\end{split}
		\end{equation*}
		if $q>\epsilon'>0$.
\end{proof}

We now obtain the desired integrated estimate for the frequency range $\mathcal{G}^{\flat}_{<}$.
	
		\begin{proposition}
			Let $a=1$ and $(\omega, m,\Lambda)\in \mathcal{G}^{\flat}_{<}$. Then for any $\delta,\epsilon>0$ arbitrarily small, there exists a constant $c=c(\epsilon, \delta,\omega_{\rm high},\eta_{\flat})>0$ and $C_m= C_m(\epsilon,\delta,\omega_{\rm high},\eta_{\flat},m)>0$ such that
			\begin{equation}
			\label{eq:freqmornearsr}
			\begin{split}
				c\int_{{\mathcal{B}_<}}&\sum_{\ell}\int_{1}^{\infty}r^{-1-\delta}(1-r^{-1})^{\epsilon}\left(\Lambda|u|^2+\frac{1}{\tomega^2+(r-1)^2}|u'|^2\right)\,drd\tomega\\
				\leq&\: \int_{{\mathcal{B}_<}} \int_{1+z_0|m|^{-1}|\tomega|}^{\infty}\chi_+\mathbf{1}_{m\tomega<0}(\tomega)y_{+}\Re\left(\frac{d(\mathfrak{w}_+^{-1}u)}{dr}\overline{\mathfrak{w}_+^{-1}H}\right)-\chi_+\mathbf{1}_{m\tomega>0}(\tomega)y_{-}\Re\left(\frac{d(\mathfrak{w}_-^{-1}u)}{dr}\overline{\mathfrak{w}_-^{-1}H}\right)\,dr d\tomega\\
				&+C_mB_{\rm hom}^4 \left[ \int_{0}^1 \int_{\Sigma_{\tau}}\Delta^{\frac{1-\epsilon}{2}}r^{-1+\epsilon}r^2|F_{
		\xi}|^2\,d\sigma dr d\tau+ \int_0^{\infty}\int_{\Sigma_{\tau}}(1+\tau)^{1+\delta}\Delta^{\frac{1-\epsilon}{2}}r^{-1+\epsilon}r^2|G|^2\,d\sigma drd\tau\right],
		\end{split}
			\end{equation}
			with $\chi_+$ a smooth cut-off function such that $\chi_+(r)=1$ for $r\leq 1+\delta$ and $\chi_+(r)=0$ for $r\geq 1+2\delta$, $y_+$ and $y_-$ behave to leading order like $(r-1)^{-2\mp 2\Im\sqrt{2m^2-\Lambda_{+}-\frac{1}{4}}}$ as $r\downarrow 1$ and  $\mathbf{1}_{r\geq 1+\delta}$, $\mathbf{1}_{m\tomega>0}$ and $\mathbf{1}_{m\tomega<0}$ denoting indicator functions.
		\end{proposition}
		\begin{proof}
		The proof proceeds in \textbf{two steps}. First, we will control the left-hand side of \eqref{eq:freqmornearsr} away from $r=1$ by using Proposition \ref{prop:infest}. We will then derive estimates near $r=1$.

\paragraph{\textbf{Step 1}: estimates away from $r=1$}
We establish control in the region $r\in [1+\delta, \infty)$, for $\delta>0$ arbitrarily small by applying Proposition \ref{prop:infest}:
\begin{equation}
\label{eq:inhomestintinft}
\begin{split}
	\int_{\mathcal{B}_{<}}&\sum_{\ell}\int_{1+\delta}^{\infty}r^{-1-\delta'}(|u'|^2+\Lambda|u|^2)\,dr\leq  C \sup_{r\geq 1+\delta}\int_{\mathcal{B}_{<}}\sum_{\ell} m^2(|u|^2+\omega^{-2}|u'|^2)(r)\\
	\leq &\: CB^4_{\rm hom}W_0m^2\left[ \int_{0}^1 \int_{\Sigma_{\tau}}\Delta^{\frac{1-\epsilon}{2}}r^{-1+\epsilon}r^2|F_{
		\xi}|^2\,d\sigma dr d\tau+ \int_0^{\infty}\int_{\Sigma_{\tau}}(1+\tau)^{1+\delta}\Delta^{\frac{1-\epsilon}{2}}r^{-1+\epsilon}r^2|G_{
		\xi}|^2\,d\sigma drd\tau\right].
	\end{split}
\end{equation}
\paragraph{\textbf{Step 2}: estimates near $r=1$}
We proceed exactly as in the proof of Proposition \ref{prop:addestuout} by considering the renormalized variables $\check{u}_{\pm}=\mathfrak{w}_{\pm}^{-1}u$, which satisfies the ODE:
\begin{align}
\label{eq:inhomrenormode1}
\check{u}_{\pm}''+2\mathfrak{w}_{\pm}^{-1}\mathfrak{w}'_{\pm}\check{u}_{\pm}'+\left[ \widetilde{\omega}^2-\check{V}-\frac{1}{4}(1-a^2)^2(1+O(r-1))+(1-a^2)O((r-1)^2)\right]\check{u}_{\pm}=&\:\mathfrak{w}_{\pm}^{-1}H
\end{align}
and we additionally keep track of the terms involving $H$. The boundary terms at $r=1+z_0|m|^{-1}|\tomega|$ are estimated via Proposition \ref{prop:horest} and we apply \eqref{eq:inhomestintinft} to estimate the terms in the region where $r\in [1+\delta,1+2\delta]$.
		\end{proof}

	\subsubsection{$\mathcal{G}^{\flat}_{>}$: bounded frequencies away from the superradiant threshold}
	\label{sec:boundfreqaway}	
		
		In this section, we consider the case of bounded frequencies with $|m \tomega|>\eta_{\flat}m^2$. The cases $|\omega|<\omega_{\rm low}$ and $|\omega|>\omega_{\rm low}$, with $\omega_{\rm low}$ a suitably small constant, need to be done separately. The estimates here follow directly from analogous estimates in \cite{part3} in the case $a<1$, so we will omit their derivation.
		
		\begin{proposition}
			Let $0<a_0<a\leq 1$, $E\geq 2$ and $R_{\infty}>1$ arbitrarily large. Let $(\omega,m,\Lambda)\in \mathcal{G}^{\flat}_{>}$.
			\begin{enumerate}[label=\emph{(\roman*)}]
				\item For $\omega_{\rm low}$ suitably small depending on $a_0$ and $E$ and $R_{\infty}$ sufficiently large, there exists constants $c(a_0),C(a_0)>0$ and functions $y_{\flat,1},\tilde{y}_{\flat},h_{\flat},\chi_1,\chi_2$, such that for $|\omega|\leq \omega_{\rm low}$
				\begin{equation*}
					c\int_1^{R_{\infty}} |u|^2+|u'|^2\,dr\leq \int_{\R}\Re(iE\omega \chi_2 u+2i\tomega \chi_1u-h_{\flat}u-2(y_{\flat,1}+\tilde{y}_{\flat})u'\overline{H})\,dr_*,
				\end{equation*}
				and
				\begin{align*}
					|\tilde{y}_{\flat}|+|y_{\flat,1}|+|h_{\flat}|+|\chi_1|+|\chi_2|\leq&\: C,\\
					|\tilde{y}_{\flat}|\leq&\: C e^{-b (r-1)}\\
					y_{\flat,1}=&\:1\:\textnormal{and}\:h_{\flat }=0\quad \textnormal{for $r\geq R_{\infty}$}.
				\end{align*}
				\item Let $|\omega|> \omega_{\rm low}$. Then there exists constants $c,C>0$ depending on $\omega_{\rm low},\omega_{\rm high},\epsilon_{\rm width}$ and a function $y_{\flat,2}$, such that for $|\omega|> \omega_{\rm low}$:
				\begin{equation*}
					c\int_1^{R_{\infty}} |u|^2+|u'|^2\,dr\leq E|\omega||\tomega||u|^2(1)+\int_{\R}\Re(iE\omega u-2y_{\flat,2}u'\overline{H})\,dr_*,
				\end{equation*}
				with
				\begin{equation*}
					|y_{\flat,2}|\leq C\:\textnormal{and}\: y_{\flat,2}=1\:\quad \textnormal{for $r\geq R_{\infty}$}.
				\end{equation*}
				Furthermore,
					\begin{multline}
					\label{eq:addestr1boundfreq}
		|\omega||\tomega||u|^2(1)\leq C B_{\rm hom}^4 W_0 m^2\Bigg[\int_{0}^1 \int_{\Sigma_{\tau}}\Delta^{\frac{1-\epsilon}{2}}r^{-1+\epsilon}r^2|F_{
		\xi}|^2\,d\sigma dr d\tau\\
		+\int_0^{\infty}\int_{\Sigma_{\tau}}(1+\tau)^{1+\delta}\Delta^{\frac{1-\epsilon}{2}}r^{-1+\epsilon}r^2|G_{
		\xi}|^2\,d\sigma drd\tau\Bigg].
		\end{multline}
			\end{enumerate}
		\end{proposition}
		\begin{proof}
		Part (i) and (ii) follow from \cite[Propositions 8.7.3]{part3} and \cite[Propositions 8.7.4]{part3}, respectively, which work without modification even when $a\uparrow 1$. In the case of (ii), we can additionally apply \eqref{eq:nearhorest} to further estimate the boundary term at $r=1$ and obtain \eqref{eq:addestr1boundfreq}.
		\end{proof}
		
\subsubsection{Combining the frequency ranges}
		We can combine all the estimates the frequency ranges $\mathcal{G}^{\sharp}_{\rm sr}, \mathcal{G}^{\sharp, \omega}_{\rm nsr}, \mathcal{G}^{\sharp, \Lambda}_{\rm nsr}, \mathcal{G}^{\sharp, \sim}_{\rm nsr}, \mathcal{G}^{\flat}$ together with the $(r-1)$-weighted estimates near $r=1$ from \S \ref{sec:rweightfreq}.
		
		In order to convert the resulting estimate in frequency space to an estimate in physical space, it is advantageous to minimize the need for frequency restrictions in front of the multipliers corresponding to the different frequency ranges. Since we have already closed the relevant estimates in each frequency range separately, it is becomes much easier to consider the same currents in the different frequency ranges and thus make the use of currents more uniform in frequency. However, as the proposition below demonstrates, we cannot completely remove frequency restrictions on the right-hand side of \eqref{eq:mainfreqmorest}.
		
		\begin{proposition}
		\label{prop:mainfreqmorest}	
		Let $a=1$. Denote with $u$ a solution to \eqref{eq:odeU}, corresponding to a fixed azimuthal mode solution $\phi_m$ to \eqref{eq:inhomwaveeq} satisfying Assumptions \ref{asm:G} and \ref{asm:phi}. Then, for $\epsilon>0$ arbitrarily small, there exists a constant $c_m>0$ such that:
		\begin{equation}
		\label{eq:mainfreqmorest}	
		\begin{split}
			c_m\int_{\R_{\omega}}\sum_{\ell} \int_{1}^{\infty}&\chi_+(1-r^{-1})^{-2+\epsilon}|u'+i\tomega u|^2+(1-r^{-1})^{-1+\epsilon}r^{-1-\epsilon}|u'|^2\\
			&\:+(\Lambda r^{-3} +\tomega^2(1-r^{-1})^{-1}r^{-3} +\omega r^{-1-\epsilon})\zeta_{\rm trap}(1-r^{-1})^{\epsilon}|u|^2\,dr\\
			\leq &\: \int_{\R_{\omega}}\sum_{\ell}\int_1^{\infty}\mathbf{1}_{\mathcal{G}}\cdot \Re( \mathbf{f} \cdot(u,u')\cdot \overline{H})\,dr_* d\omega\\
			& -\int_{\R_{\omega}} \sum_{\ell}\int_{\R_{r_*}}2y\Re(u'\overline{H})\,dr_*+  \int_{\R}E(\chi_K\tomega+\omega)\Re(i u \overline{H})\\
			&\:+(r-1)^{-1+\epsilon}\chi_+ \Re\left(v'e^{-i \widetilde{\omega} r_*}\Delta^{-\frac{ im}{2}}\overline{H}\right)\,dr_* d\omega\\
			&\:+ C\int_{\R_{\omega}}\int_1^{\infty}\sum_{\ell}\mathbf{1}_{|\tomega|\leq z_0^{-1}|m| (r-1)}\mathbf{1}_{m\tomega<0} y_{+}\Re\left(\frac{d(\mathfrak{w}_+^{-1}u)}{dr}\overline{\mathfrak{w}_+^{-1}H}\right)\\
			&-\chi_+\mathbf{1}_{|\tomega|\leq z_0^{-1}|m| (r-1)}\mathbf{1}_{m\tomega>0}y_{-}\Re\left(\frac{d(\mathfrak{w}_-^{-1}u)}{dr}\overline{\mathfrak{w}_-^{-1}H}\right)\,dr d\omega\\
		&\:+ C\int_{0}^1 \int_{\Sigma_{\tau}}\Delta^{\frac{1-\epsilon}{2}}r^{-1}r^2|F_{
		\xi}|^2\,d\sigma dr d\tau+C\int_0^{\infty}\int_{\Sigma_{\tau}}(1+\tau)^{1+\delta}\Delta^{\frac{1-\epsilon}{2}}r^{-1}r^2|G_m|^2\,d\sigma drd\tau,
		\end{split}
			\end{equation}
			where
			\begin{equation*}
				\zeta_{\rm trap}(r):= 1-\mathbf{1}_{r_{{\rm trap},1}\leq r\leq r_{{\rm trap},2}}(r),
			\end{equation*}
			and we can take $r_{{\rm trap},i}$, with $i=1,2$, arbitrarily close to $1+\sqrt{2}$. The term $\mathbf{1}_{\mathcal{G}}\cdot \Re( \mathbf{f} \cdot(u,u')\cdot \overline{H})$ is schematic notation for products involving $u$ or $u'$ with $\overline{H}$ and a function $\mathbf{f}$ belonging to the set:
			\begin{equation*}
				\mathbf{f}\in \{f,f_{\sim} , \hat{y},\tilde{y}_{\flat},{y}_{\flat,1},{y}_{\flat,2}, h,h_{\flat},\chi_1,\chi_2\},
			\end{equation*}
			where the functions $\mathbf{f}$ are bounded near $r=1$.
			The factor $\mathbf{1}_{\mathcal{G}}$ denotes an indicator function supported on one of the frequency ranges  $\mathcal{G}^{\sharp}_{\rm sr}, \mathcal{G}^{\sharp, \omega}_{\rm nsr}, \mathcal{G}^{\sharp, \Lambda}_{\rm nsr}, \mathcal{G}^{\sharp, \sim}_{\rm nsr}, \mathcal{G}^{\flat}_<$, or $\mathcal{G}^{\flat}_>$. All the relevant functions are precisely defined in \S\S \ref{sec:rweightfreq}--\ref{sec:boundfreqaway}.
		\end{proposition}
		\begin{proof}
		To obtain control over $(1-r^{-1})^{-1+\epsilon} r^{-1-\epsilon}|u'|^2$ and $\tomega^2(1-r^{-1})^{-1}$ near $r=1$ and
		\begin{equation*}
		r^{-1-\delta}(|u'|^2+\omega^2|u|^2)+r^{-3}\Lambda |u|^2
		\end{equation*}
		as required to estimate the left-hand side of \eqref{eq:mainfreqmorest}, we apply the $y$-current from the proof of Proposition \ref{prop:omegadom} for all admissible frequencies and then apply the already established estimates for each frequency range to absorb any terms with a bad sign. 
		
		We can remove the frequency restrictions in the $T$-energy and $K$-energy currents by also considering $T$- and $K$-energy currents (with cut-offs) for $\mathcal{G}^{\flat}_<$ and absorbing the resulting boundary terms using the estimates already established for $\mathcal{G}^{\flat}_<$. 
		
		Similarly, we apply the $(r-1)$-weighted estimates with $p=1-\epsilon$ for all frequencies and absorb the resulting error terms using the already established estimates. Note that we do note entirely remove the restriction to frequency ranges on the right-hand side of \eqref{eq:mainfreqmorest} in the estimates involving $f$, $y$ and $h$-currents, as indicated by the schematic notation $\mathbf{1}_{\mathcal{G}}$.
		
		Finally, we note that the estimates involving the current $\check{j}^{y_{\pm}}[\check{u}_{\rm inf}]$ can be extended to frequencies $(\omega,m,\Lambda)\notin \mathcal{G}^{\flat}_{<}$ in the region $(r-1)\geq z_0 |m|^{-1}|\tomega|$. This restriction is the reason why we need to include the indicator function $\mathbf{1}_{|\tomega|\leq z_0^{-1}|m| (r-1)}$ on the right-hand side of \eqref{eq:mainfreqmorest}.
		\end{proof}
		
		\section{Integrated energy estimates}
		\label{sec:intenestphys}
		In this section, we will convert the integrated estimates in \emph{frequency space} from Proposition \ref{prop:mainfreqmorest} to local integrated energy estimates in \emph{physical space}.
		\begin{proposition}
		Let $\phi_m$ be a fixed azimuthal mode solution to \eqref{eq:inhomwaveeq} satisfying Assumptions \ref{asm:G} and \ref{asm:phi}. Let $\epsilon>0$ be arbitrarily small. Then there exists a $C_m=C_m(m,\mathbbm{h},M, W_0, \epsilon)>0$ such that:
		\begin{multline}
				\label{eq:mainmorawetz}
			\int_{0}^{\infty}\int_{\Sigma_{\tau}}[\chi_+(1-Mr^{-1})^{2-\epsilon}|X\hphi_m|^2+(1-Mr^{-1})^{-1+\epsilon}r^{1-\epsilon}|{\mathfrak{X}}_*\hphi_m|^2\\
			+(1-Mr^{-1})^{\epsilon}(r^{-1}|\hphi_m|^2+\zeta_{\rm trap}(1-Mr^{-1})^{-1}r^{-1}|K\hphi_m|^2+\zeta_{\rm trap}r^{1-\epsilon}|T\hphi_m|^2+r^{-1}\zeta_{\rm trap}|\snabla_{\s^2}\hphi_m|^2)]d\sigma dr\,d\tau\\
				\leq C\int_{\Sigma_{0}}\left[(1-Mr^{-1})^{1-\epsilon}|X\hphi_m|^2+|T\hphi|^2+r^{-2}|\snabla_{\s^2} \hphi_m|^2\right]\,r^2d\sigma dr\\
				+C \int_{0}^{\infty}\int_{\Sigma_{\tau}}[(1+\tau)^{1+\delta}(1-Mr^{-1})^{1-\epsilon}+r^2\chi_{0\leq \tau\leq 1}+r^{1+\delta}]|G_m|^2\,r^2d\sigma drd\tau,
			\end{multline}
			with $\chi_{\tau_1\leq \tau\leq \tau_1+1}$ the indicator function supported in $\tau_1\leq \tau\leq \tau_1+1$.
					\end{proposition}
					\begin{proof}
					For the sake of notational convenience, we will omit the hat in the notation $\hphi$, $\hpsi$ and simply write $\phi$ and $\psi$, respectively. We set $r_+=1$. The key ingredients in the derivation of \eqref{eq:mainmorawetz} are the estimate \eqref{eq:mainfreqmorest}, Plancherel identities and the local energy estimate from Proposition \ref{prop:locenest}. Let $g_i: (1,\infty)\to \R$ with $i=1,2,3$ be arbitrary continuous functions and $w:(1,\infty)\to \R$ a weight function. By Plancherel:
					\begin{align}
					\label{eq:planch1a}
					\int_{\R_{\omega}}&\sum_{\ell\geq |m|} \int_{\R_{r_*}}\Re((g_1+i\omega g_2+i\tomega g_3) u \cdot \overline{u})+g_4 |u'|^2\,dr_* d\omega=\int_{\mathcal{R}} g_1 |\uppsi_m|^2+g_2 \Re( T\uppsi_m\cdot \uppsi_m)\\ \nonumber
					&\:+g_3 \Re( K\uppsi_m\cdot \uppsi_m)+g_4|{\mathfrak{X}}_*\uppsi_m|^2\,d\sigma drdt,\\
					\label{eq:planch1b}
					\int_{\R_{\omega}}&\sum_{\ell\geq |m|} \int_{\R_{r_*}} w(r)|H|^2\,dr_*=\int_{\mathcal{R}} (r^2+1)^{-2} w(r) \Delta  |\rho^2 F_m|^2\,d\sigma dr dt
					\end{align}
and
					\begin{multline}
					\label{eq:planch2}
					\int_{\R_{\omega}}\sum_{\ell\geq |m|} \int_{\R_{r_*}}\Re((g_1+i\omega g_2+i\tomega g_3) u \cdot \overline{H})+\Re(g_4 u' \cdot \overline{H})\,dr_* d\omega\\
					=\int_{\mathcal{R}} (r^2+1)^{-\frac{1}{2}}g_1 \Re( \uppsi_m \cdot \rho^2\overline{F}_m)+(r^2+1)^{-\frac{1}{2}}g_2 \Re( T\uppsi_m\cdot \rho^2\overline{F}_m)+(r^2+1)^{-\frac{1}{2}}g_3 \Re( K\uppsi_m\cdot \rho^2\overline{F}_m)\\
					+(r^2+1)^{-\frac{1}{2}}g_4\Re({\mathfrak{X}}_*\uppsi_m\cdot \rho^2\overline{F}_m)\,d\sigma drdt.
					\end{multline}
					
					Note moreover that $dt=d\tau+\left(\hfol-\frac{\Delta}{r^2+1}\right)dr$. Hence the left-hand side of \eqref{eq:mainfreqmorest} controls:
					\begin{multline}
					\label{eq:lhsmoraw}
					\tilde{c}_m\int_{1}^{\infty}\int_{\Sigma_{\tau}\cap\{r\leq r_I\}} \Big[(1-Mr^{-1})^{2+\epsilon}|X\phi_m|^2+(1-Mr^{-1})^{\epsilon}(r^{-2}|\phi_m|^2\\
					+\zeta_{\rm trap}(1-Mr^{-1})^{-1+\epsilon}|K\phi_m|^2+r^{-2}\zeta_{\rm trap}|\snabla_{\s^2}\phi_m|^2)\Big]d\sigma dr\,d\tau
					\end{multline}
					for some suitably small constant $\tilde{c}_m$.
					
					We will now estimate the terms involving restrictions to different frequency ranges
					\begin{equation*}
					 \int_{\R_{\omega}}\sum_{\ell\geq |m|}\int_{\R_{r_*}}\mathbf{1}_{\mathcal{G}}\cdot \Re( \mathbf{f} \cdot(u,u')\cdot \overline{H})\,dr_* d\omega
					\end{equation*}
					by a weighted energy along $\Sigma_0$. Since these terms involve indicator functions restricting to subsets of frequencies, we cannot apply \eqref{eq:planch2}.
					
					We will encounter in particular terms as above involving the function $g_1$, with $g_1$ compactly supported and $|g_1|\leq C(1-r^{-1})^2$. Hence, we can apply Young's inequality to conclude that there exists an $R>1$ such that
					\begin{equation*}
					\left|\int_{\R_{r_*}}\Re(g_1 u \cdot \overline{H})\,dr_*\right| \leq \int_{1}^R \eta (1-r^{-1})^{\epsilon}|u|^2+C\eta^{-1}(1-r^{-1})^{-\epsilon}|H|^2 \,dr.
					\end{equation*}
					By \eqref{eq:planch1a}, we can estimate
					\begin{equation*}
					\int_{\R_{\omega}}\sum_{\ell\geq |m|}  \int_{1}^R (1-r^{-1})^{-\epsilon}|H|^2 \,dr\leq C \int_{0}^{\infty} \int_{\Sigma_{\tau}\cap\{r\leq R\}}(1-r^{-1})^{2-\epsilon}(|F_{\xi}|^2+|G_m|^2)\,d\sigma dr d\tau.
					\end{equation*}
					Applying Proposition \ref{prop:locenest} with $p=0$, gives moreover:
					\begin{equation*}
					 \int_{0}^{\infty}\int_{\Sigma_{\tau}\cap\{r\leq R\}} (1-r^{-1})^{2-\epsilon}|F_{\xi}|^2\,d\sigma dr d\tau\leq C \int_{\Sigma_0}\mathcal{E}_{\epsilon}[\phi]\,r^2d\sigma dr+C  \int_{0}^{1}\int_{\Sigma_{\tau}} |G_m|^2\,r^2d\sigma dr d\tau.
					\end{equation*}
					
					We will also encounter terms involving $g_4$, with $g_4$ either compactly supported or exponentially decaying in $r$, which can be similarly estimated via Young's inequality, as above.

					We consider now the integral of $E\omega\Re(i u \overline{H})$, which is supported on all frequency ranges, so we can apply Plancherel in the form \eqref{eq:planch2}. We then split via \eqref{eq:Fxi}
					\begin{multline*}
					\int_{\mathcal{R}} (r^2+1)^{-\frac{1}{2}}g_2 \Re( T\uppsi_m \cdot \rho^2\overline{F}_{\xi})\,d\sigma drdt= \int_{\mathcal{R}} -2g_2 \Re( T(\xi \psi_m) \cdot T(\xi)\overline{X\psi_m})\,d\sigma drdt\\
					+ \int_{\mathcal{R}} \Re( T\upphi_m \cdot [O(1)T(\xi)\overline{X\psi_m}+O(1)T(\xi)\overline{T\phi_m}+O(1)m T(\xi)\overline{\phi_m}+O(1)T(\xi)\overline{\phi_m}+ O(1)T^2(\xi)\overline{\phi_m})\,d\sigma drdt.
					\end{multline*}
					The second integral on the right-hand side above can easily be estimated via Proposition \ref{prop:locenest} with $p=0$. We rewrite the first integral on the right-hand side above:
					\begin{equation*}
					\begin{split}
					\int_{\mathcal{R}}& -2g_2 \Re( T(\xi \psi_m) \cdot T(\xi)\overline{X\psi_m})\,d\sigma drdt=\int_{\mathcal{R}} -g_2X(T(\xi)^2  |\psi_m |^2)-g_2T(\xi^2)\Re( T\psi_m \cdot \overline{X\psi_m})\,d\sigma drdt\\
					=&\: \int_{\mathcal{R}} -X\left(g_2T(\xi)^2  |\psi_m |^2-\frac{1}{2}g_2T(\xi^2) T(|\psi_m|^2)\right)+g_2T(\xi^2) \Re( XT\psi_m \cdot \overline{\psi_m})\\
					&\:+\frac{dg_2}{dr}\left[T(\xi)^2  |\psi_m |^2+\frac{1}{2}T(\xi^2) T(|\psi_m|^2)\right]\,d\sigma drdt.
					\end{split}
										\end{equation*}
					Note that $\frac{dg_2}{dr}=0$ in this case. We use the compact support of $T(\xi)$ in $\tau$ and the fundamental theorem of calculus to infer that the total $X$-derivative can be controlled by the boundary terms
					\begin{equation*}
					\int_{\mathcal{I}^+\cap \{0\leq\tau \leq 1\}} e^{-\tau}|\psi_m|^2+|T\psi_m|^2\,d\sigma d\tau+\int_{\Sigma_0\cap\{r=1\}}|\psi_m|^2\,d\sigma+\int_{\mathcal{H}^+\cap\{0\leq \tau\leq 1\}} e^{-\tau}|\psi_m|^2+|K\psi_m|^2\,d\sigma d\tau,
					\end{equation*}
					where the exponential weight is chosen for convenience.
					
					By the fundamental theorem of calculus and a Hardy inequality, we can control the second integral above by
					\begin{equation*}
					\int_{\Sigma_0\cap \{r\leq R\}} |X\psi_m||\psi_m|\,d\sigma dr d\tau\leq C\int_{\Sigma_0\leq \{r\leq R\}} \mathcal{E}_{1+\epsilon}[\phi]\,d\sigma dr.
					\end{equation*}
					Using that $\psi_m$ is compactly supported along $\Sigma_0$, we can also apply a Hardy inequality in $\tau$ to estimate the above boundary term along $\mathcal{I}^+$ by $\int_{\mathcal{I}^+\cap \{0\leq\tau \leq 1\}} |T\psi_m|^2\,d\sigma d\tau$ and the boundary term along $\mathcal{H}^+$ by 
					\begin{equation*}
					\int_{\mathcal{H}^+\cap \{0\leq\tau \leq 1\}} |K\psi_m|^2\,d\sigma d\tau+\int_{\Sigma_0\cap\{r=1\}}|\psi_m|^2\,d\sigma.
					\end{equation*}
					The integrals of the time derivatives along $\mathcal{H}^+$ and $\mathcal{I}^+$ can then be estimated via Proposition \ref{prop:locenest} with $p=0$.
					
					Finally, the term involving a factor $XT\psi_m$ can now be estimated by using that $\square_g\psi_m=G_m$ and applying \eqref{eq:waveeqX}. The resulting terms that do not involve $G_m$ can then be estimated straightforwardly by integrating by parts in the $X$- and $T$-directions and applying Proposition \ref{prop:locenest} with $p=0$. The term involving $G_m$ leads to the following additional term on the right-hand side:
					\begin{equation*}
					\int_0^1\int_{\Sigma_{\tau}}|\rho^2G|^2\,d\sigma drd\tau.
					\end{equation*}
					
					When estimating the integral
					\begin{equation*}
					\left|\int_{\mathcal{R}\cap\{r\geq R\}} (r^2+1)^{-\frac{1}{2}}g_2 \Re( T\uppsi_m \cdot \rho^2G_m)\,d\sigma drdt\right|,
					\end{equation*}
					we cannot make use of compact support in $\tau$. Instead, we apply Young's inequality to control this integral by
					\begin{equation*}
					\int_0^{\infty} \int_{\Sigma_{\tau}\cap\{r\geq R\}}\eta r^{1-\delta}|T\phi_m|^2+C \eta^{-1}r^{1+\delta}|rG_m|^2\,d\sigma drd\tau,
					\end{equation*}

					The integral of $\chi_K\Re(i\tomega \psi_m\cdot \overline{H})$ can be estimated in an analogous way, with the role of $T$ taken on by $K$ and where we apply \eqref{eq:waveeqtildeX} instead of \eqref{eq:waveeqX}. The only difference here, is that we estimate the integral
					\begin{equation*}
					\left|\int_{\mathcal{R}\cap\{r\leq R_+\}} (r^2+1)^{-\frac{1}{2}}g_3 \Re( K\uppsi_m \cdot \rho^2G_m)\,d\sigma drdt\right|,
					\end{equation*}
					via Young's inequality by
					\begin{equation*}
					\int_0^{\infty} \int_{\Sigma_{\tau}\cap\{r\leq R_+\}}\eta (r-1)^{\epsilon-1}|K\phi_m|^2+C \eta^{-1}(r-1)^{1-\epsilon}|G_m|^2\,d\sigma drd\tau.
					\end{equation*}
					
					We now turn to the integral of $-2y\Re(u'\overline{H})$. We apply Plancherel in the form \eqref{eq:planch2} with the function $g_4=-2y$ and we use that $|g_4|\leq C$ to infer that the integral is bounded by:
					\begin{equation*}
					\begin{split}
					\int_{\mathcal{R}}& (r^2+1)^{-\frac{1}{2}}g_4 \Re( {\mathfrak{X}}_*\uppsi_m \cdot \rho^2\overline{F}_{\xi})\,d\sigma drdt\\
					\leq &\:C \int_{0}^{\infty} \int_{\Sigma_{\tau}} r^{-1}|{\mathfrak{X}}_*\psi_m|(|\rho^2F_{\xi}|+|\rho^2G_m|)\,d\sigma dr d\tau\\
					\leq &\:C \int_0^1\int_{\Sigma_{\tau}}  \mathcal{E}_0[\phi]\,r^2d\sigma dr d\tau+C\int_{0}^{\infty} \int_{\Sigma_{\tau}}  \eta (1-r^{-1})^{-1+\epsilon} r^{-1-\epsilon}|{\mathfrak{X}}_*\psi_m|^2\\
					+&\:\eta^{-1}(1-r^{-1})^{1-\epsilon}r^{-1+\epsilon}|\rho^2G|^2\,d\sigma dr d\tau.
					\end{split}
					\end{equation*}
					We can absorb the term with the factor $\eta$ for $\eta$ suitably small and estimate the integral in $\{0\leq \tau\leq 1\}$ via Proposition \ref{prop:locenest} with $p=0$.
					
					Consider now the integral of $(r-1)^{-1+\epsilon}\chi_+ \Re\left(v'e^{-i \widetilde{\omega}r_*}\Delta^{-\frac{i m}{2}}\overline{H}\right)$. Recall the definition of $v$ for $a=1$:
					\begin{equation*}
v= e^{i \widetilde{\omega} r_*}\Delta^{\frac{im}{2}} u.
\end{equation*}
Hence
\begin{equation*}
v'e^{-i \widetilde{\omega} r_*}\Delta^{-\frac{im}{2}}\overline{H}=\left[u'+i\tomega u-im(r^2+1)^{-1} (r-1)u\right]\cdot \overline{H}.
\end{equation*}
Note moreover that
\begin{equation*}
{\mathfrak{X}}_*=\frac{\Delta}{r^2+1}Y+K-\frac{r-1+\frac{1}{2}(r-1)^2}{r^2+1}\Phi.
\end{equation*}

By Plancherel in the form of \eqref{eq:planch2} with $g_1=im (r^2+1)^{-1} (r-1)$, $g_3=1$ and $g_4=1$ we can therefore estimate the corresponding integral by:
\begin{equation*}
\int_{\mathcal{R}\cap\{r\leq R_+\}} (r-1)^{-1+\epsilon}(r^2+1)^{-\frac{1}{2}}\Re\left( \frac{\Delta}{r^2+1}Y\uppsi_m \cdot \rho^2\overline{F}\right)+O((r-1)^{1+\epsilon})\Re( \Phi\uppsi_m \cdot \rho^2\overline{F})\,d\sigma drdt,
\end{equation*}
with $\chi_+$ supported in $\{r\leq R_+-1\}$.

This integral can be further controlled by:
\begin{equation*}
\begin{split}
\int_0^{\infty} &\int_{\Sigma_{\tau}} (r-1)^{1+\epsilon} (|X\psi_m|+|T\psi_m|+|\Phi \psi_m|)(|F_{\xi}|+|G_m|)\,d\sigma drd\tau\leq C \int_0^1\int_{\Sigma_{\tau}\cap \{r\leq R_+\}} \mathcal{E}_{1-\epsilon}[\phi_m]\,d\sigma dr d\tau\\
&\: + \int_0^{\infty}\int_{\Sigma_{\tau}\cap \{r\leq R_+\}} \eta (r-1)^{2+\epsilon}(|X\psi_m|^2+|K\psi_m|^2+|\Phi \psi_m|^2)+ C \eta^{-1}|G_m|^2\,d\sigma drd\tau,
\end{split}
\end{equation*}
where, as above, we absorb the term with a factor $\eta$ and estimate the integral in $\{0\leq \tau\leq 1\}$ via Proposition \ref{prop:locenest} with $p=1-\epsilon$.

Finally, we will estimate
\begin{multline*}
\int_{\R_{\omega}}\int_1^{\infty}\sum_{\ell}\mathbf{1}_{|\tomega|\leq z_0^{-1}|m| (r-1)}\mathbf{1}_{m\tomega<0} \chi_+y_{+}\Re\left(\frac{d(\mathfrak{w}_+^{-1}u)}{dr}\overline{\mathfrak{w}_+^{-1}H}\right)\,dr d\omega.
			\end{multline*}
			The integral involving $\mathfrak{w}_-$ can be estimated analogously. We would like to apply Plancherel \eqref{eq:planch2} with a suitable function $g_4$, but due to the presence of $\mathbf{1}_{|\tomega|\leq z_0^{-1}|m| (r-1)}\mathbf{1}_{m\tomega<0}$, we need to replace $\uppsi_m$ on the right-hand side of \eqref{eq:planch2} with the convolution:
			\begin{equation*}
			\uppsi_m^{\rm conv}:=\mathcal{F}^{-1}_{\omega}(\mathbf{1}_{|\tomega|\leq z_0^{-1}|m| (r-1)}\mathbf{1}_{m\tomega<0})*\uppsi_m,
			\end{equation*}
			where $\mathcal{F}^{-1}_{\omega}$ denotes the inverse Fourier transform with respect to the frequency $\omega$.
			
			Note that for $m<0$,
			\begin{equation*}
			\mathcal{F}^{-1}_{\omega}(\mathbf{1}_{|\tomega|\leq z_0^{-1}|m| (r-1)}\mathbf{1}_{m\tomega<0})(t,r)=\frac{e^{i\frac{m}{2}t}}{\sqrt{2\pi}}\int_{0}^{z_0^{-1}|m| (r-1)}e^{i\tomega t}\,d\tomega=\frac{e^{i\frac{m}{2}t}}{i \sqrt{2\pi} t}(1-e^{iz_0^{-1}m(r-1) t}).
			\end{equation*}
			For $m>0$, we similarly obtain:
			\begin{equation*}
			\mathcal{F}^{-1}_{\omega}(\mathbf{1}_{|\tomega|\leq z_0^{-1}|m| (r-1)}\mathbf{1}_{m\tomega<0})(t,r)=-\frac{e^{i\frac{m}{2}t}}{i \sqrt{2\pi} t}(1-e^{iz_0^{-1}m(r-1) t}).
			\end{equation*}
			By Plancherel, we need to control
			\begin{equation*}
\int_{\mathcal{R}\cap\{r\leq R_+\}} (r^2+1)^{-\frac{1}{2}}\chi_+y_+\Re\left({\mathfrak{X}}_*\left(\mathfrak{w}_+^{-1}\uppsi_m^{\rm conv}\right) \cdot \overline{\mathfrak{w}_+}^{-1}\rho^2\overline{F}\right)\,d\sigma drdt,
\end{equation*}
		which can be estimated by
		\begin{equation*}
		\int_0^{\infty}\int_{\Sigma_{\tau}\cap\{r\leq R_+\}} (r-1)^{-1}((r-1)^2|X\uppsi_m^{\rm conv}|+|m|(r-1)|\uppsi_m^{\rm conv}|)(|F_{\xi}|+|G_m|)\,d\sigma dr d\tau,
		\end{equation*}
		where we moreover used that $|K\uppsi_m^{\rm conv}|^2\leq z_0^{-2}m^2(r-1)^2|\uppsi_m^{\rm conv}|$ by $|\tomega|\leq z_0^{-1}|m| (r-1)$.
		
		Note that $|\uppsi_m^{\rm conv}|$ can be estimated as follows:  by Cauchy--Schwarz together with
		\begin{equation*}
		|1-e^{iz_0^{-1}m(r-1) s}|\leq C |m|(r-1)\quad \textnormal{for $|s|\leq 1$}
		\end{equation*}
		we obtain for $0\leq \tau\leq 1$:
		\begin{equation*}
		\begin{split}
		|\uppsi_m^{\rm conv}|^2(\tau,r,\theta,\varphi)\leq&\:\left(\int_{\R}\frac{1}{\sqrt{2\pi} |s|}|1-e^{iz_0^{-1}m(r-1) s}||\uppsi_m|(\tau-s,r,\theta,\varphi)\,ds\right)^2\\
		\leq&\: Cm^2\int_{\R\setminus(-1,1)} |s|^{-2}\,ds\int_{\R} (r-1)^2\xi^2|\psi_m|^2(\tau-s,r,\theta,\varphi)\,ds\\
		+&\:C\int_{-1}^1 \xi^2|\psi_m|^2(\tau-s,r,\theta,\varphi)\,ds\\
		\leq&\: C\int_0^{\infty} (r-1)^2|\psi_m|^2(\tau',r,\theta,\varphi)\,d\tau'+C\int_0^{2 }|\psi_m|^2(\tau',r,\theta,\varphi)\,d\tau'.
		\end{split}
		\end{equation*}
		Similarly,
		\begin{equation*}
		\begin{split}
		|X\uppsi_m^{\rm conv}|^2(\tau,r,\theta,\varphi)\leq&\: C\int_0^{\infty} (r-1)^2|X\psi_m|^2(\tau',r,\theta,\varphi)\,d\tau'+C\int_0^{2 }|X\psi_m|^2(\tau',r,\theta,\varphi)\,d\tau'.
		\end{split}
		\end{equation*}
		Therefore
		\begin{equation*}
		\begin{split}
		\int_0^{1}&\int_{\Sigma_{\tau}\cap\{r\leq R_+\}} (r-1)^{-1}((r-1)^2|X\uppsi_m^{\rm conv}|+|m|(r-1)|\uppsi_m^{\rm conv}|)|F_{\xi}|\,d\sigma dr d\tau\\
		\leq&\: \eta  \int_0^{2}\int_{\Sigma_{\tau}\cap\{r\leq R_+\}}(r-1)^{-1+\epsilon}m^2|\psi_m|^2+(r-1)^{1+\epsilon}|X\psi_m|^2\,d\sigma drd\tau\\
		+&\: \eta \int_0^{\infty}\int_{\Sigma_{\tau}\cap\{r\leq R_+\}}(r-1)^{1+\epsilon}m^4|\psi_m|^2+(r-1)^{3+\epsilon}|X\psi_m|^2\,d\sigma drd\tau\\
		+&\:\int_0^{1}\int_{\Sigma_{\tau}\cap\{r\leq R_+\}}C\eta^{-1}(r-1)^{1-\epsilon}|F_{\xi}|^2\,d\sigma dr d\tau.
		\end{split}
		\end{equation*}
		By taking $\eta$ suitably small (depending also on $|m|$) and applying a Hardy inequality to control the $(r-1)^{-1+\epsilon}m^2|\psi_m|^2$ term by the $(r-1)^{1+\epsilon}|X\psi_m|^2$, we can either absorb the terms with a factor $\eta$ or control them via the local energy estimates via Proposition \ref{prop:locenest} with $p=1-\epsilon$. The $F_{\xi}$ term can similarly be estimated via Proposition \ref{prop:locenest} with $p=1+\epsilon$.
		
		To estimate the terms involving $|G_m|$, we use that for all $\tau\geq 0$:
		\begin{equation*}
		\begin{split}
		\int_{\R}|\uppsi_m^{\rm conv}|^2(t,r,\theta,\varphi)\,dt\leq&\: \int_0^{\infty} |\psi_m|^2(\tau,r,\theta,\varphi)\,d\tau,
		\end{split}
		\end{equation*}
		so
		\begin{equation*}
		\begin{split}
		\int_0^{\infty}&\int_{\Sigma_{\tau}\cap\{r\leq R_+\}} (r-1)^{-1}((r-1)^2|X\uppsi_m^{\rm conv}|+|m|(r-1)|\uppsi_m^{\rm conv}|)|G_m|\,d\sigma dr d\tau\\
		\leq&\: C\int_0^{\infty}\int_{\Sigma_{\tau}\cap\{r\leq R_+\}} \eta (r-1)^{2+\epsilon}|X\psi_m|^2+\eta m^2(r-1)^{\epsilon}|\psi_m|^2+C \eta^{-1}(r-1)^{1-\epsilon}|G_m|^2\,d\sigma dr d\tau.
		\end{split}
		\end{equation*}
	We absorb the terms with a factor $\eta$.
		\end{proof}

				We can apply \eqref{eq:mainmorawetz} to obtain different (higher-order) integrated energy estimates in different spacetime regions. 
				
				\begin{theorem}
				\label{thm:morawetz}
				Assume that \eqref{assm:wronk} holds and let $0<\epsilon<1$. Let $\phi$ supported on azimuthal modes with azimuthal number $|m|\leq m_0$.
				\begin{enumerate}[label=\emph{(\roman*)}]
				\item There exists a $C=C(m_0,\mathbbm{h},M, \epsilon)>0$ such that such that for all $0\leq \tau_1< \tau_2$:
				\begin{multline}
				\label{eq:mornearhor}
				\int_{\Sigma_{\tau_2}}\mathcal{E}_{1-\epsilon}[\phi]\,r^2d\sigma dr +\int_{\tau_1}^{\tau_2} \int_{\Sigma_{\tau}\cap\{r\leq 2M\}}(1-Mr^{-1})^{\epsilon}\mathcal{E}_0[\phi]+(1-Mr^{-1})^{-1+\epsilon}|K\phi|^2\,d\sigma dr d\tau\\
				\leq C \int_{\Sigma_{\tau_1}}\mathcal{E}_{1+\epsilon}[\phi]\,r^2d\sigma dr+C  \int_{\tau_1}^{\infty}\int_{\Sigma_{\tau}}[(1+\tau)^{1+\epsilon}(1-Mr^{-1})^{1-\epsilon}+r^2\chi_{\tau_1\leq \tau\leq \tau_2+1}+r^{1+\epsilon}]|G|^2\,r^2d\sigma drd\tau.
				\end{multline}
			
					\item Let $M<r_H<2M<r_I$. Then there exists a $C=C(m,\mathbbm{h},M, \epsilon, r_H,r_I)>0$ such that such that for all $0\leq \tau_1< \tau_2$:
				\begin{multline}
				\label{eq:morneartrapps}
				\int_{\tau_1}^{\tau_2} \int_{\Sigma_{\tau}\cap\{r_H\leq r\leq r_I\}}(|\phi|^2+\mathcal{E}_2[\phi])\,r^2d\sigma dr d\tau\\
				\leq C\sum_{k\leq 1}\int_{\Sigma_{\tau_1}}\mathcal{E}_{1+\epsilon}[K^{k}\phi]\,r^2d\sigma dr\\
				+C \int_{\tau_1}^{\infty}\int_{\Sigma_{\tau}}[(1+\tau)^{1+\epsilon}(1-Mr^{-1})^{1-\epsilon}+r^2\chi_{\tau_1\leq \tau\leq \tau_2+1}+r^{1+\epsilon}]|K^{k}G|^2\,r^2d\sigma drd\tau.
				\end{multline}

				\item Let $M<r_H<2M$ and $4M<r_I<\infty$. Denote with $Z$ an element of $\{(1-Mr^{-1})\Lbar,rL,\Phi, \slashed{\nabla}_{\s^2}\}$. Then there exists a $C=C(\mathbbm{h},M, W_0, \epsilon, r_H,r_I)>0$ such that such that for all $0\leq \tau_1< \tau_2$:
				\begin{multline}
				\label{eq:mornearho}
				\sum_{|I|\leq N}\int_{\tau_1}^{\tau_2} \int_{\Sigma_{\tau}\cap\{r_I-M\leq r\leq r_I\}\cup\{r_H\leq r\leq r_H+M\}}(|Z^I\phi|^2+\mathcal{E}_2[Z^I\phi])\,r^2d\sigma dr d\tau\\
				\leq C\sum_{k\leq N}\int_{\Sigma_{\tau_1}}\mathcal{E}_{1+\epsilon}[T^{k}\phi]\,r^2d\sigma dr\\
				+C \int_{\tau_1}^{\infty}\int_{\Sigma_{\tau}}[(1+\tau)^{1+\epsilon}(1-Mr^{-1})^{1-\epsilon}+r^2\chi_{\tau_1\leq \tau\leq \tau_2+1}+r^{1+\epsilon}]|K^{k}G|^2\,r^2d\sigma drd\tau.
				\end{multline}
				\end{enumerate}
				\end{theorem}
				\begin{proof}
				Let $V$ be a timelike vector field in $\{r>M\}$ such that $V=L+(r-M)^{-1+\epsilon}\Lbar$ for $r\leq 2M$ and $V=T$ for $r\geq 2\frac{1}{4}M$. Since the ergoregion is contained in $\{r\leq 2M\}$, such a vector field exists. Since $V$ is Killing, $\nabla_{\mu}(\mathbb{T}^{\mu \nu}[\phi] V_{\nu})$ is supported in $\{r\leq 2\frac{1}{4}M\}$.
				
				We can write
				\begin{equation*}
				V=K+O_1(r-M)\Phi+\frac{1}{4M^2}(r-M)^{1+\epsilon}(1+O_1((r-M)^{1-\epsilon})\underline{Y}
				\end{equation*}
	
				\begin{equation*}
				|\nabla_{\mu}(\mathbb{T}^{\mu \nu}[\phi] V_{\nu})|\leq C (r-M)^{\epsilon}\mathcal{E}_0[\phi]+|K\phi||\Phi\phi|\leq C(r-M)^{\epsilon}\mathcal{E}_0[\phi]+(r-1)^{-\epsilon}|K\phi|^2
				\end{equation*}
				By construction, $\zeta_{\rm trap}=1$ in this region. We can therefore apply \eqref{eq:mainmorawetz} with $\Sigma_0$ replaced by $\Sigma_{\tau_1}$ to control the spacetime integral of  $|\nabla_{\mu}(\mathbb{T}^{\mu \nu}[\phi] V_{\nu})|$ and conclude \eqref{eq:mornearhor}.
				
				Consider now (ii). Since the integral to be estimated includes values of $r$ for which $\zeta_{\rm trap}$ vanishes, the left-hand side of \eqref{eq:mainmorawetz} fails to control $|T\phi|^2+|\snabla_{\s^2}\phi|^2$. By replacing $\phi$ with $T\phi$, we can obtain control over $|T\phi|^2$ without degeneracy at the expense of an extra $T$-derivative on the right-hand side. Since the degeneracy of $\zeta_{\rm trap}$ occurs outside the ergoregion where $T$ is timelike, we can apply standard elliptic estimates to control also obtain moreover $|\snabla_{\s^2}\phi|^2$. Since $|m|$ is bounded, we can equivalently express the estimate in terms of extra $K$-derivatives on the right-hand side and we obtain \eqref{eq:morneartrapps}.
				
				Finally, \eqref{eq:mornearho} follows by observing that the integration region is bounded, excludes the values of $r$ for which $\zeta_{\rm trap}$ vanishes and lies strictly outside the event horizon. The result then follows from the observation that commutation with $K$ and $\Phi$ controls all other derivatives via standard elliptic estimates.
				\end{proof}
				
			\section{$r$-weighted and $(r-M)^{-1}$-weighted energy estimates}
			\label{sec:rweight}
			In this section we derive additional $r$-weighted estimates in far-away regions and $(r-M)^{-1}$-weighted energy estimates near the event horizon. These weighted estimates will play a key role in deriving quantitative energy decay estimates and the powers appearing in the weight factors determine the corresponding decay rates. We will assume that $|a|=M$ in this section.
			
			\textbf{Throughout this section and \S \ref{sec:edecay}, $\phi$ will denote a solution to \eqref{eq:inhomwaveeq} such that Theorem \ref{thm:morawetz} applies to $\phi$.} We will later apply these estimates to $\hphi$ (and its $K$-integral) instead of $\phi$.
			
			\subsection{Estimates near $\mathcal{H}^+$}
			\label{sec:rmin1weight}
			We first consider $(r-M)^{-1}$-weighted energy estimates near the event horizon. These rely on the following key identity:
			\begin{proposition}
			\label{prop:rweightesthor}
			The following identity holds in $M\leq r\leq r_H$ for solutions $\psi$ to \eqref{eq:waveeqnull}:
			\begin{multline}
			\label{eq:keyidnearhor}
			K\Bigg[\left(1-\frac{\mathbbm{h}\Delta}{r^2+a^2}\right) (r^2+a^2)^{-1}(r-M)^{-2-p}|\Lbar \psi|^2+\mathbbm{h}\frac{(r-M)^{2-p}}{4}\left( |\widetilde{\snabla}_{\s^2}\psi|^2+2a\upomega_+ |\Phi \psi|^2\right) \\
+\frac{\mathbbm{h}}{4}a^2\sin^2 \theta(r-M)^{2-p} |K\psi|^2+(r-M)^{-p}O((|\Lbar\psi| +(r-M)|\Phi \psi|+\sin^2\theta (r-M)|K \psi|)^2)\Bigg]\\
+\widetilde{X}\Bigg[(r^2+a^2)^2(r-M)^{-p}|\Lbar \psi|^2-\frac{(r-M)^{2-p}}{4}\left( |\widetilde{\snabla}_{\s^2}\psi|^2+2a\upomega_+ |\Phi \psi|^2\right)-\frac{1}{4}a^2\sin^2 \theta(r-M)^{2-p} |K\psi|^2\\
+(r-M)^{2-p} O( |\Phi \psi| |\Lbar \psi|)\Bigg]+\Phi(\ldots)+\Divs(\ldots)\\
+\left[\frac{p (r^2+a^2)^2}{2}(r-M)^{-1-p} +O((r-M)^{-p})\right]|\Lbar \psi|^2+\frac{2-p}{8}(r-M)^{1-p}\left( |\widetilde{\snabla}_{\s^2}\psi|^2+2a\upomega_+ |\Phi \psi|^2\right)\\
+\frac{1}{8(r^2+a^2)^2}a^2\sin^2 \theta(r-M)^{1-p} |K\psi|^2\\
=+(r-M)^{-p} O(|\Phi \psi| |\Lbar \psi|)+O((r-M)^{2-p} )|\Phi \psi|^2+(r-M)^{2-p} O(|\Phi \psi||G|)+(r-M)^{-p} O(|\Lbar \psi||G|).
			\end{multline}
			\end{proposition}
			\begin{proof}
			We multiply both sides of \eqref{eq:waveeqnull} with $-\frac{1}{2}(r-M)^{2-p}\overline{\Lbar \psi}$ to obtain:
			\begin{multline}
			\label{eq:keyidrweightesthor}
-\frac{(r-M)^{2-p}}{2}\Re(\rho^2G\overline{\Lbar \psi})=\overbrace{2(r^2+a^2)^2(r-M)^{-p}\Re(L\Lbar \psi\overline{\Lbar \psi})}^{=:J_1}\overbrace{-\frac{1}{2}a^2\sin^2 \theta(r-M)^{2-p}  \Re(K^2\psi\overline{\Lbar \psi})}^{=:J_2}\\
\overbrace{-a(1-a\upomega_+\sin^2\theta)(r-M)^{2-p} \Re(K\Phi \psi\overline{\Lbar \psi})}^{=:J_3}\overbrace{-\frac{(r-M)^{2-p}}{2}\Re((\slashed{\mathcal{D}}-2a\upomega_+ \Phi^2)\psi\overline{\Lbar \psi})}^{=:J_4}\\
\overbrace{- \frac{a r (r-M)^{2-p}}{r^2+a^2}\Re(\Phi\psi\overline{\Lbar \psi})}^{=:J_5}+\overbrace{\frac{(r-M)^{2-p}}{2}\frac{(r^2+a^2)^2}{2 r\Delta} \frac{d}{dr}\left(r^2(r^2+a^2)^{-3}\Delta^2\right)\Re(\psi \overline{\Lbar \psi})}^{=:J_6}.
\end{multline}

We start by rewriting $J_1$, using that $L(r)=\frac{\Delta}{2(r^2+a^2)}$:
\begin{equation*}
\begin{split}
J_1=&\:2(r^2+a^2)^2(r-M)^{-p}\Re(L\Lbar \psi\overline{\Lbar \psi})=(r^2+a^2)^2(r-M)^{-p}L(|\Lbar \psi|^2)\\
=&\: L((r^2+a^2)^2(r-M)^{-p}|\Lbar \psi|^2)+\left[\frac{p (r^2+a^2)}{2}(r-M)^{1-p} +O((r-M)^{2-p})\right]|\Lbar \psi|^2.
\end{split}
\end{equation*}
Now, we turn to $J_2$, using that $\Lbar(r)=-\frac{\Delta}{2(r^2+a^2)}$:
\begin{equation*}
\begin{split}
J_2=&\:-\frac{1}{2}a^2\sin^2 \theta(r-M)^{2-p}  \Re(K^2\psi\overline{\Lbar \psi})\\
=&\: K\left((r-M)^{2-p}\sin^2\theta O(|K\psi| |\Lbar \psi|)  \right)+\frac{1}{4}a^2\sin^2 \theta(r-M)^{2-p}\Lbar\left( |K\psi|^2\right)\\
=&\: K\left((r-M)^{2-p}\sin^2\theta O(|K\psi| |\Lbar \psi|)  \right)+\Lbar\left(\frac{1}{4}a^2\sin^2 \theta(r-M)^{2-p} |K\psi|^2\right)\\
+&\:\frac{1}{8(r^2+a^2)}a^2\sin^2 \theta(r-M)^{3-p} |K\psi|^2
\end{split}
\end{equation*}
To rewrite $J_3$, we will use the schematic notation $\Phi(\ldots)$ and $\Divs(\ldots)$ to indicate total $\Phi$-derivative terms and divergence terms on $\s^2$, respectively, which both vanish after integration on $\s^2$. We will moreover use that: 
\begin{equation*}
K=L+\Lbar-a\frac{r^2-r_+^2}{(r^2+a^2)(r_+^2+a^2)}\Phi.
\end{equation*}
We then obtain
\begin{equation*}
\begin{split}
J_3=&\: -a(1-a\upomega_+\sin^2\theta)(r-M)^{2-p} \Re(K\Phi \psi\overline{\Lbar \psi})\\
=&\: -K\left( a(1-a\upomega_+\sin^2\theta)(r-M)^{2-p}  \Re(\Phi \psi\overline{\Lbar \psi})\right)+a(1-a\upomega_+\sin^2\theta)(r-M)^{2-p} \Re(\Phi \psi K\overline{\Lbar \psi})\\
=&\: K\left((r-M)^{2-p}O(|\Phi\psi| |\Lbar \psi|)  \right)+a(1-a\upomega_+\sin^2\theta)(r-M)^{2-p} \Re(\Phi \psi L \overline{\Lbar \psi})\\
&\:+a(1-a\upomega_+\sin^2\theta)(r-M)^{2-p} \Re(\Phi \psi \overline{\Lbar ^2 \psi})+\Phi(\ldots).
\end{split}
\end{equation*}
We estimate the term involving $L\Lbar \psi$ by applying \eqref{eq:waveeqnull} again:
\begin{multline*}
a(1-a\upomega_+\sin^2\theta)(r-M)^{2-p} \Re(\Phi \psi L \overline{\Lbar \psi})= a(1-a\upomega_+\sin^2\theta)O((r-M)^{4-p})\Big[  a^2 \sin^2\theta\Re(\Phi \psi \overline{K^2\psi})\\
+2a(1-a\upomega_+\sin^2\theta)\Re(\Phi \psi  \overline{K\Phi \psi})+\Re(\Phi \psi  \overline{(\slashed{\mathcal{D}}-2a \upomega_+ \Phi^2)\psi})+\Re(\Phi \psi \rho^2\overline{G})\Big] \\
+O((r-M)^{4-p} )|\Phi \psi|^2+\Phi(\ldots).
\end{multline*}
Hence,
\begin{multline*}
a(1-a\upomega_+\sin^2\theta)(r-M)^{2-p} \Re(\Phi \psi L \overline{\Lbar \psi})=K((r-M)^{4-p}O(|\Phi \psi|(\sin^2\theta|K\psi|+|\Phi \psi|))+\Phi(\ldots)\\
+\Divs(\ldots)+O((r-M)^{4-p} )|\Phi \psi|^2+O((r-M)^{4-p} )|\Phi \psi||G|.
\end{multline*}
We conclude that
\begin{multline*}
J_3=K((r-M)^{2-p}O( |\Phi \psi|((r-M)^2\sin^2\theta|K\psi|+(r-M)^2|\Phi \psi|+|\Lbar \psi|))+O((r-M)^{4-p} )|\Phi \psi|^2\\
+\Lbar((r-M)^{2-p} O( |\Phi \psi| |\Lbar \psi|))+\Phi(\ldots)+\Divs(\ldots)+O((r-M)^{4-p} )|\Phi \psi|^2+(r-M)^{3-p} O(|\Phi \psi| |\Lbar \psi|)\\
+(r-M)^{4-p} O(|\Phi \psi||G|).
\end{multline*}
We integrate by parts on $\s^2$ to rewrite $J_4$ as follows:
\begin{equation*}
\begin{split}
J_4=&\: -\frac{(r-M)^{2-p}}{2}\Re((\slashed{\mathcal{D}}-2a\upomega_+ \Phi^2)\psi\overline{\Lbar \psi})\\
=&\:\frac{(r-M)^{2-p}}{4}\Lbar\left( |\widetilde{\snabla}_{\s^2}\psi|^2+2a\upomega_+ |\Phi \psi|^2\right)\\
=&\: \Lbar\left[\frac{(r-M)^{2-p}}{4}\left( |\widetilde{\snabla}_{\s^2}\psi|^2+2a\upomega_+ |\Phi \psi|^2\right)\right]+\frac{2-p}{8}\frac{\Delta}{r^2+a^2}(r-M)^{1-p}\left( |\widetilde{\snabla}_{\s^2}\psi|^2+2a\upomega_+ |\Phi \psi|^2\right)\\
&+\Phi(\ldots)+\Divs(\ldots).
\end{split}
\end{equation*}
Since $J_5$ only involves first-order derivatives, we can directly write:
\begin{equation*}
J_5=(r-M)^{2-p} O(|\Phi \psi| |\Lbar \psi|).
\end{equation*}
\textbf{This term cannot be absorbed into the non-negative definite terms in $J_1-J_6$!} For this reason, the spacetime integral of this term cannot be estimated directly and is the \textbf{key source of difficulties} for obtaining $(r-M)$-weighted estimates near $\mathcal{H}^+$ in extremal Kerr.

We can also immediately write:
\begin{equation*}
J_6=(r-M)^{3-p} O(| \psi| |\Lbar \psi|).
\end{equation*}

In order to integrate in spacetime, we combine all $J_1--J_6$ and we express $\Lbar$ and $L$ as follows in terms of $K$ and $\widetilde{X}$:
\begin{align*}
\Lbar=&\: -\frac{\Delta}{r^2+a^2}(\widetilde{X}-\mathbbm{h} K),\\
L=&\: \left(1-\frac{\mathbbm{h}\Delta}{r^2+a^2}\right)K+\frac{\Delta}{r^2+a^2}\widetilde{X}+a\frac{r^2-r_+^2}{(r^2+a^2)(r_+^2+a^2)}\Phi.
\end{align*}
We obtain:
\begin{multline*}
\frac{r^2+a^2}{\Delta}\sum_{i=1}^6J_i=K\Bigg[\left(1-\frac{\mathbbm{h}\Delta}{r^2+a^2}\right) (r^2+a^2)^{-1}(r-M)^{-2-p}|\Lbar \psi|^2\\
+\mathbbm{h}\frac{(r-M)^{2-p}}{4}\left( |\widetilde{\snabla}_{\s^2}\psi|^2+2a\upomega_+ |\Phi \psi|^2\right) \\
+\frac{\mathbbm{h}}{4}a^2\sin^2 \theta(r-M)^{2-p} |K\psi|^2+(r-M)^{-p}O((|\Lbar\psi| +(r-M)|\Phi \psi|+\sin^2\theta (r-M)|K \psi|)^2)\Bigg]\\
+\widetilde{X}\Bigg[(r^2+a^2)^2(r-M)^{-p}|\Lbar \psi|^2-\frac{(r-M)^{2-p}}{4}\left( |\widetilde{\snabla}_{\s^2}\psi|^2+2a\upomega_+ |\Phi \psi|^2\right)-\frac{1}{4}a^2\sin^2 \theta(r-M)^{2-p} |K\psi|^2\\
+(r-M)^{2-p} O( |\Phi \psi| |\Lbar \psi|)\Bigg]+\Phi(\ldots)+\Divs(\ldots)\\
+\left[\frac{p (r^2+a^2)^2}{2}(r-M)^{-1-p} +O((r-M)^{-p})\right]|\Lbar \psi|^2+\frac{2-p}{8}(r-M)^{1-p}\left( |\widetilde{\snabla}_{\s^2}\psi|^2+2a\upomega_+ |\Phi \psi|^2\right)\\
+\frac{1}{8(r^2+a^2)^2}a^2\sin^2 \theta(r-M)^{1-p} |K\psi|^2+(r-M)^{-p} O(|\Phi \psi| |\Lbar \psi|)+O((r-M)^{2-p} )|\Phi \psi|^2\\
+(r-M)^{2-p} O(|\Phi \psi||G|).
\end{multline*}
Then \eqref{eq:keyidnearhor} follows immediately.
			\end{proof}
			
			The following corollary allows us to conveniently characterize solutions to \eqref{eq:inhomwaveeq} with $G=0$ that satisfy the qualitative integrability properties needed to apply the frequency-localized estimates in \S \ref{sec:freqspacean}.
			
			\begin{corollary}
			\label{cor:intassmpart}
			Assume that there exists a $r_H>M$ such that
			\begin{align*}
			\int_{0}^{\infty}\int_{\Sigma_{\tau} \cap\{r\leq r_H\}} |\phi_m|^2\,d\sigma dr d\tau&\:<\infty,\\
			\int_{0}^{\infty}\int_{\Sigma_{\tau} \cap\{r= r_H\}} |\snabla_{\s^2}\phi_m|^2+|K\phi_m|^2\,d\sigma dr d\tau&\:<\infty.
			\end{align*}
			Then, for $r_H$ suitably small, 
			\begin{equation*}
			\int_{0}^{\infty}\int_{\Sigma_{\tau} \cap\{r\leq r_H\}} |\underline{L}\phi_m|^2\frac{r^2+a^2}{\Delta}\,d\sigma dr d\tau<\infty.
			\end{equation*}
			\end{corollary}
			\begin{proof}
			We integrate \eqref{eq:keyidnearhor} with $p=1$ and $G=0$ in $\mathcal{R}$ and use that the initial data for $\phi_m$ is smooth and compactly supported to conclude that the integral along $\Sigma_0$ is finite. We moreover apply Young's inequality to absorb the integral of $(r-M)^{-p} O(|\Phi \psi| |\Lbar \psi|)$ into the integrals of $|\underline{L}\phi_m|^2$ and $|\snabla_{\s^2}\phi_m|^2$. 
			\end{proof}

			Before we apply Proposition \ref{prop:rweightesthor} to derive weighted energy estimates near $r=M$, we will derive some commutation relations that we will need when considering appropriate, weighted, higher-order energies.
			
			\begin{lemma}
			Let $M=1$ and $N\in \N$. Denote $Z=\frac{r-1}{2(r^2+1)}Y$. Then we have in the region $r\leq 2$:
			\begin{multline}
			\label{eq:commrelZ1}
			\left[Z^N,-\frac{(r^2+1)^2}{\Delta} L\Lbar \right](\cdot)=-\frac{N}{4}(1+O_{\infty}(r-1))\left(-\frac{r^2+1}{\Delta} L\Lbar Z^{N-1}(\cdot) \right)+\frac{N(r^2+1)}{2}Z^{N+1}(\cdot)\\
			+\sum_{k=1}^NO_{\infty}(1)Z^k(\cdot)+\sum_{k=1}^NO_{\infty}(1)Z^k\Phi(\cdot)+\sum_{k=0}^{N-1}O_{\infty}((r-1)^{-2})L\underline{L}Z^k(\cdot)
			\end{multline}
			\end{lemma}
			\begin{proof}
			We will prove \eqref{eq:commrelZ1} via an induction argument. We will first establish \eqref{eq:commrelZ1} for $N=1$. Let $f$ be an arbitrary function on $\mathcal{R}$. Then we can express:
			\begin{equation*}
			\begin{split}
			-\frac{(r^2+1)^2}{\Delta} L\Lbar f=&\:\frac{(r^2+1)^2}{(r-1)^2}\left( T+\frac{a}{r^2+a^2}\Phi+\frac{(r-1)^2}{2(r^2+1)}Y\right)\frac{(r-1)^2}{2(r^2+1)}Yf\\
								=&\:\frac{(r^2+1)^2}{(r-1)^2}\left( T+\frac{a}{r^2+a^2}\Phi+(r-1)Z\right)((r-1)Zf)\\
								=&\:\frac{(r^2+1)^2}{r-1}\left( T+\frac{a}{r^2+a^2}\Phi+(r-1)Z^2f+\frac{r-1}{2(r^2+1)}Zf\right).
			\end{split}
			\end{equation*}
			Acting with $Z$ on both sides of the equation above, we obtain:
			\begin{multline*}
			Z\left(-\frac{(r^2+1)^2}{\Delta} L\Lbar f\right)=-\frac{(r^2+1)^2}{\Delta} L\Lbar Zf-\frac{1}{4}(1+O_{\infty}(r-1))\left(-\frac{(r^2+1)^2}{\Delta} L\Lbar f\right)+\frac{r^2+1}{2}Z^2f\\
			+O_{\infty}(1)Zf+O_{\infty}(1)\Phi f.
			\end{multline*}
			Now we take our that \eqref{eq:commrelZ1} holds with $N$ replaced by $k$ and for all $1\leq k\leq N$. By applying \eqref{eq:commrelZ1}, we obtain
			\begin{multline*}
			ZZ^{N}\left(-\frac{(r^2+1)^2}{\Delta} L\Lbar f\right)=Z\left(-\frac{(r^2+1)^2}{\Delta} L\Lbar Z^Nf\right)-\frac{N}{4}Z\left((1+O_{\infty}(r-1))\left(-\frac{r^2+1}{\Delta} L\Lbar Z^{N-1}(\cdot) \right)\right)\\
			+Z\left(\frac{N(r^2+1)}{2}Z^{N+1}f\right)+\sum_{k=1}^{N+1}O_{\infty}(1)Z^k f+\sum_{k=1}^{N+1}O_{\infty}(1)Z^k\Phi f+\sum_{k=0}^{N}O_{\infty}((r-1)^{-2})L\underline{L}Z^kf
			\end{multline*}
			Subsequently, we apply \eqref{eq:commrelZ1} holds with $N$ replaced by $1\leq k\leq N$ to conclude that:
			\begin{multline*}
			ZZ^{N}\left(-\frac{(r^2+1)^2}{\Delta} L\Lbar f\right)=-\frac{(r^2+1)^2}{\Delta} L\Lbar Z^{N+1}f-\frac{N+1}{4}(1+O_{\infty}(r-1))\left(-\frac{r^2+1}{\Delta} L\Lbar Z^{N}f \right)\\
			+\frac{(N+1)(r^2+1)}{2}Z^{N+2}f+\sum_{k=1}^{N+1}O_{\infty}(1)Z^k f+\sum_{k=1}^{N+1}O_{\infty}(1)Z^k\Phi f+\sum_{k=0}^{N}O_{\infty}((r-1)^{-2})L\underline{L}Z^kf.
			\end{multline*}
			so \eqref{eq:commrelZ1} holds also with $N$ replaced by $N+1$, which concludes the induction argument.
			\end{proof}
			
			\begin{proposition}
			\label{prop:commZ}
			Let $Z\in \{K,\Phi, (r-M) Y \}$. Then for all $N\in \N_0$ and $0<p\leq 2$, there exists $r_H$ suitably small and a constant $C=C(r_H,\mathbbm{h},p,N)>0$, such that:
			\begin{equation}
			\label{eq:horminMestpre1}
			\begin{split}
		\sum_{0\leq k\leq N}&\int_{\Sigma_{\tau_2} \cap\{r\leq r_H\}} \mathcal{E}_{p}[Z^k\phi]\,d\sigma dr +	\int_{\tau_1}^{\tau_2}\int_{\Sigma_{\tau} \cap\{r\leq r_H\}} (r-M)^{p-1}\mathcal{E}_{0}[Z^N\phi]\,d\sigma dr d\tau \\
		\leq&\: C 	\sum_{0\leq k\leq N}\int_{\Sigma_{\tau_1} \cap\{r\leq r_H\}}\mathcal{E}_{p}[Z^k\phi]\,d\sigma dr + C\int_{0}^{\infty}\int_{\Sigma_{\tau} \cap\{r= r_H\}} |\snabla_{\s^2}Z^k\phi|^2+|KZ^k\phi|^2\,d\sigma dr d\tau\\
		+&\:C \sum_{0\leq k+l\leq N}\int_{\tau_1}^{\infty}\int_{\Sigma_{\tau} \cap\{r\leq r_H\}}(r-M)^{1-p}|\snabla_{\s^2}^lZ^kG|^2\,d\sigma drd\tau\\
		+&\:C\sum_{0\leq k_1+k_2\leq N+1}\int_{0}^{\infty}\int_{\Sigma_{\tau} \cap\{r\leq r_H\}} (r-M)^{1-p} |K^{k_1}\Phi^{k_2}\phi|^2\,d\sigma dr d\tau
		\end{split}
			\end{equation}
			or, alternatively,
			\begin{equation}
				\label{eq:horminMestpre2}
			\begin{split}
		\sum_{0\leq k\leq N}&\int_{\Sigma_{\tau_2} \cap\{r\leq r_H\}} \mathcal{E}_{p}[Z^k\phi]\,d\sigma dr +	\int_{\tau_1}^{\tau_2}\int_{\Sigma_{\tau} \cap\{r\leq r_H\}} (r-M)^{1-p}\mathcal{E}_{0}[Z^N\phi]\,d\sigma dr d\tau \\
		\leq&\: C 	\sum_{0\leq k\leq N}\int_{\Sigma_{\tau_1} \cap\{r\leq r_H\}}\mathcal{E}_{p}[Z^k\phi]\,d\sigma dr + C\int_{\tau_1}^{\infty}\int_{\Sigma_{\tau} \cap\{r= r_H\}} |\snabla_{\s^2}Z^k\phi|^2+|KZ^k\phi|^2\,d\sigma dr d\tau\\
		+&\:C \sum_{0\leq k+l\leq N}\int_{\tau_1}^{\infty}\int_{\Sigma_{\tau} \cap\{r\leq r_H\}}(r-M)^{1-p}|\snabla_{\s^2}^lZ^kG|^2\,d\sigma drd\tau\\
		+&\:C\sum_{0\leq k_1+k_2\leq N}\int_{0}^{\infty}\int_{\Sigma_{\tau} \cap\{r\leq r_H\}} (r-M)^{-1-p} |\underline{L}K^{k_1}\Phi^{k_2}\phi|^2+ (r-M)^{1-p} |K^{k_1+1}\Phi^{k_2}\phi|^2\,d\sigma dr d\tau.
		\end{split}
			\end{equation}
			
				\end{proposition}
			\begin{proof}
			Suppose $N=0$. Then the estimates follows directly by integrating \eqref{eq:keyidnearhor}. Furthermore, for any $k_1,k_1\in \N_0$ we can replace $\psi$ by $K^{k_1}\Phi^{k_2}\psi$ and apply \eqref{eq:keyidnearhor} again, since $\Phi$ and $K$ are Killing.
			Now let $Z=(r-M)^{-1}\underline{L}=\frac{r-M}{2(r^2+M^2)}Y$ and consider $Z^N\phi$. By acting with $Z^N$ on both sides of \eqref{eq:waveeqnull}, we obtain
			\begin{multline}
			\label{eq:commeqrmin1est}
Z^N((r^2+a^2)^{\frac{1}{2}}\rho^2G)=-4Z^N\left(\frac{(r^2+a^2)^2}{\Delta}L\Lbar \psi\right)+a^2\sin^2 \theta  K^2Z^N\psi+2a(1-a\upomega_+\sin^2\theta) K\Phi Z^N\psi\\
+(\slashed{\mathcal{D}}-2a\upomega_+ \Phi^2)Z^N\psi+\sum_{k\leq N}O(1)\Phi Z^k\psi+O((r-M))\sum_{k\leq N}Z^k\psi.
\end{multline}
For the sake of convenience, we now set $M=1$. By \eqref{eq:commrelZ1}, we can write:
\begin{multline}
-4Z^N\left(\frac{(r^2+1)^2}{\Delta}L\Lbar \psi\right)=-4\frac{(r^2+1)^2}{\Delta}L\Lbar Z^N\psi-N(1+O_{\infty}(r-1))\left(-\frac{r^2+1}{\Delta} L\Lbar Z^{N-1}\psi \right)\\
			+2N(r^2+1)(r-1)^{-1}\underline{L}Z^{N}\psi+\sum_{k=1}^NO_{\infty}(1)Z^k\psi+\sum_{k=1}^NO_{\infty}(1)Z^k\Phi \psi+\sum_{k=0}^{N-1}O_{\infty}((r-1)^{-2})L\underline{L}Z^k\psi.
\end{multline}
We then proceed inductively by assuming that the estimates \eqref{eq:horminMestpre1} and \eqref{eq:horminMestpre2} hold and then considering the case where $N$ is replaced by $N+1$. We then proceed as in the proof of Proposition \ref{prop:rweightesthor} by multiplying both sides of \eqref{eq:commeqrmin1est} by $-\frac{1}{2}(r-M)^{2-p}\overline{\Lbar Z^{N+1}\psi}$ and taking the real part. 

We use \eqref{eq:waveeqnull} and \eqref{eq:commeqrmin1est} to express $L\underline{L}Z^k\psi$ with $k\leq N$ in terms of $T$, $\underline{L}$ and angular derivatives.

The extra terms present when $N>0$ either result in a a good sign or can be straightforwardly estimated via Young's inequality or integration by parts in the $T$ or $\s^2$ directions together via an application of the induction assumption.
			\end{proof}
				
				In the corollary below, we combine the $(r-M)^{-1}$-weighted estimates from Proposition \ref{prop:commZ} with the integrated estimates from Theorem \ref{thm:morawetz}.
				
			\begin{corollary}
			\label{cor:horminMest}
		For all $N\in \N_0$ and $l\in \N_0$ and $\delta,\epsilon>0$, there exists $r_H$ suitably small and a constant $C=C(r_H,\mathbbm{h},\epsilon,\delta,N,l)>0$, such that for  $Z\in \{K,\Phi, (r-M) Y \}$:
			\begin{equation}
				\label{eq:horminMest}
			\begin{split}
		\sum_{0\leq k\leq N}&\int_{\Sigma_{\tau_2} \cap\{r\leq r_H\}} \mathcal{E}_{1-\epsilon}[\snabla_{\s^2}^lZ^k\phi]\,d\sigma dr +	\int_{\tau_1}^{\tau_2}\int_{\Sigma_{\tau} \cap\{r\leq r_H\}} (r-M)^{\epsilon}\mathcal{E}_{0}[\snabla_{\s^2}^lZ^k\phi]\\
		+&\:(r-M)^{-1+\epsilon}|\snabla_{\s^2}^lZ^kK\phi|^2\,d\sigma dr d\tau \\
		\leq&\: C 	\sum_{0\leq k\leq N}\int_{\Sigma_{\tau_1} \cap\{r\leq r_H\}} \mathcal{E}_{1-\epsilon}[\snabla_{\s^2}^lZ^k\phi]\,d\sigma dr+C\sum_{0\leq k_1+k_2\leq N}\int_{\Sigma_{\tau_1}} \mathcal{E}_{1+\epsilon}[K^{k_1} \Phi^{k_2}\phi]\,d\sigma dr\\
		+&\:C \sum_{0\leq k\leq N}\int_{\tau_1}^{\infty}\int_{\Sigma_{\tau} \cap\{r\leq r_H\}}(r-M)^{\epsilon}|\snabla_{\s^2}^lZ^kG|^2\,d\sigma drd\tau\\
		&\:+ C \sum_{0\leq k_1+k_2\leq N}\int_{\tau_1}^{\infty}\int_{\Sigma_{\tau}}[(1+\tau)^{1+\epsilon}(1-Mr^{-1})^{1-\epsilon}+r^2\chi_{\tau_1\leq \tau\leq \tau_2+1}+r^{1+\epsilon}]|K^{k_1} \Phi^{k_2}G|^2\,r^2d\sigma drd\tau.
		\end{split}
			\end{equation}
				\end{corollary}
			\begin{proof}
			Consider first the case $l=0$. Then the desired estimate follows by using \eqref{eq:mornearhor} together with \eqref{eq:horminMestpre1}. The case $l>0$ can be estimated by replacing $\psi$ with $(-1)^l\slashed{\Delta}_{\s^2}\psi$ and integrating an additional $l$ times by parts on $\s^2$. We omit the details of this additional step and refer to \cite[Lemma 2.5]{aagkerr} for inequalities on $\s^2$ required to obtain the necessary control of the angular derivatives.
				\end{proof}

				\subsection{Estimates near $\mathcal{I}^+$}
				\label{sec:rweightinf}
		In this section, we state the $r$-weighted estimates in regions $r>r_I$, with $r_I\gg M$ required for energy decay estimates. These do not depend on special geometric properties of extremal Kerr spacetimes, but are valid much more generally on asymptotically flat spacetimes. They involve ``standard'' $r$-weighted estimates, introduced in \cite{newmethod} as well as $r$-weighted estimates with an additional commutation $r^2L$, introduced in \cite{paper1}, which do not hold for general $\phi$, but they do hold for $\phi$ supported on azimuthal modes with $|m|\geq 1$ or $|m|\geq 2$. 
		
		In the context of rotating Kerr, with $G\equiv 0$,the estimates in the proposition below fall under the special case $\ell=0$ and $n\leq 1$ in \cite[Proposition 5.7]{aagkerr}. The generalization to $G\not \equiv0$ is straightforward.
			\begin{proposition}
			\label{prop:rpest}
			Let $0<p<2$ and $\epsilon>0$ suitably small. Then there exists a constant $C>0$ such that for all $0\leq \tau_1<\tau_2$:
			\begin{multline}
			\label{eq:rpinf}
			\int_{\Sigma_{\tau_2} \cap\{r\geq r_I\}} r^p | L\psi|^2\,d\sigma dr+\int_{\tau_1}^{\tau_2}\int_{\Sigma_{\tau} \cap\{r\geq r_I\}} pr^{p-1} | L\psi|^2+(2-p)r^{p-3}|\snabla_{\s^2}\psi|^2,d\sigma dr d\tau\\
			\leq  C\int_{\Sigma_{\tau_1} \cap\{r\geq r_I\}} r^p | L\psi|^2\,d\sigma dr+C\int_{\Sigma_{\tau_1}} \mathcal{E}_{1+\epsilon}[\Phi \phi]\,d\sigma dr+ C\int_{\tau_1}^{\infty}\int_{\Sigma_{\tau} \cap\{r \geq r_I\}} r^{p+1}|rG|^2\,d\sigma drd\tau\\
			+ C \int_{\tau_1}^{\infty}\int_{\Sigma_{\tau}} [(1+\tau)^{1+\epsilon}(1-Mr^{-1})^{1-\epsilon}+r^2\chi_{\tau_1\leq \tau\leq \tau_2+1}+r^{1+\epsilon}]| \Phi G|^2\,r^2d\sigma drd\tau
			\end{multline}
		and
		\begin{multline}
			\label{eq:rpinfcomm}
			\int_{\Sigma_{\tau_2} \cap\{r\geq r_I\}} r^p | L(r^2L)\psi|^2\,d\sigma dr+\int_{\tau_1}^{\tau_2}\int_{\Sigma_{\tau} \cap\{r\geq r_I\}} r^{p-1} | L(r^2L)\psi|^2+(2-p)r^{p-3}|\snabla_{\s^2}(r^2L)\psi|^2\,d\sigma dr d\tau \\
			\leq C\int_{\Sigma_{\tau_1} \cap\{r\geq r_I\}} r^p | L(r^2L)\psi|^2\,d\sigma dr+C\sum_{0\leq k_1+k_2\leq 1}\int_{\Sigma_{\tau_1}} \mathcal{E}_{1+\epsilon}[K^{k_1}\Phi^{k_2} \phi]\,d\sigma dr\\
			+ C\int_{\tau_1}^{\infty}\int_{\Sigma_{\tau} \cap\{r \geq r_I\}} r^{p+1}|L(r^3G)|^2\,d\sigma drd\tau\\
			+ C \sum_{0\leq k_1+k_2\leq 1}\int_{\tau_1}^{\infty}\int_{\Sigma_{\tau}}[(1+\tau)^{1+\epsilon}(1-Mr^{-1})^{1-\epsilon}+r^2\chi_{\tau_1\leq \tau\leq \tau_2+1}+r^{1+\epsilon}]|K^{k_1}\Phi^{k_2} G|^2\,r^2d\sigma drd\tau.
			\end{multline}
			\end{proposition}
		
		When we restrict to $\phi_{|m|\geq 2}$, we can obtain additional, higher-order estimates involving $(r^2L)^2\phi_{|m|\geq 2}$, which were also introduced in \cite{paper1}; see the $\ell=0$ and $n\leq 2$ case in \cite[Proposition 5.7]{aagkerr}.
		\begin{proposition}
			Let $0<p<2$ and $\epsilon>0$ suitably small. Then there exists a constant $C>0$ such that for all $0\leq \tau_1<\tau_2$:
					\begin{multline}
			\label{eq:rpinfcommmgeq2}
			\int_{\Sigma_{\tau_2} \cap\{r\geq r_I\}} r^p | L(r^2L)^2\psi_{|m|\geq 2}|^2\,d\sigma dr\\
			+\int_{\tau_1}^{\tau_2}\int_{\Sigma_{\tau} \cap\{r\geq r_I\}} r^{p-1} | L(r^2L)^2\psi_{|m|\geq 2}|^2+(2-p)r^{p-3}|\snabla_{\s^2}(r^2L)^2\psi_{|m|\geq 2}|^2\,d\sigma dr d\tau \\
			\leq C\int_{\Sigma_{\tau_1} \cap\{r\geq r_I\}} r^p | L(r^2L)^2\psi_{|m|\geq 2}|^2\,d\sigma dr+C\sum_{0\leq k_1+k_2\leq 2}\int_{\Sigma_{\tau_1}} \mathcal{E}_{1+\epsilon}[K^{k_1}\Phi^{k_2} \phi_{|m|\geq 2}]\,d\sigma dr\\
			+ C\int_{\tau_1}^{\infty}\int_{\Sigma_{\tau} \cap\{r \geq r_I\}} r^{p+1}|(L(r^2L(r^3G_{|m|\geq 2}))|^2\,d\sigma drd\tau\\
			+ C \sum_{0\leq k_1+k_2\leq 2}\int_{\tau_1}^{\infty}\int_{\Sigma_{\tau}}[(1+\tau)^{1+\epsilon}(1-Mr^{-1})^{1-\epsilon}+r^2\chi_{\tau_1\leq \tau\leq \tau_2+1}+r^{1+\epsilon}]|K^{k_1}\Phi^{k_2} G_{|m|\geq 2}|^2\,r^2d\sigma drd\tau.
			\end{multline}
			\end{proposition}

			\section{Energy decay for $K$-derivatives}
			\label{sec:edecay}
			In this section, we consider hierarchies of $(r-M)^{-1}$- and $r$-weighted energy estimates, as derived in \S \ref{sec:rmin1weight} and \S \ref{sec:rweightinf} and use them to derive decay-in-time estimates for (weighted) energies. Rather than considering $\phi$, we derive decay for $K\phi$ and $K^2\phi$. In \S \ref{sec:decay}, we will use these decay estimates in combination with the $K$-inversion theory from \S \ref{sec:ellipticKinv} to infer energy decay estimates at the level of $\phi$, without additional $K$-derivatives.
			
		The following lemma illustrates that $K$ derivatives can be ``exchanged'' for factors $(r-1)^2$.
				\begin{lemma}
		Let $M\leq r\leq r_H$. Then there exists a $C=C(r_H)>0$, such that
		\begin{multline}
		\label{eq:convKder}
		(r-M)^{-2}|\Lbar K\psi|^2\leq C \sum_{0\leq k\leq 1}|\Lbar Z^k\psi|^2\\
		+C(r-M)^2\sum_{0\leq k_1+k_2\leq 1}\left[ |\sD_{\s^2}\psi |^2+ |K^{k_1}\Phi^{1+k_2} \psi|^2+\sin^4 \theta |K^2\psi|^2+|G|^2\right]+(r-M)^3|\psi|^2.
		\end{multline}
		\end{lemma}
		We will employ the notation $\widetilde{Z}$ for differential operators in $\{ K,\Phi, (1-Mr^{-1})\Lbar, \snabla_{\s^2}\}$ and we set $M\equiv 1$ in the remainder of this section.
		
		We will first make the following \textbf{additional assumptions} on the inhomogeneity $G$: for all suitably small $\epsilon>0$ and $k,l\in \N$:
		\begin{align}
		\label{eq:decayassmG1}
		\sum_{k,l}&\int_{\tau}^{\infty}\int_{\Sigma_{\tau}\cap\{r\leq 2\}}[(1+\tau')^{1+\epsilon}(1-r^{-1})^{1-\epsilon}+1\\ \nonumber
		&+(1-r^{-1})^{-1+\epsilon}(1+\tau')^{-1}](1+\tau')^{2k}|K^k\tZ^lG|^2\,d\sigma dr d\tau<\infty,\\
		\label{eq:decayassmG2}
		\sum_{k,l}&\int_{\tau}^{\infty}\int_{\Sigma_{\tau}\cap\{r\geq 2\}}r^{2-\epsilon}(\tau'+r)(1+\tau')^{2k}|K^{k+1} (Xr^2)^lrG|^2\,d\sigma dr d\tau\\
		&+\sup_{\tau\in [0,\infty)}\int_{\Sigma_{\tau}\cap\{r\geq 2\}} |r^2G|^2\,d\sigma dr<\infty.
		\end{align}
		\begin{proposition}
		\label{prop:edecay}
		Let $\eta>0$ be arbitrarily small and assume \eqref{eq:decayassmG1} and \eqref{eq:decayassmG1}. Then there exist $\epsilon>0$ suitably small, a $r_I>1$ suitably small and a constant $C>0$, such that:
		\begin{align}
\label{eq:step2edecaylessdeg}
\int_{\Sigma_{\tau}} \mathcal{E}_{1+\epsilon}[K\phi]\, r^2d\sigma dr\leq&\:  C (1+\tau)^{-2+\eta}D_{\rm inhom}[\phi](0),\\
\label{eq:step3edecay}
\int_{\Sigma_{\tau}} \mathcal{E}_{1+\epsilon}[K^2\phi]\, r^2d\sigma dr\leq &\: C (1+\tau)^{-4+\eta}D_{{\rm inhom},1}[\phi](0),\\
\label{eq:step3edecayrsq}
\int_{\Sigma_{\tau} \cap\{r\geq r_I\}} r^{2} | LK\psi|^2\,d\sigma dr\leq &\: C (1+\tau)^{-2+\eta}D_{\rm inhom,2}[\phi](0).
\end{align}
with
\begin{align}
\label{eq:indatanormhom0}
D_{\rm hom}[\phi](\tau)=&\:\sum_{k+l\leq 2,l\leq 1}\int_{\Sigma_{\tau}} \mathcal{E}_{1+\epsilon}[K^{k+l+1}\phi]\, r^2d\sigma dr+\int_{\Sigma_{\tau} \cap\{r\leq r_H\}} \mathcal{E}_{1-\epsilon}[\tZ^lK^k\phi]\,d\sigma dr\\ \nonumber
+&\:\int_{\Sigma_{\tau} \cap\{r\geq r_I\}} r^{2-\epsilon} | LK\psi|^2\,d\sigma dr,\\ 
\label{eq:indatanorminhom0}
D_{\rm inhom}[\phi](\tau)=&\:D_{\rm hom}[\phi](\tau)\\ \nonumber
+&\:\sum_{k\leq 1}\sum_{l\leq 1} \int_{\tau}^{\infty}\int_{\Sigma_{\tau'}}[(1+\tau')^{1+\epsilon}(1-r^{-1})^{1-\epsilon}+r^2\chi_{\tau\leq \tau'\leq \tau+1}+r^{1+\epsilon}]r^2(1+\tau')^{2k}|K^{k+l}G|^2\\ \nonumber
+&\: r^{2-\epsilon}(\tau'r^{\epsilon}+r)|rKG|^2+(1-r^{-1})^{\epsilon}|\tZ^lK^kG|^2d\sigma drd\tau',\\
\label{eq:indatanormhom1}
D_{{\rm hom},1}[\phi](\tau)=&\:\sum_{k+l\leq 4,l\leq 2}\int_{\Sigma_{\tau}} \mathcal{E}_{1+\epsilon}[K^{k+l+1}\phi]\, r^2d\sigma dr+\int_{\Sigma_{\tau} \cap\{r\leq r_H\}} \mathcal{E}_{1-\epsilon}[\tZ^lK^k\phi]\,d\sigma dr\\ \nonumber
+&\:\sum_{k\leq 2}\int_{\Sigma_{\tau} \cap\{r\geq r_I\}}  r^{2-\delta} | LK^{k+1}\psi|^2+r^{2-\delta} | L((r^2L)K^2\psi)|^2\,d\sigma dr,\\
\label{eq:indatanorminhom1}
D_{{\rm inhom},1}[\phi](\tau)=&\:D_{{\rm hom},1}[\phi](\tau)\\ \nonumber
+&\:\sum_{k\leq 2}\sum_{l\leq 2}\int_{\tau}^{\infty}\int_{\Sigma_{\tau'}}[(1+\tau')^{1+\delta}(1-r^{-1})^{1-\epsilon}+r^2\chi_{\tau\leq \tau\leq \tau+1}+r^{1+\epsilon}]r^2(1+\tau')^{2k}|K^{k+l}G|^2\\ \nonumber
+&\: r^{2-\delta}(\tau'+r)(1+\tau)^{2k}|rK^kG|^2+r^{2-\delta}(r^{\delta}\tau'+r)|L(r^3K^2 G)|^2\\  \nonumber
+&\:(1-r^{-1})^{\epsilon}|\tZ^lK^kG|^2+(1-r^{-1})^{-1+\epsilon}(1+\tau')|\tZ^lKG|^2d\sigma drd\tau',\\
D_{{\rm hom},2}[\phi](\tau)=&\: D_{\rm hom}[\phi](\tau)+D_{\rm hom}[K\phi](\tau)+\int_{\Sigma_{\tau}\cap\{r\geq r_I\}}r^{2-\eta} | L(r^2LK\psi)|^2\,d\sigma dr,\\ 
D_{{\rm inhom},2}[\phi](\tau)=&\: D_{\rm inhom}[\phi]+D_{\rm inhom}[K\phi]+\int_{\Sigma_{0}} r^{2-\eta}  |L(r^2 LK\psi)|^2\,d\sigma dr\\ \nonumber
			&\:+\int_{\tau}^{\infty}\int_{\Sigma_{\tau'} \cap\{r \geq r_I\}} r^{3-\eta}(1+\tau)^2|rKG|^2+ r^{2-\eta}(r^{\eta}\tau+r)|L(r^3KG)|^2\,d\sigma drd\tau'.
\end{align}
		\end{proposition}
		\begin{proof}
		We will start by deriving decay estimates for the spacetime integral of $\mathcal{E}_{1+\epsilon}[K\phi]$. The proof proceeds in \textbf{four steps} and revolves around the repeated application of an interpolation inequality (Step 0).
		
		\paragraph{\textbf{Step 0}: an interpolation inequality}
		By using \eqref{eq:mornearhor} directly, we can only obtain a (boundedness) estimate for the energy corresponding to the energy density $\mathcal{E}_{1-\epsilon}[K\phi]$. We will use that the presence of a $K$-derivative allows us to gain factors of $r-1$ via \eqref{eq:convKder}. We apply the following interpolation inequality:
		let $0\leq \tau_A<\tau_B$ and $\gamma\geq 0$, then:
		\begin{multline*}
		\int_{\tau_A}^{\tau_B}\int_{\Sigma_{\tau}} \mathcal{E}_{1+\epsilon}[K\phi]\, r^2d\sigma dr=\int_{\tau_A}^{\tau_B}\int_{\Sigma_{\tau}\cap\{ (1+\tau_A)^{-\gamma}\leq r-1\leq r_I-1\}} \mathcal{E}_{1+\epsilon}[K\phi]\, r^2d\sigma dr\\
		+\int_{\tau_A}^{\tau_B}\int_{\Sigma_{\tau}\cap\{r-1\leq (1+\tau_A)^{-\gamma}\}} \mathcal{E}_{1+\epsilon}[K\phi]\, r^2d\sigma dr+\int_{\tau_A}^{\tau_B}\int_{\Sigma_{\tau}\cap\{ r\geq r_I\}} \mathcal{E}_{1+\epsilon}[K\phi]\, r^2d\sigma dr\\
		\lesssim 	(1+\tau_A)^{\gamma(1+2\epsilon)}\int_{\tau_A}^{\tau_B}\int_{\Sigma_{\tau}} (1-r^{-1})^{\epsilon}\mathcal{E}_{0}[K\phi]\, r^2d\sigma dr\\
		+\int_{\tau_A}^{\tau_B}\int_{\Sigma_{\tau}\cap\{r-1\leq (1+\tau_A)^{-\gamma}\}} (1-r^{-1})^{-3-\epsilon}|\Lbar K\phi|^2+|K^2\phi|^2+|\snabla_{\s^2}\phi|^2\, d\sigma dr d\tau\\
		+\int_{\tau_A}^{\tau_B}\int_{\Sigma_{\tau}\cap\{r\geq r_I\}} |LK\phi|^2+|K^2\phi|^2+|\snabla_{\s^2}\phi|^2\, d\sigma drd\tau\\
		\lesssim 	\underbrace{(1+\tau_A)^{\gamma(1+2\epsilon)}\int_{\tau_A}^{\tau_B}\int_{\Sigma_{\tau}\cap \{r\leq r_I\}} (1-r^{-1})^{\epsilon}\mathcal{E}_{0}[K\phi]\, r^2d\sigma dr}_{=:J_1}\\
		+\underbrace{(1+\tau_A)^{-\gamma(1-2\epsilon)}\int_{\tau_A}^{\tau_B}\int_{\Sigma_{\tau}\cap\{r\leq r_H\}} (r-1)^{-4+\epsilon}|\Lbar K\phi|^2\, d\sigma dr}_{=:J_{2,a}}\\
		+\underbrace{(1+\tau_A)^{-\gamma(1-2\epsilon)}\int_{\tau_A}^{\tau_B}\int_{\Sigma_{\tau}\cap\{r\leq r_H\}} (r-1)^{-1+\epsilon}(|K^2\phi|^2+|\snabla_{\s^2}K\phi|^2)\, d\sigma dr}_{=:J_{2,b}}\\
		+\underbrace{\int_{\tau_A}^{\tau_B}\int_{\Sigma_{\tau}\cap\{r\geq r_I\}} |LK\phi|^2+|K^2\phi|^2+|\snabla_{\s^2}\phi|^2\, d\sigma drd\tau}_{=:J_{3}},
		\end{multline*}
		where we used that $r-1\leq (1+\tau_A)^{-\gamma}$ implies that $r<r_H$, for $r_H$ suitably small and $\gamma\geq 0$.
		
				We start by estimating $J_{2,a}$. By applying \eqref{eq:convKder} and \eqref{eq:horminMest}, we can estimate
		\begin{multline}
		\label{eq:J2aest}
		J_{2,a}\lesssim (1+\tau_A)^{-\gamma(1-2\epsilon)}\sum_{k\leq 1} \int_{\tau_A}^{\tau_B}\int_{\Sigma_{\tau}\cap\{r\leq r_H\}}(r-1)^{\epsilon}\mathcal{E}_{0}[\tZ^k \phi] \,d\sigma dr \\
		 \leq C 	(1+\tau_A)^{-\gamma(1-2\epsilon)}\Bigg\{\sum_{k\leq 1}\int_{\Sigma_{\tau_A} \cap\{r\leq r_H\}} \mathcal{E}_{1-\epsilon}[\tZ^k\phi]\,d\sigma dr+\sum_{ k\leq 1}\int_{\Sigma_{\tau_A}} \mathcal{E}_{1+\epsilon}[K^{k}\phi]\,d\sigma dr\\
		+ \sum_{k\leq 1}\int_{\tau_A}^{\infty}\int_{\Sigma_{\tau} \cap\{r\leq r_H\}}(r-1)^{\epsilon}|\tZ^kG|^2\,d\sigma drd\tau\\
		+ \sum_{ k\leq 1}\int_{\tau_A}^{\infty}\int_{\Sigma_{\tau}}[(1+\tau)^{1+\epsilon}(1-r^{-1})^{1-\epsilon}+r^2\chi_{\tau_A\leq \tau\leq \tau_A+1}+r^{1+\epsilon}]|K^{k} G|^2\,r^2d\sigma drd\tau\Bigg\}\\
		\leq C (1+\tau_A)^{-\gamma(1-2\epsilon)} D_{2}[\phi](\tau_A),
		\end{multline}
		with
		\begin{multline*}
		 D_{2}[\phi](\tau_A):=\sum_{k\leq 1}\int_{\Sigma_{\tau_A} \cap\{r\leq r_H\}} \mathcal{E}_{1-\epsilon}[\tZ^k\phi]\,d\sigma dr+ \int_{\Sigma_{\tau_A}}\mathcal{E}_{1+\epsilon}[K^k\phi]\,r^2d\sigma dr\\
		 \sum_{k\leq 1}\int_{\Sigma_{\tau_A} \cap\{r\leq r_H\}} (r-1)^{\epsilon}|\tZ^kG|^2\,d\sigma dr\\
		 +\int_{\tau_A}^{\infty} \int_{\Sigma_{\tau}}[(1+\tau)^{1+\epsilon}(1-r^{-1})^{1-\epsilon}+r^2\chi_{\tau_A\leq \tau\leq \tau_A+1}+r^{1+\epsilon}]|K^{k} G|^2\,r^2d\sigma dr d\tau.
		\end{multline*}
	We estimate $J_{2,b}$ directly by applying \eqref{eq:horminMest}:
	\begin{equation*}
	\begin{split}
		J_{2,b} \leq &\:C(1+\tau_A)^{-\gamma(1-2\epsilon)} D_2[\phi](\tau_A).
		\end{split}
		\end{equation*}
		
		We can estimate $J_1$, by directly applying \eqref{eq:mornearhor} and \eqref{eq:morneartrapps}:
		\begin{equation*}
		\begin{split}
		J_1\lesssim &\: 	(1+\tau_A)^{\gamma(1+2\epsilon)}\int_{\Sigma_{\tau_A}}\mathcal{E}_{1+\epsilon}[K\phi]\,r^2d\sigma dr\\
		&\:+	(1+\tau_A)^{\gamma(1+2\epsilon)}\sum_{k\leq 1}\int_{\tau_A}^{\infty}\int_{\Sigma_{\tau}}[(1+\tau)^{1+\epsilon}(1-r^{-1})^{1-\epsilon}+r^2\chi_{\tau_A\leq \tau\leq \tau_A+1}+r^{1+\epsilon}]|KG|^2\,r^2d\sigma drd\tau \\
		\leq & \:(1+\tau_A)^{\gamma(1+2\epsilon)} \Bigg [ \sum_{k\leq 1}\int_{\Sigma_{\tau_A}}\mathcal{E}_{1+\epsilon}[K\phi]\,r^2d\sigma dr+C D_1[\phi](\tau_A)\Bigg],
		\end{split}
		\end{equation*}
		with
		\begin{equation*}
		D_1[\phi](\tau_A):=\int_{\tau_A}^{\infty}\int_{\Sigma_{\tau}}[(1+\tau)^{1+\epsilon}(1-r^{-1})^{1-\epsilon}+r^2\chi_{\tau_A\leq \tau\leq \tau_A+1}+r^{1+\epsilon}]|KG|^2\,d\sigma dr d\tau.		\end{equation*}
		Finally, we estimate $J_3$. By \eqref{eq:rpinf} with $p=1$, we obtain:
		\begin{multline*}
		J_{3}\lesssim \int_{\Sigma_{\tau_A}} \mathcal{E}_{1-\epsilon}[K\phi]\,d\sigma dr+C\int_{\Sigma_{\tau_1} \cap\{r\geq r_I\}} r | LK\psi|^2\,d\sigma dr+ C\int_{\tau_A}^{\infty}\int_{\Sigma_{\tau} \cap\{r \geq r_I\}} r^{2}|rKG|^2\,d\sigma drd\tau\\
			+ C \int_{\tau_A}^{\infty}\int_{\Sigma_{\tau}}[(1+\tau)^{1+\epsilon}(1-r^{-1})^{1-\epsilon}+r^2\chi_{\tau_A\leq \tau\leq \tau_A+1}+r^{1+\epsilon}]| K G|^2\,r^2d\sigma drd\tau\\
			=: C\int_{\Sigma_{\tau_A}} \mathcal{E}_{1-\epsilon}[K\phi]\,d\sigma dr+C\int_{\Sigma_{\tau_A} \cap\{r\geq r_I\}} r | LK\psi|^2\,d\sigma dr+CD_1[\phi](\tau_A)+CD_3[\phi](\tau_A),
		\end{multline*}
		with
			\begin{equation*}
		D_3[\phi](\tau_A):=C\int_{\tau_A}^{\infty}\int_{\Sigma_{\tau} \cap\{r \geq r_I\}} r^{2}|rKG|^2\,d\sigma drd\tau.
				\end{equation*}

Putting together the estimates for $J_1$, $J_{2,a}$ and $J_{2,b}$, we are left with:
\begin{multline}
\label{eq:mainrweightinterpolineq}
\int_{\tau_A}^{\tau_B}\int_{\Sigma_{\tau}} \mathcal{E}_{1+\epsilon}[K\phi]\, r^2d\sigma dr\leq (1+\tau_A)^{\gamma(1+2\epsilon)} \int_{\Sigma_{\tau_A}}\mathcal{E}_{1+\epsilon}[K\phi]\,r^2d\sigma dr+ C\int_{\Sigma_{\tau_A} \cap\{r\geq r_I\}} r | LK\psi|^2\,d\sigma dr\\
+C(1+\tau_A)^{\gamma(1+2\epsilon)}D_1[\phi](\tau_A)+C(1+\tau_A)^{-\gamma(1-2\epsilon)} D_2[\phi](\tau_A)+CD_3[\phi](\tau_A).
\end{multline}
The next two steps of the proof consist of starting with $\gamma=0$ and then taking $\gamma$ larger and larger at each step.\\
\\
\paragraph{\textbf{Step 1}: $\gamma=0$}
Let $\{\tau_i^{(0)}\}$ be a dyadic sequence. We first apply \eqref{eq:mainrweightinterpolineq} with $\gamma=0$ and $\tau_A=\tau_i^{(0)}$, $\tau_B=\tau_{i+1}^{(0)}$, together with the mean value theorem to obtain another dyadic sequence $\{\tau_i^{(1)}\}$ with $\tau^{(0)}_{2i}\leq \tau^{(1)}_i\leq \tau^{(0)}_{2i+1}$
\begin{multline}
\label{eq:mvtstep1}
\int_{\Sigma_{\tau_{i}^{(1)}}} \mathcal{E}_{1+\epsilon}[K\phi]\, r^2d\sigma dr\leq C \sum_{k\leq 1}(\tau_i^{(1)})^{-1}\Bigg[\int_{\Sigma_{\tau_{2i}^{(0)}}} \mathcal{E}_{1+\epsilon}[K^{k+1}\phi]\, r^2d\sigma dr+C\int_{\Sigma_{\tau_{2i}^{(0)}} \cap\{r\geq r_I\}} r | LK\psi|^2\,d\sigma dr\\
+D_1[\phi](\tau_{2i}^{(0)})+D_2[\phi](\tau_{2i}^{(0)})+D_3[\phi](\tau_{2i}^{(0)})\Bigg].
\end{multline}
To conclude that the integral of $\mathcal{E}_{1+\epsilon}[K\phi]$ decay along a dyadic sequence, we need to estimate the RHS by an an analogous expression evaluated at $\tau=0$.

Note that for all $\tau_A\geq 0$
\begin{equation*}
\int_{\tau_A}^{\tau_B}\int_{\Sigma_{\tau}} \mathcal{E}_{1+\epsilon}[K\phi]\, r^2d\sigma dr\leq \int_{0}^{\tau_B}\int_{\Sigma_{\tau}} \mathcal{E}_{1+\epsilon}[K\phi]\, r^2d\sigma dr,
\end{equation*}
so we can alternatively replace $\tau=\tau_{2i}^{(0)}$ on the RHS of \eqref{eq:mvtstep1} with $\tau=0$ to obtain
\begin{multline}
\label{eq:step1edecay}
\int_{\Sigma_{\tau_{i}^{(1)}}} \mathcal{E}_{1+\epsilon}[K\phi]\, r^2d\sigma dr\leq C (\tau_i^{(1)})^{-1}\sum_{k\leq 1}\Bigg[\int_{\Sigma_{0}} \mathcal{E}_{1+\epsilon}[K^{k+1}\phi]\, r^2d\sigma dr+C\int_{\Sigma_{0} \cap\{r\geq r_I\}} r | LK\psi|^2\,d\sigma dr\\
+D_1[\phi](0)+D_2[\phi](0)+D_3[\phi](0)\Bigg].
\end{multline}

\paragraph{\textbf{Step 2}: $\gamma=1-\frac{1}{2^l}$}
By \eqref{eq:rpinf} with $p=2-\delta$ together with the mean-value theorem and a standard interpolation argument (to interchange powers of $\tau$ and $r$), we can ensure that along the dyadic sequence $\{\tau_i^{(1)}\}$ from Step 1:
\begin{multline}
\label{eq:rweightestinfKpsi}
\int_{\Sigma_{\tau_i^{(1)}} \cap\{r\geq r_I\}} r | LK\psi|^2\,d\sigma dr\leq C (\tau_i^{(1)})^{-1+\delta}\Bigg[\int_{\Sigma_{0}} \mathcal{E}_{1+\epsilon}[K\phi]\, r^2d\sigma dr+\int_{\Sigma_{0} \cap\{r\geq r_I\}} r^{2-\delta} | LK\psi|^2\,d\sigma dr\\
+\int_{0}^{\infty}\int_{\Sigma_{\tau} \cap\{r \geq r_I\}} r^{3-\delta}|rKG|^2\,d\sigma drd\tau\Bigg].
\end{multline}
To extend the above estimates for all times $\tau$, we consider $\tau_i^{(1)}\leq \tau \leq \tau_{i+1}^{(1)}$ and apply \eqref{eq:rpinf} with $p=1$ to obtain:
\begin{equation}
\label{eq:rweightestinfKpsi2}
\begin{split}
\int_{\Sigma_{\tau} \cap\{r\geq r_I\}} r | LK\psi|^2\,d\sigma dr\leq &\: C \int_{\Sigma_{\tau_i^{(1)}} \cap\{r\geq r_I\}} r | LK\psi|^2\,d\sigma dr+\int_{\Sigma_{\tau_i^{(1)}}} \mathcal{E}_{1+\epsilon}[K\phi]\, r^2d\sigma dr\\
&\:+C\int_{\tau_i^{(1)}}^{\tau} \int_{\Sigma_{\tau} \cap\{r \geq r_I\}} r^{2}|rKG|^2\,d\sigma drd\tau\\
\substack{\eqref{eq:rweightestinfKpsi} \& \eqref{eq:step1edecay} \\ \leq }&\: C(1+\tau)^{-1+\delta}\sum_{k\leq 1}\Bigg[\int_{\Sigma_{0}} \mathcal{E}_{1+\epsilon}[K^{k+1}\phi]\, r^2d\sigma dr+C\int_{\Sigma_{0} \cap\{r\geq r_I\}} r^2 | LK\psi|^2\,d\sigma dr\\
&\:+D_1[\phi](0)+D_2[\phi](0)+D_3[\phi](0)\\
&\:+\int_{0}^{\infty}\int_{\Sigma_{\tau} \cap\{r \geq r_I\}} r^{2-\delta}(r^{\delta}\tau+r)|rKG|^2\,d\sigma drd\tau\Bigg].
\end{split}
\end{equation}

 Consider the dyadic sequence $\{\tau_i^{(1)}\}$ from Step 1. We apply \eqref{eq:rweightestinfKpsi2} in combination with \eqref{eq:mainrweightinterpolineq} with $\gamma=\frac{1}{2}$ and $\tau_A=\tau_i^{(1)}$, $\tau_B=\tau_{i+1}^{(1)}$, together with the mean value theorem to obtain another dyadic sequence $\{\tau_i^{(2)}\}$ with $\tau^{(1)}_{2i}\leq \tau^{(2)}_i\leq \tau^{(2)}_{2i+1}$, such that
 \begin{multline*}
\int_{\Sigma_{\tau_{i}^{(2)}}} \mathcal{E}_{1+\epsilon}[K\phi]\, r^2d\sigma dr\leq C \sum_{k\leq 1}(\tau_i^{(2)})^{-\frac{1}{2}(1+2\epsilon)}\Bigg[\int_{\Sigma_{\tau_{2i}^{(1)}}} \mathcal{E}_{1+\epsilon}[K\phi]\, r^2d\sigma dr+D_1[\phi](\tau_{2i}^{(1)})\\
+C(\tau_i^{(2)})^{-2+\frac{1}{2}(1+2\epsilon)}\int_{\Sigma_{0} \cap\{r\geq r_I\}} r^2 | LK\psi|^2\,d\sigma dr\\
+C(\tau_i^{(2)})^{-2+\frac{1}{2}(1+2\epsilon)}\int_0^{\infty}\int_{\Sigma_{\tau} \cap\{r \geq r_I\}} r^{2-\delta}(r^{\delta}\tau+r)|rKG|^2\,d\sigma drd\tau\\
+(D_2[\phi](\tau_{2i}^{(1)})+D_3[\phi](\tau_{2i}^{(1)})).
\end{multline*}
Proceeding as in Step 1, but appealing additionally to \eqref{eq:step1edecay} and using that
\begin{equation*}
D_1[\phi](\tau_A)\leq (1+\tau_A)^{-2}\int_{\tau_A}^{\infty}\int_{\Sigma_{\tau}}[(1+\tau)^{1+\epsilon}(1-r^{-1})^{1-\epsilon}+r^2\chi_{\tau_A\leq \tau\leq \tau_A+1}+r^{1+\epsilon}](1+\tau)^{2}|KG|^2\,d\sigma dr d\tau
\end{equation*}
we conclude that:
\begin{multline}
\label{eq:step2edecay}
\int_{\Sigma_{\tau_{i}^{(2)}}} \mathcal{E}_{1+\epsilon}[K\phi]\, r^2d\sigma dr\leq C (\tau_i^{(2)})^{-\frac{3+2\epsilon}{2}}\sum_{k\leq 1}\Bigg[\int_{\Sigma_{0}} \mathcal{E}_{1+\epsilon}[K^{k+1}\phi]\, r^2d\sigma dr+C\int_{\Sigma_{0} \cap\{r\geq r_I\}} r^2 | LK\psi|^2\,d\sigma dr\\
\int_{0}^{\infty}\int_{\Sigma_{\tau} \cap\{r \geq r_I\}} r^{2}(\tau+r)|rKG|^2\,d\sigma drd\tau\\
+\int_{\tau_A}^{\infty}\int_{\Sigma_{\tau}}[(1+\tau)^{1+\epsilon}(1-r^{-1})^{1-\epsilon}+r^2\chi_{\tau_A\leq \tau\leq \tau_A+1}+r^{1+\epsilon}](1+\tau)^{2}|K^{k+1}G|^2\,d\sigma dr d\tau\\
+D_2[\phi](0)+D_3[\phi](0)\Bigg].
\end{multline}
We can repeat this procedure by taking $\gamma=1-2^{-l}$ with $l\in \N$, $l\geq 2$. Given $\eta>0$ arbitrarily small, we can find a $L$ sufficiently large such that repeating the procedure until $l=L$ gives:
\begin{multline}
\label{eq:step2edecaygen}
\int_{\Sigma_{\tau_{i}^{(L+1)}}} \mathcal{E}_{1+\epsilon}[K\phi]\, r^2d\sigma dr\leq C (\tau_i^{(L+1)})^{-2+\eta+\epsilon}\sum_{k\leq 1}\Bigg[\int_{\Sigma_{0}} \mathcal{E}_{1+\epsilon}[K^{k+1}\phi]\, r^2d\sigma dr\\
+\int_{\Sigma_{0} \cap\{r\geq r_I\}} r^2 | LK\psi|^2\,d\sigma dr+\int_{0}^{\infty}\int_{\Sigma_{\tau} \cap\{r \geq r_I\}} r^{2-\delta}(r^{\delta}\tau+r)|rKG|^2\,d\sigma drd\tau\\
+\int_{\tau_A}^{\infty}\int_{\Sigma_{\tau}}[(1+\tau)^{1+\epsilon}(1-r^{-1})^{1-\epsilon}+r^2\chi_{\tau_A\leq \tau\leq \tau_A+1}+r^{1+\epsilon}](1+\tau)^{2k}|K^{k}G|^2\,d\sigma dr d\tau\\
+D_2[\phi](0)+D_3[\phi](0)\Bigg].
\end{multline}
By applying \eqref{eq:mornearhor}, we can extend the estimate \eqref{eq:step2edecay} to all $\tau$ at the expense of additional degenerate factors. We obtain:
\begin{multline}
\label{eq:step2edecaydeg}
\int_{\Sigma_{\tau}} \mathcal{E}_{1-\epsilon}[K\phi]\, r^2d\sigma dr\leq C (1+\tau)^{-2+\eta+\epsilon}\sum_{k\leq 1}\Bigg[\int_{\Sigma_{0}} \mathcal{E}_{1+\epsilon}[K^{k+1}\phi]\, r^2d\sigma dr+\int_{\Sigma_{0} \cap\{r\geq r_I\}} r^2 | LK\psi|^2\,d\sigma dr\\
\int_{0}^{\infty}\int_{\Sigma_{\tau} \cap\{r \geq r_I\}} r^{2-\delta}(r^{\delta}\tau+r)|rKG|^2\,d\sigma drd\tau\\
+\int_{\tau_A}^{\infty}\int_{\Sigma_{\tau}}[(1+\tau)^{1+\delta}(1-r^{-1})^{1-\epsilon}+r^2\chi_{\tau_A\leq \tau\leq \tau_A+1}+r^{1+\epsilon}](1+\tau)^{2}|K^{k+1}G|^2\,d\sigma dr d\tau+D_2[\phi](0)+D_3[\phi](0)\Bigg].
\end{multline}			
To obtain decay for the integral of $\mathcal{E}_{1+\epsilon}[K\phi]$ at all times $\tau$, we consider $\tau_{i}^{(L+1)}\leq \tau \leq \tau_{i+1}^{(L+1)}$.

We interpolate again to obtain:
\begin{equation*}
\begin{split}
\int_{\Sigma_{\tau}} \mathcal{E}_{1+\epsilon}[K\phi]\, r^2d\sigma dr=&\: \int_{\Sigma_{\tau}\cap\{r-1\geq (1+\tau)^{-1}\}} \mathcal{E}_{1+\epsilon}[K\phi]\, r^2d\sigma dr+ \int_{\Sigma_{\tau}\cap\{r-1\leq (1+\tau)^{-1}\}} \mathcal{E}_{1+\epsilon}[K\phi]\, r^2d\sigma dr\\
\leq &\: (1+\tau)^{2\epsilon} \int_{\Sigma_{\tau}} \mathcal{E}_{1-\epsilon}[K\phi]\, r^2d\sigma dr\\
&\:+(1+\tau)^{-2+\epsilon}\int_{\Sigma_{\tau}\cap\{r-1\leq (1+\tau)^{-1}\}} (r-1)^{-5+\epsilon}|\Lbar K\phi|^2\,d\sigma dr
\end{split}
\end{equation*}
The first term on the very RHS we can estimate using \eqref{eq:step2edecaydeg}. The second term on the very RHS we can estimate analogously to $J_{2,a}$ (which is a spacetime integral). We apply \eqref{eq:convKder} and \eqref{eq:horminMest} to obtain:
\begin{equation*}
(1+\tau)^{-2+\epsilon}\int_{\Sigma_{\tau}\cap\{r-1\geq \tau\}} (r-1)^{-5+\epsilon}|\Lbar K\phi|^2\,d\sigma dr\leq C(1+\tau)^{-2+\epsilon}D_2[\phi](0).
\end{equation*}
We conclude \eqref{eq:step2edecaylessdeg}.

In the next step, we will show that we can obtain faster decay by considering $K^2\phi$ instead of $K\phi$.

\paragraph{\textbf{Step 3}: Stronger decay for $\mathcal{E}_{1+\epsilon}[K^2\phi]$.}
To obtain stronger decay for $\mathcal{E}_{1+\epsilon}[K^2\phi]$, we need to establish appropriate decay of the terms $D_1[K\phi]$, $D_2[K\phi]$ and $D_3[K\phi]$ on the right-hand side of \eqref{eq:mvtstep1} and improve the decay when $\phi$ is replaced by $K\phi$. We also need to improve the decay rate of $\int_{\Sigma_{\tau_{2i}^{(0)}} \cap\{r\geq r_I\}} r | LK^2\psi|^2\,d\sigma dr$ by two powers.

Note first of all that
\begin{equation*}
D_1[K\phi](\tau_A)\leq (1+\tau_A)^{-4} \int_{\tau_A}^{\infty}\int_{\Sigma_{\tau}}[(1+\tau)^{1+\delta}(1-r^{-1})^{1-\epsilon}+r^2\chi_{\tau_A\leq \tau\leq \tau_A+1}+r^{1+\epsilon}](1+\tau)^{4}|K^2G|^2\,d\sigma dr d\tau.
\end{equation*}
Furthermore,
\begin{equation*}
		D_3[\phi](\tau_A)\leq C(1+\tau_A)^{-2}\int_{\tau_A}^{\infty}\int_{\Sigma_{\tau} \cap\{r \geq r_I\}} r^{2}(1+\tau)^2|rK^2G|^2\,d\sigma drd\tau.
				\end{equation*}
To obtain sufficient decay for $D_2[K\phi](\tau_A)$ and the integral of $r | LK^2\psi|^2$, we will establish in particular decay for
\begin{equation*}
\int_{\Sigma_{\tau} \cap\{r\leq r_H\}} \mathcal{E}_{1-\epsilon}[\tZ K \phi]\,d\sigma dr.
\end{equation*}

We apply \eqref{eq:horminMestpre2} with $N=1$ and $p=2-\epsilon$, together with \eqref{eq:rpinfcomm} with $p=1$ and \eqref{eq:mornearho} in order to obtain:
\begin{multline}
\int_{\tau_A}^{\tau_B}\int_{\Sigma_{\tau} \cap\{r\leq r_H\}} \mathcal{E}_{1-\epsilon}[\tZ K\phi]\,d\sigma dr\,d\tau\leq C\sum_{k\leq 1} \int_{\Sigma_{\tau_A} \cap\{r\leq r_H\}} \mathcal{E}_{2-\epsilon}[\tZ^kK\phi]\,d\sigma dr\\
+C\int_{\tau_A}^{\infty}\int_{\Sigma_{\tau} \cap\{r\leq r_H\}}(r-1)^{-1+\epsilon}|\widetilde{Z}^kKG|^2+r^2[(1+\tau)^{1+\delta}(1-r^{-1})^{1-\epsilon}+1]|K^{k+1} G|^2\,d\sigma drd\tau\\
		+C\int_{\tau_A}^{\infty}\int_{\Sigma_{\tau} \cap\{r\leq r_H\}}  \mathcal{E}_{1-\epsilon}[K^{1+k}\phi]\,d\sigma dr d\tau.
\end{multline}

We now apply the mean-value theorem along a dyadic sequence, together with \eqref{eq:horminMestpre2} with $p=1-\epsilon$ and \eqref{eq:mornearho}, to conclude that for all $\tau\geq 0$:
\begin{multline}
\int_{\Sigma_{\tau} \cap\{r\leq r_H\}} \mathcal{E}_{1-\epsilon}[\tZ K\phi]\,d\sigma dr\\
\leq C(1+\tau)^{-1} \sum_{k\leq 1}\Bigg[ \int_{\Sigma_{0} \cap\{r\leq r_H\}} \mathcal{E}_{2-\epsilon}[\tZ^kK\phi]\,d\sigma dr+(1+\tau)^{-1}\sup_{\tau'\geq \tau}\int_{\Sigma_{\tau}}(1+\tau')\mathcal{E}_{1+\epsilon}[K^{k+1}\phi]\,r^2d\sigma dr\\
+(1+\tau)^{-1}\int_{0}^{\infty}\int_{\Sigma_{\tau} \cap\{r\leq r_H\}}(r-1)^{-1+\epsilon}(1+\tau)|\widetilde{Z}KG|^2+r^2[(1+\tau)^{1+\delta}(1-r^{-1})^{1-\epsilon}\\
+r^2\chi_{0\leq \tau\leq 1}+r^{1+\epsilon}](1+\tau)|K^{k+1} G|^2\,d\sigma drd\tau\\
+C(1+\tau)^{-\delta}\sup_{\tau'\geq \tau}(1+\tau')^{\delta+1}\int_{\Sigma_{\tau'} \cap\{r\leq r_H\}}  \mathcal{E}_{1-\epsilon}[K^{1+k}\phi]\,d\sigma dr d\tau.\Bigg].
\end{multline}
In order to improve the decay above from $(1+\tau)^{-1}$ to $(1+\tau)^{-2+\eta}$, we estimate further the spacetime integral of  $\mathcal{E}_{2-\epsilon}[\tZ K\phi]$. We first use \eqref{eq:convKder} to estimate
\begin{equation*}
\begin{split}
\int_{\Sigma_{\tau} \cap\{r\leq r_H\}} \mathcal{E}_{2-\epsilon}[\tZ^kK\phi]\,d\sigma dr=&\: \sum_{k\leq 1}\int_{\Sigma_{\tau} \cap\{r\leq r_H\}} \mathcal{E}_{1-\epsilon}[\tZ^kK\phi]+(r-1)^{-4+\epsilon}|\Lbar \tZ^kK\phi|\,d\sigma dr\\
\leq &\: \sum_{k\leq 1}\int_{\Sigma_{\tau} \cap\{r\leq r_H\}} \mathcal{E}_{1-\epsilon}[\tZ^kK\phi]\,d\sigma dr\\
&\:+ C\sum_{k\leq 2}\int_{\Sigma_{\tau} \cap\{r\leq r_H\}} (r-1)^{\epsilon}\mathcal{E}_{0}[\tZ^k\phi]\,d\sigma dr.
\end{split}
\end{equation*}
Note that we have already established $\tau^{-1}$ decay for the first term on the very right-hand side and the second term can also be shown to decay like $\tau^{-1}$ by applying \eqref{eq:horminMest} with $N=2$ and $p=1-\epsilon$. We then conclude that:
\begin{multline}
\int_{\Sigma_{\tau} \cap\{r\leq r_H\}} \mathcal{E}_{1-\epsilon}[\tZ K\phi]\,d\sigma dr\\
\leq C(1+\tau)^{-1} \sum_{k,l,m\leq 1}\Bigg[(1+\tau)^{-1} \int_{\Sigma_{0} \cap\{r\leq r_H\}} \mathcal{E}_{1-\epsilon}[\tZ^{k+l}\phi]\,d\sigma dr+(1+\tau)^{-1} \int_{\Sigma_{0} \cap\{r\leq r_H\}} \mathcal{E}_{1+\epsilon}[K^{k+l}\phi]\,d\sigma dr\\
+(1+\tau)^{-1}\sup_{\tau'\geq \tau}\int_{\Sigma_{\tau}}(1+\tau')\mathcal{E}_{1+\epsilon}[K^{k+1}\phi]\,r^2d\sigma dr\\
+(1+\tau)^{-1}\int_{0}^{\infty}\int_{\Sigma_{\tau} \cap\{r\leq r_H\}}(r-1)^{\epsilon}(1+\tau)^k|\widetilde{Z}^{l+m}K^kG|^2+(r-1)^{-1+\epsilon}(1+\tau)|\widetilde{Z}KG|^2\\
+r^2[(1+\tau)^{1+\delta}(1-r^{-1})^{1-\epsilon}+r^2\chi_{0\leq \tau\leq 1}+r^{1+\epsilon}](1+\tau)^k|K^{k+l} G|^2\,d\sigma drd\tau\\
+C(1+\tau)^{-\delta}\sup_{\tau'\geq \tau}(1+\tau')^{\delta+1}\int_{\Sigma_{\tau'} \cap\{r\leq r_H\}}  \mathcal{E}_{1-\epsilon}[K^{1+k}\phi]\,d\sigma dr d\tau\Bigg].
\end{multline}
From the above estimate, we obtain:
\begin{multline}
\label{eq:decayD2}
D_2[K\phi](\tau)\leq C(1+\tau)^{-2+\eta} \sum_{k\leq 1}D_{\rm hom}[K^{k+1}\phi]+\sum_{k+l\leq 4,l\leq 2}\int_{0}^{\infty}\int_{\Sigma_{\tau}}(1-r^{-1})^{\epsilon}(1+\tau)^{2k}|\widetilde{Z}^{l}K^kG|^2 \\
+(1-r^{-1})^{-1+\epsilon}(1+\tau)|\widetilde{Z}KG|^2+r^2[(1+\tau)^{1+\delta}(1-r^{-1})^{1-\epsilon}+r^2\chi_{0\leq \tau\leq 1}+r^{1+\epsilon}](1+\tau)^{2k}|K^{k} G|^2\\
+r^{2-\delta}(r^{\delta}\tau+r)(1+\tau)^{2l}|rK^{k}G|^2\,d\sigma drd\tau.
\end{multline}
We are left with improving the decay of the integral of $r|LK^2\psi|^2$. To improve the decay rate in \eqref{eq:rweightestinfKpsi2}, we need to establish decay of the integral of $r^2|LK^2\psi|^2$.

Note first of all that by an application of a Hardy inequality and an averaging of the boundary terms, we obtain
\begin{equation*}
\int_{\Sigma_{\tau} \cap\{r\geq r_I\}} r^2 | LK^2\psi|^2\,d\sigma dr\leq C\int_{\Sigma_{\tau} \cap\{r_I-M\leq r\leq r_I\}}  |LK^2\psi|^2\,d\sigma dr+C\int_{\Sigma_{\tau} \cap\{r\geq r_I\}}  |L(r^2 LK^2\psi)|^2\,d\sigma dr.
\end{equation*}
Now we apply \eqref{eq:rpinfcomm} with $p=1$ and then $p=2-\eta$ to $K^2\phi$, together with a standard interpolation argument. We obtain:
\begin{multline*}
\int_{\Sigma_{\tau} \cap\{r\geq r_I\}} r | LK^2\psi|^2\,d\sigma dr\leq C(1+\tau)^{-3+\eta}\Bigg[\int_{\Sigma_{\tau} \cap\{r\geq r_I\}} r^{2-\delta}|L(r^2 LK^2\psi)|^2\,d\sigma dr\\
+C\int_{0}^{\infty}\int_{\Sigma_{\tau}\cap\{r\geq r_I\}} r^{2-\delta}(r^{\delta}\tau+r)|r^2L(rK^2 G)|^2\,d\sigma drd\tau+ \ldots\Bigg],
\end{multline*}
with $\ldots$ denoting the terms that are already contained on the right-hand side of \eqref{eq:decayD2}.

To improve the decay rate for $\mathcal{E}_{1+\epsilon}[K^2\phi]$, it is also necessary to obtain better decay for the term $J_3$ when $K\phi$ is replaced by $K^2\phi$.

We can improve the decay rate for $\mathcal{E}_{1+\epsilon}[K^2\phi]$ by repeating the estimates in Step 1 and Step 3, taking $\gamma=1-2^{-l}$ and starting from the already established $\tau^{-2+\eta}$ decay for $\mathcal{E}_{1+\epsilon}[K^2\phi]$ and the decay estimates for the remaining terms on the RHS of  \eqref{eq:mvtstep1} that follow from Step 3.

Via a similar argument to the one above for improving decay for $\int_{\Sigma_{\tau} \cap\{r\geq r_I\}} r | LK^2\psi|^2\,d\sigma dr$, we can also obtain decay with a rate $(1+\tau)^{-2+\eta}$ for $\int_{\Sigma_{\tau} \cap\{r\geq r_I\}} r^2 | LK^2\psi|^2\,d\sigma dr$. We first observe, as above, that
\begin{equation*}
\begin{split}
\int_{\tau_A}^{\tau_B}&\int_{\Sigma_{\tau} \cap\{r\geq r_I\}} r^2 | LK\psi|^2\,d\sigma dr\leq  C\int_{\tau_A}^{\tau_B}\int_{\Sigma_{\tau} \cap\{r_I-M\leq r\leq r_I\}}  |LK\psi|^2\,d\sigma dr\\
&\:+C\int_{\tau_A}^{\tau_B}\int_{\Sigma_{\tau} \cap\{r\geq r_I\}}  |L(r^2 LK\psi)|^2\,d\sigma dr\\
\substack{ \eqref{eq:mornearho} \& \eqref{eq:rpinfcomm}  \\ \leq }&\: C(\mathcal{E}_{1+\epsilon}[K\phi](\tau_A)+\mathcal{E}_{1+\epsilon}[K^2\phi](\tau_A))+C\int_{\Sigma_{\tau} \cap\{r\geq r_I\}}  r |L(r^2 LK\psi)|^2\,d\sigma dr\\
&\:+C\int_{\tau_A}^{\infty}\int_{\Sigma_{\tau} \cap\{r \geq r_I\}} r^{2}|L(r^3KG)|^2\,d\sigma drd\tau\\
			&\:+ C \sum_{k\leq 1}\int_{\tau_A}^{\infty}\int_{\Sigma_{\tau}}[(1+\tau)^{1+\delta}(1-r^{-1})^{1+\epsilon}+r^2\chi_{0\leq \tau\leq 1}+r^{1+\epsilon}]|K^{k+1}G|^2\,r^2d\sigma drd\tau
\end{split}
\end{equation*}
The only term for which we cannot immediately conclude $(1+\tau)^{-1+\eta}$ decay from the above estimates is $\int_{\Sigma_{\tau} \cap\{r\geq r_I\}}  r |L(r^2 LK\psi)|^2\,d\sigma dr$. To estimate this term, we apply \eqref{eq:rpinfcomm} with $N=1$ and $p=2-\eta$. We obtain
\begin{equation*}
\begin{split}
\int_{\Sigma_{\tau_B\cap\{r\geq r_I\}}}& r^{2-\eta}  |L(r^2 LK\psi)|^2\,d\sigma dr+ \int_{\tau_A}^{\tau_B}\int_{\Sigma_{\tau} \cap\{r\geq r_I\}}  r^{1-\eta}  |L(r^2 LK\psi)|^2\,d\sigma dr d\tau\\
\substack{ \leq }&\: C(\mathcal{E}_{1+\epsilon}[K\phi](\tau_A)+\mathcal{E}_{1+\epsilon}[K^2\phi](\tau_A))+C\int_{\Sigma_{\tau_A} \cap\{r\geq r_I\}}  r^{2-\eta} |L(r^2 LK\psi)|^2\,d\sigma dr\\
+&\:C\int_{\tau_A}^{\infty}\int_{\Sigma_{\tau} \cap\{r \geq r_I\}} r^{3-\eta}|L(r^3KG)|^2\,d\sigma drd\tau\\
			+&\: C \sum_{k\leq 1}\int_{\tau_A}^{\infty}\int_{\Sigma_{\tau}}[(1+\tau)^{1+\delta}(1-r^{-1})^{1+\epsilon}+r^2\chi_{\tau_A\leq \tau\leq \tau_A+1}+r^{1+\epsilon}]|K^{k+1}G|^2\,r^2d\sigma drd\tau.
\end{split}
\end{equation*}
We apply the mean-value theorem together with an interpolation argument to obtain $(1+\tau)^{-1+\eta}$ decay for $\int_{\Sigma_{\tau} \cap\{r\geq r_I\}}  r |L(r^2 LK\psi)|^2\,d\sigma dr$. We then conclude that
\begin{equation*}
\begin{split}
\int_{\Sigma_{\tau} \cap\{r\geq r_I\}}& r^2 | LK\psi|^2\,d\sigma dr\leq C(1+\tau)^{-2+\eta}\sup_{\tau'}\left[(1+\tau')^{2-\eta}\mathcal{E}_{1+\epsilon}[K\phi](\tau')+(1+\tau')^{1-\eta}\mathcal{E}_{1+\epsilon}[K^2\phi](\tau')\right]\\
+&\:C(1+\tau)^{-2+\eta}\int_{\Sigma_{0}} r^{2-\eta}  |L(r^2 LK\psi)|^2\,d\sigma dr\\
+&\:C(1+\tau)^{-2+\eta}\int_{0}^{\infty}\int_{\Sigma_{\tau'} \cap\{r \geq r_I\}} r^{3-\eta}(1+\tau)^2|rKG|^2\,d\sigma drd\tau'\\
+&\:C(1+\tau)^{-2+\eta}\int_{0}^{\infty}\int_{\Sigma_{\tau} \cap\{r \geq r_I\}} r^{2-\eta}(r^{\eta}\tau+r)|L(r^3KG)|^2\,d\sigma drd\tau\\
			+&\: C(1+\tau)^{-2+\eta} \sum_{k\leq 1}\int_{0}^{\infty}\int_{\Sigma_{\tau}}[(1+\tau')^{1+\delta}(1-r^{-1})^{1+\epsilon}+r^2\chi_{0\leq \tau\leq 1}+r^{1+\epsilon}]\\
			\times&\:(1+\tau')^{2-\eta}|K^{k+1}G|^2\,r^2d\sigma drd\tau\\
			\leq &\: C(1+\tau)^{-2+\eta}\Big[D_{\rm inhom}[\phi]+D_{\rm inhom}[K\phi]+\int_{\Sigma_{0}} r^{2-\eta}  |L(r^2 LK\psi)|^2\,d\sigma dr\\
			&\:+\int_{0}^{\infty}\int_{\Sigma_{\tau'} \cap\{r \geq r_I\}} r^{3-\eta}(1+\tau)^2|rKG|^2+ r^{2-\eta}(r^{\eta}\tau+r)|L(r^3KG)|^2\,d\sigma drd\tau'\Big]. \qedhere
\end{split}
\end{equation*}

\end{proof}
\section{Elliptic estimates and $K$-inversion theory}
\label{sec:ellipticKinv}
In this section, we will set $M\equiv 1$ for convenience. Let $u:  \Sigma\to \R$ and consider the differential operator
			\begin{equation}
			\label{def:L}
			\mathcal{L}u:=\tX(\Delta \tX u)-(r^2-1)\tX\Phi u-r\Phi u+(\slashed{\mathcal{D}}-\Phi^2)u.
			\end{equation}
			We will assume that $u\in C_c^{\infty}(\Sigma)$ in this section. The operator $\mathcal{L}$ can be obtained by considering all the terms in the wave operator $\square_g$ without $K$-derivatives; see \eqref{eq:waveeqtildeX}.
			 
			We will construct the corresponding inverse operator $\mathcal{L}^{-1}$ and use it to construct the initial data for \eqref{eq:inhomwaveeq} corresponding to the \emph{$K$-integral} of solutions $\phi$ to \eqref{eq:inhomwaveeq}, with an appropriate choices of inhomogeneity $G$. For this construction, we will need to derive appropriate, weighted elliptic estimates.\footnote{In contrast with the standard elliptic estimates, the estimates in this section will only control (weighted) $L^2$ norms of derivatives up to order 1 in terms of a (weighted) $L^2$ norm of $\mathcal{L}u$, rather than derivatives up to order 2. This is related to the fact that $\mathcal{L}$ is not really an elliptic operator, since $K$ is not a timelike vector field. Nevertheless, we will refer to these estimates as elliptic estimates.} 
			
			\subsection{Elliptic estimates}
			We first derive a preliminary, \emph{global} estimate for $u$. For this, we introduce the rescaled variable:
			\begin{equation*}
			\widetilde{u}=\sqrt{r^2-1}u.
			\end{equation*}
\begin{proposition}
\label{prop:elliptictilde}
Let $m\in \N_1$ and $0< p\leq 2$. Consider the rescaled variable $\widetilde{u}=\sqrt{r^2-1}u$. Then for $r_H>1$ arbitrarily small and $\delta\geq 0$, $\epsilon>0$, there exists a constant $C=C(r_H)>1$ such that
\begin{multline}
	\label{eq:elliptictilde}
	\int_{\Sigma} pr ^{p-1}(1-r^{-1})^2|\tX \widetilde{u}|^2+ (2-p)r^{p-3}|\widetilde{\snabla}_{\s^2}\widetilde{u}|^2\,d\sigma dr\leq \epsilon  \int_{\Sigma\cap\{r\leq r_H\}} (1-r^{-1})^{1+\delta}|\tX \widetilde{u}|^2\,d\sigma dr\\
	+C\epsilon^{-1}\int_{\Sigma} r^{p-1}(1-r^{-1})^{-\delta}|\mathcal{L}u|^2\,d\sigma dr.
	\end{multline}
\end{proposition}
	\begin{proof}
	We start by deriving an equation for $\widetilde{u}$ in terms of $\mathcal{L}u$. We have that
	\begin{equation*}
	\begin{split}
	(r^2-1)^{-\frac{1}{2}}\tX(\Delta \tX u)=&\:(r^2-1)^{-\frac{1}{2}}\tX(\Delta \tX((r^2-1)^{-\frac{1}{2}}  \widetilde{u}))=(r^2-1)^{-\frac{1}{2}}\tX\left(\frac{\Delta}{r^2-1}(r^2-1)^{\frac{1}{2}} \tX\widetilde{u}\right)\\
	-&(r^2-1)^{-\frac{1}{2}}\tX(r\Delta (r^2-1)^{-\frac{3}{2}} \widetilde{u} )\\
	=&\:\tX\left(\frac{\Delta}{r^2-1} \tX\widetilde{u}\right)-(r^2-1)^{-\frac{1}{2}}\frac{d}{dr}(r\Delta (r^2-1)^{-\frac{3}{2}} ) \widetilde{u}.
		\end{split}
	\end{equation*}
	We moreover obtain:
	\begin{equation*}
	-(r^2-1)^{-\frac{1}{2}}(r^2-1)\tX\Phi ((r^2-1)^{-\frac{1}{2}}  \widetilde{u})=-\tX\Phi\widetilde{u}+(r^2-1)^{-1}r\Phi \widetilde{u}.
	\end{equation*}
	Note that $-(r^2-1)^{-\frac{1}{2}}r\Phi u=-(r^2-1)^{-1}r\Phi \widetilde{u}$, which cancels with the last term on the RHS above.
	Hence, we obtain
	\begin{equation*}
	(r^2-1)^{-\frac{1}{2}}\mathcal{L}u=\tX\left(\frac{\Delta}{r^2-1} \tX\widetilde{u}\right)- \tX\Phi\widetilde{u}+(r^2-1)^{-1}(\slashed{\mathcal{D}}-\Phi^2)\widetilde{u}-(r^2-1)^{-\frac{1}{2}}\frac{d}{dr}(r\Delta (r^2-1)^{-\frac{3}{2}} ) \widetilde{u}.
	\end{equation*}
	We now consider the vector field multiplier $-(r+1)^p\frac{r-1}{r+1}\tX \overline{\widetilde{u}}$. We obtain the identity:
	\begin{multline*}
	\Re\left(-(r+1)^p\frac{r-1}{r+1}\tX \overline{\widetilde{u}}(r^2-1)^{-\frac{1}{2}}\mathcal{L}u\right)=\overbrace{-(r+1)^p\Re\left(\frac{r-1}{r+1}\tX \overline{\widetilde{u}}\tX\left(\frac{\Delta}{r^2-1} \tX\widetilde{u}\right)\right)}^{=:J_1}+\Phi(\ldots)\\
	\overbrace{-(r+1)^{p-2}\Re(\tX \overline{\widetilde{u}}(\slashed{\mathcal{D}}-\Phi^2)\widetilde{u})}^{=:J_2}\\
	+\overbrace{(r+1)^{p-\frac{3}{2}}(r-1)^{\frac{1}{2}} \frac{d}{dr}(r\Delta (r^2-1)^{-\frac{3}{2}} ) \widetilde{u}\tX \overline{\widetilde{u}}}^{=:J_3}.
	\end{multline*}
	We first rewrite $J_1$. We have that
	\begin{equation*}
	\begin{split}
	J_1=-\frac{1}{2}(r+1)^p\tX\left(\frac{(r-1)^2}{(r+1)^2} |\tX \widetilde{u}|^2\right)=\widetilde{X}\left(-\frac{1}{2}(r-1)^{p-2}(r-1)^2|\tX\widetilde{u}|^2\right)+\frac{p}{2}(r+1)^{p-1}(1-r^{-1})^2|X \widetilde{u}|^2.
	\end{split}
	\end{equation*}
	We then rewrite $J_2$ by applying the Leibniz rule on $\s^2$ and in the $\tX$ direction:
	\begin{equation*}
	\begin{split}
	J_2=&\:\frac{1}{2}(r+1)^{p-2}\tX(|\widetilde{\snabla}_{\s^2}\widetilde{u}|^2-|\Phi\widetilde{u}|^2)+\Phi(\ldots)+\Divs(\ldots)\\
	=&\:\tX\left(\frac{1}{2}(r+1)^{p-2}|\widetilde{\snabla}_{\s^2}\widetilde{u}|^2-|\Phi\widetilde{u}|^2\right)+\frac{1}{2}(2-p)(r+1)^{p-3}(|\widetilde{\snabla}_{\s^2}\widetilde{u}|^2-|\Phi\widetilde{u}|^2).
	\end{split}
	\end{equation*}
	Finally, we consider $J_3$. We use that
	\begin{equation*}
	\frac{d}{dr}(r\Delta (r^2-1)^{-\frac{3}{2}} )=\frac{2r-1}{(r+1)^{\frac{5}{2}}(r-1)^{\frac{1}{2}} }
	\end{equation*}
	to obtain
	\begin{equation*}
		\begin{split}
	J_3=&\:\frac{1}{2}(r+1)^{p-\frac{3}{2}}(r-1)^{\frac{1}{2}}\frac{d}{dr}(r\Delta (r^2-1)^{-\frac{3}{2}} ) \widetilde{X}(|\widetilde{u}|^2)=\frac{1}{2}(r+1)^{p-4}(2r-1) \widetilde{X}(|\widetilde{u}|^2)\\
	=&\:\frac{1}{2}\widetilde{X}((r+1)^{p-4}(2r-1) |\widetilde{u}|^2)+\frac{1}{2}(r+1)^{p-5}((6-2p)(r-1)-p)|\widetilde{u}|^2.
		\end{split}
	\end{equation*}
	Note that only term in $J_1,J_2,J_3$ that is not manifestly non-negative definite is $\frac{1}{2}(r+1)^{p-5}((6-2p)(r-1)-p)|\widetilde{u}|^2$ for $r-1<\frac{p}{6-2p}$. However, since
	\begin{equation*}
	\int_{\s^2}\frac{1}{2}(2-p)(r+1)^{p-3}(|\widetilde{\snabla}_{\s^2}\widetilde{u}|^2-|\Phi\widetilde{u}|^2)\geq\frac{1}{8}(2-p)(r+1)^{p-3}\int_{\s^2}|\Phi \widetilde{u}|^2\,d\sigma,
	\end{equation*}
	we have that $J_2+J_3$ is non-negative definite for $p$ suitably small, after integrating over $\s^2$.
	
	After integrating $J_1+J_2+J_3$ and using that all boundary terms vanish (using the compactness of the support of $\widetilde{u}$), we therefore conclude that for $p>0$ sufficiently small,
	\begin{equation}
	\label{eq:auxelliptictilde}
	\int_{\Sigma} pr ^{p-1}(1-r^{-1})^2|\tX \widetilde{u}|^2+ r^{p-3}|\widetilde{\snabla}_{\s^2}\widetilde{u}|^2\,d\sigma dr\leq C \int_{\Sigma} (1-r^{-1})^{\frac{1}{2}}r^{p-1}|\tX\widetilde{u}||\mathcal{L}u|\,d\sigma dr.
	\end{equation}
	Then, equipped with \eqref{eq:auxelliptictilde}, we can immediately extend \eqref{eq:auxelliptictilde} to $0<p\leq 2$, by estimating the non-negative definite term in $J_3$ with \eqref{eq:auxelliptictilde} for $p$ small.
	
	By applying Young's inequality, we can further estimate, for $\epsilon>0$ arbitrarily small and $\delta\geq 0$,
	\begin{equation*}
	(1-r^{-1})^{\frac{1}{2}}r^{p-1}|\tX\widetilde{u}||\mathcal{L}u|\leq \epsilon r ^{p-1}(1-r^{-1})^{1+\delta}|\tX \widetilde{u}|^2+\frac{1}{4}\epsilon^{-1}(1-r^{-1})^{1-\delta}r^{p-1}|\mathcal{L}u|^2.
	\end{equation*}
	Note that away from the horizon, we can absorb the $|\tX \widetilde{u}|^2$ term above into the LHS of \eqref{eq:auxelliptictilde} and conclude \eqref{eq:elliptictilde}.
	\end{proof}
	
To improve \eqref{eq:elliptictilde} and absorb the term on the RHS with a factor $\epsilon$, we will first assume that $u$ is supported on the $m$th azimuthal mode and decompose $u(r,\theta,\tilde{\varphi})=\sum_{\ell\in \N_{|m|}} u_{m \ell}(r)S_{m \ell}(\theta)e^{im \tilde{\varphi}}$. Let 
\begin{equation*}
u^{\flat}:=\sum_{|m|\leq \ell\leq L_m} u_{m \ell}(r)S_{m \ell}(\theta)e^{im \tilde{\varphi}},
\end{equation*}
where $L_m\geq |m|$ will be chosen suitably large later on. Denote $u^{\sharp}=u-u^{\flat}$. Recall from \S \ref{sec:sphere} that
\begin{equation*}
\slashed{\mathcal{D}}(S_{m \ell}(\theta)e^{im \tilde{\varphi}})=-\Lambda_{+,m \ell}S_{m \ell}(\theta)e^{im \tilde{\varphi}}.
\end{equation*}

We first derive an analogue of \eqref{eq:elliptictilde} for $u^{\sharp}$ with fewer degenerate factors at $r=1$ on the LHS and expressed in terms of $u$.
\begin{proposition}
\label{prop:elliptichighfreq}
For $0\leq p\leq 2$ and $L_m$ suitably large, there exists a constant $C>0$ such that
\begin{equation}
	\label{eq:elliptichighfreq}
	\int_{\Sigma} pr ^{p-1}(1-r^{-1})^2|\tX (r u^{\sharp})|^2+ r^{p-1}|\widetilde{\snabla}_{\s^2}u^{\sharp}|^2\,d\sigma dr\leq C \int_{\Sigma} r^{p-1}|(\mathcal{L}u)^{\sharp}|^2\,d\sigma dr.
	\end{equation}
\end{proposition}
\begin{proof}
 We restrict to $r\leq r_H$, where $r_H$ will be chosen suitably small, by considering the vector field multiplier $-\chi (r-1)^q \overline{u}^{\sharp}$, with $\chi: [1,\infty]\to \R_{\geq 0}$ a smooth cut-off function such that $\chi(r)=1$ for $r\leq r_H$ and $\chi(r)=0$ for $r\geq r_H+1$.
 
 Then
 \begin{equation*}
 \begin{split}
 -\chi  (r-1)^q\Re(\overline{u}^{\sharp}(\mathcal{L}u)^{\sharp})=&-\Re(\chi \overline{u}^{\sharp} \tX(\Delta \tX u^{\sharp}))+(r^2-1)(r-1)^q\chi \Re(\overline{u}^{\sharp} \tX \Phi u^{\sharp})\\
 &\:+(r-1)^q(|\widetilde{\snabla}_{\s^2}u^{\sharp}|^2-|\Phi u^{\sharp}|^2)+ \Phi(\ldots)+\Divs(\ldots).
 \end{split}
 \end{equation*}
 We further rewrite
 \begin{multline*}
 -\Re(\chi  (r-1)^q\overline{u}^{\sharp} \tX(\Delta \tX u^{\sharp}))=-\tX( \chi (r-1)^q \overline{u}^{\sharp} \Delta \tX u^{\sharp})+\chi (r-1)^{2+q}| \tX u^{\sharp}|^2\\
 +\left(\frac{d\chi}{dr}(r-1)^{q+2}+q\chi (r-1)^{q+1}\right)\Re(\overline{u}^{\sharp}  \tX u^{\sharp}).
 \end{multline*}
 We estimate $(r^2-1)(r-1)^q\chi \Re(\overline{u}^{\sharp} \tX \Phi u^{\sharp})$ via Young's inequality:
 \begin{equation*}
 (r^2-1)(r-1)^q\chi |\Re(\overline{u}^{\sharp} \tX \Phi u^{\sharp})|\leq \frac{1}{2}\chi(r-1)^{2+q}| \tX u^{\sharp}|^2+\frac{1}{2}(r+1)^2(r-1)^q | \Phi u^{\sharp}|^2+\Phi(\ldots).
 \end{equation*}
 The first term on the RHS can be absorbed into $\chi (r-1)^q | \tX u^{\sharp}|^2$. To absorb the second term into $(r-1)^q(|\widetilde{\snabla}_{\s^2}u^{\sharp}|^2-|\Phi u^{\sharp}|^2)$, we need to take $L_m$ appropriately large, so that $\Lambda_{+,m\ell}$ is suitably large for all $\ell> L_m$.
 
 Note that we can estimate the term involving $\overline{u}^{\sharp}  \tX u^{\sharp}$ further by applying the Leibniz rule:
 \begin{equation*}
 \begin{split}
& \left(\frac{d\chi}{dr}(r-1)^{q+2}+q\chi(r-1)^{q+1}\right)\Re(\overline{u}^{\sharp}  \tX u^{\sharp})=\frac{1}{2} \left(\frac{d\chi}{dr}(r-1)^{q+2}+q\chi(r-1)^{q+1}\right)\tX(| u^{\sharp}|^2)\\
 =&\:\tX\left(\left(\frac{d\chi}{dr}(r-1)^{q+2}+q\chi(r-1)^{q+1}\right)| u^{\sharp}|^2\right)\\
 &-\frac{1}{2}\left[\frac{d^2\chi}{dr^2}(r-1)^{q+2}+(2q+2)\frac{d\chi}{dr} (r-1)^{q+1}+\chi q(q+1)(r-1)^{q}\right]| u^{\sharp}|^2.
 \end{split}
 \end{equation*}
 We now fix $q=0$. Integrating over $\Sigma$ and combining with \eqref{eq:elliptictilde} applied to $u^{\sharp}$ with $\epsilon$ suitably small. Here, we use that can absorb the term involving derivatives of $\chi$ into the LHS of \eqref{eq:elliptictilde} and the term involving an $\epsilon$ on the RHS of \eqref{eq:elliptictilde} can similarly be absorbed (for $\epsilon>0$ suitably small). We conclude \eqref{eq:elliptichighfreq}.
\end{proof}
To improve \eqref{eq:elliptictilde} for $u^{\flat }$, we introduce the rescaled variables $\check{u}_{\pm,m \ell}=\frac{u_{m \ell}}{ w_{\pm,  m \ell}}$, with $w_{\pm, m \ell}$ defined in \S \ref{sec:statsol}. Then $ w_{\pm,  m \ell}$ are solutions to:
\begin{equation}
\label{eq:eqwellm}
\mathcal{L}(w_{ \pm, m \ell} S_{m \ell} e^{im\tilde{\varphi}})=0,
\end{equation}
with boundary conditions imposed at $r=1$.

By using that $\mathcal{L}(w_{\pm, m \ell} S_{m \ell} e^{im\tilde{\varphi}})=0$, it follows easily that $\check{u}_{\pm,m \ell}$ must satisfy:
\begin{equation}
\label{eq:ellipticcheck}
w_{\pm,m \ell}^{-1}(\mathcal{L}u)_{m \ell}=\tX( \Delta \tX \check{u}_{\pm,m \ell})+2\Delta w_{\pm,m \ell}^{-1}\frac{dw_{\pm,m \ell}}{dr} \tX \check{u}_{\pm,m \ell}-im (r^2-1)\tX \check{u}_{\pm,m \ell}.
\end{equation}
\begin{proposition}
\label{prop:ellipticlowfreq}
Let $\delta>0$, then there exists a constant $C=C(\delta)>0$, such that 
\begin{equation}
	\label{eq:ellipticlowfreq}
	\int_{\Sigma} pr ^{p-1}(1-r^{-1})^{2+\delta}|\tX (r u^{\flat})|^2+ r^{p-1}(1-r^{-1})^{\delta}|\widetilde{\snabla}_{\s^2}u^{\flat}|^2\,d\sigma dr\leq C \int_{\Sigma} r^{p-1}(1-r^{-1})^{-\delta}|(\Phi\mathcal{L}u)^{\flat}|^2\,d\sigma dr.
	\end{equation}
\end{proposition}
\begin{proof}
 We restrict to $r\leq r_H$, where $r_H$ will be chosen suitably small, by considering $\chi: [1,\infty]\to \R_{\geq 0}$, a smooth cut-off function such that $\chi(r)=1$ for $r\leq r_H$ and $\chi(r)=0$ for $r\geq r_H+\eta$, with $\eta>0$ suitably small.
 
 We will suppress the subscripts in the notation by writing $\check{u}=\check{u}_{\pm, m \ell}$ and $w=w_{\pm,m \ell}$. Consider the vector field multiplier $-\chi(r-1)^q \tX \overline{\check{u}}$ acting on both sides of \eqref{eq:ellipticcheck}:
 \begin{equation*}
 \begin{split}
 -\chi(r-1)^q\Re(\tX \overline{\check{u}} w^{-1}(\mathcal{L}u)_{m \ell})=&\: -\chi(r-1)^{2+q}\Re(\tX \overline{\check{u}} \tX^2\check{u})-2(r-1)^{q+1}\left[1+\Re\left( w^{-1}\frac{dw}{dr} \right)\right]|\tX \check{u}|^2\\
 =&\:-\frac{1}{2}\tX\left(\chi(r-1)^{2+q}|\tX \check{u}|^2\right)\\
 &\:+(r-1)^{q+1}\left[\frac{q}{2}-1-2\Re\left( w^{-1}\frac{dw}{dr} \right)\right]|\tX \check{u}|^2+\frac{d\chi}{dr}(r-1)^{2+q}|\tX \check{u}|^2.
 \end{split}
 \end{equation*}
 By \eqref{eq:estderwelliptic}, it follows that
 \begin{equation*}
 \frac{q}{2}-1-2\Re\left( w^{-1}\frac{dw}{dr} \right)=\frac{q}{2}\pm 2\beta+O(r-1),
 \end{equation*}
with $\beta=\Re(\sqrt{\Lambda_{m\ell}-2m^2+\frac{1}{4}})$. If $\beta\in \R$, we will only consider $w_{+, m \ell}$ and not $w_{-, m \ell}$.

Hence, for $q=2\beta+\eta$, we can integrate along $\Sigma$ to obtain
\begin{multline}
\label{eq:ellipticcheckaux1}
\int_{\Sigma \cap\{r\leq r_H\}} (r-1)^{-2\beta+\eta+1}|\tX \check{u}|^2\,d\sigma dr\leq C\int_{\Sigma\cap\{r_H\leq r\leq r_H+1\}}|\tX \check{u}|^2\\
+C\int_{\Sigma\cap \{r\leq r_H+1\}}(r-1)^{-2\beta+\eta-1}|w|^{-2}|(\mathcal{L}u)_{m \ell}|^2\,d\sigma dr.
\end{multline}
We have that in $r\leq 2$:
\begin{align*}
\tX \check{u}=&\:\tX \left(\frac{(r^2-1)^{\frac{1}{2}}}{w}\widetilde{u}\right),\\
(r-1)^{-2\beta}|\tX \check{u}|^2\leq &\:C |\tX \widetilde{u}|^2+C\left|\Lambda_{m\ell}-m^2+\frac{1}{4}\right|(r-1)^{-2+2\beta}|\widetilde{u}|^2,\\
|\tX \widetilde{u}|^2\leq &\: C(r-1)^{-2\beta }|\tX \check{u}|^2+C\left|\Lambda_{m\ell}-m^2+\frac{1}{4}\right|(r-1)^{-2-2\beta}|\check{u}|^2.
\end{align*}
The restriction $\ell\leq L_m$ ensures moreover that $|\Lambda_{m\ell}-m^2+\frac{1}{4}|\leq Bm^2$ for some $B>0$.

Hence, we can combine \eqref{eq:ellipticcheckaux1} with \eqref{eq:elliptictilde} to obtain:
\begin{multline}
\label{eq:ellipticcheckaux2}
\int_{\Sigma \cap\{r\leq r_H\}} (r-1)^{-2\beta+1+\eta}|\tX \check{u}|^2+(r-1)^{1+\eta}|\tX \widetilde{u}|^2\,d\sigma dr\leq C\epsilon m^2\int_{\Sigma\cap\{r\leq r_H\}}(r-1)^{-1-2\beta+\eta}|\check{u}|^2\\
+Cm^2\int_{\Sigma}(r-1)^{-\eta}|(\mathcal{L}u)_{m \ell}|^2\,d\sigma dr.
\end{multline}
To absorb the first term on the RHS of \eqref{eq:ellipticcheckaux2} into the LHS, we apply the following Hardy inequality: for $\beta=0$
\begin{equation*}
\int_{\Sigma \cap\{r\leq r_H\}} (r-1)^{-1+\eta}|\check{u}|^2\,d\sigma dr\leq \int_{\Sigma\cap\mathcal{H}^+} (r-1)^{-\eta}| \check{u}|^2\,d\sigma+C\int_{\Sigma \cap\{r\leq r_H\}} (r-1)^{1+\eta}|\tX\check{u}|^2\,d\sigma dr
\end{equation*}
and for $\beta>0$ and $\eta<2 \beta$:
\begin{equation*}
\int_{\Sigma \cap\{r\leq r_H\}} (r-1)^{-1+\eta}|\check{u}|^2\,d\sigma dr\leq \int_{\Sigma\cap\{r=r_H\}} (r-1)^{-\eta}| \check{u}|^2\,d\sigma+C\int_{\Sigma \cap\{r\leq r_H\}} (r-1)^{1+\eta}|\tX\check{u}|^2\,d\sigma dr.
\end{equation*}
The boundary term at $\mathcal{H}^+$ vanishes and the boundary term at $r=r_H$ can be estimated via \eqref{eq:elliptictilde}. Hence, we conclude that
\begin{multline}
\label{eq:ellipticcheckaux3}
\int_{\Sigma \cap\{r\leq r_H\}} (r-1)^{-2\beta+1+\eta}|\tX \check{u}|^2+(r-1)^{1+\eta}|\tX \widetilde{u}|^2+(r-1)^{-1-2\beta+\eta}|\check{u}|^2\,d\sigma dr\\
\leq C\epsilon m^2\int_{\Sigma\cap\{r\leq r_H\}}(r-1)^{-1-2\beta+\eta}|\check{u}|^2+Cm^2\int_{\Sigma}(r-1)^{-\eta}|(\mathcal{L}u)_{m \ell}|^2\,d\sigma dr.
\end{multline}
To absorb, we therefore need $\epsilon\leq C^{-1}m^{-2}$ and then \eqref{eq:ellipticlowfreq} follows after combining \eqref{eq:ellipticcheckaux3} with \eqref{eq:elliptictilde}.
\end{proof}

\subsection{Additional estimates for $u^{\sharp}$}
In this section, we derive additional higher-order estimates for $u^{\sharp}$ by commuting $\mathcal{L}$ with $\tX$.
\begin{lemma}
The following commutation properties holds:
\begin{align}
\label{eq:ellipticcomm1}
[\mathcal{L},\tX] u=& -2(r-1)\tX^2u-2\tX u+2r\Phi \tX u+\Phi u\\
\label{eq:ellipticcomm2}
[\mathcal{L},\tX^2] u=& -4(r-1)\tX^3u-6\tX^2 u+4r\Phi \tX^2 u+4\Phi \tX u.
\end{align}
\end{lemma}
\begin{proof}
We have that
\begin{equation*}
\tX \mathcal{L}u=\tX^2(\Delta \tX u)-(r^2-1)\tX^2\Phi u-2r \tX \Phi u-r\tX\Phi u-\Phi u+(\slashed{\mathcal{D}}-\Phi^2)\tX u
\end{equation*}
We can further rewrite:
\begin{equation*}
\tX^2(\Delta \tX u)=\tX(\Delta \tX^2 u)+\tX(2(r-1) \tX u)=\tX(\Delta \tX^2 u)+2(r-1)\tX^2u+2\tX u.
\end{equation*}
Combining the above gives \eqref{eq:ellipticcomm1}.

To obtain \eqref{eq:ellipticcomm2}, we use that
\begin{equation*}
[\mathcal{L},\tX^2] u=[\mathcal{L},\tX](\tX u)+\tX([\mathcal{L},\tX]u)
\end{equation*}
and we apply \eqref{eq:ellipticcomm1}.
\end{proof}

\begin{proposition}
\label{prop:elliptichonearhor}
Let $1\leq n\leq 2$. Then for sufficiently large $L_m$ depending on $n$,
\begin{multline}
\label{eq:elliptichonearhor}
\sum_{k\leq n}\int_{\Sigma\cap\{r\leq r_I\}} (r-1)^{1-\epsilon}|\tX((1-r^{-1})\tX )^{k} u^{\sharp}|^2+ |\widetilde{\snabla}_{\s^2}((1-r^{-1})\tX )^k u^{\sharp}|^2\,d\sigma dr\\
\leq C\sum_{ k\leq n-1}\int_{\Sigma\cap\{r\leq r_I+1\}} (r-1)^{1-\epsilon}|\tX((1-r^{-1})\tX )^k(\mathcal{L}u)^{\sharp}|^2\,d\sigma dr+C\int_{\Sigma}|(\mathcal{L}u)^{\sharp}|^2\,d\sigma dr.
\end{multline}
\end{proposition}
\begin{proof}
We restrict to $r\leq r_I$, where $r_I$ will be chosen suitably small, by considering the vector field multiplier $-\chi (r-1)^q \tX\overline{u}^{\sharp}$, with $\chi: [1,\infty]\to \R_{\geq 0}$ a smooth cut-off function such that $\chi(r)=1$ for $r\leq r_I$ and $\chi(r)=0$ for $r\geq r_I+1$.

By repeating the proof of Proposition \ref{prop:elliptichighfreq}, but with $u^{\sharp}$ replaced by $\tX u^{\sharp}$, we obtain the identity:
\begin{equation*}
 \begin{split}
 -\chi  (r-1)^q\Re(\tX \overline{u}^{\sharp}(\mathcal{L}\tX u)^{\sharp})=&\tX\left[\left(\frac{d\chi}{dr}(r-1)^{q+2}+q\chi(r-1)^{q+1}\right)| \tX u^{\sharp}|^2 -\chi (r-1)^q \tX\overline{u}^{\sharp} \Delta \tX^2 u^{\sharp}\right]\\
&\: +\Phi(\ldots)+\Divs(\ldots)+\chi (r-1)^{2+q}| \tX^2 u^{\sharp}|^2+(r-1)^q(|\widetilde{\snabla}_{\s^2}\tX u^{\sharp}|^2-|\Phi \tX u^{\sharp}|^2)\\
&\:+(r^2-1)(r-1)^q\chi \Re(\tX \overline{u}^{\sharp} \tX^2 \Phi u^{\sharp})-\frac{1}{2}\chi q(q+1)(r-1)^{q}| \tX u^{\sharp}|^2\\
&-\frac{1}{2}\left(\frac{d^2\chi}{dr^2}(r-1)^{q+2}+(2q+2)\frac{d\chi}{dr} (r-1)^{q+1}\right)| \tX u^{\sharp}|^2.
 \end{split}
 \end{equation*}
By \eqref{eq:ellipticcomm1}, we also have that
\begin{equation*}
 \begin{split}
 \chi  (r-1)^q\Re(\tX \overline{u}^{\sharp}(\mathcal{L}\tX u)^{\sharp})=&\:\chi  (r-1)^q\Re(\tX \overline{u}^{\sharp}(\tX\mathcal{L} u)^{\sharp})-2\chi  (r-1)^{q+1} \Re(\tX \overline{u}^{\sharp} \tX^2u^{\sharp})-2 \chi  (r-1)^q|\tX u^{\sharp}|^2\\
 &\:+\Phi(\ldots)+\chi  (r-1)^q\Re(\tX \overline{u}^{\sharp} \Phi u^{\sharp})\\
 =&\:(q-1)\chi  (r-1)^q|\tX u^{\sharp}|^2+\chi  (r-1)^q\Re(\tX \overline{u}^{\sharp}(\tX\mathcal{L} u)^{\sharp})+\chi  (r-1)^q\Re(\tX \overline{u}^{\sharp} \Phi u^{\sharp})\\
& - \tX\left[(r-1)^{q+1}  \chi |\tX u^{\sharp}|^2\right]+\frac{d\chi}{dr}(r-1)^{q+1} |\tX u^{\sharp}|^2.
 \end{split}
\end{equation*}
Hence after integrating over $\Sigma\cap\{r\leq r_I+1\}$ and taking $L_m$ sufficiently large to absorb the $| \tX u^{\sharp}|^2$ terms with a bad sign, we obtain
\begin{equation*}
 \begin{split}
	\int_{\Sigma\cap\{r\leq r_I\}} &(1-r^{-1})^{2+q}|\tX^2 u^{\sharp}|^2+(1-r^{-1})^q|\widetilde{\snabla}_{\s^2} \tX u^{\sharp}|^2\,d\sigma dr\\
	\leq&\: C \int_{\Sigma\cap\{r\leq r_I+1\}} (1-r^{-1})^q |\Phi u^{\sharp}|^2+ (1-r^{-1})^q|(\tX\mathcal{L} u)^{\sharp}|^2\,d\sigma dr+ C\int_{\Sigma\cap\{r_I\leq r\leq r_I+1\}} |\tX u^{\sharp}|^2\,d\sigma dr.
	 \end{split}
	\end{equation*}
	
	Hence, taking $q=1-\epsilon$, we obtain \eqref{eq:elliptichonearhor} for $n=1$ by combining the above with \eqref{eq:elliptichighfreq} . The $n=2$ case follows similarly by considering an extra $\tX$ above and applying the commutation formula \eqref{eq:ellipticcomm2}.
	\end{proof}
	
In order to obtain sufficiently strong weighted $L^2$ estimates for $u^{\sharp}$ and its derivatives away from the horizon in terms of the derivative $X$ rather than $\tX$, we need additional estimates, which require more decay of $(\mathcal{L}u)^{\sharp}$ as $r\to \infty$.
	\begin{proposition}
Let $1\leq n\leq 2$. Then for sufficiently large $L_m$ depending on $n$,
\begin{multline}
\label{eq:elliptichonearinf}
\sum_{k\leq n}\int_{\Sigma\cap\{r\geq r_I\}} r^{4+4k}|\tX^{k+1} u^{\sharp}|^2+r^{2+4k}|\widetilde{\snabla}_{\s^2}\tX ^k u^{\sharp}|^2\,d\sigma dr\\
\leq C\sum_{k\leq n}\int_{\Sigma } (r-1)^{2+4k}|\tX^k(\mathcal{L}u)^{\sharp}|^2\,d\sigma dr+C\int_{\Sigma}|(\mathcal{L}u)^{\sharp}|^2\,d\sigma dr.
\end{multline}
and
\begin{multline}
\label{eq:elliptichonearinf2}
\sum_{k\leq n}\int_{\Sigma\cap\{r\geq r_I\}} r^2 |X(r^2X)^{k} u^{\sharp}|^2+|{\snabla}_{\s^2}(r^2X)^k u^{\sharp}|^2\,d\sigma dr\\
\leq C\sum_{k\leq n}\int_{\Sigma } (r-1)^{2+4k}|\tX^k(\mathcal{L}u)^{\sharp}|^2\,d\sigma dr+C\int_{\Sigma}|(\mathcal{L}u)^{\sharp}|^2\,d\sigma dr.
\end{multline}
\end{proposition}
\begin{proof}
We repeat the estimates in the proof of Propositions \ref{prop:elliptichighfreq} and \ref{prop:elliptichonearhor}, but we consider the estimates globally on $\Sigma$ by omitting $\chi$ and taking $q=2+4k$ when considering $(\mathcal{L}\tX^ku)^{\sharp}$ with $k=0,1,2$.
\end{proof}

\subsection{Construction of $\mathcal{L}^{-1}$}
The aim of this section is to construct the time integral
\begin{equation*}
K^{-1}\phi(\tau,r,\theta,\widetilde{\varphi})=-\int_{\tau}^{\infty}\phi(\tau',r,\theta,\widetilde{\varphi})\,d\tau'
\end{equation*}
corresponding to solutions $\phi$ to \eqref{eq:inhomwaveeq} with suitable $G$. We will construct time integrals by solving \eqref{eq:inhomwaveeq} with $G$ replaced by $G^{(1)}=-\int_{\tau}^{\infty}G(\tau',r,\theta,\widetilde{\varphi})\,d\tau'$ and initial data $(K^{-1}\phi|_{\Sigma},\phi|_{\Sigma})$. We construct $K^{-1}\phi|_{\Sigma}$ by inverting the operator $\mathcal{L}$, i.e.\ by solving $\mathcal{L}u=F$ for appropriate choices of $F$, depending on the initial data for $\phi$ and the inhomogeneity $G$.

We start by constructing the coefficients $u_{m \ell}$ with $\ell\leq L_m$ corresponding to solutions of
\begin{equation*}
\mathcal{L}u=F
\end{equation*}
 via integration in $r$.
\begin{proposition}
\label{eq:invLlowfreq}
We define
\begin{equation}
u_{m \ell}(r):=-w_{\infty,m \ell}(r)\int_{1}^{r}(r'-1)^{-2}w_{\infty,m \ell}^{-2}(r')e^{im\int_{r_0}^{r'} \frac{x+1}{x-1}\,dx}\left[\int_{r'}^{\infty}w_{\infty,m \ell}(x)e^{-im\int_{r_0}^{x} \frac{y+1}{y-1}\,dy}F_{m \ell}(x)\,dx\right]\,dr',
\end{equation}
where we assume enough regularity and decay towards $r=\infty$ for $F_{m \ell}$ and for the RHS above to be well-defined and in the case $\Lambda_{+,m \ell}\leq 2m^2$, we assume additionally that
\begin{equation}
\label{eq:condF}
\int_{1}^{\infty}w_{\infty,m \ell}(r)e^{-im\int_{r_0}^r \frac{x+1}{x-1}\,dx}F_{m \ell}(r)\,dr=0.
\end{equation}
Then $u_{m \ell}\in C^2((1,\infty))$ and $\mathcal{L}u_{m \ell}=F_{m \ell}$. 

Furthermore, there exists a constant $C>0$ such that for $r\geq 2$ and $q\geq 0$:
\begin{equation}
\label{eq:timeintboundlinft}
|ru_{m \ell}|(r)+|rX (ru_{m \ell})|(r)\leq r^{-q}\sup_{r'\geq 2}|r^{1+q} F_{m \ell}(r')|,
\end{equation}
where $X=\tX-im\frac{\mathbbm{h}}{2}$.
\begin{enumerate}[label=\emph{(\roman*)}]
\item When $\Lambda_{+,m \ell}\geq 2m^2$, we have that for $r\leq 2$:
\begin{align*}
\sum_{k\leq 2}|((r-1)\tX)^ku_{m \ell}(r)|\leq &\:C(|B_+|+|B_-|+1)\sum_{k\leq 2}\sup_{r'\leq 2}|((r-1)\tX)^kF_{m \ell})(r')|\\
\sum_{k\leq 2}|X((r-1)\tX)^ku_{m \ell}(r)|\leq &\:C(|B_+|+|B_-|+1)(r-1)^{-\frac{3}{2}+\sqrt{\alpha_{m \ell}}}\sum_{\substack{k\leq 2\\j\leq 1}}\sup_{r'\leq 2}|\tX^j((r-1)\tX)^kF_{m \ell})(r')|.
\end{align*}
\item When $\Lambda_{+,m \ell}<2m^2$, we have that for $r\leq 2$:
\begin{align*}
\sum_{k\leq 2}|((r-1)\tX)^ku_{\pm,m \ell}(r)|\leq &\: C\sum_{k\leq 2}\sup_{r'\leq 2}|((r-1)\tX)^kF_{m \ell})(r')|,\\
\sum_{k\leq 2}|\tX ((r-1)\tX)^ku_{\pm,m \ell}(r)|\leq &\: C\sum_{\substack{k\leq 2\\j\leq 1}}\sup_{r'\leq 2}|\tX^j((r-1)\tX)^kF_{m \ell})(r')|.
\end{align*}
\end{enumerate}
\end{proposition}
\begin{proof}
We consider $(\mathcal{L}u_{m \ell})=F_{m \ell}$. Denote
\begin{equation*}
\check{u}_{\infty,m \ell}=\frac{u_{m\ell}}{w_{\infty,m \ell}}.
\end{equation*}
We then rewrite \eqref{eq:ellipticcheck} with $\check{u}_{\infty,m \ell}$ replacing $\check{u}_{\pm,m \ell}$ and $w_{\infty,m \ell}$ replacing $w_{\pm,m \ell}$ as the following ODE:
\begin{equation}
\label{eq:ellipticcheck2}
\tX( \Delta  w_{\infty,m \ell}^{2}e^{-im\int_{r_0}^r \frac{x+1}{x-1}\,dx}\tX \check{u}_{\infty,m \ell})=e^{-im\int_{r_0}^r \frac{x+1}{x-1}\,dx}w_{\infty,m \ell}F_{m \ell}=:\check{F}_{\infty,m \ell},
\end{equation}
with $r_0>1$ arbitrary.

We assume that $\lim_{r\to \infty}\widetilde{X}\check{u}_{\infty,m \ell}(r)=0$ and integrate from $r=\infty$ to obtain
\begin{equation*}
 \tX  \check{u}_{\infty,m \ell}(r)=-(r-1)^{-2}w_{\infty,m \ell}^{-2}(r)e^{im\int_{r_0}^r \frac{x+1}{x-1}\,dx}\int_r^{\infty}\check{F}_{\infty,m \ell }(r')\,dr'.
\end{equation*}
Integrating once more, starting from $r=1$ and assuming that $\check{u}_{\infty,m \ell}(1)=0$, we obtain:
\begin{equation*}
\check{u}_{\infty,m \ell}(r)=-\int_{1}^{r}(r'-1)^{-2}w_{\infty,m \ell}^{-2}(r')e^{im\int_{r_0}^{r'} \frac{x+1}{x-1}\,dx}\int_{r'}^{\infty}\check{F}_{\infty, m \ell}(x)\,dx\,dr'.
\end{equation*}
Note that
\begin{equation*}
|X (ru_{m \ell})|=|X (r w_{\infty}(r)\check{u}_{\infty,m \ell})|=|X(e^{imr_*}(1+O(r^{-1})) \check{u}_{\infty,m \ell})|\leq C |(X+im) \check{u}_{\infty,m \ell}|+Cr^{-2}|\check{u}_{\infty,m \ell}|.
\end{equation*}
Furthermore,
\begin{equation*}
|(X+im) \check{u}_{\infty,m \ell}|=\left|\left(\tX+im\left(1-\frac{\mathbbm{h}}{2}\right)\right) \check{u}_{\infty,m \ell}\right|=\left|\left(\tX+imO(r^{-1})\right) \check{u}_{\infty,m \ell}\right|.
\end{equation*}
Hence, when $r\geq 2$, we can estimate for $q\geq 0$
\begin{equation*}
\begin{split}
|X (ru_{m \ell})|\leq &\:Cr^{-1-q}\sup_{r'\geq 2} | r^{1+q} \tX  \check{u}_{\infty,m \ell}(r')|+Cr^{-1-q}\sup_{r'\geq 2}|r^{q}\check{u}_{\infty,m \ell}(r')|\\
\leq &\: Cr^{-1-q}\sup_{r'\geq 2}|r^{1+q} F_{m \ell}(r')|.
\end{split}
\end{equation*}
Suppose $\alpha_{m \ell}:=\Lambda_{+,m \ell}+\frac{1}{4}-2m^2\leq 0$. Then, in fact, in view of the asymptotics of $w_{\infty,m \ell}$ near $r=1$, to conclude that the above integral is well-defined, it is necessary to impose
\begin{equation}
\label{eq:auxcondF}
\int_{1}^{\infty}w_{\infty,m \ell}(r)e^{-im\int_{r_0}^r \frac{x+1}{x-1}\,dx}F_{m \ell}(r)\,dr=0,
\end{equation}
so that we can express:
\begin{equation}
\check{u}_{\infty,m \ell}(r)=\int_{1}^{r}(r'-1)^{-2}w_{\infty,m \ell}^{-2}(r')e^{im\int_{r_0}^{r'} \frac{x+1}{x-1}\,dx}\int_{1}^{r'}\check{F}_{\infty,m \ell}(x)\,dx\,dr',
\end{equation}
and hence,
\begin{equation*}
u_{m \ell}(r)=w_{\infty,m \ell}(r)\int_{1}^{r}(r'-1)^{-2}w_{\infty,m \ell}^{-2}(r')e^{im\int_{r_0}^{r'} \frac{x+1}{x-1}\,dx}\int_{1}^{r'}w_{\infty,m \ell}(x)e^{-im\int_{r_0}^{x} \frac{y+1}{y-1}\,dy}F_{m \ell}(x)\,dx\,dr'.
\end{equation*}
Suppose $0<\alpha_{m \ell}<\frac{1}{4}$. Then we will still impose \eqref{eq:auxcondF}. We can immediately estimate in $r\leq 2$:
\begin{align*}
|\check{u}_{\infty,m \ell}(r)|\leq &\:C(r-1)^{\frac{1}{2}+\sqrt{\alpha_{m \ell}}}\sup_{r'\leq 2}|F_{m \ell}(r')|,\\
|\tX \check{u}_{\infty,m \ell}(r)|\leq &\:C(r-1)^{-\frac{1}{2}+\sqrt{\alpha_{m \ell}}}\sup_{r'\leq 2}|F_{m \ell}(r')|.
\end{align*}
In order to obtain boundedness for $u_{m \ell}$ and $\tX u_{m \ell}$ near $r=1$, we now consider $u_{\pm,m \ell}$. We have that in $r\leq 2$
\begin{equation*}
\begin{split}
|\Delta w_{\pm}^2\tX \check{u}_{\pm,m \ell}(r)|=&\:\left|(r-1)^2w_{\pm}^2\tX \left(\frac{w_{\infty}}{w_{\pm}} \check{u}_{\infty,m \ell}\right)(r)\right|\\
\leq&\: C|B_{\mp }| |\check{u}_{\infty,m \ell}(r)|+C(|B_{+ }|+|B_{- }|)(r-1)^{1-\Re(\sqrt{\alpha_{m \ell}})\mp \Re(\sqrt{\alpha_{m \ell}})} |\tX\check{u}_{\infty,m \ell}(r)|\\
\leq &\: C(|B_{+ }|+|B_{- }|)  (r-1)^{\frac{1}{2}\pm \Re(\sqrt{\alpha_{m \ell}})}\sup_{r'\leq 2}|F_{m \ell}(r')|.
\end{split}
\end{equation*}
Hence, $\Delta w_{\pm}^2\tX \check{u}_{\pm,m \ell}$ vanishes at $r=1$ for $\alpha_{m \ell}<\frac{1}{4}$ and we can integrate \eqref{eq:ellipticcheck2} starting from $r=1$ to obtain:
\begin{equation}
\label{eq:derodepm}
 \tX  \check{u}_{\pm,m \ell}(r)=(r-1)^{-2}w_{\pm,m \ell}^{-2}(r)e^{im\int_{r_0}^r \frac{x+1}{x-1}\,dx}\int_1^r{F}_{\infty,m \ell }(r')w_{\pm,m \ell}(r')e^{-im\int_{r_0}^{r'} \frac{y+1}{y-1}\,dy}\,dr'.
\end{equation}
Since
\begin{equation*}
\begin{split}
|\check{u}_{\pm,m \ell}(r)|=\left|\frac{w_{\infty}}{w_{\pm}} \check{u}_{\infty,m \ell}(r)\right|\leq&\: C(|B_{+ }|+|B_{- }|)(1+(r-1)^{\pm 2\sqrt{\alpha}})\left| \check{u}_{\infty,m \ell}(r)\right|\\
\leq&\: C(|B_{+ }|+|B_{- }|)(r-1)^{\frac{1}{2}+\Re(\sqrt{\alpha})\pm 2\Re(\sqrt{\alpha}) }\sup_{r'\leq 2}|F_{m \ell}(r')|,
\end{split}
\end{equation*}
we have that $\check{u}_{\pm,m \ell}(1)=0$. And hence, we can integrate \eqref{eq:derodepm} to obtain for all $\alpha_{m\ell}<\frac{1}{4}$ an alternative expression for $u_{m \ell}(r)$:
\begin{equation*}
u_{m \ell}(r)=w_{\pm,m \ell}(r)\int_{1}^{r}(r'-1)^{-2}w_{\pm,m \ell}^{-2}(r')e^{im\int_{r_0}^{r'} \frac{x+1}{x-1}\,dx}\int_{1}^{r'}w_{\pm,m \ell}(x)e^{-im\int_{r_0}^{x} \frac{y+1}{y-1}\,dy}F_{m \ell}(x)\,dx\,dr'.
\end{equation*}
The above expression of $u_{m \ell}(r)$ will allow us to obtain more refined estimates for $u_{m \ell}(r)$ and to obtain a bound on $|\tX u_{m \ell}(r)|$.

We first expand:
\begin{equation*}
F_{m \ell}(r)=F_{m \ell}(1)+\int_1^r\tX F_{m \ell}(r')dr'.
\end{equation*}
Hence, we can expand
\begin{multline*}
w_{\pm,m \ell}(x)e^{-im\int_{r_0}^{x} \frac{y+1}{y-1}\,dy}F_{m \ell}(x)\,dx=e^{-im(r_0-1)}(r-1)^{-\frac{1}{2}\mp \sqrt{\alpha_{m \ell}}-3im}(1+O_{\infty}(r-1))\\
\times \left(F_{m \ell}(1)+\int_1^x\tX F_{m \ell}(y)dy\right).
\end{multline*}
We integrate the the above expansion to obtain:
\begin{equation*}
\begin{split}
u_{m \ell}(r)=&\:e^{-im(r_0-1)}(r-1)^{-\frac{1}{2}-im\mp \sqrt{-\alpha}}(1+O((r-1))\\
&\:\times\int_1^r(r'-1)^{-1\pm 2\sqrt{\alpha_{m \ell}}+4im}(1+O_{\infty}(r'-1))\frac{1}{\frac{1}{2}\mp \sqrt{\alpha}-3im}(r'-1)^{\frac{1}{2}\mp \sqrt{\alpha_{m \ell}}-3im}F_{m \ell}(1)\,dr'\\
&\:+w_{\pm,m \ell}(r)\int_{1}^{r}(r'-1)^{-2}w_{\pm,m \ell}^{-2}(r')e^{im\int_{r_0}^{r'} \frac{x+1}{x-1}\,dx}\int_{1}^{r'}w_{\pm,m \ell}(x)e^{-im\int_{r_0}^{x} \frac{y+1}{y-1}\,dy}\int_1^x \tX F_{m \ell}(y)dy\,dx\,dr'\\
=&\: \frac{\frac{1}{2}+im\pm \sqrt{\alpha_{m \ell}}}{\frac{1}{2}\mp \sqrt{\alpha_{m \ell}}-3im}e^{-im(r_0-1)}F_{m \ell}(1)(1+O_{\infty}(r-1))\\
&\:+w_{\pm,m \ell}(r)\int_{1}^{r}(r'-1)^{-2}w_{\pm,m \ell}^{-2}(r')e^{im\int_{r_0}^{r'} \frac{x+1}{x-1}\,dx}\int_{1}^{r'}w_{\pm,m \ell}(x)e^{-im\int_{r_0}^{x} \frac{y+1}{y-1}\,dy}\int_1^x \tX F_{m \ell}(y)dy\,dx\,dr'.
\end{split}
\end{equation*}
From this, we can estimate in $r\leq 2$:
\begin{equation*}
|\tX u_{m \ell}(r)|\leq  C\sup_{r'\leq 2}|\tX F_{m \ell}(r')|.
\end{equation*}
It is straightforward to commute the above estimates with $(r-1)\tX$ and obtain for $r\leq 2$:
\begin{equation*}
\sum_{k\leq 2}|\tX ((r-1)\tX)^ku_{m \ell}(r)|\leq  C\sum_{k\leq 2}\sup_{r'\leq 2}|\tX((r-1)^kF_{m \ell})(r')|.
\end{equation*}

For $\alpha_{m \ell}\geq  \frac{1}{4}$, we will not impose \eqref{eq:auxcondF}, as this would require $F_{m \ell}(1)=0$. Instead, we express
\begin{equation*}
u_{m \ell}(r)=-w_{\infty,m \ell}(r)\int_{1}^{r}(r'-1)^{-2}w_{\infty,m \ell}^{-2}(r')e^{im\int_{r_0}^{r'} \frac{x+1}{x-1}\,dx}\int_{r'}^{\infty}w_{\infty,m \ell}(x)e^{-im\int_{r_0}^{x} \frac{y+1}{y-1}\,dy}F_{m \ell}(x)\,dx\,dr'.
\end{equation*}
Note that for $\alpha_{m \ell}\geq  \frac{1}{4}$
\begin{equation*}
w_{\infty,m \ell}(r)=B_+ (r-1)^{-\frac{1}{2}-\sqrt{\alpha_{m \ell}}}+(|B_+|+|B_-|)O((r-1)^{\frac{1}{2}-\sqrt{\alpha_{m \ell}}}).
\end{equation*}
We first split:
\begin{multline*}
\int_{r'}^{\infty}w_{\infty,m \ell}(x)e^{-im\int_{r_0}^{x} \frac{y+1}{y-1}\,dy}F_{m \ell}(x)\,dx=\int_{r'}^{2}w_{\infty,m \ell}(x)e^{-im\int_{r_0}^{x} \frac{y+1}{y-1}\,dy}F_{m \ell}(x)\,dx\\
+\int_{2}^{\infty}w_{\infty,m \ell}(x)e^{-im\int_{r_0}^{x} \frac{y+1}{y-1}\,dy}F_{m \ell}(x)\,dx
\end{multline*}
and use that
\begin{equation*}
\left|\int_{2}^{\infty}w_{\infty,m \ell}(x)e^{-im\int_{r_0}^{x} \frac{y+1}{y-1}\,dy}F_{m \ell}(x)\,dx\right|\leq C \sup_{r'\geq 2}|rF_{m \ell}|(r').
\end{equation*}

We then write $F_{m \ell}(x)=F_{m \ell}(1)+\int_1^x F_{m \ell}(y)\,dy$ and use the asymptotic properties of $w_{\infty,m \ell}$ near $r=1$ to obtain for $r\leq 2$:
\begin{multline*}
\left|\tX^k\left(\int_{r'}^{2}w_{\infty,m \ell}(x)e^{-im\int_{r_0}^{x} \frac{y+1}{y-1}\,dy}F_{m \ell}(x)\,dx-\frac{B_+}{-\frac{1}{2}+\sqrt{\alpha_{m \ell}} +2im}(r'-1)^{\frac{1}{2}-\sqrt{\alpha_{m \ell}} -2im}F_{m \ell}(1)\right)\right|\\
\leq C(|B_+|+|B_-|) \tX^k(1+(r'-1)^{\frac{3}{2}-\sqrt{\alpha_{m \ell}}})\sup_{r'\leq 2}(|F_{m \ell}(r')|+|\tX  F_{m \ell}(r')|),
\end{multline*}
with $k=0,1$.

Hence, by plugging the above expression into $u_{m \ell}$ as above, we obtain:
\begin{align*}
|u_{m \ell}(r)|\leq &\:C(|B_+|+|B_-|)\sup_{r'}| r'F_{m \ell}(r')|\\
|Xu_{m \ell}(r)|\leq &\:C(|B_+|+|B_-|)(r-1)^{-\frac{3}{2}+\sqrt{\alpha_{m \ell}}}(\sup_{r'}|r'F_{m \ell}(r')|+\sup_{r'\leq 2}|\tX  F_{m \ell}(r')|).
\end{align*}
It is straightforward to commute the above estimates with $(r-1)\tX$ and $((r-1)\tX)^2$ in order to derive in the region $r\leq 2$:
\begin{align*}
\sum_{k\leq 2}|((r-1)\tX)^ku_{m \ell}(r)|\leq &\:C(|B_+|+|B_-|)\sum_{k\leq 2}\sup_{r'}|r' F_{m \ell}(r')|\\
\sum_{k\leq 2}|X((r-1)\tX)^ku_{m \ell}(r)|\leq &\:C(|B_+|+|B_-|)(r-1)^{-\frac{3}{2}+\sqrt{\alpha_{m \ell}}}(\sup_{r'}|F_{m \ell}(r')|+\sup_{r'\leq 2}|\tX  F_{m \ell}(r')|).
\end{align*}
\end{proof}

By combining the above estimates, we arrive at a weighted energy estimate for $u$ along $\Sigma$, which will be important for establishing the finiteness of appropriate initial energy norms for $K^{-1}\phi$ in terms of initial data for $\phi$.

\begin{proposition}
\label{prop:mainestKint}
Let $u_{m \ell}$ with $\ell\geq |m|$ be the functions defined in Proposition \ref{eq:invLlowfreq}. Then
\begin{align*}
u=&\:u^{\flat}+u^{\sharp}\in L^2(\Sigma),\quad \textnormal{with}\\
u^{\flat}(r,\theta,\tphi)=&\: \sum_{|m|\leq \ell\leq L_m} u_{m \ell}(r)S_{m \ell}(\theta)e^{im \tilde{\varphi}},\\
u^{\sharp}(r,\theta,\tphi)=&\: \sum_{L_m< \ell<\infty} u_{m \ell}(r)S_{m \ell}(\theta)e^{im \tilde{\varphi}}
\end{align*}
is well-defined. For $\delta>0$ arbitrarily small, 
\begin{equation*}
\begin{split}
\sum_{k\leq 2}&\int_{\Sigma\cap{\{r\leq 2\}}}(1-r^{-1})^{1+\delta}|\tX((1-r^{-1})\tX)^k u^{\flat}|^2+(1-r^{-1})^{-1+\delta}|((1-r^{-1})\tX)^k u^{\flat}|^2\,d\sigma dr\\
&\:+\int_{\Sigma\cap{\{r\geq 2\}}} |u^{\flat}|^2+r^2 |Xu^{\flat}|^2\,d\sigma dr\\
\leq&\: C(|B_+|+|B_-|+1)\Bigg[\sum_{\substack{k\leq 2\\j\leq 1}}\sup_{r'\leq 2}\int_{\Sigma\cap\{r=r'\}}|\tX^j((r-1)\tX)^kF^{\flat})(r')|^2\,d\sigma\\
+&\:\sup_{r'\geq 2}\int_{\Sigma\cap\{r=r'\}}|r F^{\flat}(r')|^2\,d\sigma\Bigg].
\end{split}
\end{equation*}
and
\begin{equation*}
\begin{split}
\sum_{k\leq 2}\int_{\Sigma\cap{\{r\leq 2\}}}&(1-r^{-1})^{1-\epsilon}|\tX((1-r^{-1})\tX)^k u^{\sharp}|^2+|\snabla_{\s^2}((1-r^{-1})\tX)^k u^{\sharp}|^2\,d\sigma dr\\
&+\int_{\Sigma\cap{\{r\geq 2\}}} |u^{\sharp}|^2+r^2 |Xu^{\sharp}|^2\,d\sigma dr\\
\leq&\: C\sum_{k\leq 1}\int_{\Sigma\cap\{r\leq 2\}} (1-r^{-1})^{1-\epsilon}|\tX((1-r^{-1})\tX )^kF^{\sharp}|^2+|F^{\sharp}|^2\,d\sigma dr\\
+&\:C\int_{\Sigma\cap\{r\geq 2\}} |F^{\sharp}|^2+r^2 |\tX F^{\sharp}|^2\,d\sigma dr.
\end{split}
\end{equation*}
\end{proposition}
\begin{proof}
The estimate for $u^{\flat}$ follows directly from the pointwise estimates in Proposition \ref{eq:invLlowfreq}. The estimates for $u^{\sharp}$ are the result of combining Proposition \ref{prop:elliptichighfreq} with $p=1$ with Proposition \ref{prop:elliptichonearhor}.
\end{proof}

\subsection{An inhomogeneous wave equation}
\label{sec:constKinv}
In this section, we will apply the elliptic theory from \S \ref{eq:invLlowfreq} to construct $K$-integrals of solutions to appropriate inhomogeneous wave equations.

We start by considering the homogeneous wave equation \eqref{eq:waveeq} with $|a|=M$:
\begin{equation*}
\square_{g_{M,\pm |M|}}\phi=0.
\end{equation*}

Recall that $\alpha_{m \ell}=\Lambda_{+,m \ell}+\frac{1}{4}-2m^2$. We will consider in this section $m,\ell\in \Z$ such that \eqref{eq:asmsphev} holds, i.e.\
\begin{equation*}
\alpha_{m \ell}\neq 0.
\end{equation*}

Recall the definitions of $\Pi$ and $\hphi$ from \S \ref{sec:precstatthm}. Then $\hphi$ satisfies the following inhomogeneous wave equation:
\begin{equation}
\label{eq:waveeqwPi}
\square_g\hphi=-\square_g\Pi.
\end{equation}

In Proposition \ref{prop:decayPi} below, we establish suitable decay of $\square_g\Pi$ in $\tau$ and $r$ that will be required to be able to apply the energy decay estimates from \S \ref{sec:edecay} to the time integral $K^{-1}\hphi$ that we will construct in \S \ref{sec:constKinvdata}.

\begin{proposition}
\label{prop:decayPi}
Let $K,L\in \N_0$. Then there exists a $C=C(K,N,m)>0$, such that for all $k\leq K$ and $l\leq L$:
\begin{align}
\label{eq:boxPi1}
\sup_{(\theta,\widetilde{\varphi})\in \s^2}|((r-1)\tX)^kK^l\square_g\Pi|(\tau,r,\theta,\widetilde{\varphi})\leq &\:C\left((\tau+1)(r-1)+4\right)^{-\frac{1}{2}} (\tau+1)^{-1-\frac{1}{2}\nu_m-l}\quad \textnormal{for $r\leq 2$},\\
\label{eq:boxPi2}
\sup_{(\theta,\widetilde{\varphi})\in \s^2} |(Xr^2)^kK^l\square_g\Pi|(\tau,r,\theta,\widetilde{\varphi})\leq &\:Cr^{-3}(1+\tau)^{-2-l}\quad \textnormal{for $r\geq 2$},
\end{align}
with $X=\tX-im\frac{\mathbbm{h}}{2}$ and
\begin{equation*}
\nu_m=\min \left( \{\alpha_{m \ell}\,|\, \alpha_{m \ell}>0\}\cup\{1\}\right).
\end{equation*}
\end{proposition}
\begin{proof}
We apply Proposition \ref{prop:squarePial0} and \ref{prop:inhomagtr0} to estimate $|(\tX^kK^l\square_g\Pi)_{m \ell}|^2$ and then sum over $\ell$ to obtain the estimate in $r\leq 2$. Similarly, we apply Proposition \ref{prop:inhomestrgtr2} to estimate $|(\tX^kK^l\square_g\Pi)_{m \ell}|^2$ and then sum over $\ell$ to obtain the estimate in $r\geq 2$.
\end{proof}

\subsection{Construction of $K$-integral initial data}
\label{sec:constKinvdata}
In this section, we consider the equation \eqref{eq:waveeqwPi} and we apply $\mathcal{L}^{-1}$ to appropriate functions of the initial data for $\hphi$, in order to construct the initial data along $\Sigma_0$ that evolves into the $K$-integral $K^{-1}\hphi$.

Let $\widetilde{\phi}$ be a solution to
\begin{equation}
\label{eq:waveeqinhomtimeint}
\square_g \widetilde{\phi}= -\int_{\tau}^{\infty}\square_g\Pi(\tau',r,\theta,\widetilde{\varphi})\,d\tau'
\end{equation}
such that $K\widetilde{\phi}=\phi$. Then we must have that for all $\tau\geq 0$:
\begin{equation*}
\begin{split}
\mathcal{L}\widetilde{\phi}|_{\Sigma_{\tau}}=&-2(r^2+1-\hfol \Delta)\tX(\phi-\Pi)|_{\Sigma_{\tau}}-2\left(1+\frac{1}{2}\hfol (r^2-1)-\frac{1}{2}\sin^2\theta\right)\Phi(\phi-\Pi)|_{\Sigma_{\tau}}\\
&-\left(\hfol(\Delta \hfol-2(r^2+1))-\sin^2\theta\right)K(\phi-\Pi)|_{\Sigma_{\tau}}+\left(2r-\frac{d}{dr}(\Delta \hfol)\right) (\phi-\Pi)|_{\Sigma_{\tau}}\\
&\:-\int_{\tau}^{\infty}\square_g\Pi (\tau',r,\theta,\varphi)\,d\tau'.
\end{split}
\end{equation*}
We define:
\begin{align*}
F_{\rm hom}[\phi]:=&-2(r^2+1-\hfol \Delta)\tX\phi|_{\Sigma_{0}}-2\left(1+\frac{1}{2}\hfol (r^2-1)-\frac{1}{2}\sin^2\theta\right)\Phi \phi|_{\Sigma_{0}}\\
&-\left(\hfol(\Delta \hfol-2(r^2+1))-\sin^2\theta\right)K\phi|_{\Sigma_{0}}+\left(2r-\frac{d}{dr}(\Delta \hfol)\right) \phi |_{\Sigma_{0}},\\
F_{\Pi}:=& 2(r^2+1-\hfol \Delta)\tX \Pi|_{\Sigma_{0}}+2\left(1+\frac{1}{2}\hfol (r^2-1)-\frac{1}{2}\sin^2\theta\right)\Phi \Pi|_{\Sigma_{0}}\\
&\:+\left(\hfol(\Delta \hfol-2(r^2+1))-\sin^2\theta\right)K\Pi|_{\Sigma_{0}}-\left(2r-\frac{d}{dr}(\Delta \hfol)\right) \Pi|_{\Sigma_{0}}-\int_{0}^{\infty}\square_g\Pi (\tau',r,\theta,\tphi)\,d\tau'
\end{align*}
and write $F[\phi]:=F_{\rm hom}[\phi]+F_{\Pi}$. Then
\begin{equation}
\label{eq:maineqindataKint}
\mathcal{L}\widetilde{\phi}|_{\Sigma_0}=F[\phi].
\end{equation}

In view of \eqref{eq:maineqindataKint}, we consider the following data that we will evolve into a solution $K^{-1}\phi$ to \eqref{eq:waveeqinhomtimeint} with $K(K^{-1} \phi)=\phi$ in a suitable spacetime region:
\begin{equation}
\label{eq:idKint}
(\widetilde{\phi}|_{\Sigma_0},(K\widetilde{\phi})|_{\Sigma_0})=(\chi_{*}\mathcal{L}^{-1}(F[\phi]),\phi|_{\Sigma_0}),
\end{equation}
where $\chi_{*}$ is a smooth cut-off function in $r$ such that $\chi_*(r)=1$ for $r-1\geq r_{\rm initial}$ and $\chi_*(r)=0$ for $r-1\leq \frac{1}{2}(r_{\rm initial}-1)$, with $r_{\rm initial}-1>0$ arbitrarily small. The cut-off function $\chi_*$ will ensure that the initial energies are finite and we will be able to interpolate additional growth of initial energies of $\mathcal{L}^{-1}(F[\phi])$ as $r\downarrow 1$ into growth in $\tau$.

We let $\widetilde{\phi}$ be the solution to the initial value problem corresponding to the initial data \eqref{eq:idKint} and the equation:
\begin{align}
\label{eq:waveeqinhomcutoff}
\square_g\widetilde{\phi}=&\:G,\\
\label{eq:choiceG}
G(\tau,r,\theta,\widetilde{\varphi})=&\:-\xi_*(\tau)\int_{\tau}^{\infty}\square_g\Pi(\tau',r,\theta,\widetilde{\varphi})\,d\tau',
\end{align}
with $\xi_*(\tau)$ a smooth cut-off function, satisfying $\xi_*(\tau)=1$ for $\tau\leq \tau_*$ and $\xi_*(\tau)=0$, with $\tau_*\geq 1$ arbitrarily large.

We first verify the assumptions \eqref{eq:assmG}, \eqref{eq:decayassmG1} and \eqref{eq:decayassmG2} on $G$ and establish $\tau_*$-dependent bounds on the integrals involving $G$ in $D_{1,{\rm inhom}}[\widetilde{\phi}]$ and $D_{2,{\rm inhom}}[\widetilde{\phi}]$.
			\begin{lemma}
			\label{lm:estGcutoff}
			Let $G$ be defined by \eqref{eq:choiceG}. Then $G$ satisfies \eqref{eq:assmG}, \eqref{eq:decayassmG1} and \eqref{eq:decayassmG2}. Furthermore, for $\delta>0$ arbitrarily small, there exists a constant $C>0$ (independent of $\tau_*$) such that
			\begin{align}
			\label{eq:Kinvinhomest1}
			D_{1,{\rm inhom}}[\widetilde{\phi}]-D_{1,{\rm hom}}[\widetilde{\phi}]\leq C\tau_*^{1-\nu_m+\delta}|{\mathfrak{h}}_{m \ell}|^2,\\
			\label{eq:Kinvinhomest2}
			D_{2,{\rm inhom}}[\widetilde{\phi}]-D_{2,{\rm hom}}[\widetilde{\phi}]\leq C\tau_*^{1-\nu_m+\delta}|{\mathfrak{h}}_{m \ell}|^2
			\end{align}
			\end{lemma}
			\begin{proof}
			By integrating \eqref{eq:boxPi1} and \eqref{eq:boxPi2} in $\tau$, we obtain
			\begin{align*}
\int_{\tau}^{\infty}\sup_{(\theta,\widetilde{\varphi})\in \s^2}|\tX^kK^l\square_g\Pi|(\tau',r,\theta,\widetilde{\varphi})\,d\tau'\leq &\:C\left((\tau+1)(r-1)+4\right)^{-\frac{1}{2}} (\tau+1)^{-\frac{1}{2}\nu_m-l}\quad \textnormal{($r\leq 2)$}\\
\leq &\:C(r-1)^{-\frac{1}{2}+\frac{1}{2}\delta} (\tau+1)^{-\frac{1}{2}+\frac{1}{2}\delta-\frac{1}{2}\nu_m-l},\\
\int_{\tau}^{\infty}\sup_{(\theta,\widetilde{\varphi})\in \s^2} |(Xr^2)^kK^l\square_g\Pi|(\tau',r,\theta,\widetilde{\varphi})\,d\tau'\leq &\:Cr^{-3}(1+\tau)^{-1-l}\quad \textnormal{($r\geq 2$)},
\end{align*}
with $0<\delta<\nu_m$. Then \eqref{eq:assmG}, \eqref{eq:decayassmG1} and \eqref{eq:decayassmG2}, as well as the estimates \eqref{eq:Kinvinhomest1} and \eqref{eq:Kinvinhomest1}, follow immediately.
			\end{proof}

In the proposition below, we show that $\mathcal{L}^{-1}(F[\phi])$ is well-defined.
\begin{proposition}
\label{prop:idkint}
Let $\phi$ be a solution to \eqref{eq:waveeq} arising from initial data in $C_c^{\infty}(\Sigma_0)$ and assume that $\alpha_{m\ell}\neq 0$ for all $\ell\geq |m|$. Then the expressions
\begin{equation}
u_{m \ell}(r):=-w_{\infty,m \ell}(r)\int_{1}^{r}(r'-1)^{-2}w_{\infty,m \ell}^{-2}(r')e^{im\int_{r_0}^{r'} \frac{x+1}{x-1}\,dx}\left[\int_{r'}^{\infty}w_{\infty,m \ell}(x)e^{-im\int_{r_0}^{x} \frac{y+1}{y-1}\,dy}F_{m \ell}[\phi]\,dx\right]\,dr'
\end{equation}
are well-defined, and in the case $\Lambda_{+,m \ell}<2m^2$, there exists a unique ${\mathfrak{h}}_{m \ell}\in \C$ for each $\ell, m$ such that
\begin{equation}
\label{eq:condFphi}
\int_{1}^{\infty}w_{\infty,m \ell}(r)e^{-im\int_{r_0}^r \frac{x+1}{x-1}\,dx}F_{m \ell}[\phi]\,dr=0,
\end{equation}
with ${\mathfrak{h}}_{m \ell}=0$ if and only if the initial data satisfy
\begin{equation*}
\int_{1}^{\infty}w_{\infty,m \ell}(r)e^{-im\int_{r_0}^r \frac{x+1}{x-1}\,dx}(F_{\rm hom})_{m \ell}[\phi]\,dr=0.
\end{equation*}
If we assume instead that there exists an $\ell_*\in \N_0$ such that $\alpha_{m\ell_*}= 0$, then we impose the following restriction on the initial data for $\phi$:
\begin{equation}
\label{eq:restrinitialdataalpha0}
\int_{1}^{\infty}w_{\infty,m \ell_*}(r)e^{-im\int_{r_0}^r \frac{x+1}{x-1}\,dx}(F_{\rm hom})_{m \ell_*}[\phi]\,dr=0\quad \forall \ell\geq |m|.
\end{equation}

Furthermore, $F_{m \ell}\in C^1(1,\infty)$, $u_{m \ell}\in C^2(1,\infty)$ and $\mathcal{L}u_{m \ell}=F_{m \ell}[\phi]$. 

Define $\mathcal{L}^{-1}(F[\phi])=\sum_{|m|\leq \ell}u_{m \ell}$. Then $\mathcal{L}^{-1}(F[\phi])\in L^2(\Sigma_0)$ and $\widetilde{\phi}$ satisfies: let $\epsilon>0$, then there exists a constant $C_m>0$, such that
\begin{equation}
\label{eq:estKintenergy}
D_{1,{\rm inhom}}[\widetilde{\phi}]+D_{2,{\rm inhom}}[\widetilde{\phi}]\leq C_m (\tau_*^{1-\nu_m+2\epsilon}+(1-r_{\rm initial})^{-2\epsilon})D^{(1)}[\phi],
\end{equation}
with
\begin{multline*}
 D^{(1)}[\phi]:=\sum_{|m|\leq l\leq l_m}|{\mathfrak{h}}_{m \ell}|^2+\sum_{k+l\leq 4,l\leq 2}\int_{\Sigma_{0}} \mathcal{E}_{1+\epsilon}[K^{k+l}\phi]\, r^2d\sigma dr\\
 +\int_{\Sigma_{0} \cap\{r\geq r_I\}}  \sum_{k\leq 2} r^{2-\delta} | LK^{k}\psi|^2+\sum_{k\leq 1}r^{2-\delta} | L((r^2L)K^k\psi)|^2\,d\sigma dr\\ 
 +\sum_{\substack{k\leq 2\\j\leq 1}}\sup_{r'\leq 2}\int_{\Sigma_0\cap\{r=r'\}}|\tX^j((r-1)\tX)^kF_{\rm hom}[\phi])|^2\,d\sigma+\sup_{r'\geq 2}\int_{\Sigma_0\cap\{r=r'\}}|r F_{\rm hom}[\phi]|^2\,d\sigma\\
 +\int_{\Sigma_0\cap\{r\geq 2\}} r^2 |\tX F_{\rm hom}[\phi]|^2\,d\sigma dr.
\end{multline*}
\end{proposition}
\begin{proof}
The statements in the proposition follow from an application Propositions \ref{eq:invLlowfreq} and \ref{prop:mainestKint}. We therefore need to show that the assumptions in Proposition \ref{eq:invLlowfreq} hold. That is to say, we need to show that
\begin{equation}
\label{eq:finiteFint}
\int_{1}^{\infty}w_{\infty,m \ell}(r)e^{-im\int_{r_0}^r \frac{x+1}{x-1}\,dx}(F_{\Pi})_{m \ell}[\phi]\,dr<\infty
\end{equation}
is well-defined. We split the integral into a $\int_1^2$ and $\int_2^{\infty}$ part. For $r\leq 2$, we have by Propositions \ref{prop:squarePial0} and \ref{prop:inhomagtr0} that
\begin{equation*}
\int_{\tau}^{\infty}\left|\square_g\Pi_{m \ell}\right|(\tau',r)\leq C_m |{\mathfrak{h}}_{m \ell}|.
\end{equation*}
We can directly estimate the  remaining terms in $F_{\Pi}$ from the definition of $\Pi$ to obtain: for $r\leq 2$
\begin{equation*}
|(F_{\Pi})_{m \ell}|(r)\leq C_m |{\mathfrak{h}}_{m \ell}|.
\end{equation*}
In order to estimate $|(F_{\Pi})_{m \ell}|(r)$ for $r\geq 2$, we first observe that we can rewrite:
\begin{equation*}
\begin{split}
F_{\Pi}=& 2r^2(1+O(r^{-2}))\left(\tX-\frac{\hfol}{2}\Phi \right) \Pi|_{\Sigma_{\tau}}+(2r+O(r^{-1})) \Pi|_{\Sigma_{\tau}}+2\left(O(1)-\frac{1}{2}\sin^2\theta\right)\Phi \Pi|_{\Sigma_{\tau}}\\
&\:+\left(O(1)-\sin^2\theta\right)K\Pi|_{\Sigma_{\tau}}+\int_{0}^{\infty}\square_g\Pi (\tau',r,\theta,\widetilde{\varphi})\,d\tau'\\
=& 2r(1+O(r^{-2}))\left(\tX-\frac{\hfol}{2}\Phi \right) (r\Pi|_{\Sigma_{\tau}})+O(r^{-1}) \Pi|_{\Sigma_{\tau}}+2\left(O(1)-\frac{1}{2}\sin^2\theta\right)\Phi \Pi|_{\Sigma_{\tau}}\\
&\:+\left(O(1)-\sin^2\theta\right)K\Pi|_{\Sigma_{\tau}}-\int_{0}^{\infty}\square_g\Pi (\tau',r,\theta,\widetilde{\varphi})\,d\tau'.
\end{split}
\end{equation*}
For $r\geq 2$, we apply Proposition \ref{prop:inhomestrgtr2} to obtain
\begin{equation*}
\int_{\tau}^{\infty}\left|\square_g\Pi_{m \ell}\right|(\tau',r,\theta,\widetilde{\varphi})\leq C_mr^{-3} |{\mathfrak{h}}_{m \ell}|.
\end{equation*}
We can directly estimate the remaining terms in $F_{\Pi}$ in order to obtain
\begin{equation*}
|(F_{\Pi})_{m \ell}[\phi]|(r)\leq C_m r^{-1}|{\mathfrak{h}}_{m \ell}|.
\end{equation*}
We conclude that \eqref{eq:finiteFint} holds.

We may assume without loss of generality that
\begin{equation*}
\int_{1}^{\infty}w_{\infty,m \ell}(r)e^{-im\int_{r_0}^r \frac{x+1}{x-1}\,dx}(F_{\Pi})_{m \ell}\,dr\neq 0,
\end{equation*}
for ${\mathfrak{h}}_{m \ell}\neq 0$. Indeed, if the above integral were vanishing, we could redefine $f_{\pm}^{m \ell}$ by replacing $\left(\tau+\frac{4}{r-1}+1\right)$ with  $\left(\tau+\frac{4}{r-1}+\beta\right)$ with $\beta>0$, $\beta \neq 1$ and similarly, we could redefine $f_{\pm}^{m \ell,{\rm exp}}$ by replacing $\left(\tau+\frac{4}{r-1}+e^{r-1}\right)$ with  $\left(\tau+\frac{4}{r-1}+\beta e^{r-1}\right)$. As these modifications occur at subleading order in $\tau$, $r$ and $\frac{1}{r-1}$, they do not affect the estimates involving $\Pi$.

Hence, we can choose ${\mathfrak{h}}_{m \ell}$ such that
\begin{equation*}
\int_{1}^{\infty}w_{\infty,m \ell}(r)e^{-im\int_{r_0}^r \frac{x+1}{x-1}\,dx}F_{m \ell}[\phi]\,dr=\int_{1}^{\infty}w_{\infty,m \ell}(r)e^{-im\int_{r_0}^r \frac{x+1}{x-1}\,dx}((F_{\Pi})_{m \ell}+(F_{\rm hom})_{m \ell}[\phi])\,dr=0.
\end{equation*}
The estimate \eqref{eq:estKintenergy} then follows by using that $K^{k+1}\widetilde{\phi}|_{\Sigma_0}=K^k\phi|_{\Sigma_0}$ and applying Proposition \ref{prop:mainestKint} to estimate the terms involving $\widetilde{\phi}$ in $D_{1,{\rm inhom}}[\widetilde{\phi}]+D_{2,{\rm inhom}}[\widetilde{\phi}]$ that do not involve $K\phi$ and Lemma \ref{lm:estGcutoff} to estimate the terms involving $G$. 

Note that we pick up a factor $\tau_*^{2\epsilon}$ from the terms involving $\widetilde{\phi}$ since in $r\leq 2$
\begin{equation*}
\begin{split}
(r-1)^{1-\epsilon}|\tX\widetilde{\phi}|_{\Sigma_0}|^2=&\:(r-1)^{1-\epsilon}|X(\chi_{* }\mathcal{L}^{-1}(F))|^2\\
\leq&\: \frac{1}{2}(r-1)^{1-\epsilon}|\chi_{* }|^2|\tX(\mathcal{L}^{-1}(F))|^2+\frac{1}{2}(r-1)^{1-\epsilon}|\frac{d\chi_*}{dr}| |\mathcal{L}^{-1}(F)|^2\\
\leq&\: (1-r_0)^{-2\epsilon}(r-1)^{-1+\epsilon}|\tX(\mathcal{L}^{-1}(F))|^2+(1-r_0)^{-\epsilon}|\mathcal{L}^{-1}(F)|^2.
\end{split}
\end{equation*}
We similarly pick up a factor $\tau_*^{1-\nu_m+\delta}$ from Lemma \ref{lm:estGcutoff} and we set $\delta=2\epsilon$.
\end{proof}

We now introduce in the following subset of the domain of dependence of $\Sigma_0\cap\{r\geq r_{\rm initial}\}$: 
\begin{equation*}
\{\tau\leq \tau_*\}\cap D^+(\Sigma_0\cap\{r\geq r_{\rm initial}\}),
\end{equation*}
the $K$-integral as follows:
\begin{equation*}
K^{-1}\widehat{\phi}:=\widetilde{\phi}
\end{equation*}
as $K\widetilde{\phi}=\phi$ in this region, by construction.

			\section{Decay estimates and late-time asymptotics}
			\label{sec:decay}
			In this section, we will apply the energy decay estimates for $K\phi$ and $K^2\phi$ from Proposition \ref{prop:edecay} to $K^{-1}\hphi$ for $0\leq \tau\leq \tau_*$, in order to obtain energy decay for $\hphi$ and $K\hphi$. We recall that
			\begin{equation*}
\nu_m=\min \left( \{\alpha_{m \ell}\,|\, \alpha_{m \ell}> 0\}\cup\{1\}\right).
\end{equation*}
We will moreover assume that either the condition \eqref{eq:asmsphev} holds on $\Lambda_{m\ell}$ or that the initial data for $\phi$ satisfy the restriction \eqref{eq:restrinitialdataalpha0}.

\subsection{Energy decay estimates}
We first establish appropriately weighted energy decay estimates for $\hphi$.
			\begin{proposition}
		\label{prop:edecayKinv}
		Let $\eta>0$ be arbitrarily small and $r_I>1$ suitably large. Then there exist $\delta,\epsilon>0$ suitably small, such that
		\begin{align}
\label{eq:step2edecaylessdegKinv}
\int_{\Sigma_{\tau}} (1-r^{-1})^{\delta}\mathcal{E}_{1+\epsilon}[\hphi]\, r^2d\sigma dr\leq&\:  C (1+\tau)^{-1-\nu_m+\eta} D^{(1)}[\phi],\\
\label{eq:step3edecaylessdegKKinv}
\int_{\Sigma_{\tau}}  (1-r^{-1})^{\delta}\mathcal{E}_{1+\epsilon}[K\hphi]\, r^2d\sigma dr\leq &\: C (1+\tau)^{-3-\nu_m+\eta} D^{(1)}[\phi],\\
\label{eq:step3edecayrsqKinv}
\int_{\Sigma_{\tau} \cap\{r\geq r_I\}} r^{2} | L\hpsi|^2\,d\sigma dr\leq &\: C (1+\tau)^{-1-\nu_m+\eta} D^{(1)}[\phi].
\end{align}
			\end{proposition}
			\begin{proof}
			Let $\tau_*=\tau$. There exists a constant $B>0$ such that with the choice $(r_{\rm initial}-1)^{-1}:=\frac{1}{(r_1-1)}+B\tau_*$, we have for all $0\leq \tau\leq \tau_*$ that $\Sigma_{\tau}\cap \{r\geq r_1\}\subset D_+(\Sigma_0\cap\{r\geq r_{\rm initial}\})$. Hence, $K(K^{-1}\phi)=\phi$ on $\Sigma_{\tau}\cap \{r\geq r_1\}$.

			We first take $r_1=1+\tau^{-1}$ and apply Proposition \ref{prop:edecay} together with Proposition \ref{prop:idkint} to obtain:
			\begin{equation*}
			\begin{split}
			\int_{\Sigma_{\tau}\cap\{r\geq 1+(\tau+1)^{-1}\}} (1-r^{-1})^{\delta}\mathcal{E}_{1+\epsilon}[\hphi]\, r^2d\sigma dr\leq &\: (\tau+1)^{-\delta}\int_{\Sigma_{\tau}\cap\{r\geq r_1\}} \mathcal{E}_{1+\epsilon}[K \widetilde{\phi}]\, r^2d\sigma dr\\
			\leq &\:  C (1+\tau)^{-1-\nu_m+\eta-\delta} D_{\rm inhom}[\widetilde{\phi}]\\
			\leq &\: C (1+\tau)^{-1-\nu_m+\eta-\delta+2\epsilon} D^{(1)}[\phi].
			\end{split}
			\end{equation*}
			We can take $\delta>2\epsilon$. It remains to estimate the integral when $r\in [1,1+(\tau+1)^{-1})$. We define $r_1=1+(\tau+1)^{-1}$ and $r_j=1+2^{-j}(\tau+1)^{-1}$ for $j\geq 1$. Note that $\{\frac{1}{r_j-1}\}_{j\in \N}$ is dyadic. We then have the following estimate in the dyadic intervals $[r_{j+1},r_j]$:
			\begin{equation*}
						\begin{split}
			\int_{\Sigma_{\tau}\cap\{r_{j+1}\leq r\leq r_{j}\}} (1-r^{-1})^{\delta}\mathcal{E}_{1+\epsilon}[\hphi]\, r^2d\sigma dr\leq &\: 2^{-j\delta}(\tau+1)^{-\delta}\int_{\Sigma_{\tau}\cap\{r\geq r_1\}} \mathcal{E}_{1+\epsilon}[K \widetilde{\phi}]\, r^2d\sigma dr\\
			\leq &\:  C 2^{-j\delta}(1+\tau)^{-1-\nu_m+\eta-\delta} D_{\rm inhom}[\widetilde{\phi}]\\
			\leq &\:  C 2^{-j\delta}(\tau+1)^{-\delta} (2^{j+1}+B)^{2\epsilon}(\tau+1)^{2\epsilon} D^{(1)}[\phi]\\
			\leq &\: 2^{-j(\delta-2\epsilon)}C (1+\tau)^{-1-\nu_m+\eta-\delta+2\epsilon} D^{(1)}[\phi].
			\end{split}
			\end{equation*}
			Using that $\delta>\epsilon$, we can sum over $j$, since the right-hand side forms a geometric series in $j$. We obtain:
			\begin{equation*}
			\int_{\Sigma_{\tau}\cap\{r\leq 1+(\tau+1)^{-1}\}} (1-r^{-1})^{\delta}\mathcal{E}_{1+\epsilon}[\hphi]\, r^2d\sigma dr\leq C (1+\tau)^{-1-\nu_m+\eta-\delta+2\epsilon} D^{(1)}[\phi].
			\end{equation*}
			We therefore conclude \eqref{eq:step2edecaylessdegKinv}.
			
			The estimate \eqref{eq:step3edecaylessdegKKinv} follows in the same manner, since $K\phi=K^2\widetilde{\phi}$. Finally, \eqref{eq:step3edecayrsqKinv} follows similarly with $(r_{\rm initial}-1)^{-1}:=B\tau_*$ as it concerns the region $\Sigma_{\tau}\cap \{r\geq r_I\}$, which is contained in $D_+(\Sigma_0\cap\{r\geq r_{\rm initial}\})$.
			\end{proof}

			We can obtain an additional energy decay estimate for the non-degenerate energy with energy density $\mathcal{E}_2[\hphi]$.
		\begin{proposition}
		\label{prop:nondegedecay}
		Let $\eta>0$ be arbitrarily small.  Then there exists a constant $C>0$ such that
		\begin{equation}
		\label{eq:nondegedecay}
		\int_{\Sigma_{\tau}} \mathcal{E}_{2}[\hphi]\, r^2d\sigma dr\leq  C (1+\tau)^{-\nu_m+\eta} D^{(1)}[\phi].
		\end{equation}
		Furthermore,
		\begin{equation}
		\label{eq:non1degdecay}
		\int_{\Sigma_{\tau}} \mathcal{E}_{1+\epsilon}[\hphi]\, r^2d\sigma dr\leq  C (1+\tau)^{-1-\nu_m+\eta} D^{(1)}[\phi].
		\end{equation}
		\end{proposition}
		\begin{proof}
		We let $r_1$ and $r_{\rm initial}$ be as in the proof of Proposition \ref{prop:edecayKinv}. Let $r_1-1\geq (\tau+1)$. Then it follows immediately from \eqref{eq:step2edecaylessdegKinv} that
		\begin{equation*}
		\int_{\Sigma_{\tau}\cap \{r\geq r_1\}} \mathcal{E}_{2}[\hphi]\, r^2d\sigma dr\leq  (\tau+1)^{1-\epsilon-\delta}\int_{\Sigma_{\tau}\cap \{r\geq r_1\}} (1-r^{-1})^{\delta}\mathcal{E}_{1+\epsilon}[\hphi]\, r^2d\sigma dr\leq C (1+\tau)^{-\nu_m+\eta} D^{(1)}[\phi].
		\end{equation*}
		for a suitably small $\eta>0$.
		
		Now replace $r_1$ with $r_j=1+2^j(1+\tau)$ for $j\geq 1$ and consider the corresponding $r_{\rm initial}$. We apply \eqref{eq:convKder} to express $(r-1)^{-2}\underline{L}\phi$ in terms of derivatives of $K^{-1}\phi$: let $\widetilde{Z}\in \{(r-1)\tX, K, \snabla_{\s^2}\}$, then,
		\begin{equation*}
		\begin{split}
		\int_{\Sigma_{\tau}\cap\{r_{j+1}\leq r\leq r_{j}\}} \mathcal{E}_{2}[\hphi]\, r^2d\sigma dr\leq&\:  C\sum_{k\leq 1}\int_{\Sigma_{\tau}\cap\{r_{j+1}\leq r\leq r_{j}\}} \mathcal{E}_{0}[\widetilde{Z}^k\widetilde{\phi}]\, r^2d\sigma dr\\
		\leq &\:  C2^{-j}(1+\tau)^{-1+\epsilon}\sum_{k\leq 1}\int_{\Sigma_{\tau}\cap\{r_{j+1}\leq r\leq r_{j}\}} \mathcal{E}_{1-\epsilon}[\widetilde{Z}^k\widetilde{\phi}]\, r^2d\sigma dr.
		\end{split}
		\end{equation*}
		We apply Corollary \ref{cor:horminMest} with $\phi$ replaced by $\widetilde{\phi}$ to estimate the right-hand side further and use that the term on the right-hand side of \eqref{eq:horminMest} is bounded by $((r_{\rm initial}-1)^{-2\epsilon}+(1+\tau)^{-1+\nu_m+\delta})D^{(1)}[\phi]$. This gives:
		\begin{equation*}
		\begin{split}
		\int_{\Sigma_{\tau}\cap\{r_{j+1}\leq r\leq r_{j}\}} \mathcal{E}_{2}[\hphi]\, r^2d\sigma dr\leq&\:  C2^{-j}(1+\tau)^{-1+\epsilon}[(2^{j}(1+\tau)+B)^{2\epsilon}+(1+\tau)^{-1+\nu_m+\delta}]D^{(1)}[\phi]\\
		\leq&\: C 2^{-(1-\epsilon)j}(1+\tau)^{-1+\nu_m+\delta}D^{(1)}[\phi],
		\end{split}
		\end{equation*}
		for suitable $\delta>0$. We conclude the argument by summing over $j$ and combining with the estimate in the region $r\geq 1+(1+\tau)$.
		
		We obtain \eqref{eq:non1degdecay} by interpolating \eqref{eq:nondegedecay} with \eqref{eq:step2edecaylessdegKinv} to get rid of the degenerate factor $(1-r^{-1})^{\delta}$ in \eqref{eq:step2edecaylessdegKinv}.
		\end{proof}
			
			In the proposition below, we show that energies with growing $r$-weights, restricted to far-away regions $r>r_I$ also decay suitably fast, provided we restrict to $|m|\geq 2$. This will be important for deriving the precise late-time asymptotics for $r\phi$ all the way up to $\mathcal{I}^+$.
			
			\begin{proposition}
			\label{prop:mgeq2}
			Let $\eta>0$ be arbitrarily small. Then there exist a constant $C>0$, such that
			\begin{align}
			\label{eq:extradecKderinfty}
			\int_{\Sigma_{\tau} \cap\{r\geq r_I\}} r^{2} | LK\hpsi_{|m|\geq 2}|^2\,d\sigma dr\leq C (1+\tau)^{-3} D_2^{(1)}[\phi],
			\end{align}
			with
			\begin{equation*}
			D_2^{(1)}[\phi]:=\int_{\Sigma_0 \cap\{r\geq r_I\}}r  |L(r^2L)^2 K \hpsi_{|m|\geq 2}|^2\,d\sigma dr+D^{(1)}[\phi]+D^{(1)}[K\phi].
			\end{equation*}
			\end{proposition}
			\begin{proof}
			We apply \eqref{eq:rpinfcommmgeq2} with $p=1$ (and $\psi$ replaced by $\hpsi$), and we estimate the terms on the RHS of \eqref{eq:rpinfcommmgeq2} using \eqref{eq:boxPi2} and \eqref{eq:mornearho}  to conclude that for all $0\leq  \tau_B$
			\begin{equation*}
			\int_{0}^{\tau_B}\int_{\Sigma_{\tau} \cap\{r\geq r_I\}} |L(r^2L)^2 K\hpsi_{|m|\geq 2}|^2\,d\sigma dr\leq C \int_{\Sigma_{0} \cap\{r\geq r_I\}}r  |L(r^2L)^2 K\hpsi_{|m|\geq 2}|^2\,d\sigma dr+C D^{(1)}[\phi].
			\end{equation*}
			By the mean-value theorem, there exists a dyadic sequence such that
			\begin{equation*}
			\int_{\Sigma_{\tau_n} \cap\{r\geq r_I\}} |L(r^2L)^2 K\hpsi_{|m|\geq 2}|^2\,d\sigma dr\leq C(1+\tau_n)^{-1}\left[\int_{\Sigma_0 \cap\{r\geq r_I\}}r  |L(r^2L)^2 K\hpsi_{|m|\geq 2}|^2\,d\sigma dr+D^{(1)}[\phi]\right].
			\end{equation*}
			By a Hardy inequality, this implies in particular that
			\begin{equation}
			\label{eq:intmeddecaymgeq2}
			\int_{\Sigma_{\tau_n} \cap\{r\geq r_I\}}r^2 |L(r^2L) K\hpsi_{|m|\geq 2}|^2\,d\sigma dr\leq C(1+\tau_n)^{-1}\left[\int_{\Sigma_0 \cap\{r\geq r_I\}}r  |L(r^2L)^2 K \hpsi_{|m|\geq 2}|^2\,d\sigma dr+D^{(1)}[\phi]\right].
			\end{equation}
			
			We now apply \eqref{eq:rpinfcomm} with $p=2$ and $0\leq \tau_A\leq \tau_B$ to estimate
			\begin{multline*}
			\int_{\tau_A}^{\tau_B}\int_{\Sigma_{\tau} \cap\{r\geq r_I\}} r|L(r^2L) K\hpsi_{|m|\geq 2}|^2\,d\sigma dr\leq C \int_{\Sigma_{\tau_A} \cap\{r\geq r_I\}}r^2  |L(r^2L) K\hpsi_{|m|\geq 2}|^2\,d\sigma dr\\
			+C\sum_{k\leq 1} \int_{\Sigma_{\tau_A}}\mathcal{E}_{1+\epsilon}[K^{k+1}\hphi]\,d\sigma dr.
			\end{multline*}
			
			By interpolating between \eqref{eq:step3edecaylessdegKKinv} and \eqref{eq:nondegedecay} applied to $K^k\hphi$ instead of $\hphi$, we can estimate for all $\tau \geq 0$
			\begin{equation*}
			\sum_{k\leq 1} \int_{\Sigma_{\tau_A}}\mathcal{E}_{1+\epsilon}[K^{k+1}\hphi]\,d\sigma dr\leq C(1+\tau)^{-3-\nu_m+\eta}(D^{(1)}[\phi]+D^{(1)}[K\phi]).
			\end{equation*}
			
			Hence, after another application of the mean-value theorem along an appropriate dyadic sequence, together with \eqref{eq:intmeddecaymgeq2}, we obtain
			\begin{multline*}
			\int_{\Sigma_{\tau_n} \cap\{r\geq r_I\}}r |L(r^2L) K\hpsi_{|m|\geq 2}|^2\,d\sigma dr\\
			\leq C(1+\tau_n^{(1)})^{-2}\left[\int_{\Sigma_0 \cap\{r\geq r_I\}}r  |L(r^2L)^2 K \hpsi_{|m|\geq 2}|^2\,d\sigma dr+D^{(1)}[\phi]+D^{(1)}[K\phi]\right],
			\end{multline*}
			with $\tau_n^{(1)}$ another dyadic sequence.
			
			We repeat the last step by applying \eqref{eq:rpinfcomm} with $p=1$ to obtain along a dyadic sequence $\{\tau_n^{(2)}\}$:
			\begin{multline*}
			\int_{\Sigma_{\tau_n} \cap\{r\geq r_I\}} |L(r^2L) K\hpsi_{|m|\geq 2}|^2\,d\sigma dr\\
			\leq C(1+\tau_n^{(2)})^{-3}\left[\int_{\Sigma_0 \cap\{r\geq r_I\}}r  |L(r^2L)^2 K \hpsi_{|m|\geq 2}|^2\,d\sigma dr+D^{(1)}[\phi]+D^{(1)}[K\phi]\right]
			\end{multline*}
			
			We then infer that \eqref{eq:extradecKderinfty} holds along the sequence $\{\tau_n^{(2)}\}$ by once more applying a Hardy inequality and applying \eqref{eq:step3edecaylessdegKKinv} to estimate the resulting boundary terms. We then obtain \eqref{eq:extradecKderinfty} for all $\tau\geq 0$ by taking $\tau_n^{(2)}<\tau\leq \tau_n^{(2)}$ and applying \eqref{eq:rpinf} with $p=2$ in the interval $[\tau_n^{(2)}, \tau]$.
			\end{proof}
			
			\begin{remark}
			The restriction $|m|\geq 2$ in Proposition \ref{prop:mgeq2} is in fact not necessary. It is made to ensure a long enough hierarchy of $r$-weighted energy estimates in the region $r>r_I$, obtained by commuting with $(r^2L)^2$. In the $|m|=1$ case, one cannot directly commute with $(r^2L)^2$, as in \eqref{eq:rpinfcommmgeq2}, as this estimate relies on a Poincar\'e inequality on $\s^2$ with constants that need to be suitably large. Instead, it is possible to obtain a suitably long hierarchy of $r$-weighted energy estimates by considering Newman--Penrose quantities; see \cite[\S 5.2]{aagkerr}. 
			
			Note that the expected sharp decay rates for $|m|=1$ would be larger than the decay rates for $|m|\geq 2$, since $\alpha_{m|m|}>0$ for $|m|=1$ and $<0$ for $|m|\geq 2$, so the $|m|=1$ part of $\phi$ would contribute to sub-leading order in the late-time asymptotics of $\phi$. For this reason, we have decided to exclude it from the decay estimates in this section.
			\end{remark}
			
				\begin{proposition}
				\label{prop:improvedenergydecay}
				Let $\eta>0$ be arbitrarily small and let $r_I>1$. Then there exists a constant $C>0$ such that
\begin{align}
\label{eq:rwdecaymgeq2b}
\int_{\Sigma_{\tau}}\mathcal{E}_{0}[\hphi_{|m|\geq 2}]\, r^2d\sigma dr+\int_{\Sigma_{\tau}\cap\{r\geq r_I\}}r^{1+\delta} |L\hphi_{|m|\geq 2}|^2\, d\sigma dr\leq C (1+\tau)^{-2-\frac{\nu_m}{2}+\eta} D_2^{(1)}[\phi].
\end{align}
			\end{proposition}
			\begin{proof}
			Let $\phi$ be restricted to $|m|\geq 2$. We apply \eqref{eq:morneartrapps} and \eqref{eq:rpinfcomm} with $p=2-\delta$ to estimate for all $0\leq \tau_A\leq \tau_B$:
			\begin{multline*}
			\int_{\tau_A}^{\tau_B}\int_{\Sigma_{\tau}} (1-r^{-1})^{\delta} \mathcal{E}_0[\hphi]\,r^2d\sigma dr+\int_{\tau_A}^{\tau_B}\int_{\Sigma_{\tau}\cap\{r\geq r_I\} } r^{1-\delta}|L\hpsi^2|^2d\sigma drd\tau\\
			\leq C\sum_{k\leq 1}\int_{\Sigma_{\tau_A}}\mathcal{E}_{1+\epsilon}[K^k\phi]\,r^2d\sigma dr+C\int_{\Sigma_{\tau}\cap\{r\geq r_I\}}r^{2-\delta} |L \hpsi|^2\,d\sigma dr\\
			+C(1+\tau_A)^{-1-\nu_m}D^{(1)}[\phi].
			\end{multline*}
			
			By the mean-value theorem and application of \eqref{eq:non1degdecay} and \eqref{eq:step3edecayrsqKinv}, we therefore obtain that along a suitable dyadic sequence $\{\tau_n\}$
			\begin{equation*}
			\int_{\Sigma_{\tau_n}} (1-r^{-1})^{\delta} \mathcal{E}_0[\hphi]\,r^2d\sigma dr+\int_{\Sigma_{\tau_n}\cap\{r\geq r_I\} } r^{1-\delta}|L\hpsi|^2d\sigma dr\leq C (1+\tau_n)^{-2-\nu_m+\eta}D^{(1)}[\phi].
			\end{equation*}
			It remains to extend the above estimate to all $\tau\geq 0$. Let $\tau_n\leq \tau\leq \tau_{n+1}$. Then we apply the fundamental theorem of calculus in the $K$-direction to obtain:
			\begin{equation*}
			\begin{split}
		\int_{\Sigma_{\tau}} &(1-r^{-1})^{\delta} \mathcal{E}_0[\hphi]\,r^2d\sigma dr+\int_{\Sigma_{\tau}\cap\{r\geq r_I\} } r^{1-\delta}|L\hpsi|^2d\sigma dr\\
		\leq &\: \int_{\Sigma_{\tau_n}} (1-r^{-1})^{\delta} \mathcal{E}_0[\hphi]\,r^2d\sigma dr+\int_{\Sigma_{\tau_n}\cap\{r\geq r_I\} } r^{1-\delta}|L\hpsi|^2d\sigma dr+\int_{\tau_n}^{\tau}\int_{\Sigma_{\tau}} K((1-r^{-1})^{\delta} \mathcal{E}_0[\hphi])\,r^2\,d\sigma dr d\tau\\
		&\:+\int_{\tau_n}^{\tau} \int_{\Sigma_{\tau}\cap \{r\geq r_I\}}K(r |L\hpsi|^2)\,\,d\sigma dr d\tau\\
		\leq &\:  \int_{\Sigma_{\tau_n}} (1-r^{-1})^{\delta} \mathcal{E}_0[\hphi]\,r^2d\sigma dr+\int_{\Sigma_{\tau_n}\cap\{r\geq r_I\} } r^{1-\delta}|L\hpsi|^2d\sigma dr\\
		&\:+C \int_{\tau_n}^{\tau} \int_{\Sigma_{\tau}}(1-r^{-1})^{\delta} (\tau_n^{-1}\mathcal{E}_0[\hphi]+\tau_n \mathcal{E}_0[K\hphi])\,r^2\,d\sigma dr d\tau\\
		&\:+C\int_{\tau_n}^{\tau} \int_{\Sigma_{\tau}\cap \{r\geq r_I\}}\tau_n^{-1+\frac{\nu_m}{2}}r^{1-\delta} |L\hpsi|^2+\tau_n^{1-\frac{\nu_m}{2}} r^{1-\delta} |LK\hpsi|^2\,\,d\sigma dr d\tau.
		\end{split}
			\end{equation*}
			We estimate the spacetime integrals on the very RHS above via \eqref{eq:morneartrapps} and \eqref{eq:rpinfcomm} with $p=2-\delta$ by applying additionally \eqref{eq:extradecKderinfty} and \eqref{eq:step3edecaylessdegKKinv} to the terms involving $K$-derivatives:
			\begin{equation*}
			\begin{split}
		\int_{\Sigma_{\tau}} (1-r^{-1})^{\delta} \mathcal{E}_0[\hphi]\,r^2d\sigma dr+\int_{\Sigma_{\tau}\cap\{r\geq r_I\} } r^{1-\delta}|L\hpsi|^2d\sigma dr\leq C(1+\tau)^{-2-\frac{\nu_m}{2}+\eta}(D^{(1)}[\phi]+D^{(1)}[K\phi]).
		\end{split}
\end{equation*}
			We finally interpolate the above estimate with \eqref{eq:step3edecayrsqKinv} and \eqref{eq:step2edecaylessdegKinv}.
			\end{proof}

			\subsection{Pointwise decay estimates}
			
			We can apply the energy decay estimates from Propositions \ref{prop:nondegedecay} and \ref{prop:improvedenergydecay} to obtain pointwise estimates for $\hphi$. We will assume in this section that $|m|\geq 2$.
						\begin{proposition}
						\label{prop:pointwest}
			Let $\eta>0$ be arbitrarily small and $r_I>1$. Then there exist $\epsilon>0$ suitably small, such that for $\delta,\eta>0$ arbitrarily small
			\begin{align}
			\label{eq:pointw1}
			||r^{\frac{1}{2}}(\phi-\Pi)||_{L^{\infty}(\Sigma_{\tau})}= &\:||r^{\frac{1}{2}}\hphi||_{L^{\infty}(\Sigma_{\tau})}\leq C (1+\tau)^{-\frac{1}{2}-\frac{\nu_m}{2}+\eta}\sqrt{\sum_{2k+l\leq 2}D^{(1)}[Q^lK^k\phi]},\\
			\label{eq:pointw2}
			||(1-r^{-1})^{\frac{1}{2}}(\phi-\Pi)||_{L^{\infty}(\Sigma_{\tau})}= &\:||(1-r^{-1})^{\frac{1}{2}}\hphi||_{L^{\infty}(\Sigma_{\tau}\cap\{r\geq r_I\})}\\ \nonumber
			\leq&\: C(1+\tau)^{-1-\frac{\nu_m}{4}+\eta}\sqrt{\sum_{2k+l\leq 2}D_2^{(1)}[Q^lK^k\phi]},\\
			\label{eq:pointw3}
			||r(\phi-\Pi)||_{L^{\infty}(\Sigma_{\tau}\cap\{r\geq r_I\})}= &\:||r\hphi||_{L^{\infty}(\Sigma_{\tau}\cap\{r\leq r_I\})}\\ \nonumber
			\leq&\: C(1+\tau)^{-1-\frac{\nu_m}{4}+\eta}\sqrt{\sum_{2k+l\leq 2}D_2^{(1)}[Q^lK^k\phi]},
			\end{align}
			with $Q$ the Carter operator.
			\end{proposition}
			\begin{proof}
			We use that $\lim_{r\to \infty} r\phi|_{\Sigma_{\tau}}(r,\theta,\tphi)=0$ and apply the fundamental theorem of calculus and Cauchy--Schwarz to estimate
			\begin{equation*}
			|\hphi|(\tau,r,\theta,\tphi)\leq \int_r^{\infty} |\tX \hphi |(\tau,r,\theta,\tphi)\,dr'\leq \sqrt{\int_{r}^{\infty}r'^{-2}(1-r'^{-1})^{-1+\epsilon}\,dr'}\sqrt{\int_r^{\infty} r^2(1-r^{-1})^{1-\epsilon}|\tX \hphi|^2\,dr'}.
			\end{equation*}
			Hence,
			\begin{equation*}
			\int_{\s^2_{\theta,\tphi}} |\hphi|^2(\tau,r,\theta,\tphi)\,d\sigma\leq C r^{-1}\int_{\Sigma_{\tau}} \mathcal{E}_{1+\epsilon}[\hphi]\,r^2d\sigma dr.
			\end{equation*}
			We conclude \eqref{eq:pointw1} by applying \eqref{eq:step2edecaylessdegKinv} and a Sobolev inequality on $\s^2$ with the fact that second-order derivatives on $\s^2$ are controlled by $|Q\hphi|^2+|K\hphi|^2+|\Phi^2 \hphi|^2$. Since $Q,K,\Phi$ are Killing operators, we can repeat all estimates with $\hphi$ replaced by $Q\phi,\Phi^2\phi, K\phi$ to obtain additionally control over second-order derivatives on $\s^2$.
			
			We apply the fundamental theorem of calculus again, together with Young's inequality, to estimate
			\begin{equation*}
			\begin{split}
			r^{-2+\epsilon} &(1-r^{-1})^{1+\epsilon} |r\hphi|^2(\tau,r,\theta,\tphi)=\int_{1}^{r}\tX(r^{-2+\epsilon}(1-r^{-1})^{1+\epsilon} |r\hphi|^2)(\tau,r',\theta,\tphi)\,dr'\\
			\leq &\: \int_{1}^{r}(2 r^{-2+\epsilon} (1-r^{-1})^{1+\epsilon} |r\hphi ||\tX(r\hphi)|+(1+\epsilon)r^{-4+\epsilon}(1-r^{-1})^{\epsilon}|\hphi|^2)(\tau,r',\theta,\tphi)\,dr'\\
			\leq &\: C\int_r^{\infty}( r^{-3+\epsilon}(1-r^{-1})^{\epsilon}|r\hphi |^2+r^{-1+\epsilon}(1-r^{-1})^{2+\epsilon}|\tX(r\hphi)|^2)(\tau,r',\theta,\tphi)\,dr'.
			\end{split}
			\end{equation*}
			Hence,
			\begin{equation*}
			\int_{\s^2_{\theta,\tphi}} r^{-2+\epsilon}(1-r^{-1})^{1+\epsilon}|\hphi|^2(\tau,r,\theta,\tphi)\,d\sigma\leq C \int_{\Sigma_{\tau}} (1-r^{-1})^{\epsilon}r^{-1+\epsilon}\left[(1-r^{-1})^{2}|\tX (r\hphi)|^2+|\snabla_{\s^2}\hphi|^2\right]\, d\sigma dr.
			\end{equation*}

			We derive \eqref{eq:pointw2} by proceeding as in the derivation of \eqref{eq:pointw1}, but distributing the weights in $(1-r^{-1})$ differently, so that we see an integral of $\mathcal{E}_0[\hphi]$ instead of  $\mathcal{E}_{1+\epsilon}[\hphi]$:
			\begin{equation*}
			|\hphi|(r,\theta,\tphi)\leq \int_r^{\infty} |\tX \hphi |(\tau,r,\theta,\tphi)\,dr'\leq \sqrt{\int_{r}^{\infty}r'^{-2}(1-r'^{-1})^{-2}\,dr'}\sqrt{\int_r^{\infty} r^2(1-r^{-1})^{2}|\tX \hphi|^2\,dr'}.
			\end{equation*}
			Hence,
			\begin{equation*}
			\int_{\s^2_{\theta,\tphi}} |\hphi|^2(\tau,r,\theta,\tphi)\,d\sigma\leq C (r-1)^{-1}\int_{\Sigma_{\tau}} \mathcal{E}_{0}[\hphi]\,r^2d\sigma dr.
			\end{equation*}
			Then \eqref{eq:pointw2} follows after applying \eqref{eq:rwdecaymgeq2b} together with a Sobolev inequality on $\s^2$ as above.
			
			We can improve \eqref{eq:pointw2} with a growing $r$-weight on the LHS by applying the fundamental theorem of calculus to $|\hpsi|^2$ (with respect to $X$ instead of $\tX$):
			\begin{equation*}
			\begin{split}
			|\hpsi|^2(\tau,r,\theta,\varphi_*)\leq &\:|\hpsi|^2(r_I,\theta,\varphi_*)+2\int_R^{r}|\hpsi| |X \hpsi |(\tau,r,\theta,\varphi_*)\,dr'\\
			\leq &\:|\hpsi|^2(\tau,r_I,\theta,\varphi_*)+\int_R^{\infty}r^{-1-\delta}|\hpsi|^2(\tau,r,\theta,\varphi_*)\,dr'+\int_R^{\infty}r^{1+\delta}|X\hpsi|^2(\tau,r,\theta,\varphi_*)\,dr'\\
			\leq &\:|\hpsi|^2(\tau,r_I,\theta,\varphi_*)+C\int_R^{\infty}r^{1+\delta}|X\hpsi|^2(\tau,r,\theta,\varphi_*)\,dr'\\
			\leq &\:|\hpsi|^2(\tau,r_I,\theta,\varphi_*)+C\int_R^{\infty}r^{1+\delta}|X\hpsi|^2(\tau,r,\theta,\varphi_*)\,dr'\\
			\leq &\:|\hpsi|^2(\tau,r_I,\theta,\varphi_*)+C\int_R^{\infty}(r^{1+\delta}|L\hpsi|^2+r^{-1+\delta} |K\phi|^2+r^{-1+\delta}|\Phi \phi|^2 )(\tau,r,\theta,\varphi_*)\,dr'
			\end{split}
			\end{equation*}
			Hence, after integrating over $\s^2$, we can estimate
			
			\begin{equation*}
			\int_{\s^2_{\theta,\phi_*}} |\hpsi|^2(\tau,r,\theta,\tphi)\,d\sigma\leq r_I^2\int_{\s^2_{\theta,\phi_*}} |\hphi|^2(\tau,r_I,\theta,\tphi)\,d\sigma+  C\int_{\Sigma_{\tau}\cap\{r\geq r_I\}} r^2\mathcal{E}_{0}[\hphi]+r^{1+\delta}|L\hpsi|^2\,d\sigma dr.
			\end{equation*}
			The first term on the RHS above can be estimated by \eqref{eq:pointw2} and the second term by \eqref{eq:rwdecaymgeq2b}. We conclude \eqref{eq:pointw3} by applying a Sobolev inequality on $\s^2$ again. 
					\end{proof}
			
			We now rewrite the $L^{\infty}$ estimates in Proposition \ref{prop:pointwest} to obtain global late-time asymptotics, which is the main result of this section.
			\begin{corollary}
			\label{cor:latetimeasym}
			Let $|m|\geq 2$ and assume that $\nu_m>0$. Let $0<\eta<\frac{\nu_m}{4}$ be arbitrarily small. Then there exists a constant $C_m>0$ such that
			\begin{multline}
			\label{eq:tail}
			\left|\phi_m-\Pi \right|(\tau,r,\theta,\tphi)\\
			\leq C_mr^{-1}((\tau+1)(1-r^{-1})+1)^{-\frac{1}{2}}(\tau+1)^{-\frac{1}{2}-\frac{\nu_m}{4}+\eta}\sqrt{\sum_{2k+l\leq 2}D_2^{(1)}[Q^lK^k\phi]}.
			\end{multline}
			\end{corollary}
			
			\section{Azimuthal instabilities}
			\label{sec:instab}
	In this section, we will use the decay results for $\hphi$ from \S \ref{sec:decay} to derive \emph{non-decay} and instability results. As we are proving lower bounds for non-negative definite quantities, we can restrict without loss of generality to a single (or bounded set of) $\phi_m$. We will take $|m|\geq 2$ so that by Lemma \ref{lm:negalpha} $\alpha_{ m |m|}<0$.
	
	Throughout this section we will assume that either \eqref{eq:asmsphev} holds or that we restrict to initial data satisfying \eqref{eq:restrinitialdataalpha0}.
	
	We first demonstrate the manifestation of an instability in the form of non-decay at the level of the non-degenerate energy $\int_{\Sigma_{\tau}}\mathcal{E}_2[\phi]\,r^2d\sigma dr$, with $\phi$ a solution to the homogeneous wave equation $\square_g\phi=0$.

	We define
	\begin{align*}
	\mathfrak{h}_m:=\sqrt{\sum_{\substack{|m|\leq \ell\leq l_m\\ \alpha_{m\ell}<0}}|{\mathfrak{h}}_{m \ell}|^2}
		\end{align*}

	\begin{proposition}
	\label{eq:energyinstab}
	Let $\eta>0$ be arbitrarily small. Then there exists constants $c_m,C_m>0$ such that
	\begin{equation*}
	\int_{\Sigma_{\tau}}\mathcal{E}_2[\phi_m]\,r^2d\sigma dr\geq c_m\mathfrak{h}_m^2-C_m(1+\tau)^{-\nu_m+\eta}D^{(1)}[\phi_m].
	\end{equation*}
	In particular, for each $m\in \Z$, there exists a $\tau^*_m\geq 0$ suitably large (depending on the initial data) such that for all $\tau\geq \tau^*_m$
	\begin{equation*}
	\int_{\Sigma_{\tau}}\mathcal{E}_2[\phi]\,r^2d\sigma dr\geq \frac{c_m}{2}\mathfrak{h}_m^2.
	\end{equation*}
	\end{proposition}
	\begin{proof}
	We have that
	\begin{equation*}
	\int_{\Sigma_{\tau}}\mathcal{E}_2[\phi]\,r^2d\sigma dr=\int_{\Sigma_{\tau}}\mathcal{E}_2[\hphi+\Pi]\,r^2d\sigma dr\geq (1-\epsilon)\int_{\Sigma_{\tau}}\mathcal{E}_2[\Pi]\,r^2d\sigma dr-C\int_{\Sigma_{\tau}}\mathcal{E}_2[\hphi]\,r^2d\sigma dr.
	\end{equation*}
	By Proposition \ref{prop:nondegedecay}, the second term on the very RHS above decays appropriately in $\tau$, so it remains to show that
	\begin{equation*}
	\int_{\Sigma_{\tau}}\mathcal{E}_2[\Pi]\,r^2d\sigma dr\geq c|{\mathfrak{h}}_{m \ell}|^2-C(1+\tau)^{-\nu_m+\eta}\sum_{|m|\leq \ell\leq l_m}|{\mathfrak{h}}_{m \ell}|^2.
	\end{equation*}
	
	We have that
	\begin{equation*}
	\int_{\Sigma_{\tau}}\mathcal{E}_2[\Pi]\,r^2d\sigma dr\sim \int_{\Sigma_{\tau}}\left(|\tX \Pi|^2+|\snabla_{\s^2} \Pi|^2+|K\Pi|^2\right)\,r^2d\sigma dr.
	\end{equation*}
	It follows immediately that the integral of the terms $|\snabla_{\s^2} \Pi|^2+|K\Pi|^2$ as well as the integral of $|\tX \Pi|^2$ restricted to $r\leq r_H$, with $r_H$ arbitrarily small, decays with a rate faster than $(1+\tau)^{-1}$. It therefore remains to derive a uniform lower bound for the integral:
	\begin{equation*}
	\int_{\Sigma_{\tau}\cap \{r\leq r_H\}} |\tX \Pi|^2\,d\sigma dr=\sum_{\substack{|m|\leq \ell \leq l_m\\ \alpha_{m\ell}\neq 0}} \int_{1}^{r_H}|\tX \Pi_{m \ell}|^2\,dr.
	\end{equation*}
	Let $\alpha_{m \ell}<0$. We consider without loss of generality the case $\widetilde{B}_-\neq 0$ and $\widetilde{B}_+\neq 0$, in which we can appeal to the decomposition of $\Pi_{m \ell}$ in \eqref{eq:usefulidPi1} in the case $\alpha_{m \ell}<0$:
	\begin{align*}
\Pi_{m \ell}(\tau,r)=&\:b_{+,m \ell}{\mathfrak{h}}_{m \ell}w_{+, m\ell}(r)f_+^{m \ell}(\tau,r)+b_{-,m \ell}{\mathfrak{h}}_{m \ell}w_{-, m\ell}(r)f_-^{m \ell}(\tau,r)+\Pi_{m \ell,{\rm exp}}(\tau,r),\\
\Pi_{m \ell,{\rm exp}}(\tau,r)=&\:{\mathfrak{h}}_{m \ell}\widetilde{w}_{\infty, m\ell}(r)\left[E_+(f_{+,{\rm exp}}^{m \ell}- f_+^{m \ell})(\tau,r)+E_-(f_{-,{\rm exp}}^{m \ell}- f_-^{m \ell})(\tau,r)\right].
\end{align*}
The cases $\widetilde{B}_-=0$ or $\widetilde{B}_+\neq 0$ are strictly easier as the expressions \eqref{eq:exprPi2} and \eqref{eq:exprPi2} immediately give the right decomposition, with no need for $\Pi_{m \ell,{\rm exp}}$.

By Lemma \ref{lm:fexp}, it follows that $\tX \Pi_{m \ell,{\rm exp}}$ decays at least as fast as $(1+\tau)^{-1}$. We now compute:
\begin{equation*}
\begin{split}
\tX(w_{m \ell, \pm } f_{\pm}^{m \ell})=&\:\tX\left((1+O(r-1))\left((\tau+1)(r-1)+4\right)^{-\frac{1}{2}\mp \sqrt{\alpha_{m \ell}}+im}\right)(\tau+1)^{-\frac{1}{2}\mp \sqrt{\alpha_{m \ell}}-im },\\
=&\: \left(-\frac{1}{2}\mp \sqrt{\alpha_{m \ell}}+im\right)((\tau+1)(r-1)+4)^{-\frac{3}{2}\mp \sqrt{\alpha_{m \ell}}}(\tau+1)^{\frac{1}{2}\mp \sqrt{\alpha_{m \ell}}-im}+O((\tau+1)^{-\frac{1}{2}}).
\end{split}
\end{equation*}
Hence, for $c_{\pm, m \ell}=b_{\pm, m \ell} \left(-\frac{1}{2}\mp \sqrt{\alpha_{m \ell}}+im\right)$,
\begin{equation*}
\begin{split}
&|\tX \Pi_{m \ell}|^2\\
=&\: |{\mathfrak{h}}_{m \ell}|^2|c_{+,m \ell}((\tau+1)(r-1)+4)^{+ \sqrt{\alpha_{m \ell}}}+c_{-,m \ell}((\tau+1)(r-1)+4)^{- \sqrt{\alpha_{m \ell}}}|^2((\tau+1)(r-1)+4)^{-3}(\tau+1)\\
&+\ldots\\
=&|{\mathfrak{h}}_{m \ell}|^2(|c_{+,m \ell}|^2+|c_{-,m \ell}|^2+2\Re(c_{+,m \ell}\overline{c_{-,m \ell}}((\tau+1)(r-1)+4)^{- \sqrt{\alpha_{m \ell}}})((\tau+1)(r-1)+4)^{-3}(\tau+1)\\
&+\ldots\\
=&-\frac{|{\mathfrak{h}}_{m \ell}|^2}{2} \frac{d}{dr}\left[|c_{+,m \ell}|^2+|c_{-,m \ell}|^2)((\tau+1)(r-1)+4)^{-2}\right]+\ldots\\
&\:-2|{\mathfrak{h}}_{m \ell}|^2\frac{d}{dr}\Re\left[(2+\sqrt{\alpha_{m \ell}})^{-1} c_{+,m \ell}\overline{c_{-,m \ell}}((\tau+1)(r-1)+4)^{- \sqrt{\alpha_{m \ell}} -2}\right]+\ldots,
\end{split}
\end{equation*}
where $\ldots$ 	indicates terms that decay at least as fast as $(1+\tau)^{-1}$. Integrating the right-hand side in $r\in[1,r_H]$ results in a constant that we can estimate further with a Young's inequality
	\begin{equation*}
\begin{split}
	\int_1^{r_H}|\tX \Pi_{m \ell}|^2\,dr'=&\: \frac{1}{32}|{\mathfrak{h}}_{m \ell}|^2(|c_{+,m \ell}|^2+|c_{-,m \ell}|^2)+\frac{1}{16}|{\mathfrak{h}}_{m \ell}|^2\Re\left((2+\sqrt{\alpha_{m \ell}})^{-1} c_{+,m \ell}\overline{c_{-,m \ell}}e^{-\log 4\sqrt{\alpha_{m \ell}}}\right)\\
	+&\:|{\mathfrak{h}}_{m \ell}|^2O(r_H-1)+\ldots\\
	\geq&\: \frac{1}{32}|{\mathfrak{h}}_{m \ell}|^2\left(1-\frac{1}{\sqrt{4+|\alpha_{m \ell}|}}\right)(|c_{+,m \ell}|^2+|c_{-,m \ell}|^2)+O(r_H-1)+\ldots\\
	\geq&\: \frac{1}{64}|{\mathfrak{h}}_{m \ell}|^2(|c_{+,m \ell}|^2+|c_{-,m \ell}|^2)
	\end{split}
	\end{equation*}
	for $r_H-1$ suitably small.
	\end{proof}
	
	In the next proposition, we will show that the late-time asymptotics in Corollary \ref{cor:latetimeasym} imply blowup of $\int_{\s^2}|Y\phi_m|^2\,d\sigma$ along $\mathcal{H}^+$.
	
	\begin{proposition}
	\label{prop:pointwinstab}
	 There exists a constant $c_m>0$ depending only on $m$ and a monotonically increasing sequence $\{\tau_n\}$ in $ \R_+$ with $\tau_n\to\infty$ along which
	\begin{equation*}
	\int_{\Sigma_{\tau_n}\cap \mathcal{H}^+}|Y \phi_m|^2\,d\sigma\geq  c_m \mathfrak{h}_m^2(1+\tau_n).
	\end{equation*}
	In particular, for any $m\in \Z\setminus\{-1,0,1,\}$, we can estimate
	\begin{equation*}
	4\pi||Y\phi||_{L^{\infty}(\Sigma_{\tau_n}\cap \mathcal{H}^+)}\geq \int_{\Sigma_{\tau_n}\cap \mathcal{H}^+}|Y \phi|^2\,d\sigma\geq  c_m \mathfrak{h}_m^2(1+\tau_n).
	\end{equation*}
	\end{proposition}
	\begin{proof}
	We evaluate \eqref{eq:waveeqK} with $G=0$ along $\mathcal{H}+$ and rearrange terms to obtain:
	\begin{equation*}
K(-4Y \phi)=+(\slashed{\mathcal{D}}- \Phi^2)\phi- \Phi \phi+\sin^2\theta (K^2\phi- K\Phi\phi)+2 K\Phi\phi+2K\phi.
\end{equation*}
Projecting to the angular functions $S_{m \ell}(\theta)e^{im\tphi}$, we obtain:
\begin{equation}
\label{eq:projeqinstab}
K(-4Y \phi_{m \ell})=(\Lambda_{m\ell} +m^2-im)\phi_{m \ell}+(\sin^2\theta (K^2\phi- im K\phi))_{m \ell}+2(1+ im) K\phi_{m \ell}.
\end{equation}
By Corollary \ref{cor:latetimeasym} (applied also with $\phi$ replaced by $K\phi$ and $K^2\phi$), we have that
\begin{equation}
\label{eq:inteqinstab}
\left|K(-4Y \phi_{m \ell})-(\Lambda_{m\ell} +m^2-im)\Pi_{m \ell}(\tau,1)\right|\leq O((1+\tau)^{-\frac{1+\nu_m}{2}}).
\end{equation}
Without loss of generality, we assume $\widetilde{B}_-\neq 0$ and $\widetilde{B}_+\neq 0$. The cases $\widetilde{B}_-=0$ or $\widetilde{B}_+= 0$ can be treated with the same argument and are strictly easier to handle.

We integrate both sides of \eqref{eq:inteqinstab} in $\tau$ to obtain
\begin{multline*}
-4Y \phi_{m \ell}(\tau,1)+Y \phi_{m \ell}(0,1)\\
=\frac{1}{2}(\Lambda_{m\ell} +m^2-im){\mathfrak{h}}_{m\ell}\int_0^{\tau }4^{- \sqrt{\alpha_{m \ell}}+im}b_{+,m \ell}(1+\tau)^{-\frac{1}{2}-\sqrt{\alpha_{m \ell}}-im}+4^{ \sqrt{\alpha_{m \ell}} +im}b_{-,m \ell}(1+\tau)^{-\frac{1}{2}+\sqrt{\alpha_{m \ell}}-im}\,d\tau \\
 +O((1+\tau)^{\frac{1-\nu_m}{2}}).
\end{multline*}
Recall that
\begin{equation*}
\alpha_{m\ell}=\Lambda_{m\ell}-2m^2+\frac{1}{4}.
\end{equation*}
We can therefore write
\begin{equation*}
\Lambda_{m\ell}+m^2-im=\alpha_{\ell}-m^2-\frac{1}{4}=\left(\pm \sqrt{\alpha_{\ell}}+im-\frac{1}{2}\right)\left(\pm \sqrt{\alpha_{\ell}}-im+\frac{1}{2}\right)
\end{equation*}
Hence,
\begin{equation}
\label{eq:tailYphilm}
\begin{split}
Y \phi_{m \ell}(\tau,1)=&\:-\frac{{\mathfrak{h}}_{m\ell}}{8}\left(- \sqrt{\alpha_{\ell}}+im-\frac{1}{2}\right) e^{-2\log 2 ( \sqrt{\alpha_{m \ell}}+im)}b_{+,m \ell} (1+\tau)^{\frac{1}{2}-\sqrt{\alpha_{m \ell}} -im}\\
&\: -\frac{{\mathfrak{h}}_{m\ell}}{8}\left(+ \sqrt{\alpha_{\ell}}+im-\frac{1}{2}\right) e^{2\log 2  (\sqrt{\alpha_{m \ell}}-im)}b_{-,m \ell} (1+\tau)^{\frac{1}{2}+\sqrt{\alpha_{m \ell}} -im}\\
&\: + O((1+\tau)^{\frac{1-\nu_m}{2}}).
	\end{split}
\end{equation}
Denote
\begin{align*}
\tilde{b}_{+,m \ell}=&\:\frac{1}{8}\left(- \sqrt{\alpha_{\ell}}+im-\frac{1}{2}\right)b_{+,m \ell} ,\\
\tilde{b}_{-,m \ell}=&\:\frac{1}{8}\left( \sqrt{\alpha_{\ell}}+im-\frac{1}{2}\right)b_{-,m \ell}.
\end{align*}
We therefore have that
\begin{equation*}
|Y \phi_{m \ell}|^2(\tau,1)=|{\mathfrak{h}}_{m\ell}|^2|\tilde{b}_{+,m \ell}+\tilde{b}_{-,m \ell}e^{i\sqrt{|\alpha_{m \ell}|}(\log (1+\tau)+4\log 2)-2im}|^2(1+\tau)+O((1+\tau)^{1-\nu_m}).
\end{equation*}
If $\tilde{b}_{+,m \ell}\neq 0$, it therefore follows that there exists a sequence of times $\{\tau_n\}$ along which 
\begin{equation*}
|Y \phi_{m \ell}|^2(\tau_n,1)=|{\mathfrak{h}}_{m\ell}|^2|\tilde{b}_{+,m \ell}|^2(1+\tau_n)+O((1+\tau_n)^{1-\nu_m}).
\end{equation*}
If $\tilde{b}_{+,m \ell}=0$, we instead obtain
\begin{equation*}
|Y \phi_{m \ell}|^2(\tau,1)=|{\mathfrak{h}}_{m\ell}|^2|\tilde{b}_{-,m \ell}|^2(1+\tau)+O((1+\tau)^{1-\nu_m}).
\end{equation*}
for all $\tau\geq 0$.

By taking the sum over $\ell$, we therefore have that there exists a constant $c_m>0$ and a sequence of times $\{\tau_n\}$ along which:
\begin{equation*}
\int_{\s^2}|Y\phi_m|^2(\tau_n,1,\theta,\tphi)\,d\sigma\geq c_m\sum_{\substack{|m|\leq \ell \leq l_m \\ \alpha_{m\ell}<0}}|{\mathfrak{h}}_{m\ell}|^2(1+\tau_n).
\end{equation*}
Considering the case $m=2$, we then obtain the existence of a numerical constant $c>0$ such that:
\begin{equation*}
\int_{\Sigma_{\tau_n}\cap \mathcal{H}^+}|Y \phi|^2\,d\sigma=\sum_{m\in \Z}\int_{\s^2}|Y\phi_m|^2(\tau_n,1,\theta,\tphi)\,d\sigma\geq c |H_{2|2|}|^2(1+\tau_n). \qedhere
\end{equation*}
	\end{proof}

	We obtain similarly the existence of stronger instabilities for higher-order $Y$ derivatives of $\phi$.
	\begin{proposition}
	\label{prop:hopointwinstab}
	Let $k\in \N$. There exists a constant $c_{m,k}>0$ depending only on $m$ and $k$ and a monotonically increasing sequence $\{\tau_n\}$ in $ \R_+$ with $\tau_n\to\infty$ along which
	\begin{equation*}
	\int_{\Sigma_{\tau_n}\cap \mathcal{H}^+}|Y^k \phi_m|^2\,d\sigma\geq  c_{m,k} \mathfrak{h}_m^2(1+\tau_n)^{1+2k}.
	\end{equation*}
	In particular, for any $m\in \Z\setminus\{-1,0,1,\}$, we can estimate
	\begin{equation*}
	4\pi||Y\phi||^2_{L^{\infty}(\Sigma_{\tau_n} \cap \mathcal{H}^+)}\int_{\Sigma_{\tau_n}\cap \mathcal{H}^+}|Y^k \phi|^2\,d\sigma\geq  c_{m,k}\mathfrak{h}_m^2 (1+\tau_n)^{1+2k}.
	\end{equation*}
	\end{proposition}
	\begin{proof}
	Suppose $k=1$. We first evaluate \eqref{eq:waveeqK} with $G=-\square_g\Pi$ along $\mathcal{H}+$ and rearrange terms to obtain:
	\begin{equation*}
K(-4Y \hphi)=+(\slashed{\mathcal{D}}- \Phi^2)\hphi- \Phi \hphi+\sin^2\theta (K^2\hphi- K\Phi\hphi)+2 K\Phi\hphi+2K\hphi+\square_g\Pi.
\end{equation*}
We projecting the above equation the angular functions $S_{m \ell}(\theta)e^{im\tphi}$. After applying Corollary \ref{cor:latetimeasym} (applied also with $\phi$ replaced by $K\phi$ and $K^2\phi$) together with \eqref{eq:boxPi1}, we obtain:
\begin{equation*}
|K^j(Y \hphi_{m \ell})|\leq C (1+\tau)^{-\frac{1}{2}-\nu_m}\sum_{\substack{2k+l\leq 4+2j\\ l\leq 1}}D^{(1)}[Q^kK^l\phi].
\end{equation*}
Hence, after integrating and filling in the definition of $\Pi_{m \ell}$, we obtain:
\begin{equation}
\label{eq:Yhphidecayhor}
|K^j(Y \phi_{m \ell})|\leq C (1+\tau)^{\frac{1}{2}-\nu_m}\sum_{\substack{2k+l\leq 4+2j\\ l\leq 1}}D^{(1)}[Q^kK^l\phi].
\end{equation}

Subsequently, we act with $Y$ on both sides of \eqref{eq:waveeqK}, project onto $S_{m \ell}(\theta)e^{im\tphi}$ and evaluate along $\mathcal{H}^+$ to obtain schematically:
\begin{equation*}
K(Y^2\phi_m)=(\Lambda_{m\ell} +m^2-im)Y\phi_{m \ell}+O( (1+\tau)^{\frac{1}{2}-\nu_m+\eta}).
\end{equation*}
Now we use \eqref{eq:tailYphilm} to estimate further the right-hand side and we integrate both sides in $\tau$ to obtain the late-time asymptotics of $Y^2\phi_m$ and consequently, desired lower bound.	 This argument, when applied to $K^jY^2\hphi_m$  moreover gives $K^j(Y^2 \hphi_{m \ell})=O( (1+\tau)^{\frac{3}{2}-\nu_m+\eta})$.
	
Then $k>1$ case can be treated inductively, by repeating the above argument, where we start with with late-time asymptotics of $Y^l\phi$ and $K^j(Y^{l} \hphi_{m \ell})=O( (1+\tau)^{\frac{1}{2}+l-\nu_m+\eta})$ for $l\leq k$.
	\end{proof}
		
		\appendix
		\section{Estimates for $\square_g \Pi$}
		We consider the functions $f^{m \ell}_+$, $f^{m \ell}_-$, $f^{m \ell}_{+,{\rm exp}}$ and $f^{m \ell}_{-,{\rm exp}}$, defined in \S \ref{sec:constKinv}.
		\begin{lemma}
		\label{lm:fexp}
		Assume that $\alpha_{m \ell}<0$.
	\begin{enumerate}[label=\emph{(\roman*)}]
		\item Let $r\leq 2$ and $K,L\in \N_0$. Then there exists a constant $C<0$ such that for all $0\leq k\leq K$:
		\begin{equation*}
		|((r-1)\tX)^k K^lf_{\pm}^{m \ell}-((r-1)\tX)^k K^l f_{\pm}^{m \ell,{\rm exp}}|(\tau,r)\leq C|m|^{k+l+1}(r-1)(\tau+4(r-1)^{-1}+1)^{-\frac{3}{2}}(\tau+1)^{-\frac{1}{2}-l},\\
		\end{equation*}
		\item Let $r\geq 2$ and $K,L\in \N_0$. Then there exists a $\delta>0$ such that for all $0\leq k\leq K$:
		\begin{equation*}
		|\tX^k K^{l} f_{\pm}^{m \ell,{\rm exp}}|(\tau,r)\leq C|m|^{k+l}e^{-\delta r}(\tau+1)^{-1+\delta-l}.
		\end{equation*}
		\end{enumerate}
		\end{lemma}
		\begin{proof}
		We consider first (i) with $K=L=0$. Note that
		\begin{equation*}
		\left(\tau+\frac{4}{r-1}+e^{r-1}\right)^{-\frac{1}{2}\mp \sqrt{\alpha_{m \ell}}+im}-\left(\tau+\frac{4}{r-1}+1\right)^{-\frac{1}{2}\mp \sqrt{\alpha_{m \ell}}+im}=h(e^{r-1})-h(1),
		\end{equation*}
		with $h: [0,\infty)\to \R$ defined as follows:
		\begin{equation*}
		h(x)=\left(\tau+\frac{4}{r-1}+x\right)^{-\frac{1}{2}\mp \sqrt{\alpha_{m \ell}}+im}.
		\end{equation*}
		By the mean-value theorem applied to $h$, we have that for $r\leq 2$
		\begin{equation*}
		|h(e^{r-1})-h(1)|\leq (e^{r-1}-1)\sup_{1\leq x\leq e^1}|h'(x)|\leq C|m|(r-1)\left(\tau+\frac{4}{r-1}+1\right)^{-\frac{1}{2}}.
		\end{equation*}
		Taking $L>0$, we can repeat the argument above in a straightforward manner.
		
		Denote for the sake of convenience $a=-\frac{1}{2}\mp \sqrt{\alpha_{m \ell}}+im$. In the case $K=1$, we observe that
		\begin{multline*}
		(r-1)\tX\left[\left(\tau+\frac{4}{r-1}+e^{r-1}\right)^{-\frac{1}{2}\mp \sqrt{\alpha_{m \ell}}+im}-\left(\tau+\frac{4}{r-1}+1\right)^{-\frac{1}{2}\mp \sqrt{\alpha_{m \ell}}+im}\right]\\
		=h_1(e^{r-1})-h_1(1)+a(r-1)\left(\tau+\frac{4}{r-1}+e^{r-1}\right)^{-\frac{3}{2}\mp \sqrt{\alpha_{m \ell}}+im},
		\end{multline*}
		with
		\begin{equation*}
		h_1(x)=-4a(r-1)^{-1}\left(\tau+\frac{4}{r-1}+x\right)^{-\frac{3}{2}\mp \sqrt{\alpha_{m \ell}}+im}.
		\end{equation*}
		By applying the mean-value theorem, it follows that
		\begin{equation*}
		|h_1(e^{r-1})-h_1(1)|\leq C|m|^2(r-1)\left(\tau+\frac{4}{r-1}+1\right)^{-\frac{1}{2}}.
		\end{equation*}
		The general $K$ case follows analogously.
		
		We now consider (ii) with $K=L=0$. We first split:
		\begin{equation*}
		|f_{\pm,{\rm exp}}^{m \ell}|(\tau,r):=\left(\tau+\frac{4}{r-1}+e^{r-1}\right)^{-\delta}\left(\tau+\frac{4}{r-1}+e^{r-1}\right)^{\delta-\frac{1}{2}}(\tau+1)^{-\frac{1}{2} }
		\end{equation*}
Then it follows immediately that
\begin{equation*}
|f_{\pm,{\rm exp}}^{m \ell}|(\tau,r)\leq e^{–\delta r}(\tau+1)^{-1+\delta }.
\end{equation*}
		The $K>0$ and $L>0$ case can be estimated via an analogous splitting.
			\end{proof}
		\begin{proposition}
		\label{prop:squarePial0}
		Let $K,L\in \N_0$. Let $\alpha_{m \ell}<0$ and $r\leq 2$. Then there exists a $C=C(K,N,m)>0$, such that for all $k\leq K$ and $l\leq L$
		\begin{equation*}
		|((r-1)\tX)^kK^l(\rho^2\square_g\Pi))_{m \ell}|(r)\leq C\left((\tau+1)(r-1)+4\right)^{-\frac{1}{2}} (\tau+1)^{-\frac{3}{2}-l}.
		\end{equation*}
		\end{proposition}
	\begin{proof}
	We will prove the $K=N=0$ case. The general $K$ and $L$ case follows analogously. Suppose first that $\widetilde{B}_-,\widetilde{B}_+\neq 0$. Then we can apply \eqref{eq:exprPi1}. By \eqref{eq:waveeqtildeX}, we have that
	\begin{multline}
	\label{eq:Philm}
(\rho^2\square_g\Pi)_{m \ell}=\widetilde{X}((r-1)^2 \widetilde{X}\Pi_{m \ell})-i m(r^2-1)\widetilde{X}\Pi_{m \ell}+(\lambda_{m \ell}+m^2)\Pi_{m \ell}-i m r \Pi_{m \ell}\\
+2(r^2+1-\hfol(r-1)^2)K\widetilde{X}\Pi_{m \ell}+\hfol ((r-1)^2 \hfol- 2(r^2+1))K^2\Pi_{m \ell}+im(2+\hfol(r^2-1))K\Pi_{m \ell}\\
+(\sin^2\theta(K^2\Pi-im K\Pi))_{m \ell}+ \left(2r-\frac{d}{dr}((r-1)^2 \hfol)\right)K\Pi_{m \ell}.
\end{multline}
We will also denote:
\begin{multline*}
\mathcal{L}^K_{m\ell}(\cdot):=\widetilde{X}((r-1)^2 \widetilde{X}(\cdot ))-i m(r^2-1)\widetilde{X}(\cdot )+(\lambda_{m \ell}+m^2)(\cdot )-i m r (\cdot )+2(r^2+1-\hfol(r-1)^2)K\widetilde{X}(\cdot )\\
+\hfol ((r-1)^2 \hfol- 2(r^2+1))K^2(\cdot )+im(2+\hfol(r^2-1))K(\cdot )+\left(2r-\frac{d}{dr}((r-1)^2 \hfol)\right)K(\cdot ),
\end{multline*}
so that
\begin{equation*}
(\rho^2\square_g\Pi)_{m \ell}=\mathcal{L}^K_{m\ell}(\Pi_{m \ell})+(\sin^2\theta(K^2\Pi-im K\Pi))_{m \ell}.
\end{equation*}

We split
\begin{multline*}
\Pi_{m \ell}(\tau,r):={\mathfrak{h}}_{m \ell}(w_{\infty, m\ell}(r)+E_+\widetilde{w}_{\infty, m\ell}(r)) f_+^{m \ell}(\tau,r)+{\mathfrak{h}}_{m \ell}(w_{\infty, m\ell}(r)+E_-\widetilde{w}_{\infty, m\ell}(r)) f_-^{m \ell}(\tau,r)\\
+{\mathfrak{h}}_{m \ell}\widetilde{w}_{\infty, m\ell}(r)\left[E_+(f_{+,{\rm exp}}^{m \ell}- f_+^{m \ell})(\tau,r)+E_-(f_{-,{\rm exp}}^{m \ell}- f_-^{m \ell})(\tau,r)\right]
\end{multline*}
Suppose $\widetilde{B}_-\neq 0$ and $\widetilde{B}_+\neq 0$. Then further rewrite:
\begin{align*}
{\mathfrak{h}}_{m \ell}(w_{\infty, m\ell}(r)+E_+\widetilde{w}_{\infty, m\ell}(r)) f_+^{m \ell}(\tau,r)=&\:{\mathfrak{h}}_{m \ell}(B_+-\widetilde{B}_-^{-1}B_-\widetilde{B}_+)w_{+, m\ell}(r)f_+^{m \ell}(\tau,r)\\
{\mathfrak{h}}_{m \ell}(w_{\infty, m\ell}(r)+E_-\widetilde{w}_{\infty, m\ell}(r)) f_-^{m \ell}(\tau,r)=&\:{\mathfrak{h}}_{m \ell}(B_--\widetilde{B}_+^{-1}B_+\widetilde{B}_-)w_{-, m\ell}(r)f_-^{m \ell}(\tau,r).
\end{align*}
We denote $b_{+,m \ell}:={\mathfrak{h}}_{m \ell}(B_+-\widetilde{B}_-^{-1}B_-\widetilde{B}_+)$ and $b_{-,\ell}:={\mathfrak{h}}_{m \ell}(B_--\widetilde{B}_+^{-1}B_+\widetilde{B}_-)w_{+, m\ell}(r)$, so that
\begin{align}
\label{eq:usefulidPi1}
\Pi_{m \ell}(\tau,r)=&\:b_{+,m \ell}w_{+, m\ell}(r)f_+^{m \ell}(\tau,r)+b_{-,m \ell}w_{+, m\ell}(r)f_-^{m \ell}(\tau,r)+\Pi_{m \ell,{\rm exp}}(\tau,r),\\
\label{eq:usefulidPi2}
\Pi_{m \ell,{\rm exp}}(\tau,r)=&\:{\mathfrak{h}}_{m \ell}\widetilde{w}_{\infty, m\ell}(r)\left[E_+(f_{+,{\rm exp}}^{m \ell}- f_+^{m \ell})(\tau,r)+E_-(f_{-,{\rm exp}}^{m \ell}- f_-^{m \ell})(\tau,r)\right].
\end{align}
Combining \eqref{eq:Philm} with \eqref{eq:eqwellm} implies:
\begin{multline*}
(\rho^2\square_g\Pi)_{m \ell}=\sum_{\square\in\{+,-\}}\overbrace{b_{\square ,m \ell}w_{m \ell, \square}\widetilde{X}((r-1)^2 \widetilde{X}f_{\square}^{m \ell})}^{=:J_1}\\
+\overbrace{b_{\square ,m \ell}\left(2(r-1)^2\frac{\frac{dw_{m \ell, \square}}{dr}}{w_{m \ell, \square}}-im (r^2-1)\right)w_{m \ell, \square}\tX f_{\square}^{m \ell}}^{=:J_2}\\
+\overbrace{2b_{\square ,m \ell}(r^2+1-\hfol(r-1)^2)\widetilde{X}(w_{m \ell, \square}Kf_{\square}^{m \ell})}^{=:J_3}\\
+ \overbrace{b_{\square ,m \ell}\left(im(2+\hfol(r^2-1))+2r-\frac{d}{dr}((r-1)^2 \hfol)\right)w_{m \ell, \square} Kf_{\square}^{m \ell}}^{=:J_4}\\
+\overbrace{b_{\square ,m \ell}w_{m \ell, \square}\hfol ((r-1)^2 \hfol- 2(r^2+1))K^2f_{\square}^{m \ell}}^{=:J_5}+\overbrace{\mathcal{L}^K_{m\ell}(\Pi_{m \ell,{\rm exp}})}^{=:J_6}+\overbrace{(\sin^2\theta(K^2\Pi-im K\Pi))_{m \ell}}^{=:J_7}
\end{multline*}
We will estimate the terms $J_i$ defined above. Without loss of generality, we will only consider the case $\square=+$, as the $\square=-$ estimates proceed entirely analogously. We will drop the $m \ell$ in the notation below.

\paragraph{\underline{$J_1$}:}
\begin{multline*}
J_1=4b_{+ ,m \ell}(1+O(r-1))\left(- 2\sqrt{\alpha}+2i m-3 \right) \left(- 2\sqrt{\alpha} +2i m-1\right)\\
\times \left((\tau+1)(r-1)+4\right)^{-\frac{5}{2}- \sqrt{\alpha}+im}(\tau+1)^{-\frac{1}{2}- \sqrt{\alpha}-im }\\
\end{multline*}

\paragraph{\underline{$J_2$}:}
\begin{equation*}
J_2=2b_{+ ,m \ell}(1+O(r-1))  (2 \sqrt{\alpha} +1) (-2  \sqrt{\alpha} +2 im-1)\left((\tau+1)(r-1)+4\right)^{-\frac{3}{2}-\sqrt{\alpha} +i m}  (\tau+1)^{-\frac{1}{2}-\sqrt{\alpha} -i m}
\end{equation*}

\paragraph{\underline{$J_3$}:}

\begin{equation*}
\begin{split}
J_3=&\:4b_{+ ,m \ell} (-4\alpha+2i m(4\sqrt{\alpha}+4) -8 \sqrt{\alpha}-3 +4 m^2 ) \left((\tau+1)(r-1)+4\right)^{-\frac{5}{2}-\sqrt{\alpha} +i m}(\tau+1)^{-\frac{1}{2}-im -\sqrt{\alpha}} \\
+&\:2 b_{+ ,m \ell}(2\sqrt{\alpha}+1)(2\sqrt{\alpha}+1 -2 i m ) \left(4\tau+1)(r-1)\right)^{-\frac{3}{2}-\sqrt{\alpha} +i m}(\tau+1)^{-\frac{1}{2}-im -\sqrt{\alpha}} \\
+&\: O(1)\left((\tau+1)(r-1)+4\right)^{-\frac{1}{2}} (\tau+1)^{-\frac{3}{2}}
\end{split}
\end{equation*}

It follows that the leading-order terms in $J_1$, $J_2$ and $J_3$ cancel when we sum them up. We obtain:
\begin{equation*}
J_1+J_2+J_3=O(1)\left((\tau+1)(r-1)+4\right)^{-\frac{1}{2}} (\tau+1)^{-\frac{3}{2}}.
\end{equation*}

\paragraph{\underline{$J_4+J_5$}:}
We can immediately conclude that
\begin{equation*}
J_4+J_5=O(1)\left((\tau+1)(r-1)+4\right)^{-\frac{1}{2}} (\tau+1)^{-\frac{3}{2}}.
\end{equation*}

\paragraph{\underline{$J_6$}:}
We have that:
\begin{multline*}
J_6=\sum_{\square\in\{+,-\}}{\mathfrak{h}}_{m \ell} E_{\square}\widetilde{w}_{m \ell, \infty}\widetilde{X}((r-1)^2 \widetilde{X}(f_{\square,{\rm exp}}^{m \ell}- f_{\square}^{m \ell}))\\
+{\mathfrak{h}}_{m \ell} E_{\square}\left(2(r-1)^2\frac{\widetilde{w}'_{m \ell, \infty}}{\widetilde{w}_{\infty, m\ell}}-im (r^2-1)\right)\widetilde{w}_{m \ell, \infty}\tX(f_{\square,{\rm exp}}^{m \ell}- f_{\square}^{m \ell})\\
+2{\mathfrak{h}}_{m \ell} E_{\square}(r^2+1-\hfol(r-1)^2)\widetilde{w}_{m \ell, \infty} \widetilde{X}(\widetilde{w}_{m \ell, \infty}K(f_{\square,{\rm exp}}^{m \ell}- f_{\square}^{m \ell}))\\
+ {\mathfrak{h}}_{m \ell} E_{\square}\left(im(2+\hfol(r^2-1))+2r-\frac{d}{dr}((r-1)^2 \hfol)\right)\widetilde{w}_{m \ell, \infty} K(f_{\square,{\rm exp}}^{m \ell}- f_{\square}^{m \ell})\\
+{\mathfrak{h}}_{m \ell} E_{\square}\widetilde{w}_{m \ell, \infty}\hfol ((r-1)^2 \hfol- 2(r^2+1))K^2(f_{\square,{\rm exp}}^{m \ell}- f_{\square}^{m \ell}).
\end{multline*}

In contrast with the analysis of $J_1$--$J_3$, we do need need to make use of cancellations to conclude that
\begin{equation*}
J_6=O(1)\left((\tau+1)(r-1)+4\right)^{-\frac{1}{2}} (\tau+1)^{-\frac{3}{2}}.
\end{equation*}
Instead, we apply (i) of Lemma \ref{lm:fexp}.

\paragraph{\underline{$J_7$}:}
Since $J_7$ features at least one $K$ derivative, we immediately get
\begin{equation*}
J_7=O(1)\left((\tau+1)(r-1)+4\right)^{-\frac{1}{2}} (\tau+1)^{-\frac{3}{2}}.
\end{equation*}

Now suppose $\widetilde{B}_+=0$. Then, $\widetilde{A}_+=0$ and ${A}_+\neq 0$, and by \eqref{eq:exprPi2}, we have that
\begin{equation*}
 \Pi_{m \ell}(\tau,r):={\mathfrak{h}}_{m \ell}w_{\infty, m\ell}(r)f_+^{m \ell}(\tau,r)={\mathfrak{h}}_{m \ell}A_{+}^{-1}w_{+, m\ell}(r)f_+^{m \ell}(\tau,r).
 \end{equation*}
 Hence, we can repeat the arguments above for estimating $J_1$--$J_5$ and $J_7$, with $b_{+,m \ell}$ replaced by ${\mathfrak{h}}_{m \ell}A_{+}^{-1}$.
 
 Suppose that $\widetilde{B}_-=0$. Then, $\widetilde{A}_-=0$ and ${A}_-\neq 0$, and by \eqref{eq:exprPi3}, we have that
\begin{equation*}
 \Pi_{m \ell}(\tau,r):={\mathfrak{h}}_{m \ell}w_{\infty, m\ell}(r)f_-^{m \ell}(\tau,r)={\mathfrak{h}}_{m \ell}A_{-}^{-1}w_{-, m\ell}(r)f_-^{m \ell}(\tau,r).
 \end{equation*}
In this case we can apply the argument above for estimating $J_1$--$J_5$ and $J_7$, with $b_{-,m \ell}$ replaced by ${\mathfrak{h}}_{m \ell}A_{-}^{-1}$.
	\end{proof}
	
	\begin{proposition}
	\label{prop:inhomagtr0}
		Let $K,L\in \N_0$. Let $0<\alpha_{m \ell}<\frac{1}{4}$ and $r\leq 2$. Then there exists a $C=C(K,N,m)>0$, such that for all $k\leq K$ and $l\leq L$
		\begin{equation*}
		|((r-1)\tX)^kK^l(\rho^2\square_g\Pi))_{m \ell}|(r)\leq C\left((\tau+1)(r-1)+4\right)^{-\frac{1}{2}-\sqrt{\alpha}_{m \ell}} (\tau+1)^{-\frac{1}{2}-\sqrt{\alpha}_{m \ell}}( (\tau+1)^{-1}+(r-1)^{2\sqrt{\alpha_{m \ell}}}).
		\end{equation*}
		\end{proposition}
	\begin{proof}
	By \eqref{eq:exprPi4}, we can split
	\begin{equation*}
 \Pi_{m \ell}(\tau,r)={\mathfrak{h}}_{m \ell}w_{\infty, m\ell}(r)f_+^{m \ell}(\tau,r)={\mathfrak{h}}_{m \ell}A_{+}w_{+, m\ell}(r)f_+^{m \ell}(\tau,r)+{\mathfrak{h}}_{m \ell}A_{-}w_{-, m\ell}(r)f_+^{m \ell}(\tau,r).
\end{equation*}
As in the proof of Proposition \ref{prop:squarePial0}, we write
\begin{multline*}
(\rho^2\square_g\Pi)_{m \ell}=\mathcal{L}^K_{m\ell}(\Pi_{m \ell})+(\sin^2\theta(K^2\Pi-im K\Pi))_{m \ell}={\mathfrak{h}}_{m \ell}A_{+}\mathcal{L}^K_{m\ell}(w_{+, m\ell}(r)f_+^{m \ell})\\
+{\mathfrak{h}}_{m \ell}A_{-}\mathcal{L}^K_{m\ell}(w_{-, m\ell}(r)f_+^{m \ell}).
\end{multline*}
The first term on the very right-hand side can be estimated by repeating the argument involving cancellations from the proof of Proposition \ref{prop:squarePial0} in order to obtain:
\begin{equation*}
|{\mathfrak{h}}_{m \ell}A_{+}\mathcal{L}^K_{m\ell}(w_{+, m\ell}(r)f_+^{m \ell})|=O(1)\left((\tau+1)(r-1)+4\right)^{-\frac{1}{2}-\sqrt{\alpha}} (\tau+1)^{-\frac{3}{2}-\sqrt{\alpha}}.
\end{equation*}
The term ${\mathfrak{h}}_{m \ell}A_{-}\mathcal{L}^K_{m\ell}(w_{-, m\ell}(r)f_+^{m \ell})$, however, will not feature the cancellations seen above. Instead, we use that
\begin{equation*}
\frac{w_{-, m\ell}}{w_{+, m\ell}}=O((r-1)^{2\sqrt{\alpha}})
\end{equation*}
to obtain
\begin{equation*}
{\mathfrak{h}}_{m \ell}A_{-}\mathcal{L}^K_{m\ell}(w_{-, m\ell}(r)f_+^{m \ell})=O((r-1)^{2\sqrt{\alpha}})\left((\tau+1)(r-1)+4\right)^{-\frac{1}{2}-\sqrt{\alpha}} (\tau+1)^{-\frac{1}{2}-\sqrt{\alpha}}.
\end{equation*}
	\end{proof}
	
	\begin{proposition}
	\label{prop:inhomestrgtr2}
	Let $K,L\in \N_0$ and $r\geq 2$. Then there exists a $C=C(K,N,m)>0$, such that for all $k\leq K$ and $l\leq L$
		\begin{equation*}
		|(K^l\rho^2\square_g\Pi)_{m \ell}|(r)\leq Cr^{-1}(1+\tau)^{-2-l}.
		\end{equation*}
		and
		\begin{equation*}
		|K^l(Xr^2)^lX\rho^2\square_g\Pi))_{m \ell}|(r)\leq Cr^{-1}(1+\tau)^{-2-l},
		\end{equation*}
		with $X=\tX-im\frac{\mathbbm{h}}{2}$.
	\end{proposition}
	\begin{proof}
	Split
\begin{align*}
\Pi_{m \ell}=&\:{\mathfrak{h}}_{m \ell}w_{\infty, m\ell}(r)(f_+^{m \ell}+f_-^{m \ell})+\Pi_{m \ell,{\rm exp}}(\tau,r),\\
\Pi_{m \ell,{\rm exp}}(\tau,r)=&\: {\mathfrak{h}}_{m \ell}\widetilde{w}_{\infty, m\ell}(r)\left[E_+f_{+,{\rm exp}}^{m \ell}(\tau,r)+E_-f_{-,{\rm exp}}^{m \ell}(\tau,r)\right].
\end{align*}
Then
\begin{multline*}
(\rho^2\square_g\Pi)_{m \ell}=\sum_{\square\in\{+,-\}}\overbrace{{\mathfrak{h}}_{m \ell}w_{m \ell, \infty}\widetilde{X}((r-1)^2 \widetilde{X}f_{\square}^{m \ell})}^{=:J_1}\\
+\overbrace{{\mathfrak{h}}_{m \ell}\left(2(r-1)^2\frac{\frac{dw_{m \ell, \infty}}{dr}}{w_{m \ell, \infty}}-im (r^2-1)\right)w_{m \ell, \infty}\tX f_{\square}^{m \ell}}^{=:J_2}\\
+\overbrace{2{\mathfrak{h}}_{m \ell}(r^2+1-\hfol(r-1)^2) w_{m \ell, \infty}\widetilde{X}Kf_{\square}^{m \ell}}^{=:J_3}+\overbrace{2{\mathfrak{h}}_{m \ell}(r^2+1-\hfol(r-1)^2) \frac{dw_{m \ell, \infty}}{dr}Kf_{\square}^{m \ell}}^{=:J_4}\\
+ \overbrace{{\mathfrak{h}}_{m \ell}\left(im(2+\hfol(r^2-1))+2r-\frac{d}{dr}((r-1)^2 \hfol)\right)w_{m \ell, \infty} Kf_{\square}^{m \ell}}^{=:J_5}\\
+\overbrace{{\mathfrak{h}}_{m \ell}w_{m \ell, \infty}\hfol ((r-1)^2 \hfol- 2(r^2+1))K^2f_{\square}^{m \ell}}^{=:J_6}+\overbrace{\mathcal{L}^K_{m\ell}(\Pi_{m \ell,{\rm exp}})}^{=:J_7}+\overbrace{(\sin^2\theta(K^2\Pi-im K\Pi))_{m \ell}}^{=:J_8}
\end{multline*}
Without loss of generality, we will consider the terms with $\square=+$ below.
\paragraph{\underline{$J_1$}:}
We have that
\begin{equation*}
J_1=O(r^{-3})(\tau+1)^{-2}.
\end{equation*}

\paragraph{\underline{$J_2$}:}
We have that
\begin{equation*}
J_2=O(r^{-1})(\tau+1)^{-2}.
\end{equation*}

\paragraph{\underline{$J_3$}:}
We have that
\begin{equation*}
r^2+1-\hfol(r-1)^2=-r^2+O(1)
\end{equation*}
by assumptions made on $\hfol$ in \S \ref{sec:foliations}. We therefore obtain
\begin{equation*}
J_3=O(r^{-1})(\tau+1)^{-3}.
\end{equation*}

\paragraph{\underline{$J_4$}:}
We have that
\begin{equation*}
\frac{\frac{dw_{m \ell, \infty}}{dr}}{w_{m \ell, \infty}}=\frac{r^2+1}{(r-1)^2}\frac{\frac{dw_{m \ell, \infty}}{dr*}}{w_{m \ell, \infty}}=(im -2im r^{-1}-2r+O(r^{-2}))
\end{equation*}
so that
\begin{equation*}
J_4=-2{\mathfrak{h}}_{m \ell}\left(im r^2-2im r-2r+O(1)\right)w_{m \ell, \infty}Kf_{+}^{m \ell}.
\end{equation*}

\paragraph{\underline{$J_5$}:}
We have that
\begin{equation*}
\hfol(r^2-1)=\hfol(r-1)^2\frac{r^2-1}{(r-1)^2}=-2r^2-4r+O(1).
\end{equation*}
and
\begin{equation*}
\frac{d}{dr}(\hfol(r-1)^2)=4r+O(r^{-1}).
\end{equation*}
Hence
\begin{equation*}
J_5={\mathfrak{h}}_{m \ell}\left(-2im r^2-4im r+2r+O(1)\right)w_{m \ell, \infty}Kf_{+}^{m \ell}.
\end{equation*}
We therefore have a cancellation in the leading-order terms in $J_4+J_5$ and we obtain
\begin{equation*}
J_4+J_5=O(r^{-1})(\tau+1)^{-2}.
\end{equation*}

\paragraph{\underline{$J_6$}:}
We have that
\begin{equation*}
J_5=O(r^{-1})(\tau+1)^{-3}.
\end{equation*}

\paragraph{\underline{$J_7$}:}
We do not need to make use of cancellations in the terms in $J_7$. Instead we apply (ii) of Lemma \ref{lm:fexp} to conclude that there exists a $\delta>0$ such that
\begin{equation*}
J_7=O(e^{-\delta r})(\tau+1)^{-2}
\end{equation*}

\paragraph{\underline{$J_8$}:}
We have that
\begin{equation*}
J_8=O(r^{-1})(\tau+1)^{-3}.
\end{equation*}

Since
\begin{equation*}
X^k(e^{im r_*}r^{-1})=O(r^{-2-k})
\end{equation*}
and
\begin{equation*}
((r-1)^2X^k) f_{+}^{m \ell}=O((1+\tau)^{-2}),
\end{equation*}
we conclude the estimate in the proposition.
\end{proof}

\bibliographystyle{alpha}

\end{document}